\documentclass[article]{IEEEtran}

\usepackage{array,graphicx,epsfig,amssymb,amsfonts,multirow}
\usepackage{amsmath}
\usepackage{epsfig,eufrak,eucal,subfigure,hhline,cite}
\usepackage{colortbl,color,nonfloat,twoopt,stmaryrd}
\usepackage[usenames,dvipsnames,svgnames,table]{xcolor}
\usepackage{graphics,cite}
\usepackage[geometry]{ifsym}
\usepackage{amsthm}

\newcommand{\appsection}[1]{\let\oldthesection\thesection\renewcommand{\thesection}{Appendix
\oldthesection}\section{#1}\let\thesection\oldthesection}

\newtheoremstyle{jaemin}
  {0pt}
  {0pt}
  {}
  {\parindent}
  {\itshape}
  {:}
  { }
  {#1 #2#3}
\theoremstyle{jaemin}

\newtheorem{prop}{Proposition}
\newtheorem{lem}{Lemma}
\newtheorem{cor}{Corollary}

\newcommand{\Ref}[1]{(\ref{#1})}

\newcommand{\PropRef}[1]{Proposition~\ref{#1}}
\newcommand{\PropsRef}[2]{Propositions~\ref{#1} and~\ref{#2}}
\newcommand{\PropToRef}[2]{Propositions~\ref{#1} to~\ref{#2}}
\newcommand{\LemRef}[1]{Lemma~\ref{#1}}

\newcommand{\LemToRef}[2]{Lemmas~\ref{#1} to~\ref{#2}}
\newcommand{\CorRef}[1]{Corollary~\ref{#1}}
\newcommand{\CorsRef}[2]{Corollaries~\ref{#1} and~\ref{#2}}

\newcommand{\SecRef}[1]{Section~\ref{#1}}

\newcommand{\SecToRef}[2]{Sections~\ref{#1} to~\ref{#2}}
\newcommand{\AppRef}[1]{Appendix~\ref{#1}}

\newcommand{\AppssssRef}[4]{Appendices~\ref{#1},~\ref{#2},~\ref{#3}, and~\ref{#4}}
 
\newcommand{\CAP}{\!\cap\!}                                               %
\newcommand{\CUP}{\!\cup\!}
\newcommand{\PREC}{\!\prec\!}                                             
\newcommand{\PRECEQ}{\!\preceq\!}
\newcommand{\SUCC}{\!\succ\!}                                             
\newcommand{\SUCCEQ}{\!\succeq\!}

\newcommand{\G}[1][]{G{#1}}                                              
\newcommand{\GV}{V}                                                      
\newcommand{\GE}{E}                                                      
\newcommand{\LG}[1][]{{\cal{G}{#1}}}                                     
\newcommand{\LGV}{{\cal{V}}}                                             
\newcommand{\LGE}{{\cal{E}}}                                             
\newcommand{\PATH}[3]{P_{#1#2}^{#3}}                                     
\newcommand{\PATHset}[3]{\bold{P}\!_{#1#2}^{#3}}                           
\newcommandtwoopt{\MONO}[2][][]{m_{\!#1}(#2)}                            
\newcommandtwoopt{\MONOL}[2][][]{\mathring{m}_{#1}(#2)}                  
\newcommandtwoopt{\EDGE}[2][][]{e_{#1}^{#2}}                             
\newcommandtwoopt{\NNNN}[2][][]{N_{#2#1}}                                
\newcommandtwoopt{\NNNNh}[2][][]{\hat{N}_{#2#1}}                         
\newcommandtwoopt{\NNNNt}[2][][]{\tilde{N}_{#2#1}}                       

\newcommandtwoopt{\mincut}[2][][]{{\small\textsf{mincut}}(#1;#2)}       
\newcommandtwoopt{\onecut}[2][][]{{\small\textsf{1cut}}(#1;#2)}         
\newcommandtwoopt{\EC}[2][][]{{\small\textsf{EC}}(#1;#2)}               
\newcommandtwoopt{\VC}[2][][]{{\small\textsf{VC}}(#1;#2)}               
\newcommand{\head}[1][]{{\small\textsf{head}}(#1)}                      
\newcommand{\tail}[1][]{{\small\textsf{tail}}(#1)}                      
\newcommand{\IN}[1][]{{\small\textsf{In}}(#1)}                          
\newcommand{\OUT}[1][]{{\small\textsf{Out}}(#1)}                        
\newcommandtwoopt{\GCD}[2][][]{{\small\textsf{GCD}}(\,#1,#2)}           
\newcommand{\RANK}[1][]{{\small\textsf{rank}}(#1)}                      
\newcommand{\DET}[1][]{{\small\textsf{det}}(#1)}                        
\newcommand{\VecSp}[1][]{\langle\,#1\,\rangle}                          
\newcommand{\ABS}[1][]{|#1|}                                            

\newcommand{\PolyEqual}{\!\!\equiv\!}
\newcommand{\PolyNotEqual}{\!\!\not\equiv\!}
\newcommand{\EqualEmpty}{\!\!=\!\emptyset}
\newcommand{\NotEqualEmpty}{\!\!\neq\!\emptyset}

\newcommand{\IMPLY}{\;$\Rightarrow$\;}
\newcommand{\OPPLY}{\;$\Leftarrow$\;}
\newcommand{\EQUIV}{\;$\Leftrightarrow$\;}
\newcommand{\AND}{\,$\wedge$\,}
\newcommand{\OR}{\,$\vee$\,}
\newcommand{\CONT}{\,false}

\newcommand{\GANA}[1][]{G_{\textrm{3ANA}}{#1}}                           

\newcommand{\TPid}{k}                                                    
\newcommand{\TPN}{K}                                                     
\newcommand{\Sid}{i}                                                     
\newcommand{\Did}{j}                                                     
\newcommand{\TSid}{t}                                                    
\newcommand{\TSN}{\tau}                                                  
\newcommand{\TSidBRAC}{\left(\!\TSid\!\right)}                           
\newcommand{\Nid}{n}                                                     
\newcommand{\NidBRAC}{\left(\!\Nid\!\right)}                             
\newcommand{\BOLDa}{L}                                                   
\newcommand{\BOLDb}{R}                                                   
\newcommand{\LRneq}{\BOLDa\,\PolyNotEqual\BOLDb}                           
\newcommand{\LReq}{\BOLDa\,\PolyEqual\BOLDb}                               
\newcommand{\LRneqVAR}{\BOLDa(\netvar)\,\PolyNotEqual\BOLDb(\netvar)}      
\newcommand{\LReqVAR}{\BOLDa(\netvar)\,\PolyEqual\BOLDb(\netvar)}          
\newcommand{\SMN}{N}                                                     

\newcommandtwoopt{\netvar}[2][][]{\bold{\underline{x}}^{#1}_{#2}}       
\newcommand{\netvarSUB}{\bold{\underline{x}}'}                          
\newcommandtwoopt{\netreal}[2][][]{\hat{\bold{x}}_{#1}^{#2}}            

\newcommandtwoopt{\Sover}[2][][]{\overline{S}_{#1#2}}                   %
\newcommandtwoopt{\Dover}[2][][]{\overline{D}_{#1#2}}                   %
\newcommandtwoopt{\GSP}[2][][]{\bold{h}_{#1}^{#2}}                      
\newcommandtwoopt{\GSPtilde}[2][][]{\bold{k}_{#1}^{#2}}                      

\newcommandtwoopt{\ChG}[2][][]{m_{#2;#1}}                                 
\newcommandtwoopt{\ChGsub}[2][][]{m'_{#2;#1}}                             
\newcommandtwoopt{\ChGANA}[2][][]{m_{#2#1}}                               
\newcommandtwoopt{\AllChG}[2][][]{\{\forall\,\ChGANA #1\textrm{ on } #2\}}
\newcommandtwoopt{\ChGL}[2][][]{\mathring{m}_{#2;#1}}                     
\newcommandtwoopt{\ChGLsub}[2][][]{\mathring{m}'_{#2;#1}}                 
\newcommandtwoopt{\ChGLT}[2][][]{\tilde{\mathring{m}}_{#2;#1}}            
\newcommand{\ChM}[3]{\bold{M}_{#2;#1}^{#3}}                               
\newcommand{\ZeroM}{\bold{0}}                                             
\newcommand{\EyeM}{\bold{I}}                                              
\newcommandtwoopt{\PreV}[2][][]{\bold{v}_{\!#1}^{\left(\!#2\!\right)}}    
\newcommandtwoopt{\PreM}[2][][]{\bold{V}_{\!#1}^{#2}}                     
\newcommandtwoopt{\PostV}[2][][]{\bold{u}_{\!#1}^{#2}}                    
\newcommandtwoopt{\PostM}[2][][]{\bold{U}_{\!#1}^{#2}}                    
\newcommandtwoopt{\InfV}[2][][]{\bold{z}_{#1}^{#2}}                       
\newcommandtwoopt{\RxV}[2][][]{\bold{y}_{\!#1}^{#2}}                      
\newcommandtwoopt{\DecV}[2][][]{\hat{\bold{z}}_{#1}^{#2}}                 
\newcommandtwoopt{\SupChM}[2][][]{\overline{\bold{M}}_{#2;#1}}            
\newcommand{\SupPreM}[1][]{\overline{\bold{V}}_{\!#1}}                    
\newcommand{\SupPostM}[1][]{\overline{\bold{U}}_{\!#1}}                   
\newcommand{\SupRxV}[1][]{\overline{\bold{y}}_{\!#1}}                     
\newcommandtwoopt{\FRM}[2][][]{\bold{S}_{#1}^{\left(\!#2\!\right)}}       
\newcommandtwoopt{\FF}[2][][]{\mathbb{F}_{#1}^{#2}}                       
\newcommandtwoopt{\RP}[2][][]{\mathbb{F}_{#1}[{#2}]}                      
\newcommandtwoopt{\NN}[2][][]{\mathbb{N}_{#1}^{#2}}                       
\newcommandtwoopt{\CC}[2][][]{\mathbb{C}_{#1}^{#2}}                       
\newcommandtwoopt{\gensize}[2][][]{l_{#1}^{#2}}                           

\newcommand{\ist}{i_{\text{st}}}
\newcommand{\iend}{i_{\text{end}}}
\newcommand{\jst}{j_{\text{st}}}
\newcommand{\jend}{j_{\text{end}}}

\newcommand{\HS}[1][]{{\bf H{#1}}}
\newcommand{\NotHS}[1][]{$\neg$\,\HS[#1]}
\newcommand{\HStilde}[1][]{{\bf K{#1}}}
\newcommand{\NotHStilde}[1][]{$\neg$\,\HStilde[#1]}
\newcommand{\LNR}{{\bf LNR}}
\newcommand{\NotLNR}{$\neg$\,\LNR}
\newcommand{\GS}[1][]{{\bf G{#1}}}
\newcommand{\NotGS}[1][]{$\neg$\,\GS[#1]}
\newcommand{\RS}[1][]{{\bf R{#1}}}

\renewcommand{\SS}[1][]{{\bf S{#1}}}

\newcommand{\NS}[1][]{{\bf N{#1}}}

\newcommand{\ES}[1][]{{\bf E{#1}}}
\newcommand{\NotES}[1][]{$\neg$\,{\bf E{#1}}}
\newcommand{\CS}[1][]{{\bf C{#1}}}
\newcommand{\NotCS}[1][]{$\neg$\,{\bf C{#1}}}
\newcommand{\DS}[1][]{{\bf D{#1}}}

\newcommand{\XS}[1][]{{\bf X{#1}}}
\newcommand{\NotXS}[1][]{$\neg$\,\XS[#1]}

\newcommand{\SWAPSD}{$(s,d\,)$}
\newcommandtwoopt{\FromTo}[2][][]{\makebox[1.1cm][c]{$s_{#1}$-to-$d_{#2}$}}

\newcommandtwoopt{\GHS}[2][][]{{\bf H$_{\mathbf{#1}}^{\mathbf{(\!#2\!)}}$}}
\newcommandtwoopt{\NotGHS}[2][][]{$\neg$\,\GHS[#1][#2]}
\newcommandtwoopt{\GES}[2][][]{{\bf E$_{\mathbf{#1}}^{\mathbf{(\!#2\!)}}$}}
\newcommandtwoopt{\NotGES}[2][][]{$\neg$\,\GES[#1][#2]}

\newcommand{\UPSTR}[1][]{{\small\textsf{upstr}}\!\left(#1\right)} 

\usepackage[normalem]{ulem}

\begin{document}
\IEEEoverridecommandlockouts
\title{Graph-Theoretic Characterization of The Feasibility of The Precoding-Based 3-Unicast Interference Alignment Scheme}

\author{ \parbox{4.3 in}{\centering Jaemin Han, Chih-Chun Wang \\
         Center of Wireless Systems and Applications (CWSA)\\
         School of Electrical and Computer Engineering, Purdue University\\
         {Email: \{han83,chihw\}@purdue.edu}}
         \hspace*{ 0.2 in}
         \parbox{2 in}{ \centering Ness B. Shroff \\
         Department of ECE and CSE\\
         The Ohio State University\\
         {Email: shroff@ece.osu.edu}}
}

\maketitle


\begin{abstract}
A new precoding-based intersession network coding (NC) scheme has
recently been proposed, which applies the interference alignment
technique, originally devised for wireless interference channels, to
the 3-unicast problem of directed acyclic networks. The main result of this work is a graph-theoretic characterization of the feasibility of the 3-unicast interference alignment scheme. To that end, we first investigate several key relationships between
the point-to-point network channel gains and the underlying graph structure. Such
relationships turn out to be critical when characterizing graph-theoretically the
feasibility of precoding-based solutions.

\end{abstract}

\begin{keywords}
Asymptotic interference alignment, interference channels, intersession network coding, 3-unicast networks.
\end{keywords}

\section{Introduction}\label{Sec1}

Deciding whether there exists a {\em network code}
\cite{AhlswedeCaiLiYeung:IT00} that satisfies the network traffic
demands has been a long-standing open challenge when there are multiple
coexisting source-destination pairs (sessions) in the network. For the
degenerate case in which there is only one multicast session in the
network, also termed the {\em single-multicast} setting, it is known
that linear network coding \cite{LiYeungCai:IT03} is capable of
achieving the information-theoretic capacity. Several papers have
since studied the network code construction problem for the single
multicast setting
\cite{KoetterMedard:TON03,SandersEgnerTolhuizen:ACM-SPAA03,ChouWuJain:Allerton03,HoMedardKoetterKargerEffrosShiLeong:IT06}.


On the other hand, when there are multiple coexisting sessions in the
network, the corresponding network code design/analysis problem, also
known as the {\em intersession network coding} (INC) problem, becomes
highly challenging due to the potential interference within the
network. For example, {\em linear network coding} no longer achieves
the capacity \cite{DoughertyFreilingZeger:IT05}. Deciding the
existence of a (linear) network code that satisfies general traffic
demands becomes an NP-hard problem
\cite{KoetterMedard:TON03,Lehman:PhD05}. Thus, recent INC studies have
focused on the optimal characterizations over some special networks or
under some restrictive rate constraints. The results along this
direction include the {\em index coding} problem
\cite{RouayhebCostas:ISIT08}, finding the capacity regions for {\em
  directed cycles} \cite{HarveyKleinbergLehman:IT06}, degree-2
three-layer {\em directed acyclic networks} (DAG)
\cite{YanYangZhang:IT06}, {\em node-constrained} line and star
networks \cite{YazdiKramer:IT11}, and the 1-hop broadcast {\em packet
  erasure channel} with feedback \cite{Wang:IT12a}, and for networks
with integer link capacity and two coexisting rate-1 multicast
sessions \cite{WangShroff:IT10}.

Recently, the authors in
\cite{DasVishwanathJafarMarkopoulou:ISIT10,RamakrishnanDaszMalekiMarkopoulouJafarVishwanath:Allerton10}
have applied interference alignment (IA), originally
developed for wireless interference channels \cite{CadambeJafar:IT08a},
to the scenario of 3 coexisting unicast sessions called 3-unicast
Asymptotic Network Alignment (ANA). The concept of interference
alignment leads to a new perspective on INC problems. Namely, the
network designer focuses on designing the {\em precoding} and {\em
  decoding mappings} at the source and destination node while allowing
randomly generated {\em local encoding kernels}
\cite{HoMedardKoetterKargerEffrosShiLeong:IT06} within the
network. Compared to the classic algebraic framework that fully
controls the encoder, decoder, and local encoding kernels
\cite{KoetterMedard:TON03}, this precoding-based framework trades off
the ultimate achievable throughput with a distributed,
implementation-friendly structure that exploits pure random linear NC
in the interior of the network. Their initial study on 3-unicast
networks shows that, under certain network topology and traffic
demand, the precoding-based NC can perform better than the pure
routing (non-coding) solution and a few widely-used simple linear NC
solutions. Such results strike a new balance between practicality
and throughput enhancement.

This work, motivated by its practical advantages over the classic network coding
framework, focuses exclusively on the precoding-based
framework and characterizes its corresponding properties. We then use
the newly developed results to analyze the 3-unicast ANA scheme
proposed in
\cite{DasVishwanathJafarMarkopoulou:ISIT10,RamakrishnanDaszMalekiMarkopoulouJafarVishwanath:Allerton10}. Specifically,
the existing results
\cite{DasVishwanathJafarMarkopoulou:ISIT10,RamakrishnanDaszMalekiMarkopoulouJafarVishwanath:Allerton10}
show that the 3-unicast ANA scheme achieves asymptotically half of the
interference-free throughput for each transmission pair when a set of
algebraic conditions on the {\em channel gains} of the networks are
satisfied. Note that for the case of {\em wireless interference channels}, these algebraic feasibility
conditions can be satisfied with close-to-one probability provided the
channel gains are independently and continuously distributed random
variables \cite{CadambeJafar:IT08a}. For comparison, the ``network
channel gains'' are usually highly correlated\footnote{The correlation
  depends heavily on the underlying network topology.} discrete random variables and thus the
algebraic channel gain conditions do not always hold with close-to-one
probability. Moreover, except for some very simple networks, checking
whether the algebraic channel gain conditions hold turns out to be
computationally prohibitive. As a result, we need new and efficient ways
to decide whether the network of interest admits a 3-unicast ANA
scheme that achieves half of the interference-free throughput. The
results in this work answer this question by developing new
graph-theoretic conditions that characterize the feasibility of the
3-unicast ANA scheme. The proposed graph-theoretic conditions can be easily computed and checked
within polynomial time.

The key contributions of this work are:

\begin{itemize}
\item We formulate the precoding-based framework and identify several fundamental properties (\PropToRef{Prop1}{Prop3}), which allow us to bridge the gap between the algebraic feasibility of the precoding-based NC problem and the underlying network topology.

\item Using these relationships, we characterize the graph-theoretic conditions for the feasibility of the 3-unicast ANA scheme.

\end{itemize}

The rest of the paper is organized as follows: \SecRef{Sec2} introduces some useful graph-theoretic definitions, and compares the algebraic formulation of the proposed precoding-based framework to that of the classic NC framework \cite{KoetterMedard:TON03}. In addition, the 3-unicast ANA scheme proposed by \cite{DasVishwanathJafarMarkopoulou:ISIT10,RamakrishnanDaszMalekiMarkopoulouJafarVishwanath:Allerton10} is introduced in the context of the precoding-based framework. \SecRef{Sec2} also discusses its algebraic feasibility conditions and the graph-theoretic conjectures proposed in \cite{RamakrishnanDaszMalekiMarkopoulouJafarVishwanath:Allerton10}. \SecRef{Sec3} identifies several key properties of the precoding-based framework and provides the corresponding proofs. Based on the new fundamental properties, in \SecRef{Sec4} we develop the graph-theoretic necessary and sufficient conditions for the feasibility of the 3-unicast ANA scheme. \SecRef{Sec5} concludes this work.

\section{Precoding-Based Intersession NC}\label{Sec2}
\subsection{System Model and Some Graph-Theoretic Definitions}\label{Sec2A}

Consider a DAG $\G\!=\!(\GV,\GE)$ where $\GV$ is the set of nodes and $\GE$ is the set of directed edges. Each edge $e\!\in\! \GE$ is represented by $e\!=\!uv$, where $u \!=\! \tail[e]$ and $v \!=\! \head[e]$ are the tail and head of $e$, respectively. For any node $v\!\in\! \GV$, we use $\IN[v]\!\subset\! \GE$ to denote the collection of its incoming edges $uv\!\in\! E$. Similarly, $\OUT[v]\!\subset\! E$ contains all the outgoing edges $vw\!\in\! \GE$.

A path $\PATH{}{}{}$ is a series of adjacent edges $e_1 e_2 \cdot\!\cdot\!\cdot e_{k}$ where $\head[e_i]\!=\!\tail[e_{i+1}]\,\forall\,i\!\in\{1,\cdot\!\cdot\!\cdot,k\!-\!1\}$. We say that $e_1$ and $e_k$ are the starting and ending edges of $P$, respectively. For any path $\PATH{}{}{}$, we use $e\!\in\!\PATH{}{}{}$ to indicate that an edge $e$ is used by $\PATH{}{}{}$. For a given path $\PATH{}{}{}$, $x\PATH{}{}{}y$ denotes the path segment of $\PATH{}{}{}$ from node $x$ to node $y$. A path starting from node $x$ and ending at node $y$ is sometimes denoted by $P_{xy}$. By slightly abusing the notation, we sometimes substitute the nodes $x$ and $y$ by the edges $e_1$ and $e_2$ and use $e_1Pe_2$ to denote the path segment from $\tail[e_1]$ to $\head[e_2]$ along $P$. Similarly, $P_{e_1e_2}$ denotes a path from $\tail[e_1]$ to $\head[e_2]$. We say a node $u$ is an {\em upstream} node of a node $v$ (or $v$ is a {\em downstream} node of $u$) if $u\!\neq\!v$ and there exists a path $\PATH{u}{v}{}$, and we denote it as $u\PREC v$. If neither $u\PREC v$ nor $u\SUCC v$, then we say that $u$ and $v$ are {\em not reachable} from each other. Similarly, $e_1$ is an upstream edge of $e_2$ if $\head[e_1]\PRECEQ\tail[e_2]$ (where $\PRECEQ$ means either $\head[e_1]\!\PREC\tail[e_2]$ or $\head[e_1]\!\!=\!\tail[e_2]$), and we denote it by $e_1\PREC e_2$. Two distinct edges $e_1$ and $e_2$ are not reachable from each other, if neither $e_1\PREC e_2$ nor $e_1\SUCC e_2$. Given any edge set $E_1$, we say an edge $e$ is one of the most upstream edges in $E_1$ if (i) $e\in E_1$; and (ii) $e$ is not reachable from any other edge $e'\in E_1\backslash e$. One can easily see that the most upstream edge may not be unique. The collection of the most upstream edges of $E_1$ is denoted by $\UPSTR[E_1]$. 
A {\em $k$-edge cut} (sometimes just the ``edge cut") separating node sets $U\!\subset\! \GV$ and $W\!\subset\!\GV$ is a collection of $k$ edges such that any path from any $u\!\in\!U$ to any $w\!\in\!W$ must use at least one of those $k$ edges. The value of an edge cut is the number of edges in the cut. (A $k$-edge cut has value $k$.) We denote the minimum value among all the edge cuts separating $U$ and $W$ as $\EC[U][W]$. By definition, we have $\EC[U][W]\!=\!0$ when $U$ and $W$ are already disconnected. By convention, if $U\cap W\NotEqualEmpty$, we define $\EC[U][W]\!=\!\infty$. We also denote the collection of all distinct $1$-edge cuts separating $U$ and $W$ as $\onecut[U][W]$. 

\subsection{The Algebraic Framework of Linear Network Coding}\label{Sec2B}

Given a network $\G\!=\!(\GV,\GE)$, we consider the multiple-unicast problem in which there are $\TPN$ coexisting source-destination pairs $(s_k,d_k)$, $k\!=\!1,\cdots\!,K$. Let $\gensize[\TPid]$ denote the number of information symbols that $s_{\TPid}$ wants to transmit to $d_{\TPid}$. Each information symbol is chosen independently and uniformly from a finite field $\FF[\!q]$ with some sufficiently large $q$. 

Following the widely-used instantaneous transmission model for DAGs \cite{KoetterMedard:TON03}, we assume that each edge is capable of transmitting one symbol in $\FF[\!q]$ in one time slot without delay. We consider {\em linear network coding} over the entire network, i.e., a symbol on an edge $e\!\in\! E$ is a linear combination of the symbols on its adjacent incoming edges $\IN[\tail[e]]$. The coefficients (also known as the network variables) used for such linear combinations are termed local encoding kernels. The collection of all local kernels $x_{e'\!e'\!'}\!\in\!\FF[\!q]$ for all adjacent edge pairs $(e',e'')$ is denoted by $\netvar\!=\!\{x_{e'\!e'\!'} : (e',e'')\!\in\!\GE^2\text{ where }\head[e']\!=\!\tail[e'']\}$. See \cite{KoetterMedard:TON03} for detailed discussion. Following this notation, the channel gain $\ChG[{e_2}][{e_1}](\netvar)$ from an edge $e_1$ to an edge $e_2$ can be written as a polynomial with respect to $\netvar$. More rigorously, $\ChG[{e_2}][{e_1}](\netvar)$ can be rewritten as
\begin{equation*}\label{Ch_DEF}
\ChG[{e_2}][{e_1}](\netvar) = \sum_{\forall \PATH{e_1}{\!e_2}{}\in\PATHset{e_1}{\!e_2}{}}\!\left(\prod_{\forall\,e'\!,\,e'\!'\in\PATH{e_1}{\!e_2}{} \,\text{where}\,\head[\!e']=\tail[\!e'\!'] } \!\!\!\!\!\!\!\!\!\!\!\!\!\!\!\!\!\!\!\!\!\!\!\!\!\!\!x_{e'\!e'\!'}\;\;\;\;\;\;\;\;\;\right) 
\end{equation*}
where $\PATHset{e_1}{\!e_2}{}$ denotes the collection of all distinct paths from $e_1$ to $e_2$.

By convention \cite{KoetterMedard:TON03}, we set $\ChG[e_2][e_1](\netvar)\!=\!1$ when $e_1\!=\!e_2$ and set $\ChG[e_2][e_1](\netvar)\!=\!0$ when $e_1\!\neq\!e_2$ and $e_2$ is not a downstream edge of $e_1$. The channel gain from a node $u$ to a node $v$ is defined by an $\ABS[{\IN[v]}]\!\times\!\ABS[{\OUT[u]}]$ polynomial matrix $\ChM{v}{u}{}(\netvar)$, where its $(i,j)$-th entry is the (edge-to-edge) channel gain from the $j$-th outgoing edge of $u$ to the $i$-th incoming edge of $v$. When considering source $s_i$ and destination $d_j$, we use $\ChM{\Did}{\Sid}{}(\netvar)$ as shorthand for $\ChM{d_j}{s_i}{}\!(\netvar)$.

We allow the precoding-based NC to code across $\TSN$ number of time slots, which are termed the precoding frame and $\TSN$ is the frame size. The network variables used in time slot $\TSid$ is denoted as $\netvar[\TSidBRAC]$, and the corresponding channel gain from $s_{\Sid}$ to $d_{\Did}$ becomes $\ChM{\Did}{\Sid}{}(\netvar[\TSidBRAC])$  for all $t=1,\cdots\!,\TSN$.

With these settings, 
let $\InfV[\Sid][]\!\in\!\FF[\!q][{\,\gensize[\Sid]\times 1}]$ be the set of to-be-sent information symbols from $s_i$. Then, for every time slot $\TSid\!=\!1,\cdots,\TSN$, we can define the precoding matrix $\PreM[\Sid][\TSidBRAC]\!\in\!\FF[\!q][{|\OUT[s_{\Sid}]|\times\gensize[\Sid]}]$ for each source $s_{\Sid}$. Given the precoding matrices, each $d_{\Did}$ receives an $|\IN[d_j]|$-dimensional column vector $\RxV[\Did][\TSidBRAC]$ at time $t$: \vspace{-0.015\columnwidth}
\begin{equation*}
\RxV[\Did][\TSidBRAC](\netvar[\TSidBRAC]) = \ChM{\Did}{\Did}{}(\netvar[\TSidBRAC]) \PreM[\Did][\TSidBRAC] \InfV[\Did] + \sum_{ \substack{ \Sid=1 \\ \Sid\neq\Did}}^{\TPN} \ChM{\Did}{\Sid}{}(\netvar[\TSidBRAC]) \PreM[\Sid][\TSidBRAC] \InfV[\Sid]. \vspace{-0.015\columnwidth}
\end{equation*}
where we use the input argument ``$(\netvar[\TSidBRAC])$" to emphasize that $\ChM{\Did}{\Did}{}$ and $\RxV[\Did][\TSidBRAC]$ are functions of the network variables $\netvar[\TSidBRAC]$.

This system model can be equivalently expressed as \vspace{-0.015\columnwidth}
\begin{equation}\label{pinc_sysmodel2} 
\SupRxV[\Did] = \SupChM[\Did][\Did] \SupPreM[\Did] \,\InfV[\Did]  + \sum_{\substack{ \Sid=1 \\ \Sid\neq\Did}}^{\TPN} \SupChM[\Did][\Sid] \SupPreM[\Sid] \,\InfV[\Sid], \vspace{-0.015\columnwidth}
\end{equation}
where $\SupPreM[\Sid]$ is the overall precoding matrix for each source $s_{\Sid}$ by vertically concatenating $\{\PreM[\Sid][\TSidBRAC]\}_{t=1}^{\TSN}$, and $\SupRxV[\Did]$ is the vertical concatenation of $\{\RxV[\Did][\TSidBRAC](\netvar[\TSidBRAC])\}_{t=1}^{\TSN}$. The overall channel matrix $\SupChM[\Did][\Sid]$ is a block-diagonal polynomial matrix with $\{\ChM{\Did}{\Sid}{}(\netvar[\TSidBRAC])\}_{t=1}^{\TSN}$ as its diagonal blocks. Note that $\SupChM[\Did][\Sid]$ is a polynomial matrix with respect to the network variables $\{\netvar[\TSidBRAC]\}_{t=1}^{\TSN}$.

After receiving packets for $\TSN$ time slots, each destination $d_j$ applies the overall decoding matrix $\SupPostM[\Did]\!\in\!\FF[\!q][{\,\gensize[\Did]\times(\TSN\cdot\ABS[{\IN[d_{\Did}]}\!])}]$. Then, the decoded message vector $\DecV[\Did]$ can be expressed as \vspace{-0.015\columnwidth} \begin{equation}\label{pinc_sysmodel3} 
\DecV[\Did] = \SupPostM[\Did]\SupRxV[\Did] = \SupPostM[\Did]\SupChM[\Did][\Did]\SupPreM[\Did] \,\InfV[\Did] + \sum_{\substack{ \Sid=1 \\ \Sid\neq\Did}}^{\TPN} \SupPostM[\Did] \SupChM[\Did][\Sid] \SupPreM[\Sid] \,\InfV[\Sid]. \vspace{-0.015\columnwidth}
\end{equation}

The combined effects of precoding, channel, and decoding from $s_i$ to $d_j$ is $\SupPostM[\Did] \SupChM[\Did][\Sid] \SupPreM[\Sid]$, which is termed the {\em network transfer matrix} from $s_i$ to $d_j$. We say that the precoding-based NC problem is feasible if there exists a pair of encoding and decoding matrices $\{\SupPreM[\Sid],\forall\,\Sid\}$ and $\{\SupPostM[\Did],\forall\,\Did\}$ (which may be a function of $\{\netvar[\TSidBRAC]\}_{t=1}^{\TSN}$) such that when choosing each element of the collection of network variables $\{\netvar[\TSidBRAC]\}_{t=1}^{\TSN}$ independently and uniformly randomly from $\FF[\!q]$, with high probability,
\begin{equation}\label{PF2}
\begin{split}
& \SupPostM[\Did] \SupChM[\Did][\Sid] \SupPreM[\Sid] = \EyeM\,\quad\text{(the identity matrix)} \quad\forall\,i=j,  \\
& \SupPostM[\Did] \SupChM[\Did][\Sid] \SupPreM[\Sid] = \ZeroM\quad\forall\,i\neq j.
\end{split}
\end{equation}

{\em Remark 1:} One can easily check by the cut-set bound that a necessary condition for the feasibility of a precoding-based NC problem is for the frame size $\TSN\!\geq\!\max_{\TPid}\{\gensize[\TPid] / \EC[s_{\TPid}][d_{\TPid}]\}$.

{\em Remark 2:} Depending on the time relationship of $\SupPreM[\Sid]$ and $\SupPostM[\Did]$ with respect to the network variables $\{\netvar[\TSidBRAC]\}_{t=1}^{\TSN}$, a precoding-based NC solution can be classified as causal vs.\ non-causal and  time-varying vs.\ time-invariant schemes.

For convenience to the reader, we have summarized in Table~\ref{tb:keydefs1} several key definitions used in the precoding-based framework.

\begin{table}[t]
\hfill{}
\begin{tabular}{|cl|}
\multicolumn{2}{c}{\small Notations for the precoding-based framework} \\
\hline
$K$                             & \parbox{6.4cm}{\footnotesize The number of coexisting unicast sessions } \\
$l_i$                           & \parbox{6.4cm}{\footnotesize The number of information symbols sent from $s_i$ to $d_i$ } \\
$\netvar$                       & \parbox{6.4cm}{\footnotesize The network variables / local encoding kernels} \\
\multirow{2}{*}{$\ChG[{e_2}][{e_1}](\netvar)$}  & \multirow{2}{*}{\parbox{6.4cm}{\footnotesize The channel gain from an edge $e_1$ to an edge $e_2$, which is a polynomial with respect to $\netvar$} } \\
& \\[-3pt]
\multirow{3}{*}{$\ChM{v}{u}{}(\netvar)$}       & \multirow{3}{*}{\parbox{6.4cm}{\footnotesize The channel gain matrix from a node $u$ to a node $v$ where its $(i,j)$-th entry is the channel gain from $j$-th outgoing edge of $u$ to $i$-th incoming edge of $v$} } \\
& \\
& \\[-3pt]
$\TSN$                          & \parbox{6.4cm}{\footnotesize The precoding frame size (number of time slot) } \\
$\netvar[\TSidBRAC]$            & \parbox{6.4cm}{\footnotesize The network variables corresponding to time slot $t$ } \\
$\PreM[\Sid][\TSidBRAC]$        & \parbox{6.4cm}{\footnotesize The precoding matrix for $s_i$ at time slot $t$ } \\
\multirow{2}{*}{$\ChM{j}{i}{}(\netvar[\TSidBRAC])$} & \multirow{2}{*}{\parbox{6.4cm}{\footnotesize The channel gain matrix from $s_i$ to $d_j$ at time slot $t$, shorthand for $\ChM{d_j}{s_i}{}(\netvar[\TSidBRAC])$} } \\
& \\[-3pt]
$\PostM[\Did][\TSidBRAC]$       & \parbox{6.4cm}{\footnotesize The decoding matrix for $d_j$ at time slot $t$ } \\
\multirow{2}{*}{$\SupPreM[\Sid]$} & \multirow{2}{*}{\parbox{6.4cm}{\footnotesize The overall precoding matrix for $s_i$ for the entire
precoding frame $t=1,\cdots,\tau$.} } \\
& \\[-3pt]
\multirow{2}{*}{$\SupChM[\Did][\Sid]$} & \multirow{2}{*}{\parbox{6.4cm}{\footnotesize The overall channel gain matrix from $s_i$ to $d_j$ for the entire precoding frame $t=1,\cdots,\tau$.} } \\
& \\[-3pt]
\multirow{2}{*}{$\SupPostM[\Did]$} & \multirow{2}{*}{\parbox{6.4cm}{\footnotesize The overall decoding matrix for $d_j$ for the entire precoding frame $t=1,\cdots,\tau$.} } \\
& \\[-3pt]
\hline

\multicolumn{2}{c}{\tiny } \\
\multicolumn{2}{c}{\small Notations for the 3-unicast ANA network } \\
\hline
$\ChGANA[j][i](\netvar)$        & \parbox{6.4cm}{\footnotesize The channel gain from $s_i$ to $d_j$} \\
\multirow{2}{*}{$\BOLDa(\netvar)$} & \multirow{2}{*}{\parbox{6.4cm}{\footnotesize The product of three channel gains: $\ChGANA[3][1](\netvar)\ChGANA[2][3](\netvar)$ $\ChGANA[1][2](\netvar)$} } \\
& \\[-3pt]
\multirow{2}{*}{$\BOLDb(\netvar)$} & \multirow{2}{*}{\parbox{6.4cm}{\footnotesize The product of three channel gains: $\ChGANA[2][1](\netvar)\ChGANA[3][2](\netvar)$ $\ChGANA[1][3](\netvar)$} } \\
& \\[-3pt]
\hline

\end{tabular}
\hfill{}
\caption{Key definitions of the precoding-based framework}
\label{tb:keydefs1}
\end{table}

\subsection{Comparison to the Existing Linear NC Framework}\label{Sec2C}

The authors in \cite{KoetterMedard:TON03} established the algebraic framework for linear network coding, which admits similar encoding and decoding equations as in \Ref{pinc_sysmodel2} and \Ref{pinc_sysmodel3} and the same algebraic feasibility conditions as in \Ref{PF2}. This original work focuses on a single time slot $\TSN\!=\!1$ while the corresponding results can be easily generalized for $\TSN\!>\!1$ as well. Note that $\TSN\!>\!1$ provides a greater degree of freedom when designing the coding matrices $\{\SupPreM[\Sid],\forall\,\Sid\}$ and $\{\SupPostM[\Did],\forall\,\Did\}$. Such {\em time extension} turns out to be especially critical in a precoding-based NC design as it is generally much harder (sometimes impossible) to design $\{\SupPreM[\Sid],\forall\,\Sid\}$ and $\{\SupPostM[\Did],\forall\,\Did\}$ when $\TSN\!=\!1$. An example of this time extension will be discussed in Section II-D.

The main difference between the precoding-based framework and the classic framework is that the latter allows the NC designer to control the network variables $\netvar$ while the former assumes that the entries of $\netvar$ are chosen independently and uniformly randomly. One can thus view the precoding-based NC as a distributed version of classic NC schemes that trades off the ultimate achievable performance for more practical distributed implementation (not controlling the behavior in the interior of the network).

One challenge when using algebraic feasibility conditions \Ref{PF2} is that given a network code, it is easy to verify whether or not \Ref{PF2} is satisfied, but it is difficult to decide whether there exists a NC solution satisfying \Ref{PF2}, see \cite{KoetterMedard:TON03,Lehman:PhD05}. Only in some special scenarios can we convert those algebraic conditions into some graph-theoretic conditions for which one can decide the existence of a feasible network code in polynomial time. For example, if there exists only a single session $(s_1,d_1)$ in the network, then the existence of a NC solution satisfying \Ref{PF2} is equivalent to the time-averaged rate $l_1/\TSN$ being no larger than $\EC[s_1][d_1]$. Moreover, if $\left(l_1/\TSN\right)\!\leq\!\EC[s_1][d_1]$, then we can use random linear network coding \cite{HoMedardKoetterKargerEffrosShiLeong:IT06} to construct the optimal network code. Another example is when there are only two sessions $(s_1,d_1)$ and $(s_2,d_2)$ with $l_1\!=\!l_2\!=\!\TSN\!=\!1$. Then, the existence of a network code satisfying \Ref{PF2} is equivalent to the conditions that the $1$-edge cuts in the network are properly placed in certain ways \cite{WangShroff:IT10}. Motivated by the above observation, the main focus of this work is to develop new graph-theoretic conditions for a special scenario of the precoding-based NC, the 3-unicast Asymptotic Network Alignment (ANA) scheme, which will be introduced in the next subsection.

\subsection{The 3-Unicast Asymptotic Network Alignment (ANA) Scheme}\label{Sec2D}

Before proceeding, we introduce some algebraic definitions. We say that a set of polynomials $\GSP(\netvar)\!=\!\{h_1(\netvar),...,h_{\SMN}(\netvar)\}$ is linearly dependent if and only if $\sum_{k=1}^{\SMN}\alpha_{k}h_k(\netvar)\!=\!0$ for some coefficients $\{\alpha_k\}_{k=1}^{\SMN}$ that are not all zeros. By treating $\GSP(\netvar[(\!k\!)])$ as a polynomial row vector and vertically concatenating them together, we have an $M\!\times\!\SMN$ polynomial matrix $[\GSP(\netvar[(\!k\!)])]_{k=1}^{M}$. We call this polynomial matrix a {\em row-invariant} matrix since each row is based on the same set of polynomials $\GSP(\netvar)$ but with different variables $\netvar[(\!k\!)]$ for each row $k$, respectively. We say that the row-invariant polynomial matrix $[\GSP(\netvar[(\!k\!)])]_{k=1}^{M}$ is generated from $\GSP(\netvar)$. For two polynomials $g(\netvar)$ and $h(\netvar)$, we say $g(\netvar)$ and $h(\netvar)$ are {\em equivalent}, denoted by $g(\netvar)\PolyEqual h(\netvar)$, if $g(\netvar)\!=c\cdot h(\netvar)$ for some non-zero $c\in\FF[\!q]$. If not, we say that $g(\netvar)$ and $h(\netvar)$ are {\em not equivalent}, denoted by $g(\netvar)\PolyNotEqual h(\netvar)$. We use $\GCD[g(\netvar)][h(\netvar)]$ to denote the greatest common factor of the two polynomials.

We now consider a special class of networks, called the 3-unicast ANA network: A network $\G$ is a 3-unicast ANA network if (i) there are 3 source-destination pairs, $(s_i,d_i),i\!=\!1,2,3$, where all source/destination nodes are distinct; (ii) $|\IN[s_i]|\!=\!0$ and $|\OUT[s_i]|\!=\!1$ $\forall\,i$ (We denote the only outgoing edge of $s_i$ as $e_{s_i}$, termed the $s_i$-source edge.); (iii) $|\IN[d_j]|\!=\!1$ and $|\OUT[d_j]|\!=\!0$ $\forall\,j$ (We denote the only incoming edge of $d_j$ as $e_{d_j}$, termed the $d_j$-destination edge.); and (iv) $d_j$ can be reached from $s_i$ for all $(i,j)$ pairs (including those with $i\!=\!j$).\footnote{The above {\em fully interfered} setting is the worst case scenario. For the scenario in which there is some $d_j$ who is not reachable from some $s_i$, one can devise an achievable solution by modifying the solution for the worst-case fully interfered 3-ANA networks \cite{DasVishwanathJafarMarkopoulou:ISIT10}.} We use the notation $\GANA$ to emphasize that we are focusing on this 3-unicast ANA network. Note that by (ii) and (iii) the matrix $\ChM{\Did}{\Sid}{}(\netvar)$ becomes a scalar, which we denote by $\ChGANA[\Did][\Sid](\netvar)$ instead.

The authors in \cite{DasVishwanathJafarMarkopoulou:ISIT10,RamakrishnanDaszMalekiMarkopoulouJafarVishwanath:Allerton10} applied interference alignment to construct the precoding matrices $\{\SupPreM[\Sid],\forall\,\Sid\}$ for the above 3-unicast ANA network. Namely, consider the following parameter values: $\TSN\!=\!2\Nid\!+\!1$, $l_1\!=\!\Nid+1$, $l_2\!=\!\Nid$, and $l_3\!=\!\Nid$ for some positive integer $\Nid$ termed symbol extension parameter, and assume that all the network variables $\netvar[\left(\!1\!\right)]$ to $\netvar[\left(\!\TSN\!\right)]$ are chosen independently and uniformly randomly from $\FF[\!q]$. The goal is to achieve the rate tuple $(\frac{\Nid+1}{2\Nid\!+\!1},\frac{\Nid}{2\Nid\!+\!1},\frac{\Nid}{2\Nid\!+\!1})$ in a 3-unicast ANA network by applying the following $\{\SupPreM[\Sid],\forall\,\Sid\}$ construction method: Define $\BOLDa(\netvar)=\ChGANA[3][1](\netvar)\ChGANA[2][3](\netvar)\ChGANA[1][2](\netvar)$ and $\BOLDb(\netvar)=\ChGANA[2][1](\netvar)\ChGANA[3][2](\netvar)\ChGANA[1][3](\netvar)$, and consider the following 3 row vectors of dimensions $n\!+\!1$, $n$, and $n$, respectively. (Each entry of these row vectors is a polynomial with respect to $\netvar$ but we drop the input argument $\netvar$ for simplicity.) 
\begin{align}
\PreV[1][\Nid]\!(\netvar) &=\ChGANA[3][2]\ChGANA[2][3]\left[\BOLDb^{\Nid},\,\,\,\BOLDb^{\Nid-1}\BOLDa,\,\,\cdots,\,\,\BOLDb\BOLDa^{\Nid-1},\,\,\,\BOLDa^{\Nid}\right], \label{v1}\\
\PreV[2][\Nid]\!(\netvar) &=\ChGANA[3][1]\ChGANA[2][3]\left[\BOLDb^{\Nid},\,\,\,\BOLDb^{\Nid-1}\BOLDa,\,\,\cdots,\,\,\BOLDb\BOLDa^{\Nid-1}\right],
\label{v2}\\
\PreV[3][\Nid]\!(\netvar) &=\ChGANA[2][1]\ChGANA[3][2]\left[\BOLDb^{\Nid-1}\BOLDa,\,\,\cdots,\,\,\BOLDb\BOLDa^{\Nid-1},\,\,\,\BOLDa^{\Nid}\right],
\label{v3} 
\end{align}
where the superscript ``$(\Nid)$" is to emphasize the value of the symbol extension parameter $n$ used in the construction. The precoding matrix for each time slot $t$ is designed to be $\PreM[i][\TSidBRAC]\!=\!\PreV[i][\Nid](\netvar[\TSidBRAC])$. The overall precoding matrix (the vertical concatenation of $\PreM[i][\left(\!1\right)]$ to $\PreM[i][\left(\!\TSN\!\right)]$) is thus $\SupPreM[i]\!=\![\PreV[i][\Nid]\!(\netvar[\TSidBRAC]\!)]_{t=1}^{2\Nid+\!1}$. 

The authors in \cite{DasVishwanathJafarMarkopoulou:ISIT10,RamakrishnanDaszMalekiMarkopoulouJafarVishwanath:Allerton10} prove that the above construction achieves the desired rates $(\frac{\Nid+1}{2\Nid\!+\!1},\frac{\Nid}{2\Nid\!+\!1},\frac{\Nid}{2\Nid\!+\!1})$ if the overall precoding matrices $\{\SupPreM[\Sid],\forall\,i\}$ satisfy the following six constraints:\footnote{Here the interference alignment is performed based on $(s_1,d_1)$-pair who achieves larger rate than others. Basically, any transmission pair can be chosen as an alignment-basis achieving $\frac{\Nid+1}{2\Nid+1}$, and the corresponding precoding matrices and six constraints can be constructed accordingly.} 
\begin{align}
d_1\!: & \, \VecSp[ {\SupChM[1][3]\SupPreM[3]} ] = \VecSp[ {\SupChM[1][2]\SupPreM[2]} ] \label{C1} \\
& \FRM[1][\Nid]\!\!\triangleq\!\left[\,\SupChM[1][1]\SupPreM[1] \,\,\,\, \SupChM[1][2]\SupPreM[2]\,\right]\!,\,\text{and } \RANK[{\FRM[1][\Nid]}]\!=\!2\Nid\!+\!1 \label{C2} \\
d_2\!: & \,\VecSp[ {\SupChM[2][3]\SupPreM[3]} ] \subseteq \VecSp[ {\SupChM[2][1]\SupPreM[1]} ] \label{C3} \\
& \FRM[2][\Nid]\!\!\triangleq\!\left[\,\SupChM[2][2]\SupPreM[2] \,\,\,\, \SupChM[2][1]\SupPreM[1]\,\right]\!,\,\text{and } \RANK[{\FRM[2][\Nid]}]\!=\!2\Nid\!+\!1 \label{C4} \\
d_3\!: & \,\VecSp[ {\SupChM[3][2]\SupPreM[2]} ] \subseteq \VecSp[ {\SupChM[3][1]\SupPreM[1]} ] \label{C5} \\
& \FRM[3][\Nid]\!\!\triangleq\!\left[\,\SupChM[3][3]\SupPreM[3] \,\,\,\, \SupChM[3][1]\SupPreM[1]\,\right]\!,\,\text{and } \RANK[{\FRM[3][\Nid]}]\!=\!2\Nid\!+\!1 \label{C6} \vspace{-0.03\columnwidth}
\end{align}
with close-to-one probability, where $\VecSp[{\bold{A}}]$ and $\RANK[{\bold{A}}]$ denote the column vector space and the rank, respectively, of a given matrix $\bold{A}$. The overall channel matrix $\SupChM[\Did][\Sid]$ is a $(2\Nid+1)\times(2\Nid+1)$ diagonal matrix with the $t$-th diagonal element $\ChGANA[\Did][\Sid](\netvar[\TSidBRAC])$ due to the assumption of $|\OUT[s_i]|\!=\!|\IN[d_j]|\!=\!1$. We also note that the construction in \Ref{C2}, \Ref{C4}, and \Ref{C6} ensures that the square matrices $\{\FRM[i][\Nid],\forall\,i\}$ are row-invariant.

The intuition behind \Ref{C1} to \Ref{C6} is straightforward. Whenever \Ref{C1} is satisfied, the interference from $s_2$ and from $s_3$ are aligned from the perspective of $d_1$. Further, by simple linear algebra we must have $\RANK[{\SupChM[1][2]\SupPreM[2]}]\!\leq\!n$ and $\RANK[{\SupChM[1][1]\SupPreM[1]}]\!\leq\!n\!+1$. \Ref{C2} thus guarantees that (i) the rank of $\left[\,\SupChM[1][1]\SupPreM[1] \,\,\,\, \SupChM[1][2]\SupPreM[2]\,\right]$ equals to $\RANK[{\SupChM[1][1]\SupPreM[1]}]+\RANK[{\SupChM[1][2]\SupPreM[2]}]$ and (ii) $\RANK[{\SupChM[1][1]\SupPreM[1]}]\!=\!n\!+\!1$. Jointly (i) and (ii) imply that $d_1$ can successfully remove the aligned interference while recovering all $l_1\!=\!n\!+\!1$ information symbols intended for $d_1$. Similar arguments can be used to justify \Ref{C3} to \Ref{C6} from the perspectives of $d_2$ and $d_3$, respectively.

By noticing the special Vandermonde form of $\SupPreM[i]$, it is shown in \cite{DasVishwanathJafarMarkopoulou:ISIT10,RamakrishnanDaszMalekiMarkopoulouJafarVishwanath:Allerton10} that \Ref{C1}, \Ref{C3}, and \Ref{C5} always hold. The authors in \cite{RamakrishnanDaszMalekiMarkopoulouJafarVishwanath:Allerton10} further prove that if
\begin{equation}\label{LRneqVAR}
\LRneqVAR
\end{equation}
and the following algebraic conditions are satisfied: \vspace{-0.02\columnwidth}
\begin{align}
& \ChGANA[1][1]\ChGANA[3][2]\sum_{i=0}^{\Nid}\alpha_i\big(\BOLDa/\BOLDb\big)^i \neq \ChGANA[1][2]\ChGANA[3][1]\sum_{j=0}^{\Nid-1}\beta_j\big(\BOLDa/\BOLDb\big)^j \label{SC1}
\end{align}\begin{align}
& \ChGANA[2][2]\ChGANA[3][1]\sum_{i=0}^{\Nid-1}\alpha_i\big(\BOLDa/\BOLDb\big)^i \neq \ChGANA[2][1]\ChGANA[3][2]\sum_{j=0}^{\Nid}\beta_j\big(\BOLDa/\BOLDb\big)^j \label{SC2} \\
& \ChGANA[3][3]\ChGANA[2][1]\sum_{i=1}^{\Nid}\alpha_i\big(\BOLDa/\BOLDb\big)^i \neq \ChGANA[3][1]\ChGANA[2][3]\sum_{j=0}^{\Nid}\beta_j\big(\BOLDa/\BOLDb\big)^j \label{SC3} 
\end{align}
for all $\alpha_i,\beta_j\!\in\!\FF[\!q]$ with at least one of $\alpha_i$ and at least one of $\beta_j$ being non-zero, then the constraints \Ref{C2}, \Ref{C4}, and \Ref{C6} hold with close-to-one probability (recalling that the network variables $\netvar[\left(\!1\right)]$ to $\netvar[\left(\!\TSN\!\right)]$ are chosen independently and uniformly randomly).

In summary, \cite{DasVishwanathJafarMarkopoulou:ISIT10,RamakrishnanDaszMalekiMarkopoulouJafarVishwanath:Allerton10} proves the following result.

{\em Proposition (page 3, \cite{RamakrishnanDaszMalekiMarkopoulouJafarVishwanath:Allerton10})}: For a sufficiently large finite field $\FF[\!q]$, the 3-unicast ANA scheme described in \Ref{v1} to \Ref{v3} achieves the rate tuple $(\frac{\Nid+1}{2\Nid\!+\!1},\frac{\Nid}{2\Nid\!+\!1},\frac{\Nid}{2\Nid\!+\!1})$ with close-to-one probability if \Ref{LRneqVAR}, \Ref{SC1}, \Ref{SC2}, and \Ref{SC3} hold simultaneously.

It can be easily seen that directly verifying the above sufficient conditions is computationally intractable. The following conjecture is thus proposed in \cite{RamakrishnanDaszMalekiMarkopoulouJafarVishwanath:Allerton10} to reduce the computational complexity when using the above proposition.

{\em Conjecture (Page 3, \cite{RamakrishnanDaszMalekiMarkopoulouJafarVishwanath:Allerton10}):} For any $n$ value used in the 3-unicast ANA scheme construction, if \Ref{LRneqVAR} and the following three conditions are satisfied simultaneously, then \Ref{SC1} to \Ref{SC3} must hold. \vspace{-0.015\columnwidth}
\begin{align}
& \ChGANA[1][1]\ChGANA[3][2]\not\equiv \ChGANA[1][2]\ChGANA[3][1] \,\,\text{  and  }\,\, \ChGANA[1][1]\ChGANA[2][3]\not\equiv \ChGANA[1][3]\ChGANA[2][1], \label{Cof1} \\
& \ChGANA[2][2]\ChGANA[3][1]\not\equiv \ChGANA[2][1]\ChGANA[3][2] \,\,\text{  and  }\,\, \ChGANA[2][2]\ChGANA[1][3]\not\equiv \ChGANA[2][3]\ChGANA[1][2], \label{Cof2} \\
& \ChGANA[3][3]\ChGANA[2][1]\not\equiv \ChGANA[3][1]\ChGANA[2][3] \,\,\text{  and  }\,\, \ChGANA[3][3]\ChGANA[1][2]\not\equiv \ChGANA[3][2]\ChGANA[1][3]. \label{Cof3} \vspace{-0.015\columnwidth}
\end{align}

Note that even if the conjecture is true, checking whether \Ref{LRneqVAR}, \Ref{Cof1}--\Ref{Cof3} are satisfied is still highly non-trivial for large networks. Moreover, recent results in \cite{RmakrishnanMeng:ISIT12} showed that the above conjecture is false. The main contribution of this work is to work on the original conditions \Ref{LRneqVAR}--\Ref{SC3} directly and provide an easily verifiable graph-theoretic characterization that supersedes their original algebraic forms.


{\em Remark:} In the setting of wireless interference channels, the individual channel gains are independently and continuously distributed, for which one can prove that the feasibility conditions \Ref{LRneqVAR}, \Ref{C2}, \Ref{C4}, and \Ref{C6} hold with probability one \cite{CadambeJafar:IT08a}. For a network setting, the channel gains $\ChG[j][i](\netvar)$ are no longer independently distributed for different $(i,j)$ pairs and the correlation depends on the underlying network topology. For example, one can verify that the 3-unicast ANA network described in Fig.~\ref{SecII-L=R-Fig} always leads to $\LReqVAR$ even when all network variables $\netvar$ are chosen uniformly randomly from an arbitrarily large finite field $\FF[\!q]$.

\begin{figure}[t]
\centering
\includegraphics[scale=0.17]{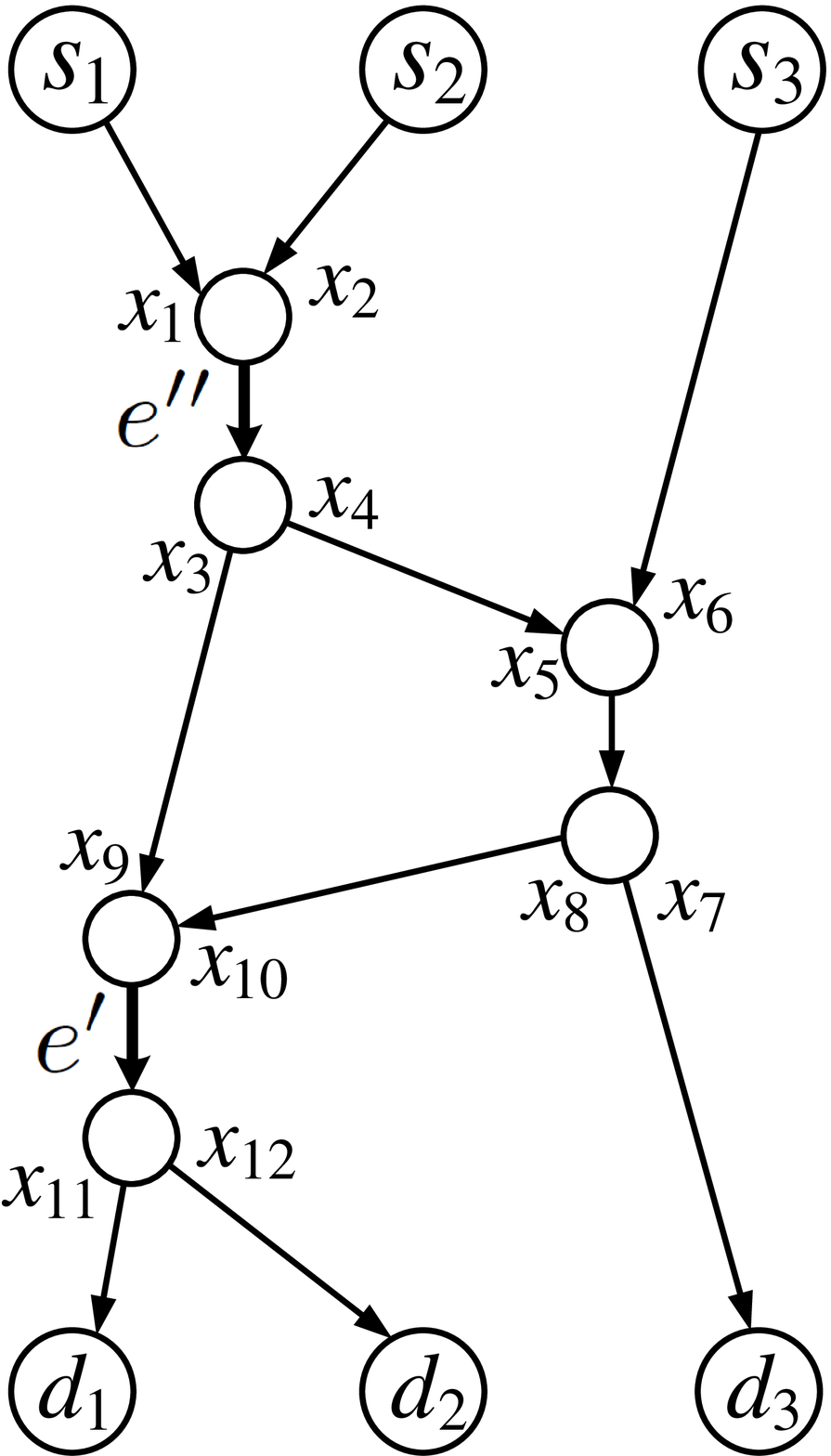} 
\caption{ Example $\GANA$ structure satisfying $\LReqVAR$ with $\netvar\!=\!\{x_1,x_2,...,x_{12}\}$. }
\label{SecII-L=R-Fig} \vspace{-0.05\columnwidth}
\end{figure}

\section{Properties of The Precoding-Based Framework}\label{Sec3}
In this section, we characterize a few fundamental relationships between the channel gains and the underlying DAG $\G$, which bridge the gap between the algebraic feasibility of the precoding-based NC problem and the underlying network structure. These properties hold for any precoding-based schemes and can be of benefit to future development of any precoding-based solution. These newly discovered results will later be used to prove the graph-theoretic characterizations of the 3-unicast ANA scheme. In \SecToRef{Sec3A}{Sec3C} we state \PropToRef{Prop1}{Prop3}, respectively. In \SecRef{Sec3D}, we discuss how these results can be applied to the existing 3-unicast ANA scheme.

\subsection{From Non-Zero Determinant to Linear Independence}\label{Sec3A}

\begin{prop}\label{Prop1}
Fix an arbitrary value of $N$. Consider any set of $N$ polynomials $\GSP(\netvar)\!=\!\{h_1(\netvar),...,h_{\SMN}(\netvar)\}$ and the polynomial matrix $[\GSP(\netvar[(\!k\!)])]_{k=1}^{\SMN}$ generated from $\GSP(\netvar)$. Then, assuming sufficiently large finite field size $q$, $\DET[{[\GSP(\netvar[(\!k\!)])]_{k=1}^{\SMN}}]$ is non-zero polynomial if and only if $\GSP(\netvar)$ is linearly independent. 
\end{prop}

The proof of \PropRef{Prop1} is relegated to \AppRef{Prop1-2Proof}.

{\em Remark:} Suppose a sufficiently large finite field $\FF[\!q]$ is used. If we choose the variables $\netvar[\left(\!1\!\right)]$ to $\netvar[\left(\!\SMN\!\right)]$ independently and uniformly randomly from $\FF[\!q]$, by Schwartz-Zippel lemma, we have $\DET[{[\GSP(\netvar[\left(\!k\!\right)])]_{k=1}^{\SMN}}]\!\neq\!0$ with close-to-one probability if and only if $\GSP(\netvar)$ is linearly independent.

The implication of Proposition 1 is as follows. Similar to the seminal work \cite{KoetterMedard:TON03}, most algebraic characterization of the precoding-based framework involves checking whether or not a determinant is non-zero. For example, the first feasibility condition of \Ref{PF2} is equivalent to checking whether or not the determinant of the network transfer matrix is non-zero. Also, \Ref{C2}, \Ref{C4}, and \Ref{C6} are equivalent to checking whether or not the determinant of the row-invariant matrix $\FRM[i][\Nid]$ is non-zero. \PropRef{Prop1} says that as long as we can formulate the corresponding matrix in a row-invariant form, then checking whether the determinant is non-zero is equivalent to checking whether the corresponding set of polynomials is linearly independent. As will be shown shortly after, the latter task admits more tractable analysis.

\subsection{The Subgraph Property of the Precoding-Based Framework}\label{Sec3B}

Consider a DAG $\G$ and recall the definition of the channel gain $\ChG[e_2][e_1](\netvar)$ from $e_1$ to $e_2$ in Section II-B. For a subgraph $\G[']\!\subseteq\!\G$ containing $e_1$ and $e_2$, let $\ChG[e_2][e_1](\netvarSUB)$ denote the channel gain from $e_1$ to $e_2$ in $\G[']$.

\begin{prop}[ (Subgraph Property)]\label{Prop2} Given a DAG $\G$, consider an arbitrary, but fixed, finite collection of edge pairs, $\{(e_i, e'_i)\!\in\!\GE^2: i\!\in\!I\}$ where $I$ is a finite index set, and consider two arbitrary polynomial functions $f:\FF[\!q][{|I|}]\!\mapsto\!\FF[\!q]$ and $g:\FF[\!q][{|I|}]\!\mapsto\!\FF[\!q]$. Then, $f(\{\ChG[e'_i][e_i](\netvar):\forall\,i\!\in\!I\})\PolyEqual g(\{\ChG[e'_i][e_i](\netvar):\forall\,i\!\in\!I\})$ if and only if for all subgraphs $\G['] \!\subseteq\! \G$ containing all edges in $\{e_i,e'_i:\forall\,i\!\in\!I\}$, $f(\{\ChG[e'_i][e_i](\netvarSUB):\forall\,i\!\in\!I\})\PolyEqual  g(\{\ChG[e'_i][e_i](\netvarSUB):\forall\,i\!\in\!I\})$.
\end{prop}

The proof of \PropRef{Prop2} is relegated to \AppRef{Prop1-2Proof}.

{\em Remark:} \PropRef{Prop2} has a similar flavor to the classic results \cite{KoetterMedard:TON03} and \cite{HoMedardKoetterKargerEffrosShiLeong:IT06}. More specifically, for the single multicast setting from a source $s$ to the destinations $\{d_{\Did}\}$, the transfer matrix $\PostM[d_j]\ChM{s}{d_j}{}(\netvar)\PreM[s]$ from $s$ to $d_j$ is of full rank (i.e., the polynomial $\DET[{\PostM[d_j]\ChM{s}{d_j}{}(\netvar)\PreM[s]}]$ is non-zero in the original graph $\G[]$) is equivalent to the existence of a subgraph $\G[']$ (usually being chosen as the subgraph induced by a set of edge-disjoint paths from $s$ to $d_{\Did}$) satisfying the polynomial $\DET[{\PostM[d_j]\ChM{s}{d_j}{}(\netvarSUB)\PreM[s]}]$ being non-zero.

Compared to \PropRef{Prop1}, \PropRef{Prop2} further connects the linear dependence of the polynomials to the subgraph properties of the underlying network. For example, to prove that a set of polynomials over a given arbitrary network is linearly independent, we only need to construct a (much smaller) subgraph and prove that the corresponding set of polynomials is linearly independent.

\subsection{The Channel Gain Property}\label{Sec3C}

Both \PropsRef{Prop1}{Prop2} have a similar flavor to the classic results of the LNC framework \cite{KoetterMedard:TON03}. The following channel gain property, on the other hand, is unique to the precoding-based framework. 

\begin{prop}[ (The Channel Gain Property)]\label{Prop3}
Consider a DAG $\G$ and two distinct edges $e_s$ and $e_d$. For notational simplicity, we denote $\head[e_s]$ by $s$ and denote $\tail[e_d]$ by $d$. Then, the following statements must hold (we drop the variables $\netvar$ for shorthand):
\begin{itemize}
\item \makebox[2.9cm][l]{If $\EC[s][d]\!=\!0$, then} $\ChG[e_d][e_s]\!=\!0$

\item \makebox[2.9cm][l]{If $\EC[s][d]\!=\!1$, then} $\ChG[e_d][e_s]$ is reducible. Moreover, let $N\!\!\stackrel{\Delta}{=}\!\!|\onecut[s][d]|$ denote the number of $1$-edge cuts separating $s$ and $d$, and we sort the $1$-edge cuts by their topological order with $e_1$ being the most upstream and $e_N$ being the most downstream. The channel gain $\ChG[e_d][e_s]$ can now be expressed as $\ChG[e_d][e_s]\!=\!\ChG[e_1][e_s]$$\left(\prod_{i=1}^{N-1}\ChG[e_{i+1}][e_i]\right)$$\ChG[e_d][e_N]$ and all the polynomial factors $\ChG[e_1][e_s]$, $\{\ChG[e_{i+1}][e_i]\}_{i=1}^{N-1}$, and $\ChG[e_d][e_N]$ are irreducible, and no two of them are equivalent.

\item \makebox[2.1cm][l]{If $\EC[s][d]\!\geq\! 2$} (including $\infty$), then $\ChG[e_d][e_s]$ is irreducible.
\end{itemize}\end{prop}

The proof of \PropRef{Prop3} is relegated to \AppRef{Prop3Proof}.

{\em Remark:} \PropRef{Prop3} only considers a channel gain between two distinct edges. If $e_s\!=\!e_d$, then by convention \cite{KoetterMedard:TON03}, we have $\ChG[e_d][e_s]\!=\!1$.

\PropRef{Prop3} relates the factoring problem of the channel gain polynomial to the graph-theoretic edge cut property. As will be shown afterwards, this observation enables us to tightly connect the algebraic and graph-theoretic conditions for the precoding-based solutions.

\subsection{Application of The Properties of The Precoding-based Framework to The 3-unicast ANA Schene}\label{Sec3D}

In this subsection, we discuss how the properties of the precoding-based framework, \PropToRef{Prop1}{Prop3}, can benefit our understanding of the 3-unicast ANA scheme.

\PropRef{Prop1} enables us to simplify the feasibility characterization of the 3-unicast ANA scheme in the following way. From the construction in \SecRef{Sec2D}, the square matrix $\FRM[i][\Nid]$ can be written as a row-invariant matrix $\FRM[i][\Nid]\!\!=
[\GSP[i][\left(\!\Nid\!\right)](\netvar[(\!t\!)])]_{t=1}^{(2\Nid\!+\!1)}$ for some set of polynomials $\GSP[i](\netvar)$. For example, by \Ref{v1}, \Ref{v2}, and \Ref{C2} we have $\FRM[1][\Nid]=[\GSP[1][\left(\!\Nid\!\right)](\netvar[(\!t\!)])]_{t=1}^{(2\Nid\!+\!1)}$ where 
\begin{equation}\label{GSP1^n}\begin{split}
\GSP[1][\NidBRAC]&(\netvar)= \{\,\ChGANA[1][1]\ChGANA[3][2]\ChGANA[2][3]\BOLDb^{\Nid},\;\ChGANA[1][1]\ChGANA[3][2]\ChGANA[2][3]\BOLDb^{\Nid-1}\BOLDa, \\
& \;\;\cdots\;,\;\ChGANA[1][1]\ChGANA[3][2]\ChGANA[2][3]\BOLDa^{\Nid},\;\ChGANA[1][2]\ChGANA[3][1]\ChGANA[2][3]\BOLDb^{\Nid},\; \\
& \;\;\ChGANA[1][2]\ChGANA[3][1]\ChGANA[2][3]\BOLDb^{\Nid-1}\BOLDa,\;\cdots\;,\;\ChGANA[1][2]\ChGANA[3][1]\ChGANA[2][3]\BOLDb\BOLDa^{\Nid-1}\,\}.
\end{split}\end{equation}

\PropRef{Prop1} implies that \Ref{C2} being true is equivalent to the set of polynomials $\GSP[1][\NidBRAC](\netvar)$ is linearly independent. Assuming the $\GANA$ of interest satisfies \Ref{LRneqVAR}, $\GSP[1][\NidBRAC](\netvar)$ being linearly independent is equivalent to \Ref{SC1} being true. As a result, \Ref{SC1} is not only sufficient but also necessary for \Ref{C2} to hold with close-to-one probability. By similar arguments \Ref{SC2} (resp. \Ref{SC3}) is both necessary and sufficient for \Ref{C4} (resp. \Ref{C6}) to hold with high probability.

\PropRef{Prop2} enables us to find the graph-theoretic equivalent counterparts of \Ref{SC1}--\,\Ref{SC3} of the {\em Conjecture} (p.~3, \cite{RamakrishnanDaszMalekiMarkopoulouJafarVishwanath:Allerton10}).

\begin{cor}[ (First stated in \cite{RamakrishnanDaszMalekiMarkopoulouJafarVishwanath:Allerton10})]\label{Cor1}
Consider a $\GANA$ and four indices $i_1$, $i_2$, $j_1$, and $j_2$ satisfying $i_1\!\neq\!i_2$ and $j_1\!\neq\!j_2$. We have $\EC[\{s_{i_1},s_{i_2}\}][\{d_{j_1},d_{j_2}\}]\!=\!1$ if and only if $\ChGANA[j_1][i_1]\ChGANA[j_2][i_2]$ $\PolyEqual\,\ChGANA[j_1][ i_2]\ChGANA[j_2][i_1]$.
\end{cor}

The main intuition behind \CorRef{Cor1} is as follows. When $\EC[\{s_{i_1},s_{i_2}\}][\{d_{j_1},d_{j_2}\}]\!=\!1$, one can show that we must have $\ChGANA[j_1][i_1](\netvar)\ChGANA[j_2][i_2](\netvar)\!=\!\ChGANA[j_1][ i_2](\netvar)\ChGANA[j_2][i_1](\netvar)$ by analyzing the underlying graph structure. When $\EC[\{s_{i_1},s_{i_2}\}][\{d_{j_1},d_{j_2}\}]\!\neq\!1$, we can construct a subgraph $\G[']$ satisfying $\ChGANA[j_1][i_1](\netvarSUB)\ChGANA[j_2][i_2](\netvarSUB)$ $\PolyNotEqual\ChGANA[j_1][ i_2](\netvarSUB)\ChGANA[j_2][i_1](\netvarSUB)$. \PropRef{Prop2} thus implies $\ChGANA[j_1][i_1](\netvar)$ $\ChGANA[j_2][i_2](\netvar)\,\PolyNotEqual\ChGANA[j_1][ i_2](\netvar)\ChGANA[j_2][i_1](\netvar)$. A detailed proof of \CorRef{Cor1} is relegated to \AppRef{Cor1-2Proof}.


\PropRef{Prop3} can be used to derive the following corollary, which studies the relationship of the channel polynomials $\ChGANA[j][i]$.

\begin{cor}\label{Cor2}
Given a $\GANA$, consider a source $s_i$ to destination $d_j$ channel gain $\ChGANA[j][i]$. Then, $\GCD[{\ChGANA[j_1][i_1]}][{\ChGANA[j_2][i_2]}]\,\PolyEqual\ChGANA[j_2][i_2]$ if and only if $(i_1,j_1)\!=\!(i_2,j_2)$. Intuitively, any channel gain $\ChGANA[j_1][i_1]$ from source $s_{i_1}$ to destination $d_{j_1}$ cannot contain another
source-destination channel gain $\ChGANA[j_2][i_2]$ as its factor.
\end{cor}

The intuition behind \CorRef{Cor2} is as follows. For example, suppose we actually have $\GCD[{\ChGANA[1][1]}][{\,\ChGANA[2][1]}]\,\PolyEqual\ChGANA[2][1]$ and assume that $\EC[{\head[e_{s_1}]}][{\tail[e_{d_2}]}]\!\geq\!2$. Then we must have the $d_2$-destination edge $e_{d_2}$ being an edge cut separating $s_1$ and $d_1$. The reason is that (i) \PropRef{Prop3} implies that any irreducible factor of the channel gain $\ChGANA[1][1]$ corresponds to the channel gain between two consecutive $1$-edge cuts separating $s_1$ and $d_1$; and (ii) The assumption $\EC[{\head[e_{s_1}]}][{\tail[e_{d_2}]}]\!\geq\!2$ implies that $\ChGANA[2][1]$ is irreducible. Thus (i), (ii), and $\GCD[{\ChGANA[1][1]}][{\ChGANA[2][1]}]\,\PolyEqual\,\ChGANA[2][1]$ together imply that $e_{d_2}\!\in\!\onecut[s_1][d_1]$. This, however, contradicts the assumption of $|\OUT[d_2]|\!=\!0$ for any 3-unicast ANA network $\GANA$. The detailed proof of \CorRef{Cor2}, which studies more general case in which $\EC[{\head[e_{s_1}]}][{\tail[e_{d_2}]}]\!=\!1$, is relegated to \AppRef{Cor1-2Proof}.

\section{Detailed Studies of The 3-unicast ANA Scheme}\label{Sec4}

In \SecRef{Sec3}, we investigated the basic relationships between the channel gain polynomials and the underlying DAG $\G$ for arbitrary precoding-based solutions. In this section, we turn our attention to a specific precoding-based solution, the 3-unicast ANA scheme, and characterize graph-theoretically its feasibility conditions.

\subsection{New Graph-Theoretic Notations and The Corresponding Properties}

We begin by defining some new notations. Consider three indices $i$, $j$, and $k$ in $\{1,2,3\}$ satisfying $j\!\neq\!k$ but $i$ may or may not be equal to $j$ (resp.\ $k$). Given a $\GANA$, define:
\begin{equation*}\begin{split}
\Sover[i][;\{j,k\}] & \triangleq \onecut[s_i][d_j]\cap\onecut[s_i][d_k]\backslash\{e_{s_i}\} \\
\Dover[i][;\{j,k\}] & \triangleq \onecut[s_j][d_i]\cap\onecut[s_k][d_i]\backslash\{e_{d_i}\}
\end{split}\end{equation*}
as the $1$-edge cuts separating $s_i$ and $\{d_j, d_k\}$ minus the $s_i$-source edge $e_{s_i}$ and the $1$-edge cuts separating $\{s_j,s_k\}$ and $d_i$ minus the $d_i$-destination edge $e_{d_i}$. 
When the values of indices are all distinct, we use $\Sover[i]$ (resp. $\Dover[i]$) as shorthand for $\Sover[i][;\{j,k\}]$ (resp. $\Dover[i][;\{j,k\}]$). The following lemmas prove some topological relationships between the edge sets $\Sover[i]$ and $\Dover[j]$ and the corresponding proofs are relegated to \AppRef{Lem1-7Proof}.

\begin{lem}\label{Lem1}
For all $i\!\neq\!j$, $e'\!\in\!\Sover[i]$, and $e''\!\in\!\Dover[j]$, one of the following three statements is true: $e'\PREC e''$, $e'\SUCC e''$, or $e'\!=\!e''$.
\end{lem} 

\begin{lem}\label{Lem2}
For any distinct $i$, $j$, and $k$ in $\{1,2,3\}$, we have $(\Dover[i]\cap\Dover[j])\!\subset\!\Sover[k]$. 
\end{lem}

\begin{lem}\label{Lem3}
For all $i\!\neq\!j$, $e'\!\in\!\Sover[i]\backslash\Dover[j]$, and $e''\!\in\!\Dover[j]$, we have $e'\PREC e''$.
\end{lem}

\begin{lem}\label{Lem4}
For any distinct $i$, $j$, and $k$ in $\{1,2,3\}$, $\Dover[j]\cap\Dover[k]\NotEqualEmpty$ if and only if both $\Sover[i]\cap\Dover[j]\NotEqualEmpty$ and $\Sover[i]\cap\Dover[k]\NotEqualEmpty$. 
\end{lem}

\begin{lem}\label{Lem5}
For all $i\neq j$ and $e''\!\in\!\Dover[i] \cap \Dover[j]$, if $\Sover[i] \cap \Sover[j] \NotEqualEmpty$, then there exists $e'\!\in\!\Sover[i] \cap \Sover[j]$ such that $e' \PRECEQ e''$.
\end{lem}

\begin{lem}\label{Lem6}
Consider four indices $i$, $j_1$, $j_2$, and $j_3$ taking values in $\{1,2,3\}$ for which the values of $j_1$, $j_2$ and $j_3$ must be distinct and $i$ is equal to one of $j_1$, $j_2$ and $j_3$. If
$\Sover[i][;\{j_1,j_2\}]\NotEqualEmpty$ and $\Sover[i][;\{j_1,j_3\}]\NotEqualEmpty$, then the following three statements are true: (i) $\Sover[i][;\{j_1,j_2\}]\cap\,\Sover[i][;\{j_1,j_3\}]\NotEqualEmpty$; (ii) $\Sover[i][;\{j_2,j_3\}]\NotEqualEmpty$; and (iii) $\Sover[i]\NotEqualEmpty$. 
\end{lem}

{\em Remark:} All the above lemmas are purely graph-theoretic. If we swap the roles of sources and destinations, then we can also derive the {\em \SWAPSD-symmetric version} of these lemmas. For example, the \SWAPSD-symmetric version of \LemRef{Lem2} becomes $(\Sover[i]\cap\Sover[j])\!\subseteq\!\Dover[k]$. The \SWAPSD-symmetric version of \LemRef{Lem5} is: For all $i\!\neq\!j$ and $e''\!\in\!\Sover[i] \cap \Sover[j]$, if $\Dover[i] \cap \Dover[j] \NotEqualEmpty$, then there exists $e'\!\in\!\Dover[i] \cap \Dover[j]$ such that $e' \SUCCEQ e''$.

\LemToRef{Lem1}{Lem6} discuss the topological relationship between the edge sets $\Sover[i]$ and $\Dover[j]$. The following lemma establishes the relationship between $\Sover[i]$ (resp. $\Dover[j]$) and the channel gains.

\begin{lem}\label{Lem7}
Given a $\GANA$, consider the corresponding channel gains as defined in Section II-D. Consider three indices $i$, $j_1$, and $j_2$ taking values in $\{1,2,3\}$ for which the values of $j_1$ and $j_2$ must be distinct. Then, $\GCD[{\ChGANA[j_1][i]}][{\ChGANA[j_2][i]}]\,\PolyEqual 1$ if and only if $\Sover[i][;\{j_1,j_2\}]\EqualEmpty$. Symmetrically, $\GCD[{\ChGANA[i][j_1]}][{\ChGANA[i][j_2]}]\,\PolyEqual 1$ if and only if $\Dover[i][;\{j_1,j_2\}]\EqualEmpty$.
\end{lem}

The proof of \LemRef{Lem7} is relegated to \AppRef{Lem1-7Proof}.

\subsection{The Graph-Theoretic Characterization of $\LRneqVAR$}

A critical condition of the 3-unicast ANA scheme \cite{DasVishwanathJafarMarkopoulou:ISIT10,RamakrishnanDaszMalekiMarkopoulouJafarVishwanath:Allerton10} is the assumption that $\LRneqVAR$, which is the fundamental reason why the Vandermonde precoding matrix $\SupPreM[i]$ is of full (column) rank. However, for some networks we may have $\LReqVAR$, for which the 3-unicast ANA scheme does not work (see Fig.~\ref{SecII-L=R-Fig}). Next, we prove the following graph-theoretic condition that fully characterizes whether $\LReqVAR$.

\begin{prop}\label{Prop4}
For a given $\GANA$, we have $\LReqVAR$ if and only if there exists a pair of distinct indices $i,j\!\in\!\{1,2,3\}$ satisfying both $\Sover[i] \cap \Sover[j] \NotEqualEmpty$ and $\Dover[i] \cap \Dover[j] \NotEqualEmpty$.
\end{prop}

\begin{proof}[Proof of the ``$\Leftarrow$" direction] Without loss of generality, suppose $\Sover[1] \cap \Sover[2] \NotEqualEmpty$ and $\Dover[1] \cap \Dover[2] \NotEqualEmpty$ (i.e., $i\!=\!1$ and $j\!=\!2$). By \LemRef{Lem5}, we can find two edges $e'\!\in\!\Sover[1] \cap \Sover[2]$ and $e''\!\in\!\Dover[1] \cap \Dover[2]$ such that $e'\PRECEQ e''$. Also note that \LemRef{Lem2} and its \SWAPSD-symmetric version imply that $e'\!\in\!\Dover[3]$ and $e''\!\in\!\Sover[3]$. Then by \PropRef{Prop3}, the channel gains $\ChGANA[j][i](\netvar)$ for all $i\!\neq\!j$ can be expressed by (we omit the variables $\netvar$ for simplicity):
\begin{equation*}\begin{split}
& \ChGANA[3][1]\!=\! \ChG[e'][{e_{s_1}}]\!\;\ChG[{e_{d_3}}][e']\! \quad\quad\quad\;\,\,\,\,\,\,  \ChGANA[2][1]\!=\!\ChG[e'][{e_{s_1}}]\!\;\ChG[e''][e']\!\;\ChG[{e_{d_2}}][e'']\! \\
& \ChGANA[2][3]\!=\! \ChG[e''][{e_{s_3}}]\!\;\ChG[{e_{d_2}}][e'']\! \quad\quad\quad\,\,\,\,\,   \ChGANA[3][2]\!=\! \ChG[e'][{e_{s_2}}]\!\;\ChG[{e_{d_3}}][e']\! \\
& \ChGANA[1][2]\!=\! \ChG[e'][{e_{s_2}}]\!\;\ChG[e''][e']\!\;\ChG[{e_{d_1}}][e'']\! \quad\,\, \ChGANA[1][3]\!=\! \ChG[e''][{e_{s_3}}]\!\;\ChG[{e_{d_1}}][e'']\! \\
\end{split}\end{equation*}
where the expressions of $\ChGANA[2][1]$ and $\ChGANA[1][2]$ are derived based on the facts that $e'\PRECEQ e''$ and $\{e',e''\}\!\subset\!\onecut[s_1][d_2]\cap \onecut[s_2][d_1]$. By plugging in the above 6 equalities to the definitions of $\BOLDa=\ChGANA[3][1]\ChGANA[2][3]\ChGANA[1][2]$ and $\BOLDb=\ChGANA[2][1]\ChGANA[3][2]\ChGANA[1][3]$, we can easily verify that $\LReq$. The proof of this direction is complete.
\end{proof}

{\em Remark:} In the example of Fig. \ref{SecII-L=R-Fig}, one can easily see that $e'\!\in\!\Sover[1] \cap \Sover[2]$ and $e''\!\in\!\Dover[1] \cap \Dover[2]$. Hence, the above proof shows that the example network in Fig. \ref{SecII-L=R-Fig} satisfies $\LReqVAR$ without actually computing the polynomials $\BOLDa(\netvar)$ and $\BOLDb(\netvar)$.

We will now focus on proving the necessity. Before proceeding, we state and prove the following lemma.

\begin{lem}\label{Lem8}
If the $\GANA$ of interest satisfies $\LReqVAR$, then $\Sover[i]\NotEqualEmpty$ and $\Dover[j]\NotEqualEmpty$ for all $i$ and $j$, respectively.
\end{lem}

\begin{proof} We prove this by contradiction. Suppose $\Sover[1]\EqualEmpty$. Denote the most upstream $1$-edge cut separating $\head[e_{s_1}\!]$ and $d_2$ by $e_{12}$ (we have at least the $d_2$-destination edge $e_{d_2}$). Also denote the most upstream $1$-edge cut separating $\head[e_{s_1}\!]$ and $d_3$ by $e_{13}$ (we have at least the $d_3$-destination edge $e_{d_3}$). Since $\Sover[1]\EqualEmpty$ and by the definition of the 3-unicast ANA network, it is obvious that $e_{12}\neq e_{13}$. Moreover, both of the two polynomials $\ChG[{e_{12}}][{e_{s_1}}]$ (a factor of $\ChGANA[2][1]$) and $\ChG[{e_{13}}][{e_{s_1}}]$ (a factor of $\ChGANA[3][1]$) are irreducible and non-equivalent to each other. Therefore, these two polynomials are coprime. If we plug in the two polynomials into $\LReqVAR$, then it means that one of the following three cases must be true: (i) $\ChG[{e_{d_3}}][{e_{13}}]$ contains $\ChG[{e_{12}}][{e_{s_1}}]$ as a factor; (ii) $\ChGANA[2][3]$ contains $\ChG[{e_{12}}][{e_{s_1}}]$ as a factor; or (iii) $\ChGANA[1][2]$ contains $\ChG[{e_{12}}][{e_{s_1}}]$ as a factor. However, (i), (ii), and (iii) cannot be true as $|\IN[s_1]|\!=\!0$ and by \PropRef{Prop3}. The proof is thus complete by applying symmetry.
\end{proof}

\begin{proof}[Proof of the ``$\Rightarrow$" direction of \PropRef{Prop4}] Suppose the $\GANA$ of interest satisfies $\LReqVAR$. By \LemRef{Lem8}, we know that $\Sover[i]\NotEqualEmpty$ and $\Dover[j]\NotEqualEmpty$ for all $i$ and $j$. Then it is obvious that $\EC[{\head[{e_{s_i}}]}][{\tail[{e_{d_j}}]}]\!=\!1$ for all $i\!\neq\!j$ because if (for example) $\EC[{\head[{e_{s_1}}]}][{\tail[{e_{d_2}}]}]\!\geq\!2$ then both $\Sover[1]$ and $\Dover[2]$ will be empty by definition. Thus by \PropRef{Prop3}, we can express each channel gain $\ChGANA[j][i]$ ($i\neq j$) as a product of irreducibles, each corresponding to the channel gain between two consecutive $1$-edge cuts (including $e_{s_i}$ and $e_{d_j}$) separating $s_i$ and $d_j$. We now consider two cases.



{\bf Case~1:} $\Sover[i]\cap\Dover[j]\EqualEmpty$ for some $i\!\neq\!j$. Assume without loss of generality that $\Sover[2]\cap\Dover[1]\EqualEmpty$ (i.e., $i\!=\!2$ and $j\!=\!1$). Let $e_2^\ast$ denote the most downstream edge in $\Sover[2]$ and let $e_1^\ast$ denote the most upstream edge in $\Dover[1]$. Since $\Sover[2]\cap\Dover[1]\EqualEmpty$, the edge $e_2^\ast$ must not be in $\Dover[1]$. By \LemRef{Lem3}, we have $e_2^{\ast} \PREC  e_1^{\ast}$.

For the following, we will prove $\{e_2^\ast, e_1^\ast\}\!\subset\!\onecut[s_1][d_2]$.  We first notice that by definition, $e_2^\ast\!\in\!\Sover[2]\!\subset\! \onecut[s_2][d_1]$ and $e_1^\ast\!\in$ $\!\Dover[1]\!\subset\!\onecut[s_2][d_1]$. Hence by \PropRef{Prop3}, we can express $\ChGANA[1][2]$ as $\ChGANA[1][2]=\ChG[e^\ast_2][e_{s_2}]\ChG[e^\ast_1][e^\ast_2]\ChG[e_{d_1}][e^\ast_1]$. Note that by our construction $e_2^\ast\PREC e_1^\ast$ we have $\ChG[e_1^\ast][e_2^\ast]\,\PolyNotEqual 1$.

We now claim $\GCD[{\ChG[e_1^\ast][e_2^\ast]}][{\,\ChGANA[3][2]\ChGANA[1][3]}]\PolyEqual 1$, i.e., $\ChGANA[3][2]\ChGANA[1][3]$ cannot contain any factor of $\ChG[e_1^\ast][e_2^\ast]$. We will prove this claim by contradiction. Suppose $\GCD[{\ChG[e_1^\ast][e_2^\ast]}][{\,\ChGANA[3][2]}]\PolyNotEqual 1$, i.e., $\ChGANA[3][2]$ contains an irreducible factor of $\ChG[e_1^\ast][e_2^\ast]$. Since that factor is also a factor of $\ChGANA[1][2]$, by \PropRef{Prop3}, there must exist at least one edge $e$ satisfying (i) $e^\ast_2\PREC e \PRECEQ e^\ast_1$; and (ii) $e\!\in\!\onecut[s_2][d_1]\cap\onecut[s_2][d_3]$. These jointly implies that we have an $\Sover[2]$ edge in the downstream of $e^\ast_2$. This, however, contradicts the assumption that $e^\ast_2$ is the most downstream edge of $\Sover[2]$. By a symmetric argument, we can also show that $\ChGANA[1][3]$ must not contain any irreducible factor of $\ChG[e_1^\ast][e_2^\ast]$. The proof of the claim $\GCD[{\ChG[e_1^\ast][e_2^\ast]}][{\,\ChGANA[3][2]\ChGANA[1][3]}]\,\PolyEqual 1$ is complete. Since the assumption $\LReqVAR$ implies that $\GCD[{\ChG[e_1^\ast][e_2^\ast]}][{\,\BOLDb}]=\ChG[e_1^\ast][e_2^\ast]$, we must have $\GCD[{\ChG[e_1^\ast][e_2^\ast]}][{\,\ChGANA[2][1]}]\!=\!\ChG[e_1^\ast][e_2^\ast]$. This implies by \PropRef{Prop3} that $\{e^\ast_2,e^\ast_1\}\!\subset\!\onecut[s_1][d_2]$.

For the following, we will prove that $e_2^\ast\!\in\!\onecut[s_1][d_3]$. To that end, we consider the factor $\ChG[e_{d_3}][e^\ast_2]$ of the channel gain $\ChGANA[3][2]$. This is possible by \PropRef{Prop3} because $e^\ast_2\!\in\!\Sover[2]\!\subset\!\onecut[s_2][d_3]$. Then similarly following the above discussion, we must have $\GCD[{\ChGANA[1][2]}][\,{\ChG[e_{d_3}][e^\ast_2]}]\,\PolyEqual 1$ otherwise there will be an $\Sover[2]$ edge in the downstream of $e^\ast_2$. Since the assumption $\LReqVAR$ means that $\GCD[{\BOLDa}][{\,\ChG[e_{d_3}][e_2^\ast]}]\!=\!\ChG[e_{d_3}][e_2^\ast]$, this further implies that $\GCD[{\ChGANA[3][1]\ChGANA[2][3]}][{\,\ChG[e_{d_3}][e_2^\ast]}]=\ChG[e_{d_3}][e_2^\ast]$.

Now consider the most upstream $\onecut[s_2][d_3]$ edge that is in the downstream of $e^\ast_2$, and denote it as $e_u$ (we have at least the $d_3$-destination edge $e_{d_3}$). Obviously, $e^\ast_2\PREC e_u \PRECEQ e_{d_3}$ and $\ChG[e_u][e^\ast_2]$ is an irreducible factor of $\ChG[e_{d_3}][e_2^\ast]$. Then we must have $\GCD[{\ChGANA[2][3]}][\,{\ChG[e_u][e^\ast_2]}]\PolyEqual 1$ and the reason is as follows. If not, then by $\ChG[e_u][e^\ast_2]$ being irreducible we have $e^\ast_2\!\in\!\onecut[s_3][d_2]$.
Then every path from $s_3$ to $\tail[e^\ast_1]$ must use $e^\ast_2$, otherwise $s_3$ can reach $e^\ast_1$ without using $e^\ast_2$ and finally arrive at $d_2$ since $e^\ast_1$ can reach $d_2$ (we showed in the above discussion that $e^\ast_1\!\in\!\onecut[s_1][d_2]$). This contradicts the previously constructed $e^\ast_2\!\in\!\onecut[s_3][d_2]$. Therefore, we must have $e^\ast_2\!\in\!\onecut[s_3][{\tail[e^\ast_1]}]$. Since $e^\ast_1\!\in\!\Dover[1]\!\subset\!\onecut[s_3][d_1]$, this in turn implies that $e^\ast_2$ is also an $1$-edge cut separating $s_3$ and $d_1$. However, note by the assumption that $e^\ast_2\!\in\!\Sover[2]\!\subset\!\onecut[s_2][d_1]$. Thus, $e^\ast_2$ will belong to $\Dover[1]$, which contradicts the assumption that $e^\ast_1$ is the most upstream $\Dover[1]$ edge. We thus have proven $\GCD[{\ChGANA[2][3]}][\,{\ChG[e_u][e^\ast_2]}]\PolyEqual 1$. Since we showed that $\GCD[{\ChGANA[3][1]\ChGANA[2][3]}][{\,\ChG[e_{d_3}][e_2^\ast]}]\!=\!\ChG[e_{d_3}][e_2^\ast]$, this further implies that the irreducible factor $\ChG[e_u][e^\ast_2]$ of $\ChG[e_{d_3}][e_2^\ast]$ must be contained by $\ChGANA[3][1]$ as a factor. Therefore, we have proven that $e^\ast_2\!\in\!\onecut[s_1][d_3]$. Symmetrically applying the above argument using the factor $\ChG[e^\ast_1][e_{s_3}]$ of the channel gain $\ChGANA[1][3]$, we can also prove that $e^\ast_1\!\in\!\onecut[s_3][d_2]$.

Thus far, we have proven that $e_2^\ast\!\in\!\onecut[s_1][d_2]$ and $e_2^\ast\!\in\!\onecut[s_1][d_3]$. However, $e^\ast_2\!=\!e_{s_1}$ is not possible since $e^\ast_2$, by our construction, is a downstream edge of $e_{s_2}$ but $e_{s_1}$ is not (since $|\IN[s_1]|\!=\!0$). As a result, we have proven $e_2^\ast\!\in\!\Sover[1]$. Recall that $e_2^\ast$ was chosen as one edge in $\Sover[2]$. Therefore, $\Sover[1]\cap\Sover[2]\NotEqualEmpty$. Similarly, we can also prove that $e_1^\ast\!\in\!\Dover[1]\cap\Dover[2]$ and thus $\Dover[1]\cap \Dover[2]\NotEqualEmpty$. The proof of {\bf Case~1} is complete.

{\bf Case~2:} $\Sover[i]\cap\Dover[j]\NotEqualEmpty$ for all $i\!\neq\!j$. By \LemRef{Lem4} and its \SWAPSD-symmetric version, we must have $\Sover[i]\cap\Sover[j]\NotEqualEmpty$ and $\Dover[i]\cap\Dover[j]\NotEqualEmpty$ $\forall\,i\neq j$. The proof of {\bf Case~2} is complete.
\end{proof}

\subsection{The Graph-Theoretic Conditions of the Feasibility of the 3-unicast ANA Scheme}\label{Sec4C}


\PropRef{Prop4} provides the graph-theoretic condition that characterizes whether or not the $\GANA$ of interest satisfies the algebraic condition of \Ref{LRneqVAR}, which implies that \Ref{C1}, \Ref{C3}, and \Ref{C5} hold simultaneously with close-to-one probability. However, to further ensure the feasibility of the 3-unicast ANA scheme, $\DET[{\FRM[\Sid][\Nid]}]$ must be non-zero polynomial (see \Ref{C2}, \Ref{C4}, and \Ref{C6}) for all $i\!\in\!\{1,2,3\}$. As a result, we need to prove the graph-theoretic characterization for the inequalities $\DET[{\FRM[\Sid][\Nid]}]\!\neq\!0$. Note by \PropRef{Prop1} that the condition $\DET[{\FRM[\Sid][\Nid]}]\!\neq\!0$ is equivalent to for all $i\!\in\!\{1,2,3\}$ the set of polynomials $\GSP[i][\left(\!\Nid\!\right)](\netvar)$ is linearly independent, where $\GSP[1][\left(\!\Nid\!\right)](\netvar)$ is defined in \Ref{GSP1^n} and $\GSP[2][\left(\!\Nid\!\right)](\netvar)$ and $\GSP[3][\left(\!\Nid\!\right)](\netvar)$ are defined as follows:
\begin{equation}\label{GSP2^n}\begin{split}
\GSP[2][\NidBRAC]&(\netvar) = \{\,\ChGANA[2][2]\ChGANA[3][1]\ChGANA[2][3]\BOLDb^{\Nid},\;\ChGANA[2][2]\ChGANA[3][1]\ChGANA[2][3]\BOLDb^{\Nid-1}\BOLDa, \\
& \qquad \;\cdots\;,\ChGANA[2][2]\ChGANA[3][1]\ChGANA[2][3]\BOLDb\BOLDa^{\Nid-1},\;\ChGANA[2][1]\ChGANA[3][2]\ChGANA[2][3]\BOLDb^{\Nid},\; \\
& \qquad\; \ChGANA[2][1]\ChGANA[3][2]\ChGANA[2][3]\BOLDb^{\Nid-1}\BOLDa,\;\cdots\;,\;\ChGANA[2][1]\ChGANA[3][2]\ChGANA[2][3]\BOLDa^{\Nid}\,\},
\end{split}\end{equation}\vspace{-0.015\columnwidth}
\begin{equation}\label{GSP3^n}\begin{split}
\GSP[3][\NidBRAC]&(\netvar) = \{\,\ChGANA[3][3]\ChGANA[2][1]\ChGANA[3][2]\BOLDb^{\Nid-1}\BOLDa,\;\cdots\;, \qquad\qquad\qquad\; \\
& \qquad \;\ChGANA[3][3]\ChGANA[2][1]\ChGANA[3][2]\BOLDb\BOLDa^{\Nid-1},\;\ChGANA[3][3]\ChGANA[2][1]\ChGANA[3][2]\BOLDa^{\Nid}, \\
& \qquad \;\ChGANA[3][1]\ChGANA[3][2]\ChGANA[2][3]\BOLDb^{\Nid},\;\ChGANA[3][1]\ChGANA[3][2]\ChGANA[2][3]\BOLDb^{\Nid-1}\BOLDa, \\
& \qquad \;\cdots\;,\;\ChGANA[3][1]\ChGANA[3][2]\ChGANA[2][3]\BOLDa^{\Nid}\,\}.
\end{split}\end{equation} \vspace{-0.01\columnwidth}

Thus in this subsection, we prove a graph-theoretic condition that characterizes the linear independence of $\GSP[i][\NidBRAC](\netvar)$ for all $i\!\in\!\{1,2,3\}$ when $\Nid\!=\!1$ and $\Nid\!\geq\!2$, respectively. Consider the following graph-theoretic conditions: \vspace{-0.015\columnwidth}

\noindent\begin{minipage}{\columnwidth+10pt}
\begin{align}
\hspace*{-10pt} & \Sover[\Sid] \cap \Sover[\Did]\,\EqualEmpty \textrm{ or } \Dover[\Sid] \cap \Dover[\Did]\,\EqualEmpty \;\; \forall\,\Sid,\Did\in\{1,2,3\},\,\Sid\neq\Did, \label{GTC1} \\
\hspace*{-10pt} & \EC[\{s_1,s_2\}][\{d_1,d_3\}]\!\geq\! 2,\, \EC[\{s_1,s_3\}][\{d_1,d_2\}]\!\geq\! 2, \label{GTC2} \\
\hspace*{-10pt} & \EC[s_1][d_1]\!\geq\! 1\;\text{on}\;\GANA\!\backslash\!\left\{\UPSTR[{(\Sover[2]\CAP\Dover[3])\CUP(\Sover[3]\CAP\Dover[2])}]\!\right\}\!,\!\!\!\! \label{GTC2A} \\
\hspace*{-10pt} & \EC[\{s_1,s_2\}][\{d_2,d_3\}]\!\geq\! 2,\, \EC[\{s_2,s_3\}][\{d_1,d_2\}]\!\geq\! 2, \label{GTC3} \\
\hspace*{-10pt} & \EC[s_2][d_2]\!\geq\! 1\;\text{on}\;\GANA\!\backslash\!\left\{\UPSTR[{(\Sover[1]\CAP\Dover[3])\CUP(\Sover[3]\CAP\Dover[1])}]\!\right\}\!,\!\!\!\! \label{GTC3A} \\
\hspace*{-10pt} & \EC[\{s_1,s_3\}][\{d_2,d_3\}]\!\geq\! 2,\, \EC[\{s_2,s_3\}][\{d_1,d_3\}]\!\geq\! 2, \label{GTC4} \\
\hspace*{-10pt} & \EC[s_3][d_3]\!\geq\! 1\;\text{on}\;\GANA\!\backslash\!\left\{\UPSTR[{(\Sover[1]\CAP\Dover[2])\CUP(\Sover[2]\CAP\Dover[1])}]\!\right\}\!.\!\!\!\! \label{GTC4A}
\end{align}
\end{minipage} \vspace{+0.015\columnwidth}

Note that (i) \Ref{GTC1} is equivalent to $\LRneqVAR$ by \PropRef{Prop4}; (ii) \Ref{GTC2}, \Ref{GTC3}, and \Ref{GTC4} are equivalent to (17) to (19) by \CorRef{Cor1}; and (iii) \Ref{GTC2A}, \Ref{GTC3A}, and \Ref{GTC4A} are the new conditions that help characterize (14) to (16).


To further simplify the analysis, we consider the following set of polynomials: \vspace{-0.01\columnwidth}
\begin{equation}\label{GSP2^n_NEW}\begin{split}
\GSPtilde[1][\NidBRAC]&(\netvar) = \{\,\ChGANA[1][1]\ChGANA[3][2]\ChGANA[1][3]\BOLDa^{\Nid},\;\ChGANA[1][1]\ChGANA[3][2]\ChGANA[1][3]\BOLDa^{\Nid-1}\BOLDb, \\
& \qquad \;\cdots\;,\ChGANA[1][1]\ChGANA[3][2]\ChGANA[1][3]\BOLDa\BOLDb^{\Nid-1},\;\ChGANA[1][2]\ChGANA[3][1]\ChGANA[1][3]\BOLDa^{\Nid},\; \\
& \qquad\; \ChGANA[1][2]\ChGANA[3][1]\ChGANA[1][3]\BOLDa^{\Nid-1}\BOLDb,\;\cdots\;,\;\ChGANA[1][2]\ChGANA[3][1]\ChGANA[1][3]\BOLDb^{\Nid}\,\},
\end{split}\vspace{-0.015\columnwidth}\end{equation}
where $\GSPtilde[1][\NidBRAC](\netvar)$ is obtained by swapping the roles of $s_1$ and $s_2$ (resp. $s_3$), and the roles of $d_1$ and $d_2$ (resp. $d_3$) to the expression of $\GSP[2][\NidBRAC](\netvar)$ in \Ref{GSP2^n} (resp. $\GSP[3][\NidBRAC](\netvar)$ in \Ref{GSP3^n}). Note that $\BOLDb=\ChGANA[2][1]\ChGANA[3][2]\ChGANA[1][3]$ becomes $\BOLDa=\ChGANA[3][1]\ChGANA[2][3]\ChGANA[1][2]$ and vice versa by such swap operation. Once we characterize the graph-theoretic conditions for the linear independence of $\GSPtilde[1][\NidBRAC](\netvar)$, then the characterization for $\GSP[2][\NidBRAC](\netvar)$ and $\GSP[3][\NidBRAC](\netvar)$ being linearly independent will be followed symmetrically.\footnote{In \SecRef{Sec2D}, $(s_1,d_1)$-pair was chosen to achieve larger rate than other pairs when aligning the interference. Thus the feasibility characterization for the other transmission pairs, $(s_2,d_2)$ and $(s_3,d_3)$ who achieve the same rate, becomes symmetric.}

\begin{prop}\label{Prop5} For a given $\GANA$, when $\Nid\!=\!1$, we have
\begin{itemize}
\item[(H1)] $\GSP[1][\left(\!1\!\right)](\netvar)$ is linearly independent if and only if $\GANA$ satisfies \Ref{GTC1} and \Ref{GTC2}.
\item[(K1)] $\GSPtilde[1][\left(\!1\!\right)](\netvar)$ is linearly independent if and only if $\GANA$ satisfies \Ref{GTC1}, \Ref{GTC2}, and \Ref{GTC2A}.
\end{itemize}
Moreover when $\Nid\!\geq\!2$, we have
\begin{itemize}
\item[(H2)] $\GSP[1][\left(\!\Nid\!\right)](\netvar)$ is linearly independent if and only if $\GANA$ satisfies \Ref{GTC1}, \Ref{GTC2}, and \Ref{GTC2A}.
\item[(K2)] $\GSPtilde[1][\left(\!\Nid\!\right)](\netvar)$ is linearly independent if and only if $\GANA$ satisfies \Ref{GTC1}, \Ref{GTC2}, and \Ref{GTC2A}.
\end{itemize}
\end{prop}


{\em Remark:} \PropRef{Prop5} proves that the conjecture in \cite{RamakrishnanDaszMalekiMarkopoulouJafarVishwanath:Allerton10} holds only for the linearly independent $\GSP[1][\left(\!1\!\right)](\netvar)$. In general, it is no longer true for the case of $\Nid\!\geq\!2$ and even for $\Nid\!=\!1$. This coincides with the recent results \cite{RmakrishnanMeng:UCI-TecRep}, which show that for the case of $\Nid\!\geq\!2$, the conjecture in \cite{RamakrishnanDaszMalekiMarkopoulouJafarVishwanath:Allerton10} no longer holds.

\begin{proof} Similar to most graph-theoretic proofs, the proofs of (H1), (K1), (H2), and (K2) involve detailed discussion of several subcases. To structure our proof, we first define the following logic statements. Each statement could be true or false. We will later use these statements to complete the proof.

\noindent $\bullet$ \HS[1]{\bf:}\label{ref:HS[1]} $\GSP[1][\left(\!\Nid\!\right)](\netvar)$ is linearly independent for $\Nid\!=\!1$.

\noindent $\bullet$ \HStilde[1]{\bf:}\label{ref:HStilde[1]} $\GSPtilde[1][\left(\!\Nid\!\right)](\netvar)$ is linearly independent for $\Nid\!=\!1$.

\noindent $\bullet$ \HS[2]{\bf:}\label{ref:HS[2]} $\GSP[1][\left(\!\Nid\!\right)](\netvar)$ is linearly independent for some $\Nid\!\geq\!2$.

\noindent $\bullet$ \HStilde[2]{\bf:}\label{ref:HStilde[2]} $\GSPtilde[1][\left(\!\Nid\!\right)](\netvar)$ is linearly independent for some $\Nid\!\geq\!2$.

\noindent $\bullet$ \LNR{\bf:}\label{ref:LNR} $\LRneqVAR$.

\noindent $\bullet$ \GS[1]{\bf:}\label{ref:GS[1]} $\ChGANA[1][1]\ChGANA[3][2]\,\PolyNotEqual\,\ChGANA[1][2]\ChGANA[3][1]$ and $\ChGANA[1][1]\ChGANA[2][3]\,\PolyNotEqual\,\ChGANA[1][3]\ChGANA[2][1]$.

\noindent \makebox[0.9cm][l]{$\bullet$ \GS[2]{\bf:}}\label{ref:GS[2]} $\EC[s_1][d_1]\!\!\geq\!\!1\!$ on $\!\GANA\!\backslash\!\!\left\{\UPSTR[{(\Sover[2]\CAP\Dover[3]\!)\CUP(\Sover[3]\CAP\Dover[2]\!)}]\!\right\}$.

One can clearly see that proving Statement~(H1) is equivalent to proving ``\LNR\AND\GS[1]\EQUIV\HS[1]" where ``$\wedge$" is the AND operator. Similarly, proving Statements~(K1),~(H2), and~(K2) is equivalent to proving ``\LNR\AND\GS[1]\AND\GS[2]\EQUIV\HStilde[1]", ``\LNR\AND\GS[1]\\\noindent\AND\GS[2]\EQUIV\HS[2]", and ``\LNR\AND\GS[1]\AND\GS[2]\EQUIV\HStilde[2]", respectively.

The reason why we use the notation of ``logic statements" (e.g., \HS[1], \LNR, etc.) is that it enables us to break down the overall proof into proving several smaller ``logic relationships" (e.g, ``\LNR\AND\GS[1]\EQUIV\HS[1]", etc.) and later assemble all the logic relationships to derive the final results. The interested readers can thus separate the verification of the proof of each individual logic relationship from the examination of the overall structure of the proof of the main results. The proof of each logic relationship is kept no longer than one page and is independent from the proof of any other logic relationship. This allows the readers to set their own pace when going through the proofs.

To give an insight how the proof works, here we provide the proof of ``\LNR\AND\GS[1]\OPPLY\HS[1]" at the bottom. All the other proofs are relegated to the appendices. Specifically, we provide the general structured proofs for the necessity direction ``$\Leftarrow$" in \AppRef{GeneralNecessityProof}. Applying this result, the proofs of ``\LNR\AND\GS[1]\\\noindent\AND\GS[2]\OPPLY\HS[2], \HStilde[1], \HStilde[2]" are provided in \AppRef{Prop5<=Proof}. Similarly, the general structured proofs for the sufficiency direction ``$\Rightarrow$" is provided in \AppRef{GeneralSufficiencyProof}. The proofs of ``\LNR\AND\GS[1]\IMPLY\HS[1]" and ``\LNR\AND\GS[1]\AND\GS[2]\IMPLY\HStilde[1], \HS[2], \HStilde[2]" are provided in \AppRef{Prop5=>Proof}.

{\em The proof of ``\LNR\AND\GS[1]\OPPLY\HS[1]":\quad}\label{Proof=>Prop5-1} We prove the following statement instead: (\NotLNR)\OR(\NotGS[1])\IMPLY(\NotHS[1]) where $\neg$ is the NOT logic operator and ``$\vee$" is the OR operator. From the expression of $\GSP[1][\left(\!\Nid\!\right)](\netvar)$ in \Ref{GSP1^n}, consider $\GSP[1][\left(\!1\!\right)](\netvar)$ which contains 3 polynomials:
\begin{equation}\label{GSP1^1}
\GSP[1][\left(\!1\!\right)](\netvar)=\{\,\ChGANA[1][1]\ChGANA[3][2]\ChGANA[2][3]\BOLDb, \;\ChGANA[1][1]\ChGANA[3][2]\ChGANA[2][3]\BOLDa, \;\ChGANA[1][2]\ChGANA[3][1]\ChGANA[2][3]\BOLDb\,\}.
\end{equation}

Suppose $\GANA$ satisfies (\NotLNR)\OR(\NotGS[1]), which means $\GANA$ satisfies either $\LReqVAR$ or $\ChGANA[1][1]\ChGANA[3][2]\PolyEqual\ChGANA[1][2]\ChGANA[3][1]$ or $\ChGANA[1][1]\ChGANA[2][3]\PolyEqual\ChGANA[1][3]\ChGANA[2][1]$. If $\LReqVAR$, then we notice that $\ChGANA[1][1]\ChGANA[3][2]\ChGANA[2][3]\BOLDb\,\PolyEqual\,\ChGANA[1][1]\ChGANA[3][2]\ChGANA[2][3]\BOLDa$ and $\GSP[1][\left(\!1\!\right)](\netvar)$, defined in \Ref{GSP1^1}, is thus linearly dependent. If $\ChGANA[1][1]\ChGANA[3][2]\PolyEqual\ChGANA[1][2]\ChGANA[3][1]$, then we notice that $\ChGANA[1][1]\ChGANA[3][2]\ChGANA[2][3]\BOLDb \,\PolyEqual\, \ChGANA[1][2]\ChGANA[3][1]\ChGANA[2][3]\BOLDb$. Similarly if $\ChGANA[1][1]\ChGANA[2][3]\PolyEqual\ChGANA[1][3]\ChGANA[2][1]$, then we have $\ChGANA[1][1]\ChGANA[3][2]\ChGANA[2][3]\BOLDa \,\PolyEqual\, \ChGANA[1][2]\ChGANA[3][1]\ChGANA[2][3]\BOLDb$. The proof is thus complete.
\end{proof}

\section{Conclusion and Future Works}\label{Sec5}
The main subject of this work is the general class of precoding-based NC schemes, which focus on designing the precoding and decoding mappings at the sources and destinations while using randomly generated local encoding kernels within the network. One example of the precoding-based structure is the 3-unicast ANA scheme, originally proposed in \cite{DasVishwanathJafarMarkopoulou:ISIT10,RamakrishnanDaszMalekiMarkopoulouJafarVishwanath:Allerton10}. In this work, we have identified new graph-theoretic relationships for the precoding-based NC solutions. Based on the findings on the general precoding-based NC, we have further characterized the graph-theoretic feasibility conditions of the 3-unicast ANA scheme. We believe that the analysis in this work will serve as a precursor to fully understand the notoriously challenging multiple-unicast NC problem and design practical, distributed NC solutions based on the precoding-based framework.

\appendices
\renewcommand{\thesubsectiondis}{\thesection-\arabic{subsection}.}
\renewcommand{\thesubsection}{\thesection-\arabic{subsection}}
\section{Proofs of \PropsRef{Prop1}{Prop2}}\label{Prop1-2Proof}
We prove \PropRef{Prop1} as follows.
\begin{proof}[Proof of $\Rightarrow$] We prove this direction by contradiction. Suppose that $\GSP(\netvar)$ is linearly dependent. Then, there exists a set of coefficients $\{\alpha_k\}_{k=1}^{\SMN}$ such that $\sum_{k=1}^\SMN \alpha_k h_k(\netvar)\!=\!0$ and at least one of them is non-zero. Since $[\GSP(\netvar[(\!k\!)])]_{k=1}^{\SMN}$ is row-invariant, we can perform elementary column operations on $[\GSP(\netvar[(\!k\!)])]_{k=1}^{\SMN}$ using $\{\alpha_k\}_{k=1}^{\SMN}$ to create an all-zero column. Thus, $\DET[{[\GSP(\netvar[(\!k\!)])]_{k=1}^{\SMN}}]$ is a zero polynomial.
\renewcommand{\IEEEQED}{} \end{proof}
\begin{proof}[Proof of $\Leftarrow$] This direction is also proven by contradiction. Suppose that $\DET[{[\GSP(\netvar[(\!k\!)])]_{k=1}^{\SMN}}]$ is a zero polynomial. We will prove that $\GSP(\netvar)$ is linearly dependent by induction on the value of $\SMN$. For $N\!=\!1$, $\DET[{[\GSP(\netvar[(\!k\!)])]_{k=1}^{\SMN}}]\!=\!0$ implies that $h_1(\netvar)$ is a zero polynomial, which by definition is linearly dependent.

Suppose that the statement holds for any $\SMN\!\!<\!\!n_0$. When $\SMN\!\!=\!n_0$, consider the (1,1)-th cofactor of $[\GSP(\netvar[(\!k\!)])]_{k=1}^{\SMN}$, which is the determinant of the submatrix of the intersection of the 2nd to $\SMN$-th rows and the 2nd to $\SMN$-th columns. Consider the following two cases. Case~1: the $(1,1)$-th cofactor is a zero polynomial. Then by the induction assumption $\{h_2(\netvar),...,h_{\SMN}(\netvar)\}$ is linearly dependent. By definition, so is $\GSP(\netvar)$. Case~2: the $(1,1)$-th cofactor is a non-zero polynomial. Since we assume a sufficiently large $q$, there exists an assignment $\netreal[2]\!\in\!\FF[\!q][{|\netvar|}]$ to $\netreal[\SMN]\!\in\!\FF[\!q][{|\netvar|}]$ such that the value of the (1,1)-th cofactor is non-zero when evaluated by $\netreal[2]$ to $\netreal[\SMN]$. But note that by the Laplace expansion, we also have $\sum_{k=1}^{\SMN} h_{k}(\netvar[(\!1\!)])\, C_{1k} = 0$ where $C_{1k}$ is the $(1,k)$-th cofactor. By evaluating $C_{1k}$ with $\{\netreal[i]\}_{i=2}^{\SMN}$, we can conclude that $\GSP(\netvar)$ is linearly dependent since at least one of $C_{1k}$ (specifically $C_{11}$) is non-zero.
\end{proof}

We prove \PropRef{Prop2} as follows.
\begin{proof}[Proof of $\Leftarrow$] This can be proved by simply choosing $\G[']\!=\!\G$.
\renewcommand{\IEEEQED}{} \end{proof}
\begin{proof}[Proof of $\Rightarrow$]
Since $f(\{\ChG[e'_i][e_i](\netvar):\forall\,i\!\in\!I\})\PolyEqual g(\{\ChG[e'_i][e_i]$ $(\netvar):\forall\,i\!\in\!I\})$, we can assume $f(\{\ChG[e'_i][e_i](\netvar):\forall\,i\!\in\!I\})\!=\!\alpha g(\{\ChG[e'_i][e_i](\netvar):\forall\,i\!\in\!I\})$ for some non-zero $\alpha\!\in\!\FF[\!q]$. Consider any subgraph $\G[']$ containing all edges in $\{e_i,e'_i:\forall\,i\!\in\!I\}$ and the channel gain $\ChG[e'_i][e_i](\netvarSUB)$ on $\G[']$. Then, $\ChG[e'_i][e_i](\netvarSUB)$ can be derived from $\ChG[e'_i][e_i](\netvar)$ by substituting those $\netvar$ variables that are not in $\G[']$ by zero. As a result, we immediately have $f(\{\ChG[e'_i][e_i](\netvarSUB):\forall\,i\!\in\!I\})\!=\!\alpha g(\{\ChG[e'_i][e_i](\netvarSUB):\forall\,i\!\in\!I\})$ for the same $\alpha$. The proof of this direction is thus complete.
\end{proof}

\section{Proofs of \CorsRef{Cor1}{Cor2}}\label{Cor1-2Proof}
We prove \CorRef{Cor1} as follows.
\begin{proof}[Proof of $\Rightarrow$] We assume $(i_1,i_2)\!=\!(1,2)$ and $(j_1,j_2)\!=\!(1,3)$ without loss of generality. Since $\EC[\{s_1,s_2\}][\{d_1,d_3\}]$ $\!=\!1$, there exists an edge $e^\ast$ that separates $\{d_1,d_3\}$ from $\{s_1,s_2\}$. Therefore, we must have $\ChGANA[1][1]\!=\!\ChG[e^\ast][e_{s_1}]\ChG[e_{d_1}][e^\ast]$, $\ChGANA[3][1]\!=\!\ChG[e^\ast][e_{s_1}]\ChG[e_{d_3}][e^\ast]$, $\ChGANA[1][2]\!=\!\ChG[e^\ast][e_{s_2}]$ $\ChG[e_{d_1}][e^\ast]$, and $\ChGANA[3][2]\!=\!\ChG[e^\ast][e_{s_2}]\ChG[e_{d_3}][e^\ast]$. As a result, $\ChGANA[1][1]\ChGANA[3][2]\,\PolyEqual\,\ChGANA[1][2]\ChGANA[3][1]$.
\renewcommand{\IEEEQED}{} \end{proof}
\begin{proof}[Proof of $\Leftarrow$] We prove this direction by contradiction. Suppose $\EC[\{s_{i_1},s_{i_2}\}][\{d_{j_1},d_{j_2}\}]\!\geq\! 2$. In a $\GANA$ network, each source (resp. destination) has only one outgoing (resp. incoming) edge. Therefore, $\EC[\{s_{i_1},s_{i_2}\}][\{d_{j_1},d_{j_2}\}]\!\geq\! 2$ implies that at least one of the following two cases must be true: Case~1: There exists a pair of edge-disjoint paths $P_{s_{i_1}d_{j_1}}$ and $P_{s_{i_2}d_{j_2}}$; Case~2: There exists a pair of edge-disjoint paths $P_{s_{i_1}d_{j_2}}$ and $P_{s_{i_2}d_{j_1}}$. For\;\,Case~1, we consider the network variables that are along the two edge-disjoint paths, i.e., consider the collection $\netvarSUB$ of network variables $x_{ee'}\!\in\!\netvar$ such that either both $e$ and $e'$ are used by $P_{s_{i_1}d_{j_1}}$ or both $e$ and $e'$ are used by $P_{s_{i_2}d_{j_2}}$. We keep those variables in $ \netvarSUB$ intact and set the other network variables to be zero. As a result, we will have $\ChGANA[j_1][i_1](\netvarSUB)\ChGANA[j_2][i_2](\netvarSUB)=\prod_{\forall x_{ee'}\in\netvarSUB} x_{ee'}$ and $\ChGANA[j_1][i_2](\netvarSUB)\ChGANA[j_2][i_1](\netvarSUB)\!=\!0$ where the latter is due the edge-disjointness between two paths $P_{s_{i_1}d_{j_1}}$ and $P_{s_{i_2}d_{j_2}}$. This implies that before hardwiring the variables outside $\netvarSUB$, we must have $\ChGANA[j_1][i_1](\netvar)\ChGANA[j_2][i_2](\netvar)$ $\,\PolyNotEqual\,\ChGANA[j_1][i_2](\netvar)\ChGANA[j_2][i_1](\netvar)$. The proof of Case~1 is complete. Case~2 can be proven by swapping the labels of $j_1$ and $j_2$.
\end{proof}

We prove \CorRef{Cor2} as follows.
\begin{proof} When $(i_1,j_1)\!=\!(i_2,j_2)$, obviously $\ChGANA[j_1][i_1]\!=\!\ChGANA[j_2][i_2]$ and $\GCD[{\ChGANA[j_1][i_1]}][{\,\ChGANA[j_2][i_2]}]\PolyEqual\ChGANA[j_2][i_2]$. Suppose that for some $(i_1,j_1)\!\neq\!(i_2,j_2)$, $\GCD[{\ChGANA[j_1][i_1]}][{\,\ChGANA[j_2][i_2]}]\,\PolyEqual\,\ChGANA[j_2][i_2]$. Without loss of generality, we assume $i_1\!\neq\!i_2$. Since the channel gains are defined for two distinct sources, we must have $\ChGANA[j_1][i_1]\PolyNotEqual\ChGANA[j_2][i_2]$. As a result, $\GCD[{\ChGANA[j_1][i_1]}][{\,\ChGANA[j_2][i_2]}]$
$\PolyEqual\ChGANA[j_2][i_2]$ implies that $\ChGANA[j_1][i_1]$ must be reducible. By \PropRef{Prop3}, $\ChGANA[j_1][i_1]$ must be expressed as $\ChGANA[j_1][i_1]\!=\ChG[e_1][e_{s_{i\!_1}}]\!\!\left(\prod_{i=1}^{N-1}\!\ChG[e_{i+1}][e_i]\!\right)\!\ChG[e_{d_{j\!_1}}][e_N]\!$ where each term corresponds to a pair of consecutive $1$-edge cuts separating $s_{i_1}$ and $d_{j_1}$. For $\ChGANA[j_1][i_1]$ to contain $\ChGANA[j_2][i_2]$ as a factor, the source edge $e_{s_{i_2}}$ must be one of the $1$-edge cuts separating $s_{i_1}$ and $d_{j_1}$. This contradicts the assumption that in a 3-unicast ANA network $|\IN[s_i]|\!=\!0$ for all $i$. The proof is thus complete.
\end{proof}

\section{Proof of \PropRef{Prop3}}\label{Prop3Proof}

\PropRef{Prop3} will be proven through the concept of the line graph, which is defined as follows: The line graph of a DAG $\G\!=\!(\GV,\GE)$ is represented as $\LG\!=\!(\LGV,\LGE)$, with the vertex set $\LGV \!=\! \GE$ and edge set $\LGE\!=\!\{(e',e'')\!\in\!\GE^2:\,\head[e']\!=\!\tail[e'']\}$ (the set representing the adjacency relationships between the edges of $\GE$). Provided that $\G$ is directed acyclic, its line graph $\LG$ is also directed acyclic. The graph-theoretic notations for $\G$ defined in \SecRef{Sec2A} are applied in the same way as in $\LG$.

Note that the line graph translates the edges into vertices. Thus, a {\em vertex cut} in the line graph is the counterpart of the edge cut in a normal graph. Specifically, a {\em $k$-vertex cut} separating vertex sets $U$ and $W$ is a collection of $k$ vertices other than the vertices in $U$ and $W$ such that any path from any $u\!\in\!U$ to any $w\!\in\!W$ must use at least one of those $k$ vertices. Moreover, the minimum value (number of vertices) of all the possible vertex cuts between vertex sets $U$ and $W$ is termed $\VC[U][W]$. For any nodes $u$ and $v$ in $\GV$, one can easily see that $\EC[u][v]$ in $\G[]$ is equal to $\VC[\tilde{u}][\tilde{v}]$ in $\LG$ where $\tilde{u}$ and $\tilde{v}$ are the vertices in $\LG$ corresponding to any incoming edge of $u$ and any outgoing edge of $v$, respectively. 

Once we focus on the line graph $\LG$, the network variables $\netvar$, originally defined over the $(e', e'')$ pairs of the normal graph, are now defined on the edges of the line graph. We can thus define the channel gain from a vertex $u$ to a vertex $v$ on $\LG$ as 
\begin{equation}\label{ChGL}
\ChGL[v][u]=\sum_{\forall\,P_{uv}\in\PATHset{u}{v}{}}\prod_{\forall\,e\in P_{uv}} x_e,
\end{equation}
where $\PATHset{u}{v}{}$ denotes the collection of all distinct paths from $u$ to $v$. For notational simplicity, we sometimes simply use ``an edge $e$" to refer to the corresponding network variable $x_e$. Each $x_e$ (or $e$) thus takes values in $\FF[\!q]$. When $u\!=\!v$, simply set $\ChGL[v][u]\!=\!1$. 

The line-graph-based version of \PropRef{Prop3} is described as follows:
\begin{cor}\label{Cor3} Given the line graph $\LG$ of a DAG $\G$, $\mathring{m}$ defined above, and two distinct vertices $s$ and $d$, the following is true:
\begin{itemize}
\item \makebox[2.9cm][l]{If $\VC[s][d]\!=\!0$, then} $\ChGL[d][s]\!=\!0$

\item \makebox[2.9cm][l]{If $\VC[s][d]\!=\!1$, then} $\ChGL[d][s]$ is reducible and can be expressed as $\ChGL[d][s]\!=\!\ChGL[u_1][s]$$\left(\prod_{i=1}^{N-1}\ChGL[u_{i+1}][u_i]\right)$$\ChGL[d][u\!_N]$ where $\{u_i\}_{i=1}^{N}$ are all the distinct $1$-vertex cuts between $s$ and $d$ in the topological order (from the most upstream to the most downstream). Moreover, the polynomial factors $\ChGL[u_1][s]$, $\{\ChGL[u_{i+1}][u_i]\}_{i=1}^{N-1}$, and $\ChGL[d][u\!_N]$ are all irreducible, and no two of them are equivalent.

\item \makebox[2.1cm][l]{If $\VC[s][d]\!\geq\! 2$} (including $\infty$), then $\ChGL[d][s]$ is irreducible.
\end{itemize}\end{cor}


\begin{proof} We use the induction on the number of edges $\ABS[\LGE]$ of $\LG\!=\!(\LGV,\LGE)$. When $\ABS[\LGE]\!=\!0$, then $\VC[s][d]\!=\!0$ since there are no edges in $\LG$. Thus $\ChGL[d][s]\!=\!0$ naturally.

Suppose that the above three claims are true for $\ABS[\LGE]\!=\!k-1$. We would like to prove that those claims also hold for the line graph $\LG$ with $\ABS[\LGE]\!=\!k$.

\makebox[3.9cm][l]{{\bf Case 1:} $\VC[s][d]\!=\!0$ on $\LG$.} In this case, $s$ and $d$ are already disconnected. Therefore, $\ChGL[d][s]\!=\!0$.

\makebox[3.9cm][l]{{\bf Case 2:} $\VC[s][d]\!=\!1$ on $\LG$.} Consider all distinct $1$-vertex cuts $u_1,\cdots\!,u\!_N$ between $s$ and $d$ in the topological order. If we define $u_0\!\triangleq\!s$ and $u\!_{N+1}\!\triangleq\!d$, then we can express $\ChGL[d][s]$ as $\ChGL[d][s]\!=\!\prod_{i=0}^{N}\ChGL[u_{i+1}][u_i]$. Since we considered all distinct $1$-vertex cuts between $s$ and $d$, we must have $\VC[u_i][u_{i+1}]\!\geq\! 2$ for $i\!=\!0,\cdots\!,N$. By induction, $\{\ChGL[u_{i+1}][u_i]\}_{i=0}^{N}$ are all irreducible. Also, since each sub-channel gain $\ChGL[u_{i+1}][u_i]$ covers a disjoint portion of $\LG$, no two of them are equivalent.

\makebox[3.9cm][l]{{\bf Case 3:} $\VC[s][d]\!\geq\!2$ on $\LG$.} Without loss of generality, we can also assume that $s$ can reach any vertex $u\!\in\!\LGV$ and $d$ can be reached from any vertex $u\!\in\!\LGV$. Consider two subcases: Case 3.1: all edges in $\LGE$ have their tails being $s$ and their heads being $d$. In this case, $\ChGL[d][s]=\sum_{e\in \LGE} x_e$. Obviously $\ChGL[d][s]$ is irreducible. Case 3.2: at least one edge in $\LGE$ is not directly connecting $s$ and $d$. In this case, there must exist an edge $e'$ such that $s \PREC \tail[e']$ and $\head[e']\!=\!d$. Arbitrarily pick one such edge $e'$ and fix it. We denote the tail vertex of the chosen $e'$ by $w$.  By the definition of \Ref{ChGL}, we have
\begin{align}
\ChGL[d][s] = \ChGL[w][s] x_{e'} + \ChGLsub[d][s], \label{Case3-1}
\end{align}
where $\ChGL[w][s]$ is the channel gain from $s$ to $w$, and $\ChGLsub[d][s]$ is the channel gain from $s$ to $d$ on the subgraph $\LG[']\!=\!\LG\backslash\{e'\}$ that removes $e'$ from $\LG$. Note that there always exists a path from $s$ to $d$ not using $w$ on $\LG[']$ otherwise $w$ will be a cut separating $s$ and $d$ on $\LG$, contradicting the assumption that $\VC[s][d]\!\geq\!2$. 

We now argue by contradiction that $\ChGL[d][s]$ must be irreducible. Suppose not, then $\ChGL[d][s]$ can be written as a product of two polynomials $A$ and $B$ with the degrees of $A$ and $B$ being larger than or equal to $1$. We can always write $A = x_{e'}A_1 +A_2$ by singling out the portion of $A$ that has $x_{e'}$ as a factor. Similarly we can write $B = x_{e'}B_1 + B_2$. We then have
\begin{align}
\ChGL[d][s]= (x_e' A_1 + A_2)( x_e' B_1 +B_2). \label{Case3-2}
\end{align}

We first notice that by \Ref{Case3-1} there is no quadratic term of $x_{e'}$ in $\ChGL[d][s]$. Therefore, one of $A_1$ and $B_1$ must be a zero polynomial. Assume $B_1 = 0$. Comparing \Ref{Case3-1} and \Ref{Case3-2} shows that $\ChGL[w][s] = A_1 B_2$ and $\ChGLsub[d][s] = A_2 B_2$. Since the degree of $B$ is larger than or equal to $1$ and $B_1 = 0$, the degree of $B_2$ must be larger than equal to $1$. As a result, we have $\GCD[{\ChGL[w][s]}][{\,\ChGLsub[d][s]}]\,\PolyNotEqual 1$ (having at least a non-zero polynomial $B_2$ as its common factor).

The facts that $\GCD[{\!\ChGL[w][s]}][{\,\ChGLsub[d][s]}]\,\PolyNotEqual 1$ and $w \PREC d$ imply that one of the following three cases must be true: (i) Both $\ChGL[w][s]$ and $\ChGLsub[d][s]$ are reducible; (ii) $\ChGL[w][s]$ is reducible but $\ChGLsub[d][s]$ is not; and (iii) $\ChGLsub[d][s]$ is reducible but $\ChGL[w][s]$ is not. For Case (i), by applying \PropRef{Prop3} to the subgraph $\LG[']\!=\!\LG\backslash\{e'\}$, we know that $\VC[s][w]\!=\!\VC[s][d]\!=\!1$ and both polynomials $\ChGL[w][s]$ and $\ChGLsub[d][s]$ can be factorized according to their $1$-vertex cuts, respectively. Since $\ChGL[w][s]$ and $\ChGLsub[d][s]$ have a common factor, there exists a vertex $u$ that is both a $1$-vertex cut separating $s$ and $w$ and a $1$-vertex cut separating $s$ and $d$ when focusing on $\LG[']$. As a result, such $u$ is a $1$-vertex cut separating $s$ and $d$ in the original graph $\LG$. This contradicts the assumption $\VC[s][d]\!\geq\!2$ in $\LG$. For Case (ii), by applying \PropRef{Prop3} to $\LG[']$, we know that $\VC[s][w]\!=\!1$ and $\ChGL[w][s]$ can be factorized according to their $1$-vertex cuts. Since $\ChGL[w][s]$ and the irreducible $\ChGLsub[d][s]$ have a common factor, $\ChGL[w][s]$ must contain $\ChGLsub[d][s]$ as a factor, which implies that $d$ is a $1$-vertex cut separating $s$ and $w$ in $\LG[']$. This contradicts the construction of $\LG[']$ where $w \PREC d$. For Case (iii), by applying \PropRef{Prop3} to $\LG[']$, we know that $\VC[s][d]\!=\!1$ and $\ChGLsub[d][s]$ can be factorized according to their $1$-vertex cuts. Since $\ChGLsub[d][s]$ and the irreducible $\ChGL[w][s]$ have a common factor, $\ChGLsub[d][s]$ must contain $\ChGL[w][s]$ as a factor, which implies that $w$ is a $1$-vertex cut separating $s$ and $d$ in $\LG[']$. As a result, $w$ is a $1$-vertex cut separating $s$ and $d$ in the original graph $\LG$. This contradicts the assumption $\VC[s][d]\!\geq\!2$ in $\LG$.
\end{proof}

\section{Proofs of Lemmas~1 to~7}\label{Lem1-7Proof}
We prove \LemRef{Lem1} as follows.
\begin{proof}
Consider indices $i\!\neq\!j$. By the definition, all paths from $s_i$ to $d_j$ must use all edges in $\Sover[i]$ and all edges in $\Dover[j]$. Thus, for any $e'\!\in\!\Sover[i]$ and any $e''\!\in\!\Dover[j]$, one of the following statements must be true: $e'\PREC e''$, $e'\SUCC e''$, or $e'\!=\!e''$.
\end{proof}

We prove \LemRef{Lem2} as follows.
\begin{proof} Consider three indices $i$, $j$, and $k$ taking distinct values in $\{1,2,3\}$. Consider an arbitrary edge $e\!\in\! \Dover[i]\cap\Dover[j]$. By definition, all paths from $s_k$ to $d_i$, and all paths from $s_k$ to $d_j$ must use $e$. Therefore, $e\!\in\!\Sover[k]$.
\end{proof}

We prove \LemRef{Lem3} as follows.
\begin{proof} Without loss of generality, let $i\!=\!1$ and $j\!=\!2$. Choose the most downstream edge in $\Sover[1]\backslash\Dover[2]$ and denote it as $e'_\ast$. Since $e'_\ast$ belongs to $\onecut[s_1][d_2]\CAP\onecut[s_1][d_3]$ but not to $\onecut[s_3][d_2]$, there must exist a \FromTo[3][2] path $P_{32}$ not using $e'_\ast$. In addition, for any $e''\!\in\!\Dover[2]$, we have either $e''\PREC e'_\ast$, $e''\SUCC e'_\ast$, or $e''\!=\!e'_\ast$ by \LemRef{Lem1}. Suppose there exists an edge $e''\!\in\!\Dover[2]$ such that $e'' \PREC  e'_\ast$. Then by definition, any \FromTo[3][2] path must use $e''$. Also note that since $e''\!\in\!\Dover[2]$, there exists a path $P_{s_1 \tail[e'']}$ from $s_1$ to $\tail[e'']$. Consider the concatenated \FromTo[1][2] path $P_{s_1 \tail[e'']}e''P_{32}$. We first note that since $e''\PREC e'_\ast$, the path segment $P_{s_1 \tail[e'']}e''$ does not use $e'_\ast$. By our construction, $P_{32}$ also does not use $e'_\ast$. Jointly, the above observations contradict the fact that $e'_\ast\!\in\!\Sover[1]$ is a $1$-edge cut separating $s_1$ and $d_2$. By contradiction, we must have $e'_\ast\PRECEQ e''$. Note that since by our construction $e'_\ast$ must not be in $\Dover[2]$ while $e''$ is in $\Dover[2]$, we must have $e'_\ast\!\neq\!e''$ and thus $e'_\ast\PREC e''$. Since $e'_\ast$ was chosen as the most downstream edge of $\Sover[1]\backslash \Dover[2]$, we have $e'\PREC e''$ for all $e'\!\in\!\Sover[1]\backslash \Dover[2]$ and $e''\!\in\!\Dover[2]$. The proof is thus complete.
\end{proof}

We prove \LemRef{Lem4} as follows.
\begin{proof}[Proof of $\Rightarrow$] We note that $(\Sover[i]\CAP\Dover[j])\!\supset\!(\Sover[i]\CAP\Dover[j]\CAP\Dover[k])\!=\!(\Dover[j]\CAP\Dover[k])$ where the equality follows from \LemRef{Lem2}. As a result, when $\Dover[j]\CAP\Dover[k]\NotEqualEmpty$, we also have $\Sover[i]\CAP\Dover[j]\NotEqualEmpty$.  \renewcommand{\IEEEQED}{} \end{proof}
\begin{proof}[Proof of $\Leftarrow$] Consider three indices $i$, $j$, and $k$ taking distinct values in $\{1,2,3\}$. Suppose that $\Sover[i]\cap\Dover[j]\NotEqualEmpty$ and $\Sover[i]\cap\Dover[k]\NotEqualEmpty$. Then, for any $e'\!\in\!\Sover[i]\cap\Dover[j]$ and any $e''\!\in\!\Sover[i]\cap\Dover[k]$, we must have either $e'\PREC e''$, $e'\SUCC e''$, or $e'\!=\!e''$ by \LemRef{Lem1}. Suppose that $\Dover[j]\CAP\Dover[k]\EqualEmpty$. Then we must have $e'\!\neq\!e''$, which leaves only two possibilities: either $e'\PREC e''$ or $e'\SUCC e''$. However, $e'\PREC e''$ contradicts \LemRef{Lem3} because $e'\!\in\!(\Sover[i]\CAP\Dover[j])\!\subset\!\Dover[j]$ and $e''\!\in\!(\Sover[i]\CAP\Dover[k])\!\subset\!(\Sover[i]\backslash\Dover[j])$, the latter of which is due to the assumption of $\Dover[j]\CAP\Dover[k]\EqualEmpty$. By swapping the roles of $j$ and $k$, one can also show that it is impossible to have $e'\SUCC e''$. By contradiction, we must have $\Dover[j]\cap\Dover[k]\NotEqualEmpty$. The proof is thus complete.
\end{proof}

We prove \LemRef{Lem5} as follows.
\begin{proof} Without loss of generality, consider $i\!=\!1$ and $j\!=\!2$. Note that by \LemRef{Lem1} any $e'\!\in\!\Sover[1]\CAP\Sover[2]$ and any $e''\!\in\!\Dover[1] \CAP\Dover[2]$ must satisfy either $e'\PREC e''$, $e'\SUCC e''$, or $e'\!=\!e''$. For the following, we prove this lemma by contradiction.

Suppose that there exists an edge $e''_\ast\!\in\!\Dover[1]\CAP\Dover[2]$ such that for all $e'\!\in\!\Sover[1]\CAP\Sover[2]$ we have $e''_\ast \PREC  e'$. For the following, we first prove that any path from $s_{i}$ to $d_{j}$ where $i,j\!\in\!\{1,2,3\}$ and $i\neq j$ must pass through $e''_\ast$. To that end, we first notice that by the definition of $\Dover[1]$ and $\Dover[2]$ and by the assumption $e''_\ast\!\in\!\Dover[1]\cap\Dover[2]$, any path from $\{s_2,s_3\}$ to $d_1$, and any path from $\{s_1,s_3\}$ to $d_2$ must use $e''_\ast$. Thus, we only need to prove that any path from $\{s_1,s_2\}$ to $d_3$ must use $e''_\ast$ as well.

Suppose there exists a \FromTo[1][3] path $P_{13}$ that does not use $e''_\ast$. By the definition of $\Sover[1]$, $P_{13}$ must use all edges of $\Sover[1]\cap\Sover[2]$, all of which are in the downstream of $e''_\ast$ by the assumption. Also $d_2$ is reachable from any $e'\!\in\!\Sover[1]\cap\Sover[2]$. Choose arbitrarily one edge $e'_\ast\!\in\!\Sover[1]\cap\Sover[2]$ and a path $P_{\head[e'_\ast]d_2}$ from $\head[e'_\ast]$ to $d_2$. Then, we can create an path $P_{13}\,e'_\ast\,P_{\head[e'_\ast]d_2}$ from $s_1$ to $d_2$ without using $e''_\ast$. The reason is that $P_{13}$ does not use $e''_\ast$ by our construction and $e'_\ast P_{\head[e'_\ast]d_2}$ does not use $e''_\ast$ since $e''_\ast\PREC e'_\ast$. Such an \FromTo[1][2] path not using $e''_\ast$ thus contradicts the assumption of $e''_\ast\!\in\!(\Dover[1]\CAP\Dover[2])\!\subset\!\onecut[s_1][d_2]$. Symmetrically, any \FromTo[2][3] path must use $e''_\ast$.

In summary, we have shown that $e''_\ast\!\in\!\cap_{i=1}^{3}\left(\Sover[i]\CAP\Dover[i]\right)$. However, this contradicts the assumption that $e''_\ast$ is in the upstream of all $e'\!\in\!\Sover[1]\cap\Sover[2]$, because we can simply choose $e'\!=\!e''_\ast\!\in\!\cap_{i=1}^{3}\left(\Sover[i]\CAP\Dover[i]\right)\!\subset\!(\Sover[1]\CAP\Sover[2])$ and $e''_\ast$ cannot be an upstream edge of itself $e'\!=\!e''_\ast$. The proof is thus complete.
\end{proof}

We prove \LemRef{Lem6} as follows.
\begin{proof} Without loss of generality, let $i\!=\!1$, $j_1\!=\!1$, $j_2\!=\!2$, and $j_3\!=\!3$. Suppose that $\Sover[1][;\{1,2\}]\NotEqualEmpty$ and $\Sover[1][;\{1,3\}]\NotEqualEmpty$. For the following, we prove this lemma by contradiction. 

Suppose that $\Sover[1][;\{1,2\}]\CAP\Sover[1][;\{1,3\}]\EqualEmpty$. For any $e'\!\in\!\Sover[1][;\{1,2\}]$ and any $e''\!\in\!\Sover[1][;\{1,3\}]$, since both $e'$ and $e''$ are $1$-edge cuts separating $s_1$ and $d_1$, it must be either $e'\PREC e''$ or $e'\SUCC e''$, or $e'\!=\!e''$. The last case is not possible since we assume $\Sover[1][;\{1,2\}]\cap\Sover[1][;\{1,3\}]\EqualEmpty$. Consider the most downstream edges $e'_\ast\!\in\!\Sover[1][;\{1,2\}]$ and $e''_\ast\!\in\!\Sover[1][;\{1,3\}]$, respectively. We first consider the case $e'_\ast\PREC e''_\ast$. If all paths from $s_1$ to $d_3$ use $e'_\ast$, which, by definition, use $e''_\ast$, then $e'_\ast$ will belong to $\onecut[s_1][d_3]$, which contradicts the assumption that $\Sover[1][;\{1,2\}]\cap\Sover[1][;\{1,3\}]\EqualEmpty$. Thus, there exists a \FromTo[1][3] path $P_{13}$ using $e''_\ast$ but not $e'_\ast$. Then, $s_1$ can follow $P_{13}$ and reach $d_1$ via $e''_\ast$ without using $e'_\ast$. Such a \FromTo[1][1] path contradicts the definition $e'_\ast\!\in\!\Sover[1][;\{1,2\}]\!\subset\!\onecut[s_1][d_1]$. Therefore, it is impossible to have $e'_\ast\PREC e''_\ast$. By symmetric arguments, it is also impossible to have $e'_\ast \SUCC e''_\ast$. By definition, any edge in $\Sover[1][;\{1,2\}]\cap\Sover[1][;\{1,3\}]$ is a $1$-edge cut separating $s_1$ and $\{d_2,d_3\}$, which implies that $\Sover[1][;\{2,3\}]\NotEqualEmpty$ and $\Sover[1]\NotEqualEmpty$.
\end{proof}

We prove \LemRef{Lem7} as follows.
\begin{proof}[Proof of $\Rightarrow$] Suppose $\Sover[i][;\{j_1,j_2\}]\NotEqualEmpty$. By definition, there exists an edge $e\!\in\!\onecut[s_i][d_{j_1}]\cap\onecut[s_i][d_{j_2}]$ in the downstream of the $s_i$-source edge $e_{s_i}$. Then, the channel gains $\ChGANA[j_1][i]$ and $\ChGANA[j_2][i]$ have a common factor $\ChG[e][e_{s_i}]$ and we thus have $\GCD[{\ChGANA[j_1][i]}][{\,\ChGANA[j_2][i]}]\,\PolyNotEqual 1$. \renewcommand{\IEEEQED}{} \end{proof}


\begin{proof}[Proof of $\Leftarrow$] We prove this direction by contradiction. Suppose $\GCD[{\ChGANA[j_1][i]}][{\ChGANA[j_2][i]}]\,\PolyNotEqual 1$. By \CorRef{Cor2}, we know that $\GCD[{\ChGANA[j_1][i]}][{\,\ChGANA[j_2][i]}]$ must not be $\ChGANA[j_1][i]$ nor $\ChGANA[j_2][i]$. Thus, both must be reducible and by \PropRef{Prop3} can be expressed as the product of irreducibles, for which each factor corresponds to the consecutive $1$-edge cuts in $\onecut[s_i][d_{j_1}]$ and $\onecut[s_i][d_{j_2}]$, respectively. Since they have at least one common irreducible factor, there exists an edge $e\!\in\!\onecut[s_i][d_{j_1}]\cap\onecut[s_i][d_{j_2}]$ in the downstream of the $s_i$-source edge $e_{s_i}$. Thus, $e\!\in\!\Sover[i][;\{j_1,j_2\}]$. The case for $\GCD[{\ChGANA[i][j_1]}][{\,\ChGANA[i][j_2]}]\,\PolyEqual 1$ can be proven symmetrically. The proof is thus complete.
\end{proof}

\newcommand{\PageGeneralNecessityProof}{\thepage}
\section{General Structured Proof for the Necessity}\label{GeneralNecessityProof}

In this appendix, we provide \CorRef{Cor4}, which will be used to prove the graph-theoretic necessary direction of 3-unicast ANA network for arbitrary $\Nid$ values. Since we already provided the proof for ``\LNR\AND\GS[1]\OPPLY\HS[1]" in \PropRef{Prop5}, here we focus on proving ``\LNR\AND\GS[1]\AND\GS[2]\OPPLY\HS[2], \HStilde[1], \HStilde[2]". After introducing \CorRef{Cor4}, the main proof of  "\LNR\AND\GS[1]\AND\GS[2]\OPPLY\HS[2], \HStilde[1], \HStilde[2]" will be provided in \AppRef{Prop5<=Proof}.

Before proceeding, we need the following additional logic statements to describe the general proof structure.

\subsection{The first set of logic statements}

Consider the following logic statements.

\noindent $\bullet$ \GS[0]{\bf:}\label{ref:GS[0]} $\ChGANA[1][1]\ChGANA[3][2]\ChGANA[2][3]=\BOLDb+\BOLDa$.

\noindent $\bullet$ \GS[3]{\bf:}\label{ref:GS[3]} $\Sover[2]\cap\Dover[3]\EqualEmpty$.

\noindent $\bullet$ \GS[4]{\bf:}\label{ref:GS[4]} $\Sover[3]\cap\Dover[2]\EqualEmpty$.

Several implications can be made when \GS[3] is true. We term those implications {\em the properties of \GS[3]}. Several properties of \GS[3] are listed as follows, for which their proofs are provided in \AppRef{ProofsG3G4}.

{\em Consider the case in which \GS[3] is true.} Use $e^\ast_2$ to denote the most downstream edge in $\onecut[s_2][d_1]\cap\onecut[s_2][d_3]$. Since the source edge $e_{s_2}$ belongs to both $\onecut[s_2][d_1]$ and $\onecut[s_2][d_3]$, such $e^\ast_2$ always exists. Similarly, use $e^\ast_3$ to denote the most upstream edge in $\onecut[s_1][d_3]\cap\onecut[s_2][d_3]$. The properties of \GS[3] can now be described as follows.

\noindent \makebox[3.3cm][l]{$\diamond$ {\bf Property~1 of \GS[3]:}}$e^\ast_2 \prec e^\ast_3$ and the channel gains $\ChGANA[3][1]$, $\ChGANA[1][2]$, and $\ChGANA[3][2]$ can be expressed as $\ChGANA[3][1] = \ChG[e^\ast_3][e_{s_1}]\ChG[e_{d_3}][e^\ast_3]$, $\ChGANA[1][2] = \ChG[e^\ast_2][e_{s_2}]\ChG[e_{d_1}][e^\ast_2]$, and $\ChGANA[3][2] = \ChG[e^\ast_2][e_{s_2}]\ChG[e^\ast_3][e^\ast_2]\ChG[e_{d_3}][e^\ast_3]$.

\noindent \makebox[3.3cm][l]{$\diamond$ {\bf Property~2 of \GS[3]:}}$\GCD[{\ChG[e^\ast_3][e_{s_1}]}][\;{\ChG[e^\ast_2][e_{s_2}]\ChG[e^\ast_3][e^\ast_2]}]\PolyEqual 1$,
$\GCD[{\ChG[e^\ast_3][e^\ast_2]\ChG[e_{d_3}][e^\ast_3]}][\;{\ChG[e_{d_1}][e^\ast_2]}]\PolyEqual 1$, $\GCD[{\ChGANA[3][1]}][\;{\ChG[e^\ast_3][e^\ast_2]}]\PolyEqual 1$, and $\GCD[{\ChGANA[1][2]}][\;{\ChG[e^\ast_3][e^\ast_2]}]\PolyEqual 1$.

On the other hand, when \GS[3] is false, or equivalently when \NotGS[3] is true where ``$\neg$" is the NOT operator, we can also derive several implications, termed {\em the properties of \NotGS[3]}.

{\em Consider the case in which \GS[3] is false}.  Use $e^{23}_u$ (resp.\ $e^{23}_v$) to denote the most upstream (resp.\ the most downstream) edge in $\Sover[2]\cap\Dover[3]$. By definition, it must be $e^{23}_u \PRECEQ e^{23}_v$. We now describe the following properties of \NotGS[3].

\noindent \makebox[3.45cm][l]{$\diamond$ {\bf Property~1 of \NotGS[3]:}}The channel gains $\ChGANA[3][1]$, $\ChGANA[1][2]$, and $\ChGANA[3][2]$ can be expressed as $\ChGANA[3][1] = \ChG[e^{23}_u][e_{s_1}]\ChG[e^{23}_v][e^{23}_u]\ChG[e_{d_3}][e^{23}_v]$, $\ChGANA[1][2]=$ $\ChG[e^{23}_u][e_{s_2}]\ChG[e^{23}_v][e^{23}_u]\ChG[e_{d_1}][e^{23}_v]\!,\!\!\text{ and }\!\ChGANA[3][2]\!\!=\! \ChG[e^{23}_u][e_{s_2}]\ChG[e^{23}_v][e^{23}_u]\ChG[e_{d_3}][e^{23}_v]\!$.

\noindent \makebox[3.6cm][l]{$\diamond$~{\bf Property~2 of \NotGS[3]:}}$\GCD[{\ChG[e^{23}_u][e_{s_1}]}][\;{\ChG[e^{23}_u][e_{s_2}]}]\PolyEqual 1$ and $\GCD[{\ChG[e_{d_1}][e^{23}_v]}][\;{\ChG[e_{d_3}][e^{23}_v]}]\,\PolyEqual 1$.

\noindent \makebox[3.6cm][l]{$\diamond$~{\bf Property~3 of \NotGS[3]:}}$\{e^{23}_u,e^{23}_v\}\!\subset\!\onecut[s_1][{\head[e^{23}_v]}]$ and $\{e^{23}_u,e^{23}_v\}\!\subset\!\onecut[{\tail[e^{23}_u]}][d_1]$. This further implies that for any \FromTo[1][1] path $P$, if there exists a vertex $w\in P$ satisfying $\tail[e^{23}_u]\PRECEQ w\PRECEQ \head[e^{23}_v]$, then we must have $\{e^{23}_u, e^{23}_v\}\subset P$.

Symmetrically, we define the following properties of \GS[4] and \NotGS[4].

{\em Consider the case in which \GS[4] is true.} Use $e^\ast_3$ to denote the most downstream edge in $\onecut[s_3][d_1]\cap\onecut[s_3][d_2]$, and use $e^\ast_2$ to denote the most upstream edge in $\onecut[s_1][d_2]\cap\onecut[s_3][d_2]$. We now describe the following properties of \GS[4].

\noindent \makebox[3.3cm][l]{$\diamond$~{\bf Property~1 of \GS[4]:}}$e^\ast_3 \prec e^\ast_2$ and
the channel gains $\ChGANA[2][1]$, $\ChGANA[1][3]$, and $\ChGANA[2][3]$ can be expressed as $\ChGANA[2][1] = \ChG[e^\ast_2][e_{s_1}]\ChG[e_{d_2}][e^\ast_2]$, $\ChGANA[1][3] = \ChG[e^\ast_3][e_{s_3}]\ChG[e_{d_1}][e^\ast_3]$, and $\ChGANA[2][3] = \ChG[e^\ast_3][e_{s_3}]\ChG[e^\ast_2][e^\ast_3]\ChG[e_{d_2}][e^\ast_2]$.

\noindent \makebox[3.3cm][l]{$\diamond$~{\bf Property~2 of \GS[4]:}}$\GCD[{\ChG[e^\ast_2][e_{s_1}]}][\;{\ChG[e^\ast_3][e_{s_3}]\ChG[e^\ast_2][e^\ast_3]}]\PolyEqual 1$,
$\GCD[{\ChG[e^\ast_2][e^\ast_3]\ChG[e_{d_2}][e^\ast_2]}][\;{\ChG[e_{d_1}][e^\ast_3]}]\PolyEqual 1$, $\GCD[{\ChGANA[2][1]}][\;{\ChG[e^\ast_2][e^\ast_3]}]\PolyEqual 1$, and $\GCD[{\ChGANA[1][3]}][\;{\ChG[e^\ast_2][e^\ast_3]}]\PolyEqual 1$.

{\em Consider the case in which \GS[4] is false.} Use $e^{32}_u$ (resp.\ $e^{32}_v$) to denote the most upstream (resp.\ the most downstream) edge in $\Sover[3]\cap\Dover[2]$. By definition, it must be $e^{32}_u \PRECEQ e^{32}_v$. We now describe the following properties of \NotGS[4].

\noindent \makebox[3.45cm][l]{$\diamond$~{\bf Property~1 of \NotGS[4]:}}The channel gains $\ChGANA[2][1]$, $\ChGANA[1][3]$, and $\ChGANA[2][3]$ can be expressed as $\ChGANA[2][1] = \ChG[e^{32}_u][e_{s_1}]\ChG[e^{32}_v][e^{32}_u]\ChG[e_{d_2}][e^{32}_v]$, $\ChGANA[1][3]=$ $\ChG[e^{32}_u][e_{s_3}]\ChG[e^{32}_v][e^{32}_u]\ChG[e_{d_1}][e^{32}_v]\!,\!\!\text{ and }\!\ChGANA[2][3]\!\!=\! \ChG[e^{32}_u][e_{s_3}]\ChG[e^{32}_v][e^{32}_u]\ChG[e_{d_2}][e^{23}_v]\!$.

\noindent \makebox[3.6cm][l]{$\diamond$~{\bf Property~2 of \NotGS[4]:}}$\GCD[{\ChG[e^{32}_u][e_{s_1}]}][\;{\ChG[e^{32}_u][e_{s_3}]}]\PolyEqual 1$ and $\GCD[{\ChG[e_{d_1}][e^{32}_v]}][\;{\ChG[e_{d_2}][e^{32}_v]}]\,\PolyEqual 1$

\noindent \makebox[3.6cm][l]{$\diamond$~{\bf Property~3 of \NotGS[4]:}}$\{e^{32}_u,e^{32}_v\}\!\subset\!\onecut[s_1][{\head[e^{32}_v]}]$ and $\{e^{32}_u,e^{32}_v\}\!\subset\!\onecut[{\tail[e^{32}_u]}][d_1]$. This further implies that for any \FromTo[1][1] path $P$, if there exists a vertex $w\in P$ satisfying $\tail[e^{32}_u]\PRECEQ w\PRECEQ \head[e^{32}_v]$, then we must have $\{e^{32}_u, e^{32}_v\}\subset P$.

The following logic statements are well-defined if and only if (\NotGS[3])\AND(\NotGS[4]) is true. Recall the definition of $e^{23}_u$, $e^{23}_v$, $e^{32}_u$, and $e^{32}_v$ when (\NotGS[3])\AND(\NotGS[4]) is true.

\noindent $\bullet$ \GS[5]{\bf:}\label{ref:GS[5]} Either $e^{23}_u \PREC\,e^{32}_u$ or $e^{23}_u \SUCC e^{32}_u$. 

\noindent $\bullet$ \GS[6]{\bf:}\label{ref:GS[6]} Any vertex $w'$ where $\tail[e^{23}_u]\PRECEQ w' \PRECEQ\head[e^{23}_v]$ and any vertex $w''$ where $\tail[e^{32}_u]\PRECEQ w'' \PRECEQ\head[e^{32}_v]$ are not reachable from each other. (That is, neither $w'\PRECEQ w''$ nor $w''\PRECEQ w'$.)


It is worth noting that a {\em statement} being well-defined does not mean that it is true. Any well-defined logic statement can be either true or false. For comparison, a {\em property} of \GS[3] is both well-defined and true whenever \GS[3] is true.

\subsection{General Necessity Proof Structure}

The following ``logic relationships" are proved in \AppRef{ProofsNS[1]-NS[9]}, which will be useful for the proof of the following \CorRef{Cor4}.

\noindent $\bullet$ \NS[1]{\bf:}\label{ref:NS[1]} \HS[2]\IMPLY \LNR\AND\GS[1].

\noindent $\bullet$ \NS[2]{\bf:}\label{ref:NS[2]} \HStilde[1]\IMPLY \LNR\AND\GS[1].

\noindent $\bullet$ \NS[3]{\bf:}\label{ref:NS[3]} \HStilde[2]\IMPLY \LNR\AND\GS[1].

\noindent $\bullet$ \NS[4]{\bf:}\label{ref:NS[4]} (\NotGS[2])\AND\GS[3]\AND\GS[4]\IMPLY \CONT.

\noindent $\bullet$ \NS[5]{\bf:}\label{ref:NS[5]} \GS[1]\AND(\NotGS[2])\AND(\NotGS[3])\AND\GS[4]\IMPLY \CONT.

\noindent $\bullet$ \NS[6]{\bf:}\label{ref:NS[6]} \GS[1]\AND(\NotGS[2])\AND\GS[3]\AND(\NotGS[4])\IMPLY \CONT.

\noindent $\bullet$ \NS[7]{\bf:}\label{ref:NS[7]} \LNR\AND(\NotGS[3])\AND(\NotGS[4])\AND(\NotGS[5])\IMPLY \GS[6].

\noindent $\bullet$ \NS[8]{\bf:}\label{ref:NS[8]} \GS[1]\AND(\NotGS[2])\AND(\NotGS[3])\AND(\NotGS[4])\AND\GS[5]\IMPLY \CONT.

\noindent $\bullet$ \NS[9]{\bf:}\label{ref:NS[9]} (\NotGS[2])\AND(\NotGS[3])\AND(\NotGS[4])\AND(\NotGS[5])\AND\GS[6]\IMPLY \GS[0].


\begin{cor}\label{Cor4} \quad Let $\GSP(\netvar)$ be a set of (arbitrarily chosen) polynomials based on the $9$ channel gains $\ChGANA[j][i]$ of the 3-unicast ANA network, and define \XS[] to be the logic statement that $\GSP(\netvar)$ is linearly independent. If we can prove that \XS[]\IMPLY\LNR\AND\GS[1]\\\noindent and \XS[]\AND\GS[0]\IMPLY\CONT, then the logic relationship \XS[]\IMPLY\LNR\AND\\\noindent\GS[1]\AND\GS[2] must hold.
\end{cor}

\begin{proof} Suppose \XS[]\IMPLY \LNR\AND\GS[1] and \XS[]\AND\GS[0]\IMPLY\CONT. We first see that \NS[7] and \NS[9] jointly imply
\begin{equation*}\label{Nec-1}
\textbf{LNR}\wedge(\neg\,\textbf{G2})\wedge(\neg\,\textbf{G3})\wedge(\neg\,\textbf{G4})\wedge(\neg\,\textbf{G5})\Rightarrow\,\textbf{G0}.
\end{equation*}

Combined with \NS[8], we thus have
\begin{equation*}\label{Nec-2}
\textbf{LNR}\wedge\textbf{G1}\wedge(\neg\,\textbf{G2})\wedge(\neg\,\textbf{G3})\wedge(\neg\,\textbf{G4})\Rightarrow\,\textbf{G0}.
\end{equation*}

This, jointly with \NS[4], \NS[5], and \NS[6], further imply
\begin{equation*}\label{Nec-3}
\textbf{LNR}\wedge\textbf{G1}\wedge(\neg\,\textbf{G2})\Rightarrow\,\textbf{G0}.
\end{equation*}

Together with the assumption that \XS[]\AND\GS[0]\IMPLY\CONT, we have \XS[]\AND\LNR\AND\GS[1]\AND(\NotGS[2])\IMPLY\CONT. Combining with the assumption that \XS[]\IMPLY\LNR\AND\GS[1] then yields
\begin{equation*}\label{Nec-4}
\textbf{X}\wedge(\neg\,\textbf{G2})\Rightarrow\,\text{false},
\end{equation*}
which equivalently implies that \XS[]\IMPLY\GS[2]. The proof is thus complete. 
\end{proof}

\section{The Reference Table}\label{reftable}

For the ease of exposition, we provide the Table~\ref{tb:reftable}, the reference table. The reference table helps finding where to look for the individual logic statements and relationships for the entire proof of \PropRef{Prop5}.
\begin{table}[h]
\hfill{}
\begin{tabular}{|cl|cl|}
\multicolumn{4}{c}{\small The Logic Statements for the Proof of \PropRef{Prop5}} \\
\hline
{\footnotesize\CS[0] to \CS[6]} & {\footnotesize defined in p.~\pageref{ref:CS[0]}.} & {\footnotesize\GS[7] to \GS[15]} & {\footnotesize defined in p.~\pageref{ref:GS[7]}.} \\[-2pt]
{\footnotesize\DS[1] to \DS[6]} & {\footnotesize defined in p.~\pageref{ref:DS[1]}.} & {\footnotesize\GS[16] to \GS[26]} & {\footnotesize defined in p.~\pageref{ref:GS[16]}.} \\[-2pt]
{\footnotesize\ES[0] to \ES[2]} & {\footnotesize defined in p.~\pageref{ref:ES[0]}.} & {\footnotesize\GS[27] to \GS[31]} & {\footnotesize defined in p.~\pageref{ref:GS[27]}.} \\[-2pt]
{\footnotesize\GS[0]} & {\footnotesize defined in p.~\pageref{ref:GS[0]}.} & {\footnotesize\GS[32] to \GS[36]} & {\footnotesize defined in p.~\pageref{ref:GS[32]}.} \\[-2pt]
{\footnotesize\GS[1], \GS[2]} & {\footnotesize defined in p.~\pageref{ref:GS[1]}.} & {\footnotesize\GS[37] to \GS[43]} & {\footnotesize defined in p.~\pageref{ref:GS[37]}.} \\[-2pt]
{\footnotesize\GS[3], \GS[4]} & {\footnotesize defined in p.~\pageref{ref:GS[3]}.} & {\footnotesize\HS[1], \HS[2], \HStilde[1], \HStilde[2]} & {\footnotesize defined in p.~\pageref{ref:HS[1]}.} \\[-2pt]
{\footnotesize\GS[5], \GS[6]} & {\footnotesize defined in p.~\pageref{ref:GS[5]}.} & {\footnotesize\LNR} & {\footnotesize defined in p.~\pageref{ref:LNR}.} \\
\hline
\end{tabular}
\hfill{}
\end{table}
\begin{table}[h]
\begin{tabular}{|cl|}
\multicolumn{2}{c}{\small The Logic Relationships for the Proof of \PropRef{Prop5}} \\
\hline
& \\[-9pt]
{\footnotesize\NS[1] to \NS[9]} & {\parbox{6.4cm}{\footnotesize defined in p.~\pageref{ref:NS[1]}, to help proving \CorRef{Cor4}, the general structured proof for the necessity of \PropRef{Prop5}.}} \\[+2pt]
{\footnotesize\RS[1] to \RS[10]} & {\footnotesize defined in p.~\pageref{ref:RS[1]}, to help proving \SS[11].} \\
{\footnotesize\RS[11] to \RS[25]} & {\footnotesize defined in p.~\pageref{ref:RS[11]}, to help proving \SS[13].} \\
{\footnotesize\RS[26] to \RS[33]} & {\footnotesize defined in p.~\pageref{ref:RS[26]}, to help proving \SS[14].} \\
{\footnotesize\RS[34] to \RS[40]} & {\footnotesize defined in p.~\pageref{ref:RS[34]}, to help proving \RS[28].} \\
{\footnotesize\RS[41] to \RS[47]} & {\footnotesize defined in p.~\pageref{ref:RS[41]}, to help proving \RS[29].} \\[+2pt]
{\footnotesize\SS[1] to \SS[14]} & {\parbox{6.4cm}{\footnotesize defined in p.~\pageref{ref:SS[1]}, to help proving \CorRef{Cor5}, the general structured proof for the sufficiency of \PropRef{Prop5}.}} \\[3pt]
\hline
\end{tabular}
\caption{The reference table for the proof of \PropRef{Prop5}.}
\label{tb:reftable}
\end{table}

\section{Proof of ``\LNR\AND\GS[1]\AND\GS[2]\OPPLY\HStilde[1]\OR\HS[2]\OR\HStilde[2]"}\label{Prop5<=Proof}

Thanks to \CorRef{Cor4} and the logic relationships \NS[1], \NS[2], and \NS[3] in \AppRef{GeneralNecessityProof}, we only need to show that (i) \HStilde[1]\AND\GS[0]\\\noindent\IMPLY\CONT; (ii) \HS[2]\AND\GS[0]\IMPLY\CONT; and (iii) \HStilde[2]\AND\GS[0]\IMPLY\CONT.

We prove ``\HStilde[1]\AND\GS[0]\IMPLY\CONT" as follows.
\begin{proof} We prove an equivalent form: \GS[0]\IMPLY(\NotHStilde[1]). Suppose \GS[0] is true. Consider $\GSPtilde[1][\left(\!1\!\right)](\netvar)$ which contains 3 polynomials (see \Ref{GSP2^n_NEW} when $\Nid\!=\!1$):
\begin{equation}\label{GSP2^1_NEW}
\GSPtilde[1][\left(\!1\!\right)](\netvar)=\{\,\ChGANA[1][1]\ChGANA[3][2]\ChGANA[1][3]\BOLDa, \;\ChGANA[1][2]\ChGANA[3][1]\ChGANA[1][3]\BOLDa, \;\ChGANA[1][2]\ChGANA[3][1]\ChGANA[1][3]\BOLDb\,\}.
\end{equation}

Since $\BOLDa=\ChGANA[3][1]\ChGANA[2][3]\ChGANA[1][2]$, the first polynomial in $\GSPtilde[1][\left(\!1\!\right)](\netvar)$ is equivalent to $\ChGANA[1][1]\ChGANA[3][2]\ChGANA[2][3]\ChGANA[1][2]\ChGANA[3][1]\ChGANA[1][3]$. Then $\GSPtilde[1][\left(\!1\!\right)](\netvar)$ becomes linearly dependent by substituting $\BOLDb + \BOLDa$ for $\ChGANA[1][1]\ChGANA[3][2]\ChGANA[2][3]$ (from \GS[0] being true). The proof is thus complete. 
\end{proof}

We prove ``\HS[2]\AND\GS[0]\IMPLY\CONT" as follow.
\begin{proof} We prove an equivalent form: \GS[0]\IMPLY(\NotHS[2]). Suppose \GS[0] is true. Consider $\GSP[1][\left(\!\Nid\!\right)](\netvar)$ in \Ref{GSP1^n}. Substituting $\BOLDb + \BOLDa$ for $\ChGANA[1][1]\ChGANA[3][2]\ChGANA[2][3]$ (from \GS[0] being true) and $\BOLDa=\ChGANA[1][2]\ChGANA[3][1]\ChGANA[2][3]$ to the expression of $\GSP[1][\left(\!\Nid\!\right)](\netvar)$, then we have
\begin{equation*}\begin{split}
\GSP[1][\NidBRAC](\netvar) &= \{\,(\BOLDb + \BOLDa)\BOLDb^{\Nid},\;(\BOLDb + \BOLDa)\BOLDb^{\Nid-1}\BOLDa,\;\cdots,\;(\BOLDb + \BOLDa)\BOLDa^{\Nid}, \\
& \qquad ,\;\BOLDb^{\Nid}\BOLDa,\;\BOLDb^{\Nid-1}\BOLDa^2,\;\cdots,\;\BOLDb\BOLDa^{\Nid}\,\}.
\end{split}\end{equation*}

One can see that $\GSP[1][\left(\!n\!\right)](\netvar)$ becomes linearly dependent when $\Nid\!\geq\!2$. The proof is thus complete.
\end{proof}

We prove ``\HStilde[2]\AND\GS[0]\IMPLY\CONT" as follow.
\begin{proof} Similarly following the proof of ``\HStilde[1]\AND\GS[0]\IMPLY\\\CONT", we further have
\begin{equation*}\begin{split}
&\GSPtilde[1][\NidBRAC](\netvar) = \ChGANA[1][2]\ChGANA[3][1]\ChGANA[1][3]\,\{\,(\BOLDb+\BOLDa)\BOLDa^{\Nid-1},\;(\BOLDb+\BOLDa)\BOLDa^{\Nid-2}\BOLDb, \\
& \qquad \;\cdots\;,(\BOLDb+\BOLDa)\BOLDb^{\Nid-1},\;\BOLDa^{\Nid},\; \BOLDa^{\Nid-1}\BOLDb,\;\cdots\;,\BOLDa\BOLDb^{\Nid-1},\;\BOLDb^{\Nid}\,\},
\end{split}\end{equation*}
which becomes linearly dependent when $\Nid\!\geq\!2$. The proof is thus complete.
\end{proof}


\section{General Structured Proof for the Sufficiency}\label{GeneralSufficiencyProof}

In this appendix, we provide \CorRef{Cor5}, which will be used to prove the graph-theoretic sufficient direction of 3-unicast ANA network for arbitrary $\Nid\!>\!0$ values. We need the following additional logic statements to describe the general proof structure.

\subsection{The second set of logic statements}

Given a 3-unicast ANA network $\GANA$, recall the definitions $\BOLDa=\ChGANA[3][1]\ChGANA[2][3]\ChGANA[1][2]$ and $\BOLDb=\ChGANA[2][1]\ChGANA[3][2]\ChGANA[1][3]$ (we drop the input argument $\netvar$ for simplicity). By the definition of $\GANA$, any channel gains are non-trivial, and thus $\BOLDb$ and $\BOLDa$ are non-zero polynomials. Let  $\psi^{(\!\Nid\!)}_\alpha(\BOLDb,\BOLDa)$ and $\psi^{(\!\Nid\!)}_\beta(\BOLDb,\BOLDa)$ to be some polynomials with respect to $\netvar$, represented by 
\begin{align*}
\psi^{(\!\Nid\!)}_\alpha(\BOLDb,\BOLDa) \!=\! \sum_{i=0}^{\Nid} \alpha_i \BOLDb^{\Nid-i}\BOLDa^{i},\quad
\psi^{(\!\Nid\!)}_\beta(\BOLDb,\BOLDa) \!=\! \sum_{j=0}^{\Nid} \beta_j \BOLDb^{\Nid-j}\BOLDa^{j},
\end{align*}
with some set of coefficients $\{\alpha_i\}_{i=0}^{\Nid}$ and $\{\beta_j\}_{j=0}^{\Nid}$, respectively. Basically, given a value of $n$ and the values of $\{\alpha_i\}_{i=0}^{\Nid}$ and $\{\beta_j\}_{j=0}^{\Nid}$, $\psi^{(\!\Nid\!)}_\alpha(\BOLDb,\BOLDa)$ (resp. $\psi^{(\!\Nid\!)}_\beta(\BOLDb,\BOLDa)$) represents a linear combination of $\{ \BOLDb^{\Nid},\,\BOLDb^{\Nid-1}\BOLDa,\,\cdots\,,\,\BOLDb\BOLDa^{\Nid-1},\,\BOLDa^{\Nid}\}$, the set of Vandermonde polynomials

We need the following additional logic statements.

\noindent \makebox[1cm][l]{$\bullet$ \ES[0]{\bf:}}\label{ref:ES[0]}Let $I_{\textrm{3ANA}}$ be a finite index set defined by $I_{\textrm{3ANA}}=\{(i,j):\;i,j\!\in\!\{1,2,3\}\;\textrm{and}\;i\neq j\}$. Consider two non-zero polynomial functions $f: \FF[\!q][|I_{\textrm{3ANA}}|]\!\mapsto\!\FF[\!q]$ and $g: \FF[\!q][|I_{\textrm{3ANA}}|]\!\mapsto\!\FF[\!q]$. Then given a $\GANA$ of interest, there exists some coefficient values $\{\alpha_i\}_{i=0}^{\Nid}$ and $\{\beta_j\}_{j=0}^{\Nid}$ such that
\begin{equation*}\begin{split}
\ChGANA[1][1]&\,f(\{\ChGANA[j][i]:\forall\,(i,j)\!\in\!I_{\textrm{3ANA}}\})\,\psi^{(\!\Nid\!)}_\alpha(\BOLDb,\BOLDa) \\
& \qquad\;\; = \; g(\{\ChGANA[j][i]:\forall\,(i,j)\!\in\!I_{\textrm{3ANA}}\})\,\psi^{(\!\Nid\!)}_\beta(\BOLDb,\BOLDa),
\end{split}\end{equation*}
\noindent with (i) At least one of coefficients $\{\alpha_i\}_{i=0}^{\Nid}$ is non-zero; and (ii) At least one of coefficients $\{\beta_j\}_{j=0}^{\Nid}$ is non-zero.

Among $\{\alpha_{i}\}_{i=0}^{\Nid}$ and $\{\beta_{j}\}_{j=0}^{\Nid}$, define $\ist$ (resp. $\jst$) as the smallest $i$ (resp. $j$) such that $\alpha_{i}\neq 0$ (resp. $\beta_{j}\neq 0$). Similarly, define $\iend$ (resp. $\jend$) as the largest $i$ (resp. $j$) such that $\alpha_{i}\neq 0$ (resp. $\beta_{j}\neq 0$).\footnote{From definition, $0\!\leq\!\ist\!\leq\!\iend\!\leq\!\Nid$ and $0\!\leq\!\jst\!\leq\!\jend\!\leq\!\Nid$.} Then, we can rewrite the above equation as follows:
\begin{equation}\label{E0}\begin{split}
\sum_{i=\ist}^{\iend}& \alpha_{i}\,\ChGANA[1][1]\,f(\{\ChGANA[j][i]:\forall\,(i,j)\!\in\!I_{\textrm{3ANA}}\})\,\BOLDb^{\Nid-i}\BOLDa^{i} \\
& = \; \sum_{j=\jst}^{\jend}\beta_{j}\,g(\{\ChGANA[j][i]:\forall\,(i,j)\!\in\!I_{\textrm{3ANA}}\})\,\BOLDb^{\Nid-j}\BOLDa^{j}.
\end{split}\end{equation}

\noindent \makebox[1cm][l]{$\bullet$ \ES[1]{\bf:}}\label{ref:ES[1]}Continue from the definition of \ES[0]. The considered $\GANA$ satisfies \Ref{E0} with (i) $f(\{\ChGANA[j][i]:\forall\,(i,j)\!\in\!I_{\textrm{3ANA}}\})=\ChGANA[3][2]$; and (ii) $g(\{\ChGANA[j][i]:\forall\,(i,j)\!\in\!I_{\textrm{3ANA}}\})=\ChGANA[3][1]\ChGANA[1][2]$. Then, \Ref{E0} reduces to
\begin{equation}\label{E1}
\sum_{i=\ist}^{\iend} \alpha_{i}\ChGANA[1][1]\ChGANA[3][2]\BOLDb^{\Nid-i}\BOLDa^{i} = \sum_{j=\jst}^{\jend}\beta_{j}\ChGANA[3][1]\ChGANA[1][2]\BOLDb^{\Nid-j}\BOLDa^{j}.
\end{equation}


\noindent \makebox[1cm][l]{$\bullet$ \ES[2]{\bf:}}\label{ref:ES[2]}Continue from the definition of \ES[0]. The chosen coefficients $\{\alpha_i\}_{i=0}^{\Nid}$ and $\{\beta_j\}_{j=0}^{\Nid}$ which satisfy \Ref{E0} in the given $\GANA$ also satisfy (i) $\alpha_k\!\neq\!\beta_k$ for some $k\!\in\!\{0,...,\Nid\}$; and (ii) either $\alpha_0\!\neq\!0$ or $\beta_{\Nid}\!\neq\!0$ or $\alpha_{k}\!\neq\!\beta_{k-1}$ for some $k\!\in\!\{1,...,\Nid\}$.

One can see that whether the above logic statements are true or false depends on the polynomials $\ChGANA[j][i]$ and on the $\{\alpha_i\}_{i=0}^{\Nid}$ and $\{\beta_j\}_{j=0}^{\Nid}$ values being considered. 

The following logic statements are well-defined if and only if \ES[0] is true. Whether the following logic statements are true depends on the values of $\ist$, $\iend$, $\jst$, and $\jend$.

\noindent \makebox[1cm][l]{$\bullet$ \CS[0]{\bf:}}\label{ref:CS[0]}$\ist\!>\!\jst$ and $\iend\!=\!\jend$.

\noindent \makebox[1cm][l]{$\bullet$ \CS[1]{\bf:}}\label{ref:CS[1]}$\ist\!<\!\jst$.

\noindent \makebox[1cm][l]{$\bullet$ \CS[2]{\bf:}}\label{ref:CS[2]}$\ist\!>\!\jst$.

\noindent \makebox[1cm][l]{$\bullet$ \CS[3]{\bf:}}\label{ref:CS[3]}$\ist\!=\!\jst$.

\noindent \makebox[1cm][l]{$\bullet$ \CS[4]{\bf:}}\label{ref:CS[4]}$\iend\!<\!\jend$.

\noindent \makebox[1cm][l]{$\bullet$ \CS[5]{\bf:}}\label{ref:CS[5]}$\iend\!>\!\jend$.

\noindent \makebox[1cm][l]{$\bullet$ \CS[6]{\bf:}}\label{ref:CS[6]}$\iend\!=\!\jend$.

We also define the following statements for the further organization.

\noindent \makebox[1cm][l]{$\bullet$ \DS[1]{\bf:}}\label{ref:DS[1]}$\GCD[\!{\ChGANA[2][1]^{l_1}\ChGANA[3][2]^{l_1}\ChGANA[1][3]^{l_1}}][{\,\ChGANA[2][3]}]\!=\!\ChGANA[2][3]$ for some integer $l_1\!\!>\!0$.

\noindent \makebox[1cm][l]{$\bullet$ \DS[2]{\bf:}}\label{ref:DS[2]}$\GCD[\!{\ChGANA[3][1]^{l_2}\ChGANA[2][3]^{l_2}\ChGANA[1][2]^{l_2}}][{\,\ChGANA[3][2]}]\!=\!\ChGANA[3][2]$ for some integer $l_2\!\!>\!0$.

\noindent \makebox[1cm][l]{$\bullet$ \DS[3]{\bf:}}\label{ref:DS[3]}$\GCD[{\ChGANA[1][1]\ChGANA[3][1]^{l_3}\ChGANA[2][3]^{l_3}\ChGANA[1][2]^{l_3}}][{\,\ChGANA[2][1]\ChGANA[1][3]}]\!=\!\ChGANA[2][1]\ChGANA[1][3]$ for some integer $l_3\!>\!0$.

\noindent \makebox[1cm][l]{$\bullet$ \DS[4]{\bf:}}\label{ref:DS[4]}$\GCD[{\ChGANA[1][1]\ChGANA[2][1]^{l_4}\ChGANA[3][2]^{l_4}\ChGANA[1][3]^{l_4}}][{\,\ChGANA[3][1]\ChGANA[1][2]}]\!=\!\ChGANA[3][1]\ChGANA[1][2]$ for some integer $l_4\!>\!0$.

\noindent \makebox[1cm][l]{$\bullet$ \DS[5]{\bf:}}\label{ref:DS[5]}$\GCD[{\ChGANA[1][1]\ChGANA[2][1]^{l_5}\ChGANA[3][2]^{l_5}\ChGANA[1][3]^{l_5}}][{\,\ChGANA[2][3]}]\!=\!\ChGANA[2][3]$
for some integer $l_5\!>\!0$.

\noindent \makebox[1cm][l]{$\bullet$ \DS[6]{\bf:}}\label{ref:DS[6]}$\GCD[{\ChGANA[1][1]\ChGANA[3][1]^{l_6}\ChGANA[2][3]^{l_6}\ChGANA[1][2]^{l_6}}][{\,\ChGANA[3][2]}]\!=\!\ChGANA[3][2]$ for some integer $l_6\!>\!0$.

\subsection{General Sufficiency Proof Structure}

We prove the following ``logic relationships," which will be used for the proof of \CorRef{Cor5}. 

\noindent $\bullet$ \SS[1]:\label{ref:SS[1]} \DS[1]\IMPLY\DS[5].

\noindent $\bullet$ \SS[2]:\label{ref:SS[2]} \DS[2]\IMPLY\DS[6].

\noindent $\bullet$ \SS[3]:\label{ref:SS[3]} \ES[0]\AND\ES[1]\AND\CS[1]\IMPLY\DS[4]\AND\DS[5].

\noindent $\bullet$ \SS[4]:\label{ref:SS[4]} \ES[0]\AND\ES[1]\AND\CS[2]\IMPLY\DS[1].

\noindent $\bullet$ \SS[5]:\label{ref:SS[5]} \GS[1]\AND\ES[0]\AND\ES[1]\AND\CS[3]\IMPLY\DS[4].

\noindent $\bullet$ \SS[6]:\label{ref:SS[6]} \ES[0]\AND\ES[1]\AND\CS[4]\IMPLY\DS[2]\AND\DS[3].

\noindent $\bullet$ \SS[7]:\label{ref:SS[7]} \ES[0]\AND\ES[1]\AND\CS[5]\IMPLY\DS[3].

\noindent $\bullet$ \SS[8]:\label{ref:SS[8]} \GS[1]\AND\ES[0]\AND\ES[1]\AND\CS[6]\IMPLY\DS[2].

\noindent $\bullet$ \SS[9]:\label{ref:SS[9]} \ES[0]\AND\ES[1]\AND\CS[0]\IMPLY\ES[2].

\noindent $\bullet$ \SS[10]:\label{ref:SS[10]} \GS[1]\AND\ES[0]\AND\ES[1]\AND(\NotCS[0])\IMPLY(\DS[1]\AND\DS[3])\OR(\DS[2]\AND\DS[4])\OR\\\noindent(\DS[3]\AND\DS[4]).

\noindent $\bullet$ \SS[11]:\label{ref:SS[11]} \LNR\AND\GS[1]\AND\ES[0]\AND\DS[1]\AND\DS[3]\IMPLY\CONT.

\noindent $\bullet$ \SS[12]:\label{ref:SS[12]} \LNR\AND\GS[1]\AND\ES[0]\AND\DS[2]\AND\DS[4]\IMPLY\CONT.

\noindent $\bullet$ \SS[13]:\label{ref:SS[13]} \LNR\AND\GS[1]\AND\GS[2]\AND\ES[0]\AND\ES[1]\AND\ES[2]\AND\DS[1]\AND\DS[2]\IMPLY\CONT.


\noindent $\bullet$ \SS[14]:\label{ref:SS[14]} \LNR\AND\GS[1]\AND\ES[0]\AND\DS[3]\AND\DS[4]\IMPLY\CONT.

The proofs of \SS[1] to \SS[10] are relegated to \AppRef{ProofsSS[1]-SS[10]}. The proofs of \SS[11] to \SS[14] are relegated to \AppssssRef{ProofSS[11]}{ProofSS[12]}{ProofSS[13]}{ProofSS[14]}, respectively. Note that the above \SS[1] to \SS[14] relationships greatly simplify the analysis of finding the graph-theoretic conditions for the feasibility of the 3-unicast ANA network. This observation is summarized in \CorRef{Cor5}.

\begin{cor}\label{Cor5} \quad Let $\GSP(\netvar)$ be a set of (arbitrarily chosen) polynomials based on the $9$ channel gains $\ChGANA[j][i]$ of the 3-unicast ANA network, and define \XS[] to be the logic statement that $\GSP(\netvar)$ is linearly independent. Let \GS[] to be an arbitrary logic statement in the 3-unicast ANA network. If we can prove that
\begin{itemize}
\item[(A)] \GS[]\AND(\NotXS[])\IMPLY \ES[0]\AND\ES[1]\AND(\NotCS[0]),
\end{itemize}
then the logic relationship \LNR\AND\GS[1]\AND\GS[]\AND(\NotXS[])\IMPLY\CONT\; must also hold.

Also, if we can prove that
\begin{itemize}
\item[(B)] \GS[]\AND(\NotXS[])\IMPLY \ES[0]\AND\ES[1]\AND\CS[0],
\end{itemize}
then the logic relationship \LNR\AND\GS[1]\AND\GS[2]\AND\GS[]\AND(\NotXS[])\IMPLY\\\noindent\CONT\;must also hold.
\end{cor}

\begin{proof} First, notice that \SS[11], \SS[12], and \SS[14] jointly imply
\begin{equation}\label{Suf-1}\begin{split}
\textbf{LNR}&\wedge\textbf{G1}\wedge\textbf{E0}\,\wedge \\
& \{(\textbf{D1}\wedge\textbf{D3})\vee(\textbf{D2}\wedge\textbf{D4})\vee(\textbf{D3}\wedge\textbf{D4})\}\Rightarrow\,\text{false}.
\end{split}\end{equation}

Then, \Ref{Suf-1}, jointly with \SS[10] further imply
\begin{equation}\label{Suf-2}
\textbf{LNR}\wedge\textbf{G1}\wedge\textbf{E0}\,\wedge\textbf{E1}\wedge(\neg\,\textbf{C0})\Rightarrow\,\text{false}.
\end{equation}

Note that by definition \CS[0] is equivalent to \CS[2]\AND\CS[6]. Then \SS[4] and \SS[8] jointly imply
\begin{equation}\label{Suf-3}
\textbf{G1}\wedge\textbf{E0}\,\wedge\textbf{E1}\wedge\textbf{C0}\Rightarrow\,\textbf{D1}\wedge\textbf{D2}.
\end{equation}

Then, \Ref{Suf-3}, \SS[9], and \SS[13] jointly imply
\begin{equation}\label{Suf-4}
\textbf{LNR}\wedge\textbf{G1}\wedge\textbf{G2}\wedge\textbf{E0}\,\wedge\textbf{E1}\wedge\textbf{C0}\Rightarrow\,\text{false}.
\end{equation}

Now we prove the result using \Ref{Suf-2} and \Ref{Suf-4}. Suppose we can also prove (A) \GS[]\AND(\NotXS[])\IMPLY\ES[0]\AND\ES[1]\AND(\NotCS[0]). Then, one can see that this, jointly with \Ref{Suf-2}, implies \LNR\AND\GS[1]\AND\GS[]\AND\\\noindent(\NotXS[])\IMPLY\CONT. Similarly, (B) \GS[]\AND(\NotXS[])\IMPLY\ES[0]\AND\ES[1]\AND\CS[0] and \Ref{Suf-4} jointly imply \LNR\AND\GS[1]\AND\GS[2]\AND\GS[]\AND(\NotXS[])\IMPLY\CONT. The proof is thus complete.
\end{proof}

\subsection{The insight on proving the sufficiency}

To prove the sufficiency directions, we need to show that a set of polynomials is linearly independent given any 3-unicast ANA network, for example, ``\LNR\AND\GS[]\IMPLY\XS[]". To that end, we prove the equivalent relationship ``\LNR\AND\GS[]\AND(\NotXS)\IMPLY\CONT."  Focusing on the linear dependence condition \NotXS, although there are many possible cases, allows us to use the subgraph property (Proposition 2) to simplify the proof. Further, we use the logic statements \SS[3] to \SS[10] to convert all the cases of the linear dependence condition into the greatest common divisor statements \DS[1] to \DS[6], for which the channel gain property (Proposition 3) further helps us to find the corresponding graph-theoretic implication. 

\section{Proofs of ``\LNR\AND\GS[1]\IMPLY\HS[1]" and ``\LNR\AND\GS[1]\AND\GS[2]\IMPLY\HStilde[1]\OR\HS[2]\OR\HStilde[2]"}\label{Prop5=>Proof}

As discussed in \AppRef{GeneralSufficiencyProof}, we use \CorRef{Cor5} to prove the sufficiency directions. We first show that (i) \LNR\AND\GS[1]\AND\\\noindent\GS[2]\IMPLY\HS[2]; and (ii) \LNR\AND\GS[1]\AND\GS[2]\IMPLY\HStilde[2]. Then the remaining sufficiency directions ``\LNR\AND\GS[1]\IMPLY\HS[1]" and ``\LNR\AND\GS[1]\\\noindent\AND\GS[2]\IMPLY\HStilde[1]" are derived using simple facts of ``\HS[2]\IMPLY\HS[1]" and ``\HStilde[2]\IMPLY\HStilde[1]", respectively. Note that \HS[2]\IMPLY\HS[1] is straightforward since $\GSP[1][\left(\!1\!\right)](\netvar)$ is a subset of the polynomials $\GSP[1][\left(\!\Nid\!\right)](\netvar)$ (multiplied by a common factor) and whenever $\GSP[1][\left(\!\Nid\!\right)](\netvar)$ is linearly independent, so is $\GSP[1][\left(\!1\!\right)](\netvar)$. Similarly, we have \HStilde[2]\IMPLY\HStilde[1].

We prove ``\LNR\AND\GS[1]\AND\GS[2]\IMPLY\HS[2]" as follows.
\begin{proof} By the definition of linear dependence, \NotHS[2] implies that there exist two sets of coefficients $\{\alpha_i\}_{i=0}^{\Nid}$ and $\{\beta_j\}_{j=0}^{\Nid-1}$ such that
\begin{equation}\label{hn-LD}
\sum_{i=0}^{\Nid} \alpha_{i} \ChGANA[1][1]\ChGANA[3][2]\BOLDb^{\Nid-i}\BOLDa^{i} = \sum_{j=0}^{\Nid-1} \beta_{j}\ChGANA[3][1]\ChGANA[1][2]\BOLDb^{\Nid-j}\BOLDa^{j}.
\end{equation}

We will now argue that at least one of $\{\alpha_i\}_{i=0}^{\Nid}$ and at least one of $\{\beta_{j}\}_{j=0}^{\Nid-1}$ are non-zero if $\LRneq$. The reason is as follows. For example, suppose that all $\{\beta_{j}\}_{j=0}^{\Nid-1}$ are zero. By definition (iv) of the 3-unicast ANA network, any channel gain is non-trivial. Thus $\ChGANA[1][1]\ChGANA[3][2]$ is a non-trivial polynomial. Then, \Ref{hn-LD} becomes $\sum_{i=0}^{\Nid} \alpha_{i} \BOLDb^{\Nid-i}\BOLDa^{i}=0$, which implies that the set of $(\Nid+1)$ polynomials, $\tilde{\bold{h}}(\netvar)=\{\BOLDb^{\Nid},\,\BOLDb^{\Nid\!-\!1}\BOLDa,\,...,\BOLDb\BOLDa^{\Nid\!-\!1},\,\BOLDa^{\Nid}\}$, is linearly dependent. By \PropRef{Prop1}, the determinant of the Vandermonde matrix $[\tilde{\bold{h}}(\netvar[(\!k\!)])]_{k=1}^{\Nid+1}$ is thus zero, which implies $\LReqVAR$. This contradicts the assumption \LNR. The fact that not all $\{\alpha_{i}\}_{i=0}^{\Nid}$ are zero can be proven similarly.

As a result, there exist two sets of coefficients $\{\alpha_i\}_{i=0}^{\Nid}$ and $\{\beta_{j}\}_{j=0}^{\Nid-1}$ with at least one of each group being non-zero such that the following logic relationship holds:
\begin{equation}\label{HS[2]proof-1}
\textbf{LNR}\wedge(\neg\,\textbf{H2})\Rightarrow\,\textbf{E0}\wedge\textbf{E1}.
\end{equation}

Then, note that \Ref{HS[2]proof-1} implies
\begin{equation*}\begin{split}
\textbf{LNR}\wedge(\neg\,\textbf{C0})\wedge(\neg\,\textbf{H2})&\Rightarrow\,\textbf{E0}\wedge\textbf{E1}\wedge(\neg\,\textbf{C0}), \\
\text{and}\quad\textbf{LNR}\wedge\textbf{C0}\wedge(\neg\,\textbf{H2})&\Rightarrow\,\textbf{E0}\wedge\textbf{E1}\wedge\textbf{C0}.
\end{split}\end{equation*}

Applying \CorRef{Cor5}(A) (substituting \GS[] by \LNR\AND(\NotCS[0]) and \XS[] by \HS[2], respectively), the former implies \LNR\AND\GS[1]\AND\\\noindent(\NotCS[0])\AND(\NotHS[2])\IMPLY\CONT. By \CorRef{Cor5}(B), the latter implies \LNR\AND\GS[1]\AND\GS[2]\AND\CS[0]\AND(\NotHS[2])\IMPLY\CONT. These jointly imply
\begin{equation*}
\textbf{LNR}\wedge\textbf{G1}\wedge\textbf{G2}\wedge(\neg\,\textbf{H2})\Rightarrow\,\text{false},
\end{equation*}
which is equivalent to \LNR\AND\GS[1]\AND\GS[2]\IMPLY\HS[2]. The proof is thus complete.
\end{proof}

We prove ``\LNR\AND\GS[1]\AND\GS[2]\IMPLY\HStilde[2]" as follows.
\begin{proof} We will only show the logic relationship ``\LNR\AND\\\noindent(\NotHStilde[2])\IMPLY\ES[0]\AND\ES[1]" so that the rest can be proved by \CorRef{Cor5} as in the proof of ``\LNR\AND\GS[1]\AND\GS[2]\IMPLY\HS[2]". Suppose \NotHStilde[2] is true. Then, there exists two sets of coefficients $\{\alpha_i\}_{i=1}^{\Nid}$ and $\{\beta_j\}_{j=0}^{\Nid}$ such that
\begin{equation}\label{kn-LD}
\sum_{i=1}^{\Nid} \alpha_{i} \ChGANA[1][1]\ChGANA[3][2]\BOLDb^{\Nid-i}\BOLDa^{i} = \sum_{j=0}^{\Nid} \beta_{j}\ChGANA[3][1]\ChGANA[1][2]\BOLDb^{\Nid-j}\BOLDa^{j}.
\end{equation}

One can easily see that, similarly to the above proof, the assumption \LNR\;results in the not-being-all-zero condition on both $\{\alpha_i\}_{i=1}^{\Nid}$ and  $\{\beta_j\}_{j=0}^{\Nid}$, which in turn implies that ``\LNR\AND\\\noindent(\NotHStilde[2])\IMPLY\ES[0]\AND\ES[1]". The proof is thus complete.
\end{proof}


\section{Proofs of the properties of \GS[3], \GS[4], \NotGS[3], and \NotGS[4]}\label{ProofsG3G4}

We prove Properties~1 and~2 of \GS[3] as follows.
\begin{proof} Suppose \GS[3] is true, that is, $\Sover[2]\cap\Dover[3]\EqualEmpty$. Consider $e^\ast_2$, the most downstream edge of $\onecut[s_2][d_1]\cap\onecut[s_2][d_3]$ and $e^\ast_3$, the most upstream edge of $\onecut[s_1][d_3]\cap\onecut[s_2][d_3]$. If either $e^\ast_2\!=\!e_{s_2}$ or $e^\ast_3\!=\!e_{d_3}$ (or both), we must have $e^\ast_2 \PREC e^\ast_3$ otherwise it contradicts definitions (ii) and (iii) of the 3-unicast ANA network. Consider the case in which  both $e^\ast_2\neq\!e_{s_2}$ and $e^\ast_3\neq\!e_{d_3}$.
Recall the definitions of $\Sover[2]\triangleq\!\onecut[s_2][d_1]\cap\onecut[s_2][d_3]\backslash\{e_{s_2}\}$ and $\Dover[3]\!\triangleq\!\onecut[s_1][d_3]\cap\onecut[s_2][d_3]\backslash\{e_{d_3}\}$. We thus have $e^\ast_2\!\in\!\Sover[2]$ and $e^\ast_3\!\in\!\Dover[3]$. By the assumption $\Sover[2]\cap\Dover[3]\,\EqualEmpty$ and \LemRef{Lem3}, we must have $e^\ast_2\PREC e^\ast_3$ as well.

From the construction of $e^\ast_2$ and $e^\ast_3$, the channel gains $\ChGANA[3][1]$, $\ChGANA[1][2]$, and $\ChGANA[3][2]$ can be expressed as $\ChGANA[3][1] = \ChG[e^\ast_3][e_{s_1}]\ChG[e_{d_3}][e^\ast_3]$, $\ChGANA[1][2] = \ChG[e^\ast_2][e_{s_2}]\ChG[e_{d_1}][e^\ast_2]$, and $\ChGANA[3][2] = \ChG[e^\ast_2][e_{s_2}]\ChG[e^\ast_3][e^\ast_2]\ChG[e_{d_3}][e^\ast_3]$. Moreover, we have both $\GCD[{\ChG[e^\ast_3][e_{s_1}]}][\;{\ChG[e^\ast_2][e_{s_2}]\ChG[e^\ast_3][e^\ast_2]}]\PolyEqual 1$ and
$\GCD[{\ChG[e^\ast_3][e^\ast_2]\ChG[e_{d_3}][e^\ast_3]}][\;{\ChG[e_{d_1}][e^\ast_2]}]\PolyEqual 1$ otherwise it violates that $e^\ast_2$ (resp. $e^\ast_3$) is the most downstream (resp. upstream) edge of $\Sover[2]$ (resp. $\Dover[3]$). The same argument also leads to $\GCD[{\ChGANA[3][1]}][\;{\ChG[e^\ast_3][e^\ast_2]}]\PolyEqual 1$ and $\GCD[{\ChGANA[1][2]}][\;{\ChG[e^\ast_3][e^\ast_2]}]\PolyEqual 1$.
\end{proof}

We prove Properties~1,~2, and~3 of \NotGS[3] as follows.
\begin{proof} Suppose \NotGS[3] is true, i.e., $\Sover[2]\cap\Dover[3]\NotEqualEmpty$. Choose the most upstream $e^{23}_u$ and most downstream $e^{23}_v$ edges in $\Sover[2]\cap\Dover[3]$. Then, the channel gains $\ChGANA[3][1]$, $\ChGANA[1][2]$, and $\ChGANA[3][2]$ can be expressed as $\ChGANA[3][1]\!=\ChG[e^{23}_u][e_{s_1}]\ChG[e^{23}_v][e^{23}_u]\ChG[e_{d_3}][e^{23}_v]$,
$\ChGANA[1][2]\!=\ChG[e^{23}_u][e_{s_2}]\ChG[e^{23}_v][e^{23}_u]$ $\ChG[e_{d_1}][e^{23}_v]$, and $\ChGANA[3][2]\!= \ChG[e^{23}_u][e_{s_2}]\ChG[e^{23}_v][e^{23}_u]\ChG[e_{d_3}][e^{23}_v]$. Moreover, we must have both $\GCD[{\!\ChG[e^{23}_u][e_{s_1}]}][{\ChG[e^{23}_u][e_{s_2}]\!}]\PolyEqual 1$ and $\GCD[{\!\ChG[e_{d_3}][e^{23}_v]}][{\ChG[e_{d_1}][e^{23}_v]\!}] \PolyEqual 1$ otherwise it violates \LemRef{Lem3} and/or $e^{23}_u$ (resp. $e^{23}_v$) being the most upstream (resp. downstream) edge among $\Sover[2]\cap\Dover[3]$. For example, if $\GCD[{\ChG[e^{23}_u][e_{s_1}]}][\;{\ChG[e^{23}_u][e_{s_2}]}]\PolyNotEqual 1$, then by \LemRef{Lem7} and the assumption $e^{23}_u\!\in\!\Sover[2]\CAP\Dover[3]\!\subset\!\Dover[3]$, there must exist an edge $e\!\in\!\Dover[3]$ such that $e \PREC e^{23}_u$. If such edge $e$ is also in $\Sover[2]$, then this $e$ violates the construction that $e^{23}_u$ is the most upstream edge of $\Sover[2]\cap\Dover[3]$. If such edge $e$ is not in $\Sover[2]$, then it contradicts the conclusion in \LemRef{Lem3}.


We now prove Property~3 of \NotGS[3]. Suppose that at least one of $\{e^{23}_u,e^{23}_v\}$ is not an $1$-edge cut separating $s_1$ and $\head[e^{23}_v]$. Say $e^{23}_u\!\not\in\!\onecut[s_1][{\head[e^{23}_v]}]$, then $s_1$ can reach $\head[e^{23}_v]$ without using $e^{23}_u$. Since $\head[e^{23}_v]$ reaches $d_3$, we can create an \FromTo[1][3] path not using $e^{23}_u$. This contradicts the construction that $e^{23}_u\!\in\!\Sover[2]\cap\Dover[3]\!\subset\!\Dover[3]$. Similarly, we can also prove that $e^{23}_v\!\not\in\!\onecut[s_1][{\head[e^{23}_v]}]$ leads to a contradiction. Therefore, we have proven  $\{e^{23}_u,e^{23}_v\}\!\subset\!\onecut[s_1][{\head[e^{23}_v]}]$. Symmetrically applying the above arguments, we can also prove that $\{e^{23}_u,e^{23}_v\}\!\subset\!\onecut[{\tail[e^{23}_u]}][d_1]$.

Now consider an \FromTo[1][1] path $P$ such that there exists one vertex $w\!\in\!P$ satisfying $\tail[e^{23}_u]\PRECEQ w\PRECEQ \head[e^{23}_v]$. If the path of interest $P$ does not use $e^{23}_u$ and $w\!=\!\tail[e^{23}_u]$, then $\tail[e^{23}_u]$ can follow $P$ to $d_1$ without using $e^{23}_u$, which contradicts $e^{23}_u\!\in\!\onecut[{\tail[e^{23}_u]}][d_1]$. If $P$ does not use $e^{23}_u$ and $\tail[e^{23}_u]\PREC w\PRECEQ \head[e^{23}_v]$, then $s_1$ can follow $P$ to $w$ and reach $\head[e^{23}_v]$ without using $e^{23}_u$, which contradicts $e^{23}_u\!\in\!\onecut[s_1][{\head[e^{23}_v]}]$. By the similar arguments, we can also prove the case when $P$ does not use $e^{23}_v$ leads to a contradiction. Therefore, we must have $\{e^{23}_u,e^{23}_v\}\subset P$. The proof is complete.
\end{proof}

By swapping the roles of $s_2$ and $s_3$, and the roles of $d_2$ and $d_3$, the above proofs can also be used to prove Properties 1 and 2 of \GS[4] and Properties~1, 2, and 3, of \NotGS[4].

\section{Proofs of \NS[1] to \NS[9]}\label{ProofsNS[1]-NS[9]}

We prove \NS[1] as follows.
\begin{proof} Instead of proving directly, we prove \HS[2]\IMPLY\HS[1] and use the existing result of ``\LNR\AND\GS[1]\OPPLY\HS[1]" established in the proof of \PropRef{Prop5}. \HS[2]\IMPLY\HS[1] is straightforward since $\GSP[1][\left(\!1\!\right)](\netvar)$ is a subset of the polynomials $\GSP[1][\left(\!\Nid\!\right)](\netvar)$ (multiplied by a common factor) and whenever $\GSP[1][\left(\!\Nid\!\right)](\netvar)$ is linearly independent, so is $\GSP[1][\left(\!1\!\right)](\netvar)$. The proof is thus complete.
\end{proof}

%

We prove \NS[2] as follows.
\begin{proof} We prove an equivalent relationship: (\NotLNR)\OR\\(\NotGS[1])\IMPLY(\NotHStilde[1]). Consider $\GSPtilde[1][\left(\!1\!\right)](\netvar)$ as in \Ref{GSP2^1_NEW}. Suppose $\GANA$ satisfies (\NotLNR)\OR(\NotGS[1]), which means $\GANA$ satisfies either $\LReqVAR$ or $\ChGANA[1][1]\ChGANA[3][2]\PolyEqual\ChGANA[1][2]\ChGANA[3][1]$ or $\ChGANA[1][1]\ChGANA[2][3]\PolyEqual\ChGANA[1][3]$ $\ChGANA[2][1]$. If $\LReqVAR$, then we notice that $\ChGANA[1][2]\ChGANA[3][1]\ChGANA[1][3]\BOLDa\,\PolyEqual\,\ChGANA[1][2]\ChGANA[3][1]\ChGANA[1][3]\BOLDb$ and $\GSPtilde[1][\left(\!1\!\right)](\netvar)$ is thus linearly dependent. If $\ChGANA[1][1]$ $\ChGANA[3][2]\,\PolyEqual\ChGANA[1][2]\ChGANA[3][1]$, then we notice $\ChGANA[1][1]\ChGANA[3][2]\ChGANA[1][3]\BOLDa\,\PolyEqual \,\ChGANA[1][2]\ChGANA[3][1]\ChGANA[1][3]$ $\BOLDa$. Similarly if $\ChGANA[1][1]\ChGANA[2][3]\,\PolyEqual\,\ChGANA[1][3]\ChGANA[2][1]$, then we have $\ChGANA[1][1]\ChGANA[3][2]\ChGANA[1][3]$ $\BOLDa \,\PolyEqual\, \ChGANA[1][2]\ChGANA[3][1]\ChGANA[1][3]\BOLDb$. The proof is thus complete.
\end{proof}

Following similar arguments used in proving \NS[2], i.e., \HStilde[2]\\\noindent\IMPLY\HStilde[1], one can easily prove \NS[3].

We prove \NS[4] as follows.
\begin{proof} (\NotGS[2])\AND\GS[3]\AND\GS[4] implies that $s_1$ cannot reach $d_1$ on $\GANA$. This violates the definition (iv) of the 3-unicast ANA network.
\end{proof}

We prove \NS[5] as follow.
\begin{proof} We prove an equivalent relationship: (\NotGS[2])\AND \\
\noindent (\NotGS[3])\AND\GS[4]\IMPLY(\NotGS[1]). \!Suppose (\NotGS[2])\AND(\NotGS[3])\AND\GS[4] is true.
Then the most upstream edge of $\Sover[2]\cap\Dover[3]$ is an $1$-edge cut separating $s_1$ and $d_1$. Therefore we have $\EC[\{s_1,s_2\}][\{d_1,d_3\}]=1$ and thus by \CorRef{Cor1}, $\ChGANA[1][1]\ChGANA[3][2]\PolyEqual\ChGANA[1][2]\ChGANA[3][1]$. This further implies that \GS[1] is false.
\end{proof}

By swapping the roles of $s_2$ and $s_3$, and the roles of $d_2$ and $d_3$, the above \NS[5] proof can also be used to prove \NS[6].

We prove \NS[7] as follows.
\begin{proof} Suppose \LNR\AND(\NotGS[3])\AND(\NotGS[4])\AND(\NotGS[5]) is true. From \LNR\;being true, any $\Sover[2]\cap\Dover[3]$ edge and any $\Sover[3]\cap\Dover[2]$ edge must be distinct, otherwise (if there exists an edge $e\!\in\!\Sover[2]\cap\Sover[3]\cap\Dover[2]\cap\Dover[3]$) it contradicts the assumption \LNR\;by \PropRef{Prop4}. From \GS[5] being false, we have either $e^{23}_u\!=\!e^{32}_u$ or both $e^{23}_u$ and $e^{32}_u$ are not reachable from each other. But $e^{23}_u\!=\!e^{32}_u$ cannot be true by the assumption \LNR.

Now we prove \GS[6], i.e., any vertex $w'$ where $\tail[e^{23}_u]\PRECEQ w' \PRECEQ\head[e^{23}_v]$ and any vertex $w''$ where $\tail[e^{32}_u]\PRECEQ w'' \PRECEQ\head[e^{32}_v]$ are not reachable from each other. Suppose not and assume that some vertex $w'$ satisfying $\tail[e^{23}_u]\PRECEQ w' \PRECEQ\head[e^{23}_v]$ and some vertex $w''$ satisfying $\tail[e^{32}_u]\PRECEQ w'' \PRECEQ\head[e^{32}_v]$ are reachable from each other. Since $s_1$ can reach $\tail[e^{23}_u]$ or $\tail[e^{32}_u]$ and $d_1$ can be reached from $\head[e^{23}_v]$ or $\head[e^{32}_v]$ by Property~1 of \NotGS[3] and \NotGS[4], we definitely have an \FromTo[1][1] path $P$ who uses both $w'$ and $w''$. The reason is that if $w'\PRECEQ w''$, then $s_1$ can first reach $\tail[e^{23}_u]$, visit $w'$, $w''$, and $\head[e^{32}_v]$, and finally arrive at $d_1$. The case when $w''\PRECEQ w'$ can be proven by symmetry. By Property~3 of \NotGS[3], such path must use $\{e^{23}_u, e^{23}_v\}$. Similarly by Property~3 of \NotGS[4], such path must also use $\{e^{32}_u, e^{32}_v\}$. Together with the above discussion that any $\Sover[2]\cap\Dover[3]$ edge and any $\Sover[3]\cap\Dover[2]$ edge are distinct, this implies that all four edges $\{e^{23}_u, e^{23}_v, e^{32}_u, e^{32}_v\}$ are not only distinct but also used by a single path $P$. However, this contradicts the assumption \LNR\AND(\NotGS[5]) that $e^{23}_u$ and $e^{32}_u$ are not reachable from each other. 
\end{proof}

We prove \NS[8] as follows.
\begin{proof} Suppose \GS[1]\AND(\NotGS[2])\AND(\NotGS[3])\AND(\NotGS[4])\AND\GS[5] is true. Consider $e^{23}_u$ and $e^{32}_u$, the most upstream edges of $\Sover[2]\cap\Dover[3]$ and $\Sover[3]\cap\Dover[2]$, respectively. Say we have $e^{23}_u\PREC e^{32}_u$. Then \NotGS[2] implies that removing $e^{23}_u$ will disconnect $s_1$ and $d_1$. Therefore, $e^{23}_u\!\in\!\Sover[2]\cap\Dover[3]$ also belongs to $\onecut[s_1][d_1]$. This further implies that we have $\EC[\{s_1,s_2\}][\{d_1,d_3\}]\!=\!1$ and thus $\GANA$ satisfies $\ChGANA[1][1]\ChGANA[3][2]\PolyEqual\ChGANA[3][1]\ChGANA[1][2]$. However, this contradicts the assumption that \GS[1] is true. Similar arguments can be applied to show that the case when $e^{32}_u \PREC e^{23}_u$ also contradicts \GS[1]. The proof of \NS[8] is thus complete.
\end{proof}

We prove \NS[9] as follows.
\begin{proof} Suppose that (\NotGS[2])\AND(\NotGS[3])\AND(\NotGS[4])\AND(\NotGS[5])\AND\\\noindent\GS[6] is true. Consider $e^{23}_u$ and $e^{32}_u$, the most upstream edges of $\Sover[2]\cap\Dover[3]$ and $\Sover[3]\cap\Dover[2]$, respectively. From (\NotGS[5])\AND\GS[6] being true, one can see that $e^{23}_u$ and $e^{32}_u$ are not only distinct but also not reachable from each other. Thus by \NotGS[2] being true, $\{e^{23}_u, e^{32}_u\}$ constitutes an edge cut separating $s_1$ and $d_1$. Note from Property 1 of \NotGS[3] and \NotGS[4] that $s_1$ can reach $d_1$ through either $e^{23}_u$ or $e^{32}_u$. Since $e^{23}_u$ and $e^{32}_u$ are not reachable from each other, both have to be removed to disconnect $s_1$ and $d_1$ (removing only one of them is not enough).

From \GS[6] being true, any vertex $w'$ where $\tail[e^{23}_u]\PRECEQ w' \PRECEQ\head[e^{23}_v]$ and any vertex $w''$ where $\tail[e^{32}_u]\PRECEQ$ \noindent$w'' \PRECEQ\head[e^{32}_v]$ are not reachable from each other. Thus $e^{23}_u$ (resp. $e^{32}_u$) cannot reach $e^{32}_v$ (resp. $e^{23}_v$). Moreover, $e^{23}_v$ and $e^{32}_v$ are not only distinct but also not reachable from each other. This implies that $e^{23}_u$ (resp. $e^{32}_u$) can only reach $e^{23}_v$ (resp. $e^{32}_v$) if $e^{23}_u\neq e^{23}_v$ (resp. $e^{32}_u\neq e^{32}_v$). Then the above discussions further that imply $\{e^{23}_v,e^{32}_v\}$ is also an edge cut separating $s_1$ and $d_1$.



Let $m'_{11} = \ChG[e^{23}_u][e_{s_1}]\ChG[e^{23}_v][e^{23}_u]\ChG[e_{d_1}][e^{23}_v]$, which takes into account the overall path gain from $s_1$ to $d_1$ for all paths that use both $e^{23}_u$ and $e^{23}_v$. Similarly denote $m''_{11} = \ChG[e^{32}_u][e_{s_1}]\ChG[e^{32}_v][e^{32}_u]$ $\ChG[e_{d_1}][e^{32}_v]$ to be the overall path gain from $s_1$ to $d_1$ for all paths that use both $e^{32}_u$ and $e^{32}_v$. Then the discussions so far imply that the channel gain $\ChGANA[1][1]$ consists of two polynomials: $\ChGANA[1][1]=m'_{11} + m''_{11}$. Then, it follows that 
\begin{equation*}\begin{split}
& \ChGANA[1][1]\ChGANA[3][2]\ChGANA[2][3] = (m'_{11}+m''_{11})\,\ChGANA[3][2]\ChGANA[2][3]
\\
& = (\ChG[e^{23}_u][e_{s_1}]\ChG[e^{23}_v][e^{23}_u]\ChG[e_{d_1}][e^{23}_v])\,\ChGANA[3][2]\ChGANA[2][3] \\
& \; + (\ChG[e^{32}_u][e_{s_1}]\ChG[e^{32}_v][e^{32}_u]\ChG[e_{d_1}][e^{32}_v]\!)\,\ChGANA[3][2]\ChGANA[2][3]
\\
& = (\ChG[e^{23}_u][e_{s_1}]\ChG[e^{23}_v][e^{23}_u]\ChG[e_{d_1}][e^{23}_v]\!)(\ChG[e^{23}_u][e_{s_2}]\ChG[e^{23}_v][e^{23}_u]\ChG[e_{d_3}][e^{23}_v]\!)\,\ChGANA[2][3] \\
& \; + (\ChG[e^{32}_u][e_{s_1}]\ChG[e^{32}_v][e^{32}_u]\ChG[e_{d_1}][e^{32}_v]\!)\,\ChGANA[3][2]\,(\ChG[e^{32}_u][e_{s_3}]\ChG[e^{32}_v][e^{32}_u]\ChG[e_{d_2}][e^{32}_v]\!) \\
& = (\ChG[e^{23}_u][e_{s_1}]\ChG[e^{23}_v][e^{23}_u]\ChG[e_{d_3}][e^{23}_v]\!)\,\ChGANA[2][3]\,(\ChG[e^{23}_u][e_{s_2}]\ChG[e^{23}_v][e^{23}_u]\ChG[e_{d_1}][e^{23}_v]\!) \\
& \; + (\ChG[e^{32}_u][e_{s_1}]\ChG[e^{32}_v][e^{32}_u]\ChG[e_{d_2}][e^{32}_v]\!)\,\ChGANA[3][2]\,(\ChG[e^{32}_u][e_{s_3}]\ChG[e^{32}_v][e^{32}_u]\ChG[e_{d_1}][e^{32}_v]\!) \\
& = \ChGANA[3][1]\ChGANA[2][3]\ChGANA[1][2] + \ChGANA[2][1]\ChGANA[3][2]\ChGANA[1][3] = \BOLDa + \BOLDb.
\end{split}\end{equation*}
where the third and fourth equalities follow from the Property 1 of both \NotGS[3] and \NotGS[4]. The proof is thus complete.
\end{proof}


\section{Proofs of \SS[1] to \SS[10]}\label{ProofsSS[1]-SS[10]}

We prove \SS[1] as follows.
\begin{proof} Suppose \DS[1] is true, that is, $\GANA$ satisfies $\GCD[{\ChGANA[2][1]^{l_1}\ChGANA[3][2]^{l_1}\ChGANA[1][3]^{l_1}}][{\,\ChGANA[2][3]}]\!=\!\ChGANA[2][3]$ for some integer $l_1\!>\!0$. Then $\GANA$ also satisfies $\GCD[{\ChGANA[1][1]\ChGANA[2][1]^{l_1}\ChGANA[3][2]^{l_1}\ChGANA[1][3]^{l_1}}][{\,\ChGANA[2][3]}]\!=\!\ChGANA[2][3]$ obviously. Thus we have \DS[5].
\end{proof}

By swapping the roles of $s_2$ and $s_3$, and the roles of $d_2$ and $d_3$, the proof for \SS[1] can be applied symmetrically to the proof for \SS[2].

We prove \SS[3] as follows.
\begin{proof} Suppose \ES[0]\AND\ES[1]\AND\CS[1] is true. By \ES[0]\AND\ES[1] being true, $\GANA$ of interest satisfies \Ref{E1}.
By the definition of \CS[1], we have $\ist\!<\!\jst$.

By \Ref{E1}, we can divide $\BOLDa^{\ist}$ on both sides. Then we have
\begin{align*}
\sum_{i=\ist}^{\iend} \alpha_{i} \ChGANA[1][1]\ChGANA[3][2]\BOLDb^{\Nid-i}\BOLDa^{i-\ist} & = \sum_{j=\jst}^{\jend}\beta_{j}\ChGANA[3][1]\ChGANA[1][2]\BOLDb^{\Nid-j}\BOLDa^{j-\ist}.
\end{align*}

Since $\ist\!<\!\jst$, each term with non-zero $\beta_j$ in the right-hand side (RHS) has $\BOLDa$ as a common factor. Similarly, each term with non-zero $\alpha_i$ on the left-hand side (LHS) has $\BOLDa$ as a common factor except for the first term (since $\alpha_{\ist}\!\neq\!0$). Therefore the first term $\alpha_{\ist}\ChGANA[1][1]\ChGANA[3][2]\BOLDb^{\Nid-\ist}$ must contain  $\BOLDa\!=\!\ChGANA[3][1]\ChGANA[2][3]\ChGANA[1][2]$ as a factor, which implies $\GCD[{\ChGANA[1][1]\ChGANA[2][1]^{\Nid-\ist}\ChGANA[3][2]^{\Nid-\ist+1}\ChGANA[1][3]^{\Nid-\ist}}][{\,\ChGANA[3][1]\ChGANA[2][3]\ChGANA[1][2]}]\!=\!\ChGANA[3][1]\ChGANA[2][3]$ $\ChGANA[1][2]$. Since $\ist\!<\!\jst\!\leq\!\Nid$, we have $n-\ist\!\geq\!1$. Hence, we have $\GCD[{\ChGANA[1][1]\ChGANA[2][1]^{k}\ChGANA[3][2]^{k+1}\ChGANA[1][3]^{k}}][{\,\ChGANA[3][1]\ChGANA[2][3]\ChGANA[1][2]}]$ $=\ChGANA[3][1]\ChGANA[2][3]\ChGANA[1][2]$ for some integer $k\!\geq\!1$. This observation implies the following two statements. Firstly, $\GCD[{\ChGANA[1][1]\ChGANA[2][1]^{l_4}\ChGANA[3][2]^{l_4}\ChGANA[1][3]^{l_4}}][{\,\ChGANA[3][1]\ChGANA[1][2]}]$ $\!=\!\ChGANA[3][1]\ChGANA[1][2]$
when $l_4\!=\!k+1\!\geq\!2$ and thus we have proven \DS[4]. Secondly,  $\GCD[{\ChGANA[1][1]\ChGANA[2][1]^{l_5}\ChGANA[3][2]^{l_5}\ChGANA[1][3]^{l_5}}][{\,\ChGANA[2][3]}]=\ChGANA[2][3]$ when $l_5\!=\!k+1\!\geq\! 2$ and thus we have proven \DS[5]. The proof is thus complete.
\end{proof}

We prove \SS[4] as follows.
\begin{proof} Suppose \ES[0]\AND\ES[1]\AND\CS[2] is true. Then $\GANA$ of interest satisfies \Ref{E1} and we have $\ist\!>\!\jst$.

We now divide $\BOLDa^{\jst}$ on both sides of \Ref{E1}, which leads to
\begin{align*}
\sum_{i=\ist}^{\iend} \alpha_{i} \ChGANA[1][1]\ChGANA[3][2]\BOLDb^{\Nid-i}\BOLDa^{i-\jst} & = \sum_{j=\jst}^{\jend}\beta_{j}\ChGANA[3][1]\ChGANA[1][2]\BOLDb^{\Nid-j}\BOLDa^{j-\jst}.
\end{align*}

Each term with non-zero $\alpha_i$ on the LHS has $\BOLDa$ as a common factor. Similarly, each term with non-zero $\beta_j$ on the RHS has $\BOLDa$ as a common factor except for the first term (since $\beta_{\jst}\!\neq\!0$). As a result, the first term $\beta_{\jst}\ChGANA[3][1]\ChGANA[1][2]\BOLDb^{\Nid-\jst}$ must contain $\BOLDa\!=\!\ChGANA[3][1]\ChGANA[2][3]\ChGANA[1][2]$ as a factor. This implies that $\GCD[{\BOLDb^{\Nid-\jst}}][{\,\ChGANA[2][3]}]=\ChGANA[2][3]$. Since $\jst\!<\!\ist\!\leq\!\Nid$, we have $\Nid - \jst$ $\!\geq\!1$ and thus $\GCD[{\BOLDb^{k}}][{\,\ChGANA[2][3]}]=\ChGANA[2][3]$ for some positive integer $k$, which is equivalent to \DS[1]. The proof is thus complete.
\end{proof}

We prove \SS[5] as follows.
\begin{proof} Suppose \GS[1]\AND\ES[0]\AND\ES[1]\AND\CS[3] is true. By \ES[0]\AND\ES[1] being true, $\GANA$ of interest satisfies \Ref{E1}. Since $\ist\!=\!\jst$, we can divide $\BOLDa^{\ist}\!=\!\BOLDa^{\jst}$ on both sides of \Ref{E1}, which leads to
\begin{align*}
\sum_{i=\ist}^{\iend} \alpha_{i} \ChGANA[1][1]\ChGANA[3][2]\BOLDb^{\Nid-i}\BOLDa^{i-\ist} & = \sum_{j=\jst}^{\jend}\beta_{j}\ChGANA[3][1]\ChGANA[1][2]\BOLDb^{\Nid-j}\BOLDa^{j-\jst}.
\end{align*}

Note that if $\ist\!=\!\jst\!=\!\Nid$ meaning that $\ist\!=\!\jst\!=\!\iend\!=\!\jend\!=\!\Nid$, then \Ref{E1} reduces to $\ChGANA[1][1]\ChGANA[3][2]\,\PolyEqual\ChGANA[3][1]\ChGANA[1][2]$ (since $\alpha_{\ist}\!\neq\!0$ and $\beta_{\jst}\!\neq\!0$). This contradicts the assumption \GS[1].

Thus for the following, we only consider the case when $\ist\!=\!\jst\!\leq\!\Nid-1$. Note that each term with non-zero $\beta_j$ on the RHS has a common factor $\ChGANA[3][1]\ChGANA[1][2]$. Similarly, each term with non-zero $\alpha_i$ on the LHS has a common factor $\BOLDa=\ChGANA[3][1]\ChGANA[2][3]\ChGANA[1][2]$ except for the first term ($i\!=\!\ist$). As a result, the first term $\alpha_{\ist} \ChGANA[1][1]\ChGANA[3][2]\BOLDb^{\Nid-\ist}$ must contain $\ChGANA[3][1]\ChGANA[1][2]$ as a factor. Since $\ist\!\leq\!\Nid-1$, we have $\GCD[{\ChGANA[1][1]\ChGANA[2][1]^k\ChGANA[3][2]^{k+1}\ChGANA[1][3]^k}][\,{\ChGANA[3][1]\ChGANA[1][2]}]=\ChGANA[3][1]\ChGANA[1][2]$ for some integer $k\!\geq\!1$. Therefore, we have \DS[4].
\end{proof}

We prove \SS[6] as follows.
\begin{proof} Suppose \ES[0]\AND\ES[1]\AND\CS[4] is true. By \ES[0]\AND\ES[1] being true, $\GANA$ of interest satisfies \Ref{E1}. Since $\iend\!<\!\jend$, we can divide $\BOLDb^{\Nid-\jend}$ on both sides of \Ref{E1}. Then, we have
\begin{align*}
\sum_{i=\ist}^{\iend} \alpha_{i} \ChGANA[1][1]\ChGANA[3][2]\BOLDb^{\jend-i}\BOLDa^{i} & = \sum_{j=\jst}^{\jend}\beta_{j}\ChGANA[3][1]\ChGANA[1][2]\BOLDb^{\jend-j}\BOLDa^{j}.
\end{align*}

Each term with non-zero $\alpha_i$ on the LHS has $\BOLDb$ as a common factor. Similarly, each term with non-zero $\beta_j$ on the RHS has $\BOLDb$ as a common factor except for the last term (since $\beta_{\jend}\!\neq\!0$). Thus, the last term $\beta_{\jend}\ChGANA[3][1]\ChGANA[1][2]\BOLDa^{\jend}$ must be divisible by $\BOLDb\!=\!\ChGANA[2][1]\ChGANA[3][2]\ChGANA[1][3]$, which implies that $\GCD[{\ChGANA[3][1]^{k+1}\ChGANA[2][3]^{k}\ChGANA[1][2]^{k+1}}][{\,\ChGANA[2][1]\ChGANA[3][2]\ChGANA[1][3]}]\!=\!\ChGANA[2][1]\ChGANA[3][2]\ChGANA[1][3]$ for some integer $k\!=\!\jend \!\geq\! \iend+1 \!\geq\! 1$. This observation has two implications. Firstly, $\GCD[{\ChGANA[1][1]\ChGANA[3][1]^{l_3}\ChGANA[2][3]^{l_3}\ChGANA[1][2]^{l_3}}][{\,\ChGANA[2][1]\ChGANA[1][3]}]\!=\!\ChGANA[2][1]\ChGANA[1][3]$
for some positive integer $l_3\!=\!k+1$ and thus we have proven \DS[3]. Secondly, we also have $\GCD[{\ChGANA[3][1]^{l_2}\ChGANA[2][3]^{l_2}\ChGANA[1][2]^{l_2}}][{\,\ChGANA[3][2]}]=\ChGANA[3][2]$ for some positive integer $l_2\!=\!k+1$ and thus we have proven \DS[2]. The proof is thus complete.
\end{proof}

We prove \SS[7] as follows.
\begin{proof} Suppose \ES[0]\AND\ES[1]\AND\CS[5] is true. By \ES[0]\AND\ES[1] being true, $\GANA$ of interest satisfies \Ref{E1}. Since $\iend\!>\!\jend$, we can divide $\BOLDb^{\Nid-\iend}$ on both sides of \Ref{E1}. Then we have
\begin{align*}
\sum_{i=\ist}^{\iend} \alpha_{i} \ChGANA[1][1]\ChGANA[3][2]\BOLDb^{\iend-i}\BOLDa^{i} & = \sum_{j=\jst}^{\jend}\beta_{j}\ChGANA[3][1]\ChGANA[1][2]\BOLDb^{\iend-j}\BOLDa^{j}.
\end{align*}

Each term on the RHS has $\BOLDb$ as a common factor. Similarly, each term on the LHS has $\BOLDb$ as a common factor except for the last term (since $\alpha_{\iend}\!\neq\!0$). Thus, the last term $\alpha_{\iend}\ChGANA[1][1]\ChGANA[3][2]\BOLDa^{\iend}$
must be divisible by $\BOLDb\!=\!\ChGANA[2][1]\ChGANA[3][2]\ChGANA[1][3]$, which implies that $\GCD[{\ChGANA[1][1]\BOLDa^{k}}][{\,\ChGANA[2][1]\ChGANA[1][3]}]=\ChGANA[2][1]\ChGANA[1][3]$ for some integer $k\!=\!\iend$
\\\noindent $\geq\!\jend+1 \!\geq\! 1$. This further implies \DS[3].
\end{proof}

We prove \SS[8] as follows.
\begin{proof} Suppose \GS[1]\AND\ES[0]\AND\ES[1]\AND\CS[6] is true. By \ES[0]\AND\ES[1] being true, $\GANA$ of interest satisfies \Ref{E1}. Since \GS[5], $\iend=\jend$, is true, define $t\!=\!\iend\!=\!\jend$ and $m\!=\!\text{min}\{\ist,\jst\}$. Then by dividing $\BOLDb^{\Nid-t}$ and $\BOLDa^{m}$ from both sides of \Ref{E1}, we have
\begin{align}\label{Lem15C1}
\sum_{i=\ist}^{t} \alpha_{i} \ChGANA[1][1]\ChGANA[3][2]\BOLDb^{t-i}\BOLDa^{i-m} = \sum_{j=\jst}^{t}\beta_{j}\ChGANA[3][1]\ChGANA[1][2]\BOLDb^{t-j}\BOLDa^{j-m}.
\end{align}

Each term with non-zero $\alpha_i$ on the LHS has a common factor $\ChGANA[3][2]$. We first consider the case of $m\!<\!t$. Then each term with non-zero $\beta_j$ on the RHS has a common factor $\BOLDb=\ChGANA[2][1]\ChGANA[3][2]\ChGANA[1][3]$ except the last term $\beta_{t} \ChGANA[3][1]\ChGANA[1][2]\BOLDa^{t-m}$. As a result, $\beta_{t} \ChGANA[3][1]\ChGANA[1][2]\BOLDa^{t-m}$ must be divisible by $\ChGANA[3][2]$, which implies that $\GCD[{\ChGANA[3][1]^{k+1}\ChGANA[2][3]^{k}\ChGANA[1][2]^{k+1}}][{\,\ChGANA[3][2]}]=\ChGANA[3][2]$ for some $k\!=\!t-m\!\geq\!1$. This implies \DS[2].

On the other hand, we argue that we cannot have $m\!=\!t$. If so, then $\ist\!=\!\jst\!=\!\iend\!=\!\jend$ and \Ref{E1} reduces to $\ChGANA[1][1]\ChGANA[3][2]\PolyEqual\ChGANA[3][1]\ChGANA[1][2]$. However, this contradicts the assumption \GS[1].
The proof is thus complete.
\end{proof}

We prove \SS[9] as follows.
\begin{proof} Suppose \ES[0]\AND\ES[1]\AND\CS[0] is true. By \ES[0]\AND\ES[1] being true, $\GANA$ of interest satisfies \Ref{E1} with not-being-all-zero coefficients $\{\alpha_i\}_{i=0}^{\Nid}$ and $\{\beta_j\}_{j=0}^{\Nid}$. Our goal is to prove that, when $\ist\!>\!\jst$ and $\iend\!=\!\jend$, we have \ES[2]: (i) $\alpha_k\!\neq\!\beta_k$ for some $k\!\in\!\{0,...,\Nid\}$; and (ii) either $\alpha_0\!\neq\!0$ or $\beta_{\Nid}\!\neq\!0$ or $\alpha_{k}\!\neq\!\beta_{k-1}$ for some $k\!\in\!\{1,...,\Nid\}$.

Note that (i) is obvious since $\ist\!>\!\jst$. Note by definition that $\ist$ (resp. $\jst$) is the smallest $i$ (resp. $j$) among $\alpha_i\!\neq\!0$ (resp. $\beta_j\!\neq\!0$). Then, $\ist\!>\!\jst$ implies that $\alpha_{\jst}\!=\!0$ while $\beta_{\jst}\!\neq\!0$. Thus simply choosing $k\!=\!\jst$ proves (i).

We now prove (ii). Suppose (ii) is false such that $\alpha_0\!=\!0$; $\beta_{\Nid}\!=\!0$; and $\alpha_{k}\!=\!\beta_{k-1}$ for all $k\!\in\!\{1,...,\Nid\}$. Since $\beta_{\Nid}\!=\!0$, by definition, $\jend$ must be less than or equal to $\Nid-1$. Since we assumed $\iend\!=\!\jend$, this in turn implies that $\alpha_{\Nid}\!=\!0$. Then $\beta_{\Nid-1}\!$ must be zero because $\beta_{\Nid-1}\!=\!\alpha_{\Nid}$. Again this implies $\jend\!\leq\!\Nid-2$. Applying iteratively, we have all zero coefficients $\{\alpha_i\}_{i=0}^{\Nid}$ and $\{\beta_j\}_{j=0}^{\Nid}$. However, this contradicts the assumption \ES[0] since we assumed that at least one of each coefficient group is non-zero. The proof of \SS[9] is thus complete.
\end{proof}

%

We prove \SS[10] as follows.
\begin{proof} Suppose \GS[1]\AND\ES[0]\AND\ES[1]\AND(\NotCS[0]) is true. By \ES[0]\AND\ES[1] being true, $\GANA$ of interest satisfies \Ref{E1} with some values of $\ist$, $\jst$, $\iend$, and $\jend$. Investigating their relationships, there are total 9 possible cases that $\GANA$ can satisfy \Ref{E1}: (i) $\ist < \jst$ and $\iend < \jend$; (ii) $\ist < \jst$ and $\iend > \jend$; (iii) $\ist < \jst$ and $\iend = \jend$; (iv) $\ist > \jst$ and $\iend < \jend$; (v) $\ist > \jst$ and $\iend > \jend$; (vi) $\ist > \jst$ and $\iend = \jend$; (vii) $\ist = \jst$ and $\iend < \jend$; (viii) $\ist = \jst$ and $\iend > \jend$; and (ix) $\ist = \jst$ and $\iend = \jend$.

Note that \CS[0] is equivalent to (vi). Since we assumed that \CS[0] is false, $\GANA$ can satisfy \Ref{E1} with all the possible cases except (vi). We also note that (i) is equivalent to \CS[1]\AND\CS[4], (ii) is equivalent to \CS[1]\AND\CS[5], etc. By applying \SS[3] and \SS[6], we have

\noindent $\bullet$ \ES[0]\AND\ES[1]\AND(i)\IMPLY(\DS[4]\AND\DS[5])\AND(\DS[2]\AND\DS[3]).

By similarly applying \SS[3] to \SS[8], we have the following relationships:

\noindent $\bullet$ \ES[0]\AND\ES[1]\AND(ii)\IMPLY(\DS[4]\AND\DS[5])\AND\DS[3].

\noindent $\bullet$ \GS[1]\AND\ES[0]\AND\ES[1]\AND(iii)\IMPLY(\DS[4]\AND\DS[5])\AND\DS[2].

\noindent $\bullet$ \ES[0]\AND\ES[1]\AND(iv)\IMPLY\DS[1]\AND(\DS[2]\AND\DS[3]).

\noindent $\bullet$ \ES[0]\AND\ES[1]\AND(v)\IMPLY\DS[1]\AND\DS[3].


\noindent $\bullet$ \GS[1]\AND\ES[0]\AND\ES[1]\AND(vii)\IMPLY\DS[4]\AND(\DS[2]\AND\DS[3]).

\noindent $\bullet$ \GS[1]\AND\ES[0]\AND\ES[1]\AND(viii)\IMPLY\DS[4]\AND\DS[3].

\noindent $\bullet$ \GS[1]\AND\ES[0]\AND\ES[1]\AND(ix)\IMPLY\DS[4]\AND\DS[2].

Then, the above relationships jointly imply \GS[1]\AND\ES[0]\AND\ES[1]\\\AND(\NotCS[0])\IMPLY(\DS[1]\AND\DS[3])\OR(\DS[2]\AND\DS[4])\OR\noindent(\DS[3]\AND\DS[4]). The proof of \SS[10] is thus complete.
\end{proof}

\section{Proof of \SS[11]}\label{ProofSS[11]}

\subsection{The third set of logic statements}

To prove \SS[11], we need the third set of logic statements.

\noindent $\bullet$ \GS[7]{\bf:}\label{ref:GS[7]} There exists an edge $\tilde{e}$ such that both the following conditions are satisfied: (i) $\tilde{e}$ can reach $d_1$ but cannot reach any of $d_2$ and $d_3$; and (ii) $\tilde{e}$ can be reached from $s_1$ but not from any of $s_2$ nor $s_3$.

\noindent $\bullet$ \GS[8]{\bf:}\label{ref:GS[8]} $\Sover[3]\NotEqualEmpty$ and $\Dover[2]\NotEqualEmpty$.

The following logic statements are well-defined if and only if \GS[4]\AND\GS[8] is true. Recall the definition of $e^\ast_3$ and $e^\ast_2$ when \GS[4] is true.

\noindent $\bullet$ \GS[9]{\bf:}\label{ref:GS[9]} $\{e^\ast_3,e^\ast_2\}\subset \onecut[s_2][d_3]$.

\noindent $\bullet$ \GS[10]{\bf:}\label{ref:GS[10]} $e^\ast_3\in \onecut[s_2][d_1]$.

\noindent $\bullet$ \GS[11]{\bf:}\label{ref:GS[11]} $e^\ast_3\in \onecut[s_1][d_1]$.

\noindent $\bullet$ \GS[12]{\bf:}\label{ref:GS[12]} $e^\ast_2\in \onecut[s_1][d_3]$.

\noindent $\bullet$ \GS[13]{\bf:}\label{ref:GS[13]} $e^\ast_2\in \onecut[s_1][d_1]$.

The following logic statements are well-defined if and only if \NotGS[4] is true. Recall the definition of $e^{32}_u$ and $e^{32}_v$ when \NotGS[4] is true.

\noindent $\bullet$ \GS[14]{\bf:}\label{ref:GS[14]} $e^{32}_u\!\not\in\!\onecut[s_1][d_1]$. 

\noindent \makebox[1cm][l]{$\bullet$ \GS[15]{\bf:}}\label{ref:GS[15]} Let $\tilde{e}_u$ denote the most downstream edge among $\onecut[s_1][d_1]\cap\onecut[s_1][{\tail[e^{32}_u]}]$. Also let $\tilde{e}_v$ denote the most upstream edge among $\onecut[s_1][d_1]\cap\onecut[{\head[e^{32}_v]}][d_1]$. Then we have (a) $\head[\tilde{e}_u]\PREC\tail[e^{32}_u]$ and $\head[e^{32}_v]\PREC\tail[\tilde{e}_v]$; there exists a \FromTo[1][1] path $P_{11}^\ast$ through $\tilde{e}_u$ and $\tilde{e}_v$ satisfying the following two conditions: (b) $P_{11}^\ast$ is vertex-disjoint from any \FromTo[3][2] path; and (c) there exists an edge $\tilde{e}\in P_{11}^\ast$ where $\tilde{e}_u\PREC\tilde{e}\PREC\tilde{e}_v$ that is not reachable from any of $\{e^{32}_u, e^{32}_v\}$.

%
%
%

\subsection{The skeleton of proving \SS[11]}

We prove the following relationships, which jointly prove \SS[11]. The proofs for the following statements are relegated to \AppRef{ProofsRS[1]-RS[10]}.

\noindent $\bullet$ \RS[1]{\bf:}\label{ref:RS[1]} \DS[1]\IMPLY\GS[8].

\noindent $\bullet$ \RS[2]{\bf:}\label{ref:RS[2]} \GS[4]\AND\GS[8]\AND\DS[1]\IMPLY \GS[9].

\noindent $\bullet$ \RS[3]{\bf:}\label{ref:RS[3]} \GS[4]\AND\GS[8]\AND\GS[9]\AND\DS[3]\IMPLY (\GS[10]\OR\GS[11])\AND(\GS[12]\OR\GS[13]).

\noindent $\bullet$ \RS[4]{\bf:}\label{ref:RS[4]} \GS[4]\AND\GS[8]\AND\GS[9]\AND(\NotGS[10])\AND\GS[11]\AND\ES[0]\IMPLY \CONT.

\noindent $\bullet$ \RS[5]{\bf:}\label{ref:RS[5]} \GS[4]\AND\GS[8]\AND\GS[9]\AND(\NotGS[12])\AND\GS[13]\AND\ES[0]\IMPLY \CONT.

\noindent $\bullet$ \RS[6]{\bf:}\label{ref:RS[6]} \GS[4]\AND\GS[8]\AND\GS[9]\AND\GS[10]\AND\GS[12]\IMPLY (\NotLNR).

\noindent $\bullet$ \RS[7]{\bf:}\label{ref:RS[7]} \GS[1]\AND(\NotGS[4])\IMPLY \GS[14].

\noindent $\bullet$ \RS[8]{\bf:}\label{ref:RS[8]} (\NotGS[4])\AND\GS[14]\IMPLY \GS[15].

\noindent $\bullet$ \RS[9]{\bf:}\label{ref:RS[9]} (\NotGS[4])\AND\GS[14]\AND\DS[3]\IMPLY \GS[7].

\noindent $\bullet$ \RS[10]{\bf:}\label{ref:RS[10]} \GS[7]\AND\ES[0]\IMPLY \CONT.

One can easily verify that jointly \RS[4] to \RS[6] imply
\begin{equation}\label{LRS1}\begin{split}
&\textbf{LNR}\wedge\textbf{G4}\wedge\textbf{G8}\wedge\textbf{G9}\wedge\textbf{E0}\wedge(\textbf{G10}\vee\textbf{G11})\wedge \\
& \qquad\qquad\qquad\qquad\qquad\qquad (\textbf{G12}\vee\textbf{G13})\Rightarrow\,\text{false}.
\end{split}\end{equation}

Together with \RS[3], \Ref{LRS1} reduces to
\begin{equation}\label{LRS2}
\textbf{LNR}\wedge\textbf{G4}\wedge\textbf{G8}\wedge\textbf{G9}\wedge\textbf{E0}\wedge\textbf{D3}\Rightarrow\,\text{false}.
\end{equation}

Jointly with \RS[1] and \RS[2], \Ref{LRS2} further reduces to
\begin{equation}\label{LRS3}
\textbf{LNR}\wedge\textbf{G4}\wedge\textbf{E0}\wedge\textbf{D1}\wedge\textbf{D3}\Rightarrow\,\text{false}.
\end{equation}


In addition, \RS[7], \RS[9], and \RS[10] jointly imply
\begin{equation}\label{LRS5}
\textbf{G1}\wedge(\neg\,\textbf{G4})\wedge\textbf{E0}\wedge\textbf{D3}\Rightarrow\,\text{false}.
\end{equation}


One can easily verify that jointly \Ref{LRS3} and \Ref{LRS5} imply \SS[11]. The skeleton of the proof of \SS[11] is complete.

\section{Proofs of \RS[1] to \RS[10]}\label{ProofsRS[1]-RS[10]}

We prove \RS[1] as follows.
\begin{proof}
Suppose \DS[1] is true. By \CorRef{Cor2}, any channel gain cannot have the other channel gain as a factor. Therefore, $\ChGANA[2][3]$ must be reducible. Furthermore we must have $\GCD[{\ChGANA[2][1]}][{\,\ChGANA[2][3]}]\PolyNotEqual 1$ since $\ChGANA[2][1]$ is the only channel gain in the LHS of \DS[1] that reaches $d_2$. (See the proof of \LemRef{Lem8} for detailed discussion). Similarly, we must have $\GCD[{\ChGANA[1][3]}][{\,\ChGANA[2][3]}]$ $\PolyNotEqual 1$. \LemRef{Lem7} then implies $\Sover[3]\NotEqualEmpty$ and $\Dover[2]\NotEqualEmpty$.
\end{proof}


We prove \RS[2] as follows.
\begin{proof}
Suppose \GS[4]\AND\GS[8]\AND\DS[1] is true. From \GS[4]\AND\GS[8] being true, by definition, $e^\ast_3$ (resp. $e^\ast_2$) is the most downstream (resp. upstream) edge of $\Sover[3]$ (resp. $\Dover[2]$) and $e^\ast_3\PREC e^\ast_2$. 
For the following, we will prove that $\{e^\ast_3,e^\ast_2\}\!\subset\!\onecut[s_2][d_3]$.

We now consider $\ChG[e^\ast_2][e^\ast_3]$, a part of $\ChGANA[2][3]$. From \DS[1] and Property 2 of \GS[4], we have
\begin{align}\label{I.2-3}
\GCD[{\ChGANA[3][2]^{l_1}}][\;{\ChG[e^\ast_2][e^\ast_3]}] = \ChG[e^\ast_2][e^\ast_3],
\end{align}
for some positive integer $l_1$. This implies that $\ChG[e^\ast_2][e^\ast_3]$ is a factor of $\ChGANA[3][2]$. By \PropRef{Prop3}, we have $\{e^\ast_3,e^\ast_2\}\!\subset\!\onecut[s_2][d_3]$. The proof is thus complete.
\end{proof}

We prove \RS[3] as follows.
\begin{proof} Suppose \GS[4]\AND\GS[8]\AND\GS[9]\AND\DS[3] is true. Therefore, the $e^\ast_3$ (resp. $e^\ast_2$) defined in the properties of \GS[4] must also be the most downstream (resp. upstream) edge of $\Sover[3]$ (resp. $\Dover[2]$). Moreover, since $\{e^\ast_3,e^\ast_2\}\!\subset\!\onecut[s_2][d_3]$, we can express $\ChGANA[3][2]$ as $\ChGANA[3][2] = \ChG[e^\ast_3][e_{s_2}]\ChG[e^\ast_2][e^\ast_3]\ChG[e_{d_3}][e^\ast_2]$. For the following, we will prove that $e^\ast_3\!\in\!\onecut[s_1][d_1]\cup\onecut[s_2][d_1]$.

We use the following observation: For any edge $e'\!\in\!\onecut[s_3][d_2]$ that is in the upstream of $e^\ast_2$, there must exist a path from $s_1$ to $\tail[e^\ast_2]$ that does not use such $e'$. Otherwise, $e'\!\in\!\onecut[s_3][d_2]$ is also a $1$-edge cut separating $s_1$ and $d_2$, which contradicts that $e^\ast_2$ is the most upstream edge of $\Dover[2]$.

We now consider $\ChG[e_{d_1}][e^\ast_3]$, a factor of $\ChGANA[1][3]$. From \DS[3] and Property 2 of \GS[4], we have $\GCD[{\ChGANA[1][1]\ChGANA[3][1]^{l_3}\ChGANA[1][2]^{l_3}}][\,{\ChG[e_{d_1}][e^\ast_3]}]=\ChG[e_{d_1}][e^\ast_3]$.
By \PropRef{Prop3}, we must have $e^\ast_3\in\onecut[s_1][d_1]\cup\onecut[s_1][d_3]\cup\onecut[s_2][d_1]$. We also note that by the observation in the beginning of this proof, there exists a path from $s_1$ to $\tail[e^\ast_2]$  not using $e^\ast_3$. Furthermore, $e^\ast_2\in\onecut[s_2][d_3]$ implies that $e^\ast_2$ can reach $d_3$. These jointly shows that there exists a path from $s_1$ through $e^\ast_2$ to $d_3$ without using $e^\ast_3$, which means $e^\ast_3\not\in\onecut[s_1][d_3]$. Therefore, $e^\ast_3$ belongs to $\onecut[s_1][d_1]\cup\onecut[s_2][d_1]$. The proof of $e^\ast_2\!\in\!\onecut[s_1][d_1]\cup\onecut[s_1][d_3]$ can be derived similarly. The proof \RS[3] is thus complete.
\end{proof}

We prove \RS[4] as follows.
\begin{proof}
Assume \GS[4]\AND\GS[8]\AND\GS[9]\AND (\NotGS[10])\AND\GS[11]\AND \ES[0] is true. Recall that $e^\ast_3$ is the most downstream edge in $\Sover[3]$ and $e^\ast_2$ is the most upstream edge in $\Dover[2]$. For the following we construct 8 path segments that interconnects $s_1$ to $s_3$, $d_1$ to $d_3$, and two edges $e^\ast_3$ and $e^\ast_2$.

\noindent \makebox[1cm][l]{$\bullet$ $P_1$:}a path from $s_1$ to $\tail[e^\ast_2]$ without using $e^\ast_3$. This is always possible due to Properties 1 and 2 of \GS[4].

\noindent \makebox[1cm][l]{$\bullet$ $P_2$:}a path from $s_2$ to $\tail[e^\ast_3]$. This is always possible due to \GS[8] and \GS[9] being true.

\noindent \makebox[1cm][l]{$\bullet$ $P_3$:}a path from $s_3$ to $\tail[e^\ast_3]$. This is always possible due to \GS[4] and \GS[8] being true.

\noindent \makebox[1cm][l]{$\bullet$ $P_4$:}a path from $s_2$ to $d_1$ without using $e^\ast_3$. This is always possible due to \GS[10] being false.

\noindent \makebox[1cm][l]{$\bullet$ $P_5$:}a path from $\head[e^\ast_3]$ to $d_1$ without using $e^\ast_2$. This is always possible due to Properties 1 and 2 of \GS[4].

\noindent \makebox[1cm][l]{$\bullet$ $P_6$:}a path from $\head[e^\ast_3]$ to $\tail[e^\ast_2]$. This is always possible due to Property 1 of \GS[4].

\noindent \makebox[1cm][l]{$\bullet$ $P_7$:}a path from $\head[e^\ast_2]$ to $d_2$. This is always possible due to \GS[4] and \GS[8] being true.

\noindent \makebox[1cm][l]{$\bullet$ $P_8$:}a path from $\head[e^\ast_2]$ to $d_3$. This is always possible due to \GS[8] and \GS[9] being true.

\begin{figure}[t]
\centering
\includegraphics[scale=0.2]{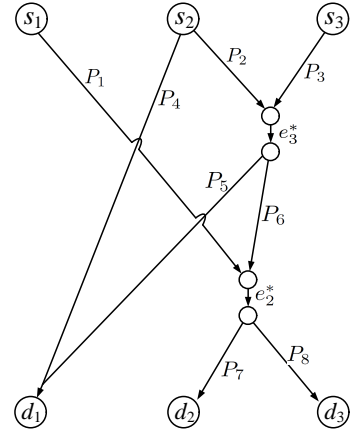}
\caption{The subgraph $\G[']$ of the 3-unicast ANA network $\GANA$ induced by the union of the 8 paths plus two edges $e^\ast_3$ and $e^\ast_2$ in the proof of \RS[4].} \vspace{-0.06\columnwidth}
\label{fig:RS[4]}
\end{figure}

Fig. \ref{fig:RS[4]} illustrates the relative topology of these 8 paths. We now consider the subgraph $\G[']$ induced by the 8 paths and two edges $e^\ast_3$ and $e^\ast_2$. One can easily check that $s_i$ can reach $d_j$ for any $i\!\neq\!j$. In particular, $s_1$ can reach $d_2$ through $P_1 e^\ast_2 P_7$; $s_1$ can reach $d_3$ through $P_1 e^\ast_2 P_8$; $s_2$ can reach $d_1$ through either $P_4$ or $P_2 e^\ast_3 P_5$; $s_2$ can reach $d_3$ through $P_2 e^\ast_3 P_6 e^\ast_2 P_8$; $s_3$ can reach $d_1$ through $P_3 e^\ast_3 P_5$; and $s_3$ can reach $d_2$ through $P_3 e^\ast_3 P_6 e^\ast_2 P_7$.

We first show the following topological relationships: $P_1$ is vertex-disjoint with $P_2$, $P_3$, and $P_4$, respectively, in the induced subgraph $\G[']$. From \GS[9], $\{P_1, P_2\}$ must be vertex-disjoint paths otherwise $s_2$ can reach $d_3$ without using $e^\ast_3\!\in\!\onecut[s_2][d_3]$. Similarly from the fact that $e^\ast_3\!\in\!\Sover[3]$, $\{P_1, P_3\}$ must be vertex-disjoint paths. Also notice that by \GS[11], $e^\ast_3$ is a $1$-edge cut separating $s_1$ and $d_1$ in the original graph. Therefore any \FromTo[1][1] path in the subgraph must use $e^\ast_3$ as well. But by definition, both $P_1$ and $P_4$ do not use $e^\ast_3$ and $s_1$ can reach $d_1$ if they share a vertex. This thus implies that $\{P_1, P_4\}$ are vertex-disjoint paths.

The above topological relationships further imply that $s_1$ cannot reach $d_1$ in the induced subgraph $\G[']$. The reason is as follows. We first note that $P_1$ is the only path segment that $s_1$ can use to reach other destinations, and any \FromTo[1][1] path, if exists, must use path segment $P_1$ in the very beginning. Since $P_1$ ends at $\tail[e^\ast_2]$, using path segment $P_1$ alone is not possible to reach $d_1$. Therefore, if a \FromTo[1][1] path exists, then at some point, it must use one of the other $7$ path segments $P_2$ to $P_8$. On the other hand, we also note that $e^\ast_3\!\in\!\onecut[s_1][d_1]$ and the path segments $P_5$ to $P_8$ are in the downstream of $e^\ast_3$. Therefore, for any \FromTo[1][1] path, if it uses any of the vertices of $P_5$ to $P_8$, it must first go through $\tail[e^\ast_3]$, the end point of path segments $P_2$ and $P_3$. As a result, we only need to consider the scenario in which one of $\{P_2, P_3, P_4\}$ is used by the \FromTo[1][1] path when this path switches from $P_1$ to a new path segment. But we have already showed that $P_1$ and $\{P_2,P_3,P_4\}$ are vertex-disjoint with each other. As a result, no \FromTo[1][1] path can exist. Thus $s_1$ cannot reach $d_1$ on the induced graph $\G[']$.


By \ES[0] being true and \PropRef{Prop2}, any subgraph who contains the source and destination edges (hence $\G[']$) must satisfy \ES[0]. Note that we already showed there is no \FromTo[1][1] path on $\G[']$. Recalling \Ref{E0}, its LHS becomes zero. Thus, we have $g(\{\ChGANA[j][i]:\forall\,(i,j)\!\in\!I_{\textrm{3ANA}}\})\,\psi^{(n)}_\beta(\BOLDb,\BOLDa) = 0$ with at least one non-zero coefficient $\beta_j$. But note also that any channel gain $\ChGANA[j][i]$ where $i\!\neq\!j$ is non-trivial on $\G[']$. Thus $\BOLDb$, $\BOLDa$, and $g(\{\ChGANA[j][i]:\forall\,(i,j)\!\in\!I_{\textrm{3ANA}}\})$ are all non-zero polynomials. Therefore, $\G[']$ must satisfy $\psi^{(n)}_\beta(\BOLDb,\BOLDa) = 0$ with at least one non-zero coefficient $\beta_j$ and this further implies that the set of polynomials $\{\BOLDb^{n},\BOLDb^{n-1}\BOLDa, \cdots, \BOLDb\BOLDa^{n-1}, \BOLDa^{n}\}$ is linearly dependent on $\G[']$. Since this is the Vandermonde form, it is equivalent to that $\LReq$ holds on $\G[']$.

For the following, we further show that in the induced graph $\G[']$, the following three statements are true: (a) $\Sover[2]\cap\Sover[3]\EqualEmpty$; (b) $\Sover[1]\cap\Sover[2]\EqualEmpty$; and (c) $\Sover[1]\cap\Sover[3]\EqualEmpty$, which implies by \PropRef{Prop4} that $\G[']$ must have $\LRneq$. We thus have a contradiction.

(a) $\Sover[2]\cap\Sover[3]\EqualEmpty$ on $\G[']$: Suppose there is an edge $e\!\in\!\Sover[2]\cap\Sover[3]$ on $\G[']$. Since $e\!\in\!\Sover[2]$, such $e$ must belong to $P_4$ and any \FromTo[2][3] path. Since both $e\!\in\!P_4$ and $e^\ast_3\!\not\in\!P_4$ belong to $\onecut[s_2][d_3]$, we have either $e\PREC e^\ast_3$ or $e\SUCC e^\ast_3$. We first note that $e$ must not be in the downstream of $e^\ast_3$. Otherwise, $s_2$ can use $P_4$ to reach $e$ without using $e^\ast_3$ and finally to $d_3$ (since $e\!\in\!\Sover[2])$, which contradicts the assumption of \GS[9] that $e^\ast_3\!\in\!\onecut[s_2][d_3]$. As a result, $e\PREC e^\ast_3$ and any path from $s_2$ to $\tail[e^\ast_3]$ must use $e$. This in turn implies that $P_2$ uses $e$. We now argue that $P_3$ must also use $e$. The reason is that the \FromTo[3][1] path $P_3 e^\ast_3 P_5$ must use $e$ since $e\!\in\!\Sover[3]$ and $e\PREC e^\ast_3$. Then these jointly contradict that $e^\ast_3\!\in\!\Sover[3]$ since $s_3$ can follow $P_3$, switch to $P_4$ through $e$, and reach $d_1$ without using $e^\ast_3$.


(b) $\Sover[1]\cap\Sover[2]\EqualEmpty$ on $\G[']$: Suppose there is an edge $e\!\in\!\Sover[1]\cap\Sover[2]$. Since $e\!\in\!\Sover[2]$, by the same arguments as used in proving (a), we know that $e\PREC e^\ast_3$ and $e$ must be used by both $P_2$ and $P_4$. We then note that $e$ must also be used by the \FromTo[1][3] path $P_1 e^\ast_2 P_8$ since $e\!\in\!\Sover[1]$. This in turn implies that $P_1$ must use $e$ since $e\PREC e^\ast_3 \PREC e^\ast_2$. However, these jointly contradict the fact $P_1$ and $\{P_2,P_3,P_4\}$ being vertex-disjoint, which were proved previously. The proof of (b) is complete.


(c) $\Sover[1]\cap\Sover[3]\EqualEmpty$ on $\G[']$: Suppose there is an edge $e\!\in\!\Sover[1]\cap\Sover[3]$. We then note that $e$ must be used by the \FromTo[1][3] path $P_1 e^\ast_2 P_8$ since $e\!\in\!\Sover[1]$. Then $e$ must be either $e^\ast_3$ or used by $P_3$ since $e^\ast_3$ is the most downstream edge of $\Sover[3]$. Therefore, $P_1$ must use $e$ (since $e^\ast_3\PREC e^\ast_2$). In addition, since by our construction $P_1$ does not use $e^\ast_3$, it is $P_3$ who uses $e$. However, $P_1$ and $P_3$ are vertex-disjoint with each other, which contradicts what we just derived $e\!\in\!P_1\cap P_3$. The proof of (c) is complete.
\end{proof}

We prove \RS[5] as follows.
\begin{proof} We notice that \RS[5] is a symmetric version of \RS[4] by simultaneously reversing the roles of sources and destinations and relabeling flow $2$ by flow $3$, i.e., we swap the roles of the following three pairs: $(s_1,d_1)$, $(s_2, d_3)$, and $(s_3,d_2)$. We can then reuse the proof of \RS[4].
\end{proof}

We prove \RS[6] as follows.
\begin{proof} Assume \GS[4]\AND\GS[8] is true and recall that $e^\ast_3$ is the most downstream edge in $\Sover[3]$ and $e^\ast_2$ is the most upstream edge in $\Dover[2]$. From \GS[9]\AND\GS[10]\AND\GS[12] being true, we further have $e^\ast_3\!\in\!\onecut[s_2][d_1]\cap\onecut[s_2][d_3]$ and $e^\ast_2\!\in\!\onecut[s_1][d_3]\cap\onecut[s_2][d_3]$. This implies that $e^\ast_3$ (resp. $e^\ast_2$) belongs to $\Sover[2]\cap\Sover[3]$ (resp. $\Dover[2]\cap\Dover[3]$). We thus have \NotLNR\;by \PropRef{Prop4}.
\end{proof}

We prove \RS[7] as follows.
\begin{proof} We prove an equivalent relationship: (\NotGS[4])\AND\\\noindent(\NotGS[14])\IMPLY(\NotGS[1]). From \GS[4] being false, we have $e^{32}_u\!\in\!\Sover[3]\CAP\Dover[2]\!\subset\!\onecut[s_3][d_2]\cap\onecut[s_1][d_2]\cap\onecut[s_3][d_1]$. From \GS[14] being false, we have $e^{32}_u\!\in\! \onecut[s_1][d_1]$. As a result, $e^{32}_u$ is a $1$-edge cut separating $\{s_1, s_3\}$ and $\{d_1,d_2\}$. This implies $\ChGANA[1][1]\ChGANA[2][3]\PolyEqual\ChGANA[2][1]\ChGANA[1][3]$ and thus \NotGS[1]. The proof of \RS[7] is thus complete.
\end{proof}

We prove \RS[8] as follows.
\begin{proof} Suppose that (\NotGS[4])\AND\GS[14] is true. From Property 3 of \NotGS[4], any \FromTo[1][1] path who uses a vertex $w$ where $\tail[e^{32}_u]\PRECEQ w \PRECEQ \head[e^{32}_v]$ must use both $e^{32}_u$ and $e^{32}_v$. Since we have $e^{32}_u\!\not\in\onecut[s_1][d_1]$ from \GS[14], there must exist a \FromTo[1][1] path not using $e^{32}_u$. Then, these jointly imply that there exists a \FromTo[1][1] path which does not use any vertex in-between $\tail[e^{32}_u]$ and $\head[e^{32}_v]$. Fix arbitrarily one such path as $P_{11}^\ast$.

If the chosen $P_{11}^\ast$ shares a vertex with any path segment from $s_3$ to $\tail[e^{32}_u]$, then $s_3$ can reach $d_1$ without using $e^{32}_u$, contradicting $e^{32}_u\!\in\!\Sover[3]\cap\Dover[2]\!\subset\!\onecut[s_3][d_1]$. By the similar argument, $P_{11}^\ast$ should not share a vertex with any path segment from $\head[e^{32}_v]$ to $d_2$. Then jointly with the above discussion, we can conclude that $P_{11}^\ast$ is vertex-disjoint with any \FromTo[3][2] path. We thus have proven (b) of \GS[15].

Now consider $\tilde{e}_u$ (we have at least the $s_1$-source edge $e_{s_1}$) and $\tilde{e}_v$ (we have at least the $d_1$-destination edge $e_{d_1}$) defined in \GS[15]. By definition, $\tilde{e}_u\PREC e^{32}_u$ and $e^{32}_v \PREC \tilde{e}_v$, and the chosen $P_{11}^\ast$ must use both $\tilde{e}_u$ and $\tilde{e}_v$. Thus if $\head[\tilde{e}_u]\!=\!\tail[e^{32}_u]$, then this contradicts the above discussion since $\tail[e^{32}_u]\!\in\!P_{11}^\ast$. Therefore, it must be $\head[\tilde{e}_u]\PREC\tail[e^{32}_u]$. Similarly, it must also be $\head[e^{32}_v]\PREC\tail[\tilde{e}_v]$. Thus we have proven (a) of \GS[15].

We now prove (c) of \GS[15]. We prove this by contradiction. Fix arbitrarily one edge $e\!\in\!P_{11}^\ast$ where $\tilde{e}_u\PREC e\PREC \tilde{e}_v$ and assume that this edge $e$ is reachable from either $e^{32}_u$ or $e^{32}_v$ or both. We first prove that whenever $e^{32}_u$ reaches $e$, then $e$ must be in the downstream of $e^{32}_v$. The reason is as follows. If $e^{32}_u$ reaches $e$, then $e\!\in\!P_{11}^\ast$ should not reach $e^{32}_v$ because it will be located in-between $e^{32}_u$ and $e^{32}_v$, and this contradicts the above discussion. The case when $e$ are $e^{32}_v$ are not reachable from each other is also not possible because $s_1$ can first reach $e$ through $e^{32}_u$ and follow $P_{11}^\ast$ to $d_1$ without using $e^{32}_v$, which contradicts the Property~3 of \NotGS[4]. Thus, if $e^{32}_u\PREC e$, then it must be $e^{32}_v \PREC e$. By the similar argument, we can show that if $e \PREC e^{32}_v$, it must be $e \PREC e^{32}_u$. Therefore, only two cases are possible when $e$ is reachable from either $e^{32}_u$ or $e^{32}_v$ or both: either $e \PREC e^{32}_u$ or $e^{32}_v \PREC e$. Extending this result to every edges of $P_{11}^\ast$ from $\tilde{e}_u$ to $\tilde{e}_v$, we can group them into two: edges in the upstream of $e^{32}_u$; and edges in the downstream of $e^{32}_v$. Since $\tilde{e}_u\PREC e^{32}_u\PREC e^{32}_v\PREC \tilde{e}_v$, this further implies that the chosen $P_{11}^\ast$ must be disconnected. This, however, contradicts the construction $P_{11}^\ast$. Therefore, there must exist an edge $\tilde{e}\!\in\!P_{11}^\ast$ where $\tilde{e}_u\PREC e\PREC \tilde{e}_v$ that is not reachable from any of $\{e^{32}_u, e^{32}_v\}$. We thus have proven (c) of \GS[15]. The proof of \RS[8] is complete.
\end{proof}

We prove \RS[9] as follows.
\begin{proof} Suppose (\NotGS[4])\AND\GS[14]\AND\DS[3] is true. From \RS[8], \GS[15] must also be true, and we will use the \FromTo[1][1] path $P_{11}^\ast$, the two edges $\tilde{e}_u$ and $\tilde{e}_v$, and the edge $\tilde{e}\!\in\!P_{11}^\ast$ defined in \GS[15]. 
For the following, we will prove that the specified $\tilde{e}$ satisfies \GS[7]. Since $\tilde{e}\!\in\!P_{11}^\ast$, we only need to prove that $\tilde{e}$ cannot be reached by any of $\{s_2, s_3\}$ and cannot reach any of $\{d_2,d_3\}$.

We first claim that $\tilde{e}$ cannot be reached from $s_3$. Suppose not. Then we can consider a new \FromTo[3][1] path: $s_3$ can reach $\tilde{e}$ and follow $P_{11}^\ast$ to $d_1$. Since $\tilde{e}$ is not reachable from any of $\{e^{32}_u,e^{32}_v\}$ by (c) of \GS[15], this new \FromTo[3][1] path must not use any of $\{e^{32}_u, e^{32}_v\}$. However, this contradicts the construction $\{e^{32}_u,e^{32}_v\}\!\subset\!\Sover[3]\CAP\Dover[2]\!\subset\!\onecut[s_3][d_1]$. We thus have proven the first claim that $\tilde{e}$ cannot be reached from $s_3$. Symmetrically, we can also prove that $\tilde{e}$ cannot reach $d_2$.

What remains to be proven is that $\tilde{e}$ cannot be reached from $s_2$ and cannot reach $d_3$. Since \DS[3] is true, there exists a positive integer $l_3$ satisfying $\GCD[{\ChGANA[1][1]\ChGANA[3][1]^{l_3}\ChGANA[2][3]^{l_3}\ChGANA[1][2]^{l_3}}][{\,\ChGANA[2][1]\ChGANA[1][3]}]\!=\!\ChGANA[2][1]\ChGANA[1][3]$. Consider $\ChG[e^{32}_u][e_{s_1}]$, a part of $\ChGANA[2][1]$, and $\ChG[e_{d_1}][e^{32}_v]$, a part of $\ChGANA[1][3]$. By Property~1 of \NotGS[4], we have
\begin{equation*}
\GCD[{\ChGANA[1][1]\ChGANA[3][1]^{l_3}\ChGANA[1][2]^{l_3}}][\;{\ChG[e^{32}_u][e_{s_1}]\ChG[d_{d_1}][e^{32}_v]}] = \ChG[e^{32}_u][e_{s_1}]\ChG[d_{d_1}][e^{32}_v].
\end{equation*}

Recall the definition of $\tilde{e}_u$ (resp. $\tilde{e}_v$) being the most downstream (resp. upstream) edge among $\onecut[s_1][{\tail[e^{32}_v]}]\cap\onecut[s_1][d_1]$ (resp. $\onecut[{\head[e^{32}_v]}][d_1]\cap\onecut[s_1][d_1]$). Then we can further factorize $\ChG[e^{32}_u][e_{s_1}]=\ChG[\tilde{e}_u][e_{s_1}]\ChG[e^{32}_u][\tilde{e}_u]$ and $\ChG[e_{d_1}][e^{32}_v]=\ChG[\tilde{e}_v][e^{32}_v]\ChG[e_{d_1}][\tilde{e}_v]$, respectively. Since both $\tilde{e}_u$ and $\tilde{e}_v$ separate $s_1$ and $d_1$, we can express $\ChGANA[1][1]$ as $\ChGANA[1][1]=\ChG[\tilde{e}_u][e_{s_1}]\ChG[\tilde{e}_v][\tilde{e}_u]\ChG[e_{d_1}][\tilde{e}_v]$. Then one can see that the middle part of $\ChGANA[1][1]$, i.e., $\ChG[\tilde{e}_v][\tilde{e}_u]$, must be co-prime to both $\ChG[e^{32}_u][\tilde{e}_u]$ and $\ChG[\tilde{e}_v][e^{32}_v]$, otherwise it violates the construction of $\tilde{e}_u$ (resp. $\tilde{e}_v$) being the most downstream (resp. upstream) edge among $\onecut[s_1][{\tail[e^{32}_v]}]\cap\onecut[s_1][d_1]$ (resp. $\onecut[{\head[e^{32}_v]}][d_1]\cap\onecut[s_1][d_1]$). 
The above equation thus reduces to
\begin{equation}\label{I.1-3}
\GCD[{\ChGANA[3][1]^{l_3}\ChGANA[1][2]^{l_3}}][\;{\ChG[e^{32}_u][\tilde{e}_u]\ChG[\tilde{e}_v][e^{32}_v]}] = \ChG[e^{32}_u][\tilde{e}_u]\ChG[\tilde{e}_v][e^{32}_v].
\end{equation}

Using \Ref{I.1-3} and the previous constructions, we first prove that $\tilde{e}$ cannot reach $d_3$. Since $\head[\tilde{e}_u]\PREC\tail[e^{32}_u]$ by (a) of \GS[15], we must have $0\!<\!\EC[{\head[\tilde{e}_u]}][{\tail[e^{32}_u]}]\!<\!\infty$. By \PropRef{Prop3}, $\ChG[e^{32}_u][\tilde{e}_u]$ is either irreducible or the product of irreducibles corresponding to the consecutive edges among $\tilde{e}_u$, $\onecut[{\head[\tilde{e}_u]}][{\tail[e^{32}_u]}]$, and $e^{32}_u$. Consider the following edge set $E_u\!=\!\{\tilde{e}_u\}\cup\onecut[{\head[\tilde{e}_u]}][{\tail[e^{32}_u]}]\cup\{e^{32}_u\}$, the collection of $\onecut[{\head[\tilde{e}_u]}][{\tail[e^{32}_u]}]$ and two edges $\tilde{e}_u$ and $e^{32}_u$. Note that in the proof of \RS[8], $P_{11}^\ast$ was chosen arbitrarily such that $\tilde{e}_u\!\in\!P_{11}^\ast$ and $e^{32}_u\!\not\in\!P_{11}^\ast$ but there was no consideration for the $1$-edge cuts from $\head[\tilde{e}_u]$ to $\tail[e^{32}_u]$ if non-empty. In other words, when $s_1$ follow the chosen $P_{11}^\ast$, it is obvious that it first meets $\tilde{e}_u$ but it is not sure when it starts to deviate not to use $e^{32}_u$ if we have non-empty $\onecut[{\head[\tilde{e}_u]}][{\tail[e^{32}_u]}]$. Let $e^{u}_1$ denote the most downstream edge of $E_u\cap P_{11}^\ast$ (we have at least $\tilde{e}_u$) and let $e^{u}_2$ denote the most upstream edge of $E_u \backslash P_{11}^\ast$ (we have at least $e^{32}_u$). From the constructions of $P_{11}^\ast$ and $E_u$, the defined edges $e^{u}_1\!\in\!P_{11}^\ast$ and $e^{u}_2\!\not\in\!P_{11}^\ast$ are edges of $E_u$ such that $\tilde{e}_u\PRECEQ e^{u}_1\PREC e^{u}_2 \PRECEQ e^{32}_u$; $\ChG[e^{u}_2][e^{u}_1]$ is irreducible; and $\ChG[e^{32}_u][\tilde{e}_u]$ contain $\ChG[e^{u}_2][e^{u}_1]$ as a factor. By doing this way, we can clearly specify the location (in-between $e^{u}_1\!\in\!P_{11}^\ast$ and $e^{u}_2\!\not\in\!P_{11}^\ast$) when $P_{11}^\ast$ starts to deviate not to use $e^{32}_u$.

For the following, we first argue that $\GCD[{\ChGANA[3][1]}][\,{\ChG[e^{u}_2][e^{u}_1]}]\,\PolyNotEqual 1$. Suppose not then we have $\GCD[{\ChGANA[1][2]}][\,{\ChG[e^{u}_2][e^{u}_1]}]\!=\!\ChG[e^{u}_2][e^{u}_1]$ from \Ref{I.1-3}. By \PropRef{Prop3}, we have $\{e^{u}_1,e^{u}_2\}\!\subset\!\onecut[s_2][d_1]$. However from the above construction, $e^{u}_1\!\in\!\onecut[s_2][d_1]$ implies that $s_2$ can first reach $e^{u}_1\!\in\!P_{11}^\ast$ and then follow $P_{11}^\ast$ to $d_1$ without using $e^{u}_2$ since $e^{u}_1\PREC e^{u}_2$ and $e^{u}_2\!\not\in\!P_{11}^\ast$. This contradicts $e^{u}_2\!\in\!\onecut[s_2][d_1]$ that we just established. This thus proves that $\GCD[{\ChGANA[3][1]}][\,{\ChG[e^{u}_2][e^{u}_1]}]\,\PolyNotEqual 1$. Since $\ChG[e^{u}_2][e^{u}_1]$ is irreducible, again by \PropRef{Prop3}, we further have $\{e^{u}_1,e^{u}_2\}\!\subset\!\onecut[s_1][d_3]$.


We now argue that $\tilde{e}$ cannot reach $d_3$. Suppose not and assume that there exists a path segment $Q$ from $\tilde{e}$ to $d_3$. Since $\tilde{e}\!\in\!P_{11}^\ast$ is not reachable from any of $\{e^{32}_u,e^{32}_v\}$ by (c) of \GS[15], it is obvious that $\tilde{e}$ must be in the downstream of $e^{u}_1\!\in\!P_{11}^\ast$ since $e^{u}_1 \PREC e^{32}_u$ from the above construction. Then when $s_1$ follow $P_{11}^\ast$ to $\tilde{e}$ (through $e^{u}_1$) and switch to $Q$ to reach $d_3$, it will not use $e^{u}_2$ unless $\tilde{e}\PREC e^{u}_2$ and $e^{u}_2\!\in\!Q$, but $\tilde{e}$ cannot be in the upstream of $e^{u}_2$ since $e^{u}_2\PRECEQ e^{32}_u$ from the above construction. Therefore, this \FromTo[1][3] path $P_{11}^\ast \tilde{e} Q$ will not use $e^{u}_2$ and thus contradicts $e^{u}_2\!\in\!\onecut[s_1][d_3]$ that we just established. As a result, $\tilde{e}$ cannot reach $d_3$.

The proof that $\tilde{e}$ cannot be reached from $s_2$ can be derived symmetrically. In particular, we can apply the above proof arguments ($\tilde{e}$ cannot reach $d_3$) by symmetrically using the following: the edge set $E_v\!=\!\{e^{32}_v\}\cup\onecut[{\head[e^{32}_v]}][{\tail[\tilde{e}_v]}]\cup\{\tilde{e}_v\}$ and denote $e^{v}_1$ (resp. $e^{v}_2$) be the most downstream (resp. upstream) edge of $E_v\backslash P_{11}^\ast$ (resp. $E_v\cap P_{11}^\ast$) such that $\{e^{v}_1,e^{v}_2\}\!\subset\!\onecut[s_2][d_1]$ from \Ref{I.1-3}.

Therefore we have proven that $\tilde{e}$ cannot be reached from $s_2$ and cannot reach $d_3$. The proof of \RS[9] is thus complete.
\end{proof}

We prove \RS[10] as follows.
\begin{proof} We prove an equivalent relationship: \GS[7]\!\IMPLY\!(\NotES[0]). Suppose \GS[7] is true and consider the edge $\tilde{e}$ defined in \GS[7]. Consider an \FromTo[1][1] path $P_{11}$ that uses $\tilde{e}$ and an edge $e\!\in\!P_{11}$ that is immediate downstream of $\tilde{e}$ along this path, i.e., $\head[\tilde{e}]\!=\!\tail[e]$. Such edge $e$ always exists since $\tilde{e}$ cannot be the $d_1$-destination edge $e_{d_1}$. (Recall that $\tilde{e}$ cannot be reached by $s_2$.) We now observe that since \GS[7] is true, such $e$ cannot reach any of $\{d_2,d_3\}$ (otherwise $\tilde{e}$ can reach one of $\{d_2,d_3\}$). Now consider a local kernel $x_{\tilde{e}e}$ from $\tilde{e}$ to $e$. Then, one can see that by the facts that $\tilde{e}$ cannot be reached by any of $\{s_2,s_3\}$ and $e$ cannot reach any of $\{d_2,d_3\}$, any channel gain $\ChGANA[j][i]$  where $i\!\neq\!j$ cannot depend on $x_{\tilde{e}e}$. On the other hand, the channel gain polynomial $\ChGANA[1][1]$ has degree $1$ in $x_{\tilde{e}e}$ since both $\tilde{e}$ and $e$ are used by a path $P_{11}$.

Since any channel gain $\ChGANA[j][i]$ where $i\!\neq\!j$ is non-trivial on a given $\GANA$, the above discussion implies that $f(\{\ChGANA[j][i]:\forall\,(i,j)\!\in\!I_{\textrm{3ANA}}\})$, $g(\{\ChGANA[j][i]:\forall\,(i,j)\!\in\!I_{\textrm{3ANA}}\})$, $\BOLDb$, and $\BOLDa$ become all non-zero polynomials, any of which does not depend on $x_{\tilde{e}e}$. Thus recalling \Ref{E0}, its RHS does not depend on $x_{\tilde{e}e}$. However, the LHS of \Ref{E0} has a common factor $\ChGANA[1][1]$ and thus has degree $1$ in $x_{\tilde{e}e}$. This implies that $\GANA$ does not satisfy \Ref{E0} if we have at least one non-zero coefficient $\alpha_i$ and $\beta_j$, respectively. This thus implies \NotES[0].
\end{proof}
%

\section{Proof of \SS[12]}\label{ProofSS[12]}

If we swap the roles of sources and destinations, then the proof of \SS[11] in \AppRef{ProofSS[11]} can be directly applied to show \SS[12]. More specifically, note that \DS[1] (resp. \DS[3]) are converted back and forth from \DS[2] (resp. \DS[4]) by such \SWAPSD-swapping. Also, one can easily verify that \LNR, \GS[1], and \ES[0] remain the same after the index swapping. Thus we can see that \SS[11] becomes \SS[12] after reverting flow indices. The proofs of \SS[11] in \AppRef{ProofSS[11]} can thus be used to prove \SS[12].

\section{Proof of \SS[13]}\label{ProofSS[13]}

\subsection{The fourth set of logic statements}

To prove \SS[13], we need the fourth set of logic statements.

\noindent $\bullet$ \GS[16]{\bf:}\label{ref:GS[16]} There exists a subgraph $\G[']\!\subset\!\GANA$ such that in $\G[']$ both the following conditions are true: (i) $s_i$ can reach $d_j$ for all $i\!\neq\!j$; and (ii) $s_1$ can reach $d_1$.

\noindent $\bullet$ \GS[17]{\bf:}\label{ref:GS[17]} Continue from the definition of \GS[16]. The considered subgraph $\G[']$ also contains an edge $\tilde{e}$ such that both the following conditions are satisfied: (i) $\tilde{e}$ can reach $d_1$ but cannot reach any of $\{d_2,d_3\}$; (ii) $\tilde{e}$ can be reached from $s_1$ but not from any of $\{s_2,s_3\}$.

\noindent $\bullet$ \GS[18]{\bf:}\label{ref:GS[18]} Continue from the definition of \GS[16]. There exists a subgraph $\G['']\!\subset\!\G[']$ such that (i) $s_i$ can reach $d_j$ for all $i\!\neq\!j$; and (ii) $s_1$ can reach $d_1$. Moreover, the considered subgraph $\G['']$ also satisfies (iii) $\ChGANA[1][1]\ChGANA[3][2] = \ChGANA[3][1]\ChGANA[1][2]$; and (iv) $\LRneq$.

\noindent $\bullet$ \GS[19]{\bf:}\label{ref:GS[19]} Continue from the definition of \GS[16]. There exists a subgraph $\G['']\!\subset\!\G[']$ such that (i) $s_i$ can reach $d_j$ for all $i\!\neq\!j$; and (ii) $s_1$ can reach $d_1$. Moreover, the considered subgraph $\G['']$ also satisfies (iii) $\ChGANA[1][1]\ChGANA[2][3] = \ChGANA[2][1]\ChGANA[1][3]$; and (iv) $\LRneq$.

\noindent $\bullet$ \GS[20]{\bf:}\label{ref:GS[20]} $\Sover[2]\NotEqualEmpty$ and $\Dover[3]\NotEqualEmpty$.

The following logic statements are well-defined if and only if \GS[3]\AND\GS[20] is true. Recall the definition of $e^\ast_2$ and $e^\ast_3$ when \GS[3] is true.

\noindent $\bullet$ \GS[21]{\bf:}\label{ref:GS[21]} $\{e^\ast_2,e^\ast_3\}\subset \onecut[s_3][d_2]$.

The following logic statements are well-defined if and only if (\NotGS[3])\AND(\NotGS[4]) is true. Recall the definition of $e^{23}_u$, $e^{23}_v$, $e^{32}_u$, and $e^{32}_v$ when (\NotGS[3])\AND(\NotGS[4]) is true.

\noindent $\bullet$ \GS[22]{\bf:}\label{ref:GS[22]} There exists a path $P_{11}^\ast$ from $s_1$ to $d_1$ who does not use any vertex in-between $\tail[e^{23}_u]$ and $\head[e^{23}_v]$, and any vertex in-between $\tail[e^{32}_u]$ and $\head[e^{32}_v]$.


\noindent $\bullet$ \GS[23]{\bf:}\label{ref:GS[23]} $e^{23}_u \PREC e^{32}_u$.

\noindent $\bullet$ \GS[24]{\bf:}\label{ref:GS[24]} $e^{32}_u \PREC e^{23}_u$.

\noindent $\bullet$ \GS[25]{\bf:}\label{ref:GS[25]} $e^{32}_u \PREC e^{23}_v$.

\noindent $\bullet$ \GS[26]{\bf:}\label{ref:GS[26]} $e^{23}_u \PREC e^{32}_v$.

\subsection{The skeleton of proving \SS[13]}

We prove the following relationships, which jointly prove \SS[13]. The proofs for the following statements are relegated to \AppRef{ProofsRS[11]-RS[25]}.

\noindent $\bullet$ \RS[11]{\bf:}\label{ref:RS[11]} \DS[1]\IMPLY \GS[8] (identical to \RS[1]).

\noindent $\bullet$ \RS[12]{\bf:}\label{ref:RS[12]} \GS[4]\AND\GS[8]\AND\DS[1]\IMPLY \GS[9] (identical to \RS[2]).

\noindent $\bullet$ \RS[13]{\bf:}\label{ref:RS[13]} \LNR\AND\GS[4]\AND\GS[8]\AND\GS[9]\AND\DS[2]\IMPLY \CONT.

\noindent $\bullet$ \RS[14]{\bf:}\label{ref:RS[14]} \DS[2]\IMPLY \GS[20].

\noindent $\bullet$ \RS[15]{\bf:}\label{ref:RS[15]} \GS[3]\AND\GS[20]\AND\DS[2]\IMPLY \GS[21].

\noindent $\bullet$ \RS[16]{\bf:}\label{ref:RS[16]} \LNR\AND\GS[3]\AND\GS[20]\AND\GS[21]\AND\DS[1]\IMPLY \CONT.

\noindent $\bullet$ \RS[17]{\bf:}\label{ref:RS[17]} \LNR\AND\GS[2]\AND(\NotGS[3])\AND(\NotGS[4])\AND(\NotGS[5])\IMPLY \GS[7].

\noindent $\bullet$ \RS[18]{\bf:}\label{ref:RS[18]} \GS[16]\AND\GS[17]\AND\ES[0]\IMPLY \CONT.

\noindent $\bullet$ \RS[19]{\bf:}\label{ref:RS[19]} \GS[16]\AND(\GS[18]\OR\GS[19])\AND\ES[0]\AND\ES[1]\AND\ES[2]\IMPLY \CONT.

\noindent $\bullet$ \RS[20]{\bf:}\label{ref:RS[20]} \GS[1]\AND(\NotGS[3])\AND(\NotGS[4])\AND(\NotGS[22])\AND\GS[23]\IMPLY \GS[16]\AND\GS[18].

\noindent $\bullet$ \RS[21]{\bf:}\label{ref:RS[21]} \LNR\AND(\NotGS[3])\AND(\NotGS[4])\AND\GS[22]\AND\GS[23]\AND\GS[25]\IMPLY \GS[16]\AND\\\noindent\GS[17].

\noindent $\bullet$ \RS[22]{\bf:}\label{ref:RS[22]} \LNR\AND(\NotGS[3])\AND(\NotGS[4])\AND\GS[22]\AND\GS[23]\AND(\NotGS[25])\IMPLY \\\noindent\GS[16]\AND(\GS[17]\OR\GS[18]).

\noindent $\bullet$ \RS[23]{\bf:}\label{ref:RS[23]} \GS[1]\AND(\NotGS[3])\AND(\NotGS[4])\AND(\NotGS[22])\AND\GS[24]\IMPLY \GS[16]\AND\GS[19].

\noindent $\bullet$ \RS[24]{\bf:}\label{ref:RS[24]} \LNR\AND(\NotGS[3])\AND(\NotGS[4])\AND\GS[22]\AND\GS[24]\AND\GS[26]\IMPLY \GS[16]\AND\\\noindent\GS[17].

\noindent $\bullet$ \RS[25]{\bf:}\label{ref:RS[25]} \LNR\AND(\NotGS[3])\AND(\NotGS[4])\AND\GS[22]\AND\GS[24]\AND(\NotGS[26])\IMPLY \\\noindent\GS[16]\AND(\GS[17]\OR\GS[19]).

One can easily verify that jointly \RS[11] to \RS[13] imply
\begin{equation}\label{R15-1}
\textbf{LNR}\wedge\textbf{G4}\wedge\textbf{D1}\wedge\textbf{D2}\Rightarrow\,\text{false}.
\end{equation}

Similarly, \RS[14] to \RS[16] jointly imply
\begin{equation}\label{R15-2}
\textbf{LNR}\wedge\textbf{G3}\wedge\textbf{D1}\wedge\textbf{D2}\Rightarrow\,\text{false}.
\end{equation}

Thus, \Ref{R15-1} and \Ref{R15-2} together imply
\begin{equation}\label{R15-3}
\textbf{LNR}\wedge(\textbf{G3}\vee\textbf{G4})\wedge\textbf{D1}\wedge\textbf{D2}\Rightarrow\,\text{false}.
\end{equation}

Now recall \RS[10], i.e., \GS[7]\AND\ES[0]\IMPLY\CONT. Then, jointly \RS[10] and \RS[17] imply
\begin{align}
\textbf{LNR}\wedge\textbf{G2}\wedge(\neg\,\textbf{G3})\wedge(\neg\,\textbf{G4})\wedge(\neg\,\textbf{G5})\wedge\textbf{E0}\Rightarrow\text{false}. \label{R15-5}
\end{align}

One can easily verify that jointly \RS[18] and \RS[19] imply
\begin{align}
\textbf{G16}\wedge(\textbf{G17}\vee\textbf{G18}\vee\textbf{G19})\wedge\textbf{E0}\wedge\textbf{E1}\wedge\textbf{E2}\Rightarrow\text{false}.
\label{R15-6}
\end{align}

One can see that jointly \Ref{R15-6}, \RS[20], \RS[21], and \RS[22] imply
\begin{equation}\begin{split}
& \textbf{LNR}\wedge\textbf{G1}\wedge(\neg\,\textbf{G3})\wedge(\neg\,\textbf{G4})\wedge\textbf{G23} \\
& \qquad\qquad\qquad\qquad\qquad \wedge\textbf{E0}\wedge\textbf{E1}\wedge\textbf{E2}\Rightarrow\text{false}.
\label{R15-7}
\end{split}\end{equation}

By similar arguments as used in deriving \Ref{R15-7}, jointly \Ref{R15-6}, \RS[23], \RS[24], and \RS[25] imply
\begin{equation}\begin{split}
& \textbf{LNR}\wedge\textbf{G1}\wedge(\neg\,\textbf{G3})\wedge(\neg\,\textbf{G4})\wedge\textbf{G24} \\
& \qquad\qquad\qquad\qquad\qquad \wedge\textbf{E0}\wedge\textbf{E1}\wedge\textbf{E2}\Rightarrow\text{false}.
\label{R15-8}
\end{split}\end{equation}

Since by definition (\NotGS[3])\AND(\NotGS[4])\AND\GS[5]\IMPLY (\NotGS[3])\AND(\NotGS[4])\\\noindent\AND(\GS[23]\OR\GS[24]),
jointly \Ref{R15-7} and \Ref{R15-8} imply
\begin{equation}\begin{split}
& \textbf{LNR}\wedge\textbf{G1}\wedge(\neg\,\textbf{G3})\wedge(\neg\,\textbf{G4})\wedge\textbf{G5} \\
& \qquad\qquad\qquad\qquad\qquad \wedge\textbf{E0}\wedge\textbf{E1}\wedge\textbf{E2}\Rightarrow\,\text{false}.
\label{R15-9}
\end{split}\end{equation}

By similar arguments as used in deriving \eqref{R15-7}, \Ref{R15-9} and \Ref{R15-5} further imply
\begin{equation}\begin{split}
& \textbf{LNR}\wedge\textbf{G1}\wedge\textbf{G2}\wedge(\neg\,\textbf{G3})\wedge(\neg\,\textbf{G4}) \\
& \qquad\qquad\qquad\qquad\qquad \wedge\textbf{E0}\wedge\textbf{E1}\wedge\textbf{E2}\Rightarrow\,\text{false}.
\label{R15-10}
\end{split}\end{equation}

Finally, one can easily verify that jointly \Ref{R15-3} and \Ref{R15-10} imply that we have \LNR\AND\GS[1]\AND\GS[2]\AND\ES[0]\AND\ES[1]\AND\ES[2]\AND\DS[1]\AND\DS[2]\IMPLY\\\CONT, which proves \SS[13]. The skeleton of the proof of \SS[13] is complete.

\section{Proofs of \RS[11] to \RS[25]}\label{ProofsRS[11]-RS[25]}

Since \RS[11] and \RS[12] is identical to \RS[1] and \RS[2], respectively, see \AppRef{ProofsRS[1]-RS[10]} for their proofs.



We prove \RS[13] as follows.
\begin{proof} We prove an equivalent relationship: \GS[4]\AND\GS[8]\AND\\ \noindent\GS[9]\AND\DS[2]\IMPLY \NotLNR. Suppose \GS[4]\AND\GS[8]\AND\GS[9] is true. The $e^\ast_3$ (resp. $e^\ast_2$) defined in the properties of \GS[4] must be the most downstream (resp. upstream) edge of $\Sover[3]$ (resp. $\Dover[2]$), both of which belongs to $\onecut[s_2][d_3]$.

For the following, we will prove that there exists an edge in-between $\{e_{s_2},e_{s_3}\}$ and $e^\ast_3$ who belongs to $\Sover[2]\cap\Sover[3]$. We will also prove that there exists an edge in-between $e^\ast_2$ and $\{e_{d_2},e_{d_3}\}$ who belongs to $\Dover[2]\cap\Dover[3]$. By \PropRef{Prop4} we thus have \LNR\;being false.

Define a node $u\!=\!\tail[e^\ast_3]$. Since $e^\ast_3\!\in\!\onecut[s_2][d_3]$, $u$ is reachable from $s_2$. Since $e^\ast_3\!\in\!\Sover[3]$, $u$ is also reachable form $s_3$. Consider the set of edges $\{\onecut[s_2][u]\cap\onecut[s_3][u]\}\cup\{e^\ast_3\}$ and choose $e''$ as the most upstream one (we have at least $e^\ast_3$). Let $e'$ denote the most downstream edge of $\onecut[s_2][{\tail[e'']}]$ (we have at least the $s_2$-source edge $e_{s_2}$). Since we choose $e'$ to be the most downstream one, by \PropRef{Prop3} the channel gain $\ChG[e''][e']$ must be irreducible.

Moreover, since $e^\ast_3\!\in\! \onecut[s_2][d_3]$, both $e'$ and $e''$ must be in $\onecut[s_2][d_3]$. The reason is that by $e^\ast_3\!\in\!\onecut[s_2][d_3]$ any path from $s_2$ to $d_3$ must use $e^\ast_3$, which in turn implies that any path from $s_2$ to $d_3$ must use $e''$ since $e''$ separates $s_2$ and $\tail[e^\ast_3]$. Therefore $e''\!\in\!\onecut[s_2][d_3]$. Similarly, any \FromTo[2][3] path must use $e''$, which means any \FromTo[2][3] path must use $e'$ as well since $e'\!\in\!\onecut[s_2][{\tail[e'']}]$. As a result, the channel gain $\ChGANA[3][2]$ contains $\ChG[e''][e']$ as a factor.

Since \DS[2] is true, it implies that $\ChG[e''][e']$ must be a factor of one of the following three channel gains $\ChGANA[3][1]$, $\ChGANA[2][3]$, and $\ChGANA[1][2]$. We first argue that $\ChG[e''][e']$ is not a factor of $\ChGANA[2][3]$. The reason is that if $\ChG[e''][e']$ is a factor of $\ChGANA[2][3]$, then $e'\!\in\!\onecut[s_3][d_2]$, which means that $e'\!\in\!\onecut[s_3][{\tail[e^\ast_3]}]$. Since $e'$ is also in $\onecut[s_2][{\tail[e^\ast_3]}]$, this contradicts the construction that $e''$ is the most upstream edge of $\onecut[s_2][{\tail[e^\ast_3]}]\cap\onecut[s_3][{\tail[e^\ast_3]}]$.

Now we argue that $\GCD[{\ChGANA[3][1]}][\,{\ChG[e''][e']}]\PolyEqual 1$. Suppose not. Then since $\ChG[e''][e']$ is irreducible, \PropRef{Prop3} implies that $\{e',e''\}$ are $1$-edge cuts separating $s_1$ and $d_3$. Also from Property~1 of \GS[4], there always exists a path segment from $s_1$ to $e^\ast_2$ without using $e^\ast_3$. Since $e^\ast_2\!\in\!\onecut[s_2][d_3]$, $e^\ast_2$ can reach $d_3$ and we thus have a \FromTo[1][3] path without using $e^\ast_3$. However by the assumption that $e'\!\in\!\onecut[s_1][d_3]$, this chosen path must use $e'$. As a result, $s_2$ can first reach $e'$ and then reach $d_3$ through the chosen path without using $e^\ast_3$, which contradicts the assumption \GS[9], i.e., $e^\ast_3\!\in\!\onecut[s_2][d_3]$.

From the above discussion $\ChG[e''][e']$ must be a factor of $\ChGANA[1][2]$, which by \PropRef{Prop3} implies that $\{e',e''\}$ also belong to $\onecut[s_2][d_1]$. Since by our construction $e''$ satisfies $e''\!\in\!\Sover[3]\cap\onecut[s_2][d_3]$, we have thus proved that $e''\!\in\!\Sover[2]\cap\Sover[3]$. The proof for the existence of an edge satisfying $\Dover[2]\cap\Dover[3]$ can be followed symmetrically. The proof of \RS[12] is thus complete.
\end{proof}

By swapping the roles of $s_2$ and $s_3$, and the roles of $d_2$ and $d_3$, the proofs of \RS[11] to \RS[13] can also be used to prove \RS[14] to \RS[16], respectively. More specifically, \DS[1] and \DS[2] are converted back and forth from each other when swapping the flow indices. The same thing happens between \GS[3] and \GS[4]; between \GS[20] and \GS[8]; and between \GS[21] and \GS[9]. Moreover, \LNR\;remains the same after the index swapping. The above proofs can thus be used to prove \RS[14] to \RS[16].

We prove \RS[17] as follows.
\begin{proof} Suppose \LNR\AND\GS[2]\AND(\NotGS[3])\AND(\NotGS[4])\AND(\NotGS[5]) is true. Recall the definitions of  $e^{23}_u$, $e^{32}_u$, $e^{23}_v$, and $e^{32}_v$ from Properties of both \NotGS[3] and \NotGS[4]. Since
\LNR\AND(\NotGS[3])\AND(\NotGS[4]) is true, we have \GS[6] if we recall \NS[7]. Together with \NotGS[5], $e^{23}_u$ and $e^{32}_u$ are distinct and not reachable from each other. Thus from \GS[2] being true, there must exist a \FromTo[1][1] path who does not use any of $\{e^{23}_u, e^{32}_u\}$. Combined with Property 3 of \NotGS[3] and \NotGS[4], this further implies that such \FromTo[1][1] path does not use any of $\{e^{23}_v, e^{32}_v\}$. Fix one such \FromTo[1][1] path as $P_{11}^\ast$. 

We will now show that there exists an edge in $P_{11}^\ast$ satisfying \GS[7]. To that end, we will show that if an edge $e\!\in\!P_{11}^\ast$ can be reached from $s_2$, then it must be in the downstream of $e^{23}_v$. We first argue that $e^{23}_v$ and $e$ are reachable from each other. The reason is that we now have a \FromTo[2][1] path by first going from $s_2$ to $e\!\in\!P_{11}^\ast$ and then use $P_{11}^\ast$ to $d_1$. Since $e^{23}_v\!\in\!\onecut[s_2][d_1]$ by definition, such path must use $e^{23}_v$. As a result, we either have $e^{23}_v \PREC e$ or $e\PREC e^{23}_v$. ($e\!=\!e^{23}_v$ is not possible since $e^{23}_v\!\not\in\!P_{11}^\ast$.) We then prove that $e\PREC e^{23}_v$ is not possible. The reason is that $P_{11}^\ast$ does not use $e^{23}_u$ and thus $s_1$ must not reach $e^{23}_v$ through $P_{11}^\ast$ due to Property~3 of \NotGS[3]. As a result, we must have $e^{23}_v\PREC e$. By symmetric arguments, any $e\!\in\!P_{11}^\ast$ that can be reached from reach $s_3$  must be in the downstream of $e^{32}_v$ and any $e\!\in\!P_{11}^\ast$ that can reach $d_3$ (resp. $d_2$) must be in the upstream of $e^{23}_u$ (resp. $e^{32}_u$).

For the following, we prove that there exists an edge $\tilde{e}\!\in\!P_{11}^\ast$ that cannot reach any of $\{d_2,d_3\}$, and that cannot be reached from any of $\{s_2,s_3\}$. Since $\tilde{e}\!\in\!P_{11}^\ast$, this will imply \GS[7]. Let $e'$ denote the most downstream edge of $P_{11}^\ast$ that can reach at least one of $\{d_2, d_3\}$ (we have at least the $s_1$-source edge $e_{s_1}$). Among all the edges in $P_{11}^\ast$ that are downstream of $e'$, let $e''$ denote the most upstream one that can be reached by at least one of $\{s_2, s_3\}$ (we have at least the $d_1$-destination edge $e_{d_1}$). In the next paragraph, we argue that $e''$ is not the immediate downstream edge of $e'$, i.e., $\head[e']\PREC \tail[e'']$. This conclusion directly implies that we have at least one edge $\tilde{e}$ that satisfies \GS[7] (which is in-between $e'$ and $e''$).

Without loss of generality, assume that $\head[e']\!=\!\tail[e'']$ and $e'$ can reach $d_2$. Then, by our previous arguments, $e'$ is an upstream edge of $e^{32}_u$. Consider two cases: Case 1: Suppose $e''$ is reachable from $s_3$, then by our previous arguments, $e''$ is a downstream edge of $e^{32}_v$. However, this implies that we can go from $\head[e']$ through $e^{32}_u$ to $e^{32}_v$ and then back to $\tail[e'']\!=\!\head[e']$, which contradicts the assumption that $\G[]$ is acyclic. Consider the Case 2: $e''$ is reachable from $s_2$. Then by our previous arguments, $e''$ is a downstream edge of $e^{23}_v$. Then we can go from $e^{23}_u$ to $e^{23}_v$, then to $\tail[e'']\!=\!\head[e']$ and then to $e^{32}_u$. This contradicts the assumption of \NotGS[5]. The proof of \RS[17] is thus complete.
\end{proof}

We prove \RS[18] as follows.
\begin{proof} Suppose \GS[16]\AND\GS[17]\AND\ES[0] is true. From \ES[0] being true, $\GANA$ satisfies \Ref{E0} with at least two non-zero coefficients $\alpha_i$ and $\beta_j$. From \GS[16] being true, the considered subgraph $\G[']$ has the non-trivial channel gain polynomials $\ChGANA[j][i]$ for all $i\!\neq\!j$ and $\ChGANA[1][1]$. By \PropRef{Prop2}, $\G[']$ also satisfies \Ref{E0} with the same set of non-zero coefficients $\alpha_i$ and $\beta_j$.

From \GS[17] being true, consider the defined edge $\tilde{e}\!\in\!\G[']$ that cannot reach any of $\{d_2,d_3\}$ (but reaches $d_1$) and cannot be reached by any of $\{s_2,s_3\}$ (but reached from $s_1$). This chosen $\tilde{e}$ must not be the $s_1$-source edge $e_{s_1}$ otherwise ($\tilde{e}\!=\!e_{s_1}$) $\tilde{e}$ will reach $d_2$ or $d_3$ and thus contradict the assumption \GS[17].

Choose an edge $e\!\in\!\G[']$ such that $e_{s_1}\PRECEQ e$ and $\head[e]\!=\!\tail[{\tilde{e}}]$. This is always possible because $s_1$ can reach $\tilde{e}$ and $e_{s_1}\PREC \tilde{e}$ on $\G[']$. Then, this chosen edge $e$ should not be reached from $s_2$ or $s_3$ otherwise $s_2$ or $s_3$ can reach $\tilde{e}$ and this contradicts the assumption \GS[17]. Now consider a local kernel $x_{e\tilde{e}}$ from $e$ to $\tilde{e}$. Then, one can quickly see that the channel gains $\ChGANA[1][2]$, $\ChGANA[3][2]$, $\ChGANA[1][3]$, and $\ChGANA[2][3]$ must not have $x_{e\tilde{e}}$ as a variable since $e$ is not reachable from $s_2$ nor $s_3$. Also $\ChGANA[2][1]$ and $\ChGANA[3][1]$ must not have $x_{e\tilde{e}}$ as a variable since $\tilde{e}$ doe not reach any of $\{d_2,d_3\}$.


This further implies that the RHS of \Ref{E0} does not depend on $x_{e\tilde{e}}$. However, the LHS of \Ref{E0} has a common factor $\ChGANA[1][1]$ and thus has degree $1$ in $x_{e\tilde{e}}$. This contradicts the above discussion that $\G[']$ also satisfies \Ref{E0}.
\end{proof}


We prove \RS[19] as follows.
\begin{proof} Equivalently, we prove the following two relationships: \GS[16]\AND\GS[18]\AND\ES[0]\AND\ES[1]\AND\ES[2]\IMPLY\CONT;\;and \GS[16]\AND\GS[19]\AND\\\noindent\ES[0]\AND\ES[1]\AND\ES[2]\IMPLY\,\CONT.

We first prove the former. Suppose that \GS[16]\AND\GS[18]\AND\ES[0]\AND\\\noindent\ES[1]\AND\ES[2] is true. From \ES[0]\AND\ES[1]\AND\ES[2] being true, there exists some coefficient values $\{\alpha_i\}_{i=0}^{\Nid}$ and $\{\beta_j\}_{j=0}^{\Nid}$ such that $\GANA$ of interest satisfies
\begin{align}\label{R34-1}
\ChGANA[1][1]\ChGANA[3][2]\,\psi^{(n)}_\alpha(\BOLDb,\BOLDa) = \ChGANA[3][1]\ChGANA[1][2]\,\psi^{(n)}_\beta(\BOLDb,\BOLDa),
\end{align}
with (i) At least one of $\alpha_i$ is non-zero; (ii) At least one of $\beta_j$ is non-zero; (iii) $\alpha_k\!\neq\!\beta_k$ for some $k\!\in\!\{0,...,\Nid\}$; and (iv) either $\alpha_0\!\neq\!0$ or $\beta_{\Nid}\!\neq\!0$ or $\alpha_{k}\!\neq\!\beta_{k-1}$ for some $k\!\in\!\{1,...,\Nid\}$.

From the assumption that \GS[16] is true, consider a subgraph $\G[']$ which has the non-trivial channel gain polynomials $\ChGANA[j][i]$ for all $i\!\neq\!j$ and $\ChGANA[1][1]$. Thus by \PropRef{Prop2}, $\G[']$ also satisfies \Ref{R34-1} with the same coefficient values.

Now from \GS[18] being true, we will prove the first relationship, i.e., \GS[16]\AND\GS[18]\AND\ES[0]\AND\ES[1]\AND\ES[2]\IMPLY\CONT. Since \GS[18] is true, there exists a subgraph $\G['']\!\subset\!\G[']$ which also has the non-trivial channel gains $\ChGANA[j][i]$ for all $i\!\neq\!j$ and $\ChGANA[1][1]$. Thus again by \PropRef{Prop2}, $\G['']$ satisfies \Ref{R34-1} with the same coefficients. Since $\G['']$ also satisfies $\ChGANA[1][1]\ChGANA[3][2] = \ChGANA[3][1]\ChGANA[1][2]$, by \Ref{R34-1}, we know that $\G['']$ satisfies
\begin{align}\label{R34-3}
\psi^{(n)}_\alpha(\BOLDb,\BOLDa) = \psi^{(n)}_\beta(\BOLDb,\BOLDa).
\end{align}

Note that by (iii), the coefficient values were chosen such that $\alpha_k\!\neq\!\beta_k$ for some $k\!\in\!\{0,...,\Nid\}$. Then \Ref{R34-3} further implies that $\G['']$ satisfies $\sum_{k=0}^{\Nid} \gamma_k\BOLDb^{\Nid-k}\BOLDa^{k}=0$ with at least one non-zero $\gamma_k$. Equivalently, this means that the set of polynomials $\{\BOLDb^{n}, \BOLDb^{n-1}\BOLDa , \cdots, \BOLDb\BOLDa^{n-1}, \BOLDa^{n}\}$ is linearly dependent. Since this is the Vandermonde form, it is equivalent to that $\LReq$ holds on $\G['']$. However, this contradicts the assumption \GS[18] that $\G['']$ satisfies $\LRneq$.

To prove the second relationship, i.e., \GS[16]\AND\GS[19]\AND\ES[0]\AND\ES[1]\\\noindent\AND\ES[2]\IMPLY\CONT, we assume \GS[19] is true. Since \GS[19] is true, there exists a subgraph $\G['']\!\subset\!\G[']$ which also has the non-trivial channel gains $\ChGANA[j][i]$ for all $i\!\neq\!j$ and $\ChGANA[1][1]$. Thus again by \PropRef{Prop2}, $\G['']$ satisfies \Ref{R34-1} with the same coefficients. Moreover, $\G['']$ satisfies $\ChGANA[1][1]\ChGANA[2][3] = \ChGANA[2][1]\ChGANA[1][3]$, which together with \Ref{R34-1} imply that $\G['']$ also satisfies
\begin{align}\label{R34-4}
\BOLDb\,\psi^{(n)}_\alpha(\BOLDb,\BOLDa) = \BOLDa\,\psi^{(n)}_\beta(\BOLDb,\BOLDa),
\end{align}
where we first multiply $\ChGANA[2][3]$ on both sides of \Ref{R34-1}.

Expanding \Ref{R34-4}, we have
\begin{equation}\label{R34-5}\begin{split}
& \BOLDb\,\psi^{(n)}_\alpha(\BOLDb,\BOLDa) - \BOLDa\,\psi^{(n)}_\beta(\BOLDb,\BOLDa) \\
& = \alpha_0 \BOLDb^{\Nid+1} + \sum_{k=1}^{\Nid} (\alpha_k - \beta_{k-1}) \BOLDb^{\Nid+1-k}\BOLDa^{k} + \beta_{\Nid}\BOLDa^{\Nid+1} \\
& = \sum_{k=0}^{\Nid+1} \gamma_k \BOLDb^{\Nid+1-k} \BOLDa^{k} = 0
\end{split}\end{equation}

By (iv), the coefficient values were chosen such that either $\alpha_0\!\neq\!0$ or $\beta_{\Nid}\!\neq\!0$ or $\alpha_{k}\!\neq\!\beta_{k-1}$ for some $k\!\in\!\{1,...,\Nid\}$. Then \Ref{R34-5} further implies that $\G['']$ satisfies $\sum_{k=0}^{\Nid+1} \gamma_k\BOLDb^{\Nid+1-k}\BOLDa^{k}=0$ with some non-zero $\gamma_k$. Equivalently, this means that the set of polynomials $\{\BOLDb^{\Nid+1}, \BOLDb^{\Nid}\BOLDa , \cdots, \BOLDb\BOLDa^{\Nid}, \BOLDa^{\Nid+1}\}$ is linearly dependent, and thus $\G['']$ satisfies $\LReq$. This contradicts the assumption \GS[19] that $\LRneq$ holds on $\G['']$. The proof of \RS[19] is thus complete.
\end{proof}

We prove \RS[20] as follows.
\begin{proof} Suppose \GS[1]\AND(\NotGS[3])\AND(\NotGS[4])\AND(\NotGS[22])\AND\GS[23] is true. Recall the definitions of  $e^{23}_u$, $e^{32}_u$, $e^{23}_v$, and $e^{32}_v$ when (\NotGS[3])\AND(\NotGS[4]) is true. From Property 1 of both \NotGS[3] and \NotGS[4], we know that $s_1$ can reach $e^{23}_u$ (resp. $e^{32}_u$) and then use $e^{23}_v$ (resp. $e^{32}_v$) to arrive at $d_1$. Note that \NotGS[22] being true implies that every \FromTo[1][1] path must use a vertex $w$ in-between $\tail[e^{23}_u]$ and $\head[e^{23}_v]$ or in-between $\tail[e^{32}_u]$ and $\head[e^{32}_v]$ or both. Combined with Property 3 of both \NotGS[3] and \NotGS[4], this further implies that every \FromTo[1][1] path must use $\{e^{23}_u,e^{23}_v\}$ or $\{e^{32}_u,e^{32}_v\}$ or both. 

From \GS[23] being true, we have $e^{23}_u \PREC e^{32}_u$. For the following we prove that (i) $\head[e^{23}_v] \PREC \tail[e^{32}_u]$; and (ii) there exists a path segment from $\head[e^{23}_v]$ to $d_1$ which is vertex-disjoint with any vertex in-between $\tail[e^{32}_u]$ and $\head[e^{32}_v]$. First we note that $e^{23}_u$ is not an $1$-edge cut separating $s_1$ and $\tail[e^{32}_u]$. The reason is that if $e^{23}_u\!\in\!\onecut[s_1][{\tail[e^{32}_u]}]$, then $e^{23}_u$ must be an $1$-edge cut separating $s_1$ and $d_1$ since any \FromTo[1][1] path must use $\{e^{23}_u,e^{23}_v\}$ or $\{e^{32}_u,e^{32}_v\}$ or both. However, since $e^{23}_u\!\in\!\Sover[2]\cap\Dover[3]$, this implies $e^{23}_u\!\in\!\onecut[{\{s_1, s_2\}}][{\{d_1,d_3\}}]$. This contradicts the assumption \GS[1]. We now consider all the possible cases: either $e^{23}_v\PREC e^{32}_u$ or $e^{32}_u\PRECEQ e^{23}_v$ or not reachable from each other. We first show that the last case is not possible. The reason is that suppose $e^{23}_v$ and $e^{32}_u$ are not reachable from each other, then $s_1$ can first reach $e^{23}_u$, then reach $e^{32}_u$ to $d_1$ without using $e^{23}_v$. This contradicts Property~3 of \NotGS[3]. Similarly, the second case is not possible because when $e^{32}_u\PRECEQ e^{23}_v$, we can find a path from $s_1$ to $e^{32}_u$ to $e^{23}_v$ to $d_1$ not using $e^{23}_u$ since $e^{23}_u\!\not\in\!\onecut[s_1][{\tail[e^{32}_u]}]$. This also contradicts Property~3 of \NotGS[3]. We thus have shown $e^{23}_v\PREC e^{32}_u$. Now we still need to show that $e^{23}_v$ and $e^{32}_u$ are not immediate neighbors: $\head[e^{23}_v]\PREC \tail[e^{32}_u]$. Suppose not, i.e., $\head[e^{23}_v]$
$=\!\tail[e^{32}_u]$. Then by Property~3 of \NotGS[3], we know that any path from $\head[e^{23}_v]\!=\!\tail[e^{32}_u]$ to $d_1$ must use both $e^{32}_u$ and $e^{32}_v$. By the conclusion in the first paragraph of this proof, we know that this implies $\{e^{32}_u,e^{32}_v\}\!\subset\!\onecut[s_1][d_1]$. However, this further implies that $\{e^{32}_u,e^{32}_v\}\!\subset\!\onecut[{\{s_1,s_3\}}][{\{d_1,d_2\}}]$, which contradicts \GS[1].  The proof of (i) is complete.

We now prove (ii). Suppose that every path from $\head[e^{23}_v]$ to $d_1$ has at least one vertex $w$ that satisfies $\tail[e^{32}_u]\PRECEQ w\PRECEQ \head[e^{32}_v]$.  Then by Property~3 of \NotGS[3], every \FromTo[1][1] path that uses $e^{23}_v$ must use both $e^{32}_u$ and $e^{32}_v$. By the findings in the first paragraph of this proof, this also implies that any \FromTo[1][1] path must use both $e^{32}_u$ and $e^{32}_v$. However, this further implies that $\{e^{32}_u,e^{32}_v\}\!\subset\!\onecut[{\{s_1, s_3\}}][{\{d_1,d_2\}}]$. This contradicts \GS[1]. We have thus proven (ii).


\begin{figure}[t]
\centering
\includegraphics[scale=0.2]{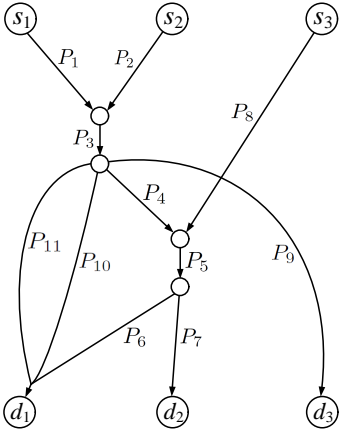}
\caption{The subgraph $\G[']$ of the 3-unicast ANA network $\GANA$ induced by the union of the 11 paths in the proof of \RS[20].} \vspace{-0.04\columnwidth}
\label{fig:RS[20]}
\end{figure}

Using the assumptions and the above discussions, we construct the following 11 path segments.

\noindent \makebox[1cm][l]{$\bullet$ $P_1$:}a path from $s_1$ to $\tail[e^{23}_u]$. This is always possible due to \GS[3] being false.

\noindent \makebox[1cm][l]{$\bullet$ $P_2$:}a path from $s_2$ to $\tail[e^{23}_u]$, which is edge-disjoint with $P_1$. This is always possible due to \GS[3] being false.

\noindent \makebox[1cm][l]{$\bullet$ $P_3$:}a path starting from $e^{23}_u$ and ending at $e^{23}_v$. This is always possible due to \GS[3] being false.

\noindent \makebox[1cm][l]{$\bullet$ $P_4$:}a path from $\head[e^{23}_v]$ to $\tail[e^{23}_u]$. This is always possible since we showed (i) in the above discussion.

\noindent \makebox[1cm][l]{$\bullet$ $P_5$:}a path starting from $e^{32}_u$ and ending at $e^{32}_v$. This is always possible due to \GS[4] being false.

\noindent \makebox[1cm][l]{$\bullet$ $P_6$:}a path from $\head[e^{32}_v]$ to $d_1$. This is always possible due to \GS[4] being false.

\noindent \makebox[1cm][l]{$\bullet$ $P_7$:}a path from $\head[e^{32}_v]$ to $d_2$, which is edge-disjoint with $P_6$. This is always possible due to \GS[4] being false and Property~2 of \NotGS[4].

\noindent \makebox[1cm][l]{$\bullet$ $P_8$:}a path from $s_3$ to $\tail[e^{32}_u]$. This is always possible due to \GS[4] being false.

\noindent \makebox[1cm][l]{$\bullet$ $P_9$:}a path from $\head[e^{23}_v]$ to $d_3$. This is always possible due to \GS[3] being false.

\noindent \makebox[1.15cm][l]{$\bullet$ $P_{10}$:}a path from $\head[e^{23}_v]$ to $d_1$, which is vertex-disjoint with $P_5$. This is always possible since we showed (ii) in the above discussion.

\noindent \makebox[1.15cm][l]{$\bullet$ $P_{11}$:}a path from $\head[e^{23}_v]$ to $d_1$, which is edge-disjoint with $P_9$. This is always possible due to \GS[3] being false.

Fig. \ref{fig:RS[20]} illustrates the relative topology of these 11 paths. We now consider the subgraph $\G[']$ induced by the above $11$ path segments. First, one can see that $s_i$ can reach $d_j$ for all $i\!\neq\!j$. In particular, $s_1$ can reach $d_2$ through $P_1 P_3 P_4 P_5 P_7$; $s_1$ can reach $d_3$ through $P_1 P_3 P_9$; $s_2$ can reach $d_1$ through either $P_2 P_3 P_4 P_5 P_6$ or $P_2 P_3 P_{10}$ or $P_2 P_3 P_{11}$; $s_2$ can reach $d_3$ through $P_2 P_3 P_9$; $s_3$ can reach $d_1$ through $P_8 P_5 P_6$; and $s_3$ can reach $d_2$ through $P_8 P_5 P_7$. Moreover, $s_1$ can reach $d_1$ through either $P_1 P_3 P_4 P_5 P_6$ or $P_1 P_3 P_{10}$ or $P_1 P_3 P_{11}$. Thus we showed \GS[16].

For the following, we will prove that $\ChGANA[1][1]\ChGANA[3][2]=\ChGANA[3][1]\ChGANA[1][2]$ and $\LRneq$ hold in the above $\G[']$. Note that $\{P_1, P_2, P_3, P_{10}\}$ must be vertex-disjoint with $P_8$, otherwise $s_3$ can reach $d_1$ without using $P_5$ and this contradicts $\{e^{32}_u,e^{32}_v\}\!\subset\!\Sover[3]\cap\Dover[2]\!\subset\!\onecut[s_3][d_1]$. Since $P_8$ is vertex-disjoint from $\{P_1,P_2\}$, one can easily see that removing $P_3$ separates $\{s_1,s_2\}$ and $\{d_1,d_3\}$. Thus $\G[']$ satisfies $\ChGANA[1][1]\ChGANA[3][2] = \ChGANA[3][1]\ChGANA[1][2]$.


To show that $\LRneq$ holds on $\G[']$, we make the following arguments. First, we show that $\G[']$ satisfies $\Sover[2]\cap\Sover[3]\EqualEmpty$. Note that any $\Sover[2]$ edge can exist only as one of four cases: (i) $P_2$; (ii) $P_3$; (iii) an edge that $P_{4}$, $P_{9}$, $P_{10}$, and $P_{11}$ share; and (iv) an edge that $P_{6}$, $P_{9}$, $P_{10}$, and $P_{11}$ share. Note also that any $\Sover[3]$ edge can exist only as one of three cases: (i) $P_8$; (ii) $P_5$; and (iii) an edge that $P_6$ and $P_7$ shares. But since $P_6$ and $P_7$ were chosen to be edge-disjoint with each other from the above construction, any $\Sover[3]$ edge can exist on either $P_8$ or $P_5$. However, $P_5$ was chosen to be vertex-disjoint with $P_{10}$ from the above construction and we also showed that $P_8$ is vertex-disjoint with $\{P_2,P_3,P_{10}\}$. Thus, $\Sover[2]\cap\Sover[3]\EqualEmpty$ on $\G[']$.

Second, we show that $\G[']$ satisfies $\Dover[1]\cap\Dover[2]\EqualEmpty$. Note that any $\Dover[1]$ edge can exist on an edge that all $P_{6}$, $P_{10}$, and $P_{11}$ share since $P_6$ cannot share an edge with any of its upstream paths (in particular $P_2$, $P_3$, $P_4$, and $P_5$); $P_5$ cannot share an edge with $P_{10}$ due to vertex-disjointness; and $P_8$ cannot share edge with $\{P_2,P_3,P_{10}\}$ otherwise there will be an \FromTo[3][1] path not using $P_5$. Note also that any $\Dover[2]$ edge can exist on (i) an edge that both $P_4$ and $P_8$ share; (ii) $P_5$; and (iii) $P_7$. However, $P_7$ was chosen to be edge-disjoint with $P_6$, and $P_5$ was chosen to be vertex-disjoint with $P_{10}$. Moreover, we already showed that $P_8$ is vertex-disjoint with $P_{10}$. Thus, $\Dover[1]\cap\Dover[2]\EqualEmpty$ on $\G[']$.

Third, we show that $\G[']$ satisfies $\Dover[1]\cap\Dover[3]\EqualEmpty$. Note that any $\Dover[1]$ edge can exist on an edge that $P_6$, $P_{10}$ and $P_{11}$ share. Note also that any $\Dover[3]$ edge can exist on (i) $P_3$; and (ii) $P_9$. However, all  $P_6$, $P_{10}$ and $P_{11}$ are the downstream paths of $P_3$. Moreover, $P_9$ was chosen to be edge-disjoint with $P_{11}$ by our construction. Thus, $\Dover[1]\cap\Dover[3]\EqualEmpty$ on $\G[']$.

Hence, the above discussions, together with \PropRef{Prop4}, implies that the considered $\G[']$ satisfies $\LRneq$. Thus we have proven \GS[18] being true. The proof is thus complete.
\end{proof}

We prove \RS[21] as follows.
\begin{proof} Suppose \LNR\AND(\NotGS[3])\AND(\NotGS[4])\AND\GS[22]\AND\GS[23]\AND\\\noindent\GS[25] is true. Recall the definitions of $e^{23}_u$, $e^{32}_u$, $e^{23}_v$, and $e^{32}_v$ when (\NotGS[3])\AND(\NotGS[4]) is true. 
From Property~1 of both \NotGS[3] and \NotGS[4], $s_1$ reaches $e^{23}_u$ and $e^{32}_u$, respectively. From \GS[22] being true, there exists a \FromTo[1][1] path $P_{11}^\ast$ who does not use any vertex in-between $\tail[e^{23}_u]$ and $\head[e^{23}_v]$, and any vertex in-between $\tail[e^{32}_u]$ and $\head[e^{32}_v]$.

Note that \GS[23]\AND\GS[25] implies $e^{23}_u\PREC e^{32}_u \PREC e^{23}_v$. For the following, we prove that $e^{32}_v \PREC e^{23}_v$. Note that by our construction $e^{32}_u\PRECEQ e^{32}_v$. As a result, we have $e^{23}_u\PREC$ $e^{32}_u \PRECEQ e^{32}_v \PREC e^{23}_v$.  To that end, we consider all the possible cases between $e^{32}_v$ and $e^{23}_v$: $e^{32}_v \PREC e^{23}_v$; or $e^{23}_v \PREC e^{32}_v$; or $e^{32}_v\!=\!e^{23}_v$; or they are not reachable from each other. We first show that the third case is not possible. The reason is that if $e^{32}_v\!=\!e^{23}_v$, then we have $\Sover[2]\cap\Sover[3]\cap\Dover[2]\cap\Dover[3]\NotEqualEmpty$, which contradicts the assumption \LNR. The last case in which $e^{32}_v$ and $e^{23}_v$ are not reachable from each other is also not possible. The reason is that by our construction, there is always an \FromTo[1][1] path through $e^{23}_u$, $e^{32}_u$, and $e^{32}_v$ without using $e^{23}_v$. Note that by Property~3 of \NotGS[3], such \FromTo[1][1] path must use $e^{23}_v$, which is a contradiction. We also claim that the second case, $e^{23}_v \PREC e^{32}_v$, is not possible. The reason is that if $e^{23}_v \PREC e^{32}_v$, then together with the assumption \GS[23]\AND\GS[25] we have $e^{23}_u\PREC e^{32}_u \PREC e^{23}_v \PREC e^{32}_v$. We also note that $e^{32}_u$ must be an $1$-edge cut separating $s_1$ and $\tail[e^{23}_v]$, otherwise $s_1$ can reach $\tail[e^{23}_v]$ without using $e^{32}_u$ and then use $e^{23}_v$ and $e^{32}_v$ to arrive at $d_2$. This contradicts the construction $e^{32}_u\!\in\!\Sover[3]\cap\Dover[2]\!\subset\!\onecut[s_1][d_2]$. Since $e^{23}_v\!\in\!\Sover[2]\cap\Dover[3]$ is also an $1$-edge cut separating $s_1$ and $d_3$, this in turn implies that $e^{32}_u\!\in\!\onecut[s_1][d_3]$. Symmetrically following this argument, we can also prove that $e^{23}_v\!\in\!\onecut[s_3][d_1]$. Since $e^{32}_u\!\in\!\Sover[3]\CAP\Dover[2]$ and $e^{23}_v\!\in\!\Sover[2]\CAP\Dover[3]$, these further imply that $e^{32}_u\!\in\!\Sover[1]\CAP\Sover[3]\CAP\Dover[2]$ and $e^{23}_v\!\in\!\Sover[2]\CAP\Dover[1]\CAP\Dover[3]$, which contradicts the assumption \LNR\;by \PropRef{Prop4}. We have thus established $e^{32}_v\PREC e^{23}_v$ and together with the assumption \GS[23]\AND\GS[25], we have $e^{23}_u$ $\PREC$ $e^{32}_u$ $\PRECEQ$ $e^{32}_v$ $\PREC e^{23}_v$.

Using the assumptions and the above discussions, we construct the following 7 path segments.

\noindent \makebox[1cm][l]{$\bullet$ $P_{1}$:}a path from $s_1$ to $\tail[e^{23}_u]$. This is always possible due to \GS[3] being false.

\noindent \makebox[1cm][l]{$\bullet$ $P_{2}$:}a path from $s_2$ to $\tail[e^{23}_u]$ which is edge-disjoint with $P_1$. This is always possible due to \GS[3] being false and Property~2 of \NotGS[3].

\noindent \makebox[1cm][l]{$\bullet$ $P_{3}$:}a path starting from $e^{23}_u$, using $e^{32}_u$ and $e^{32}_v$, and ending at $e^{23}_v$. This is always possible from the above discussion.

\noindent \makebox[1cm][l]{$\bullet$ $P_{4}$:}a path from $\head[e^{23}_v]$ to $d_1$. This is always possible due to \GS[3] being false.

\noindent \makebox[1cm][l]{$\bullet$ $P_{5}$:}a path from $\head[e^{23}_v]$ to $d_3$ which is edge-disjoint with $P_4$. This is always possible due to \GS[3] being false and Property 2 of \NotGS[3].

\noindent \makebox[1cm][l]{$\bullet$ $P_{6}$:}a path from $s_3$ to $\tail[e^{32}_u]$. This is always possible due to \GS[4] being false.

\noindent \makebox[1cm][l]{$\bullet$ $P_{7}$:}a path from $\head[e^{32}_v]$ to $d_2$. This is always possible due to \GS[4] being false.

We now consider the subgraph $\G[']$ induced by the above 7 path segments and $P_{11}^\ast$. First, one can easily check that $s_i$ can reach $d_j$ for all $i\!\neq\!j$. In particular, $s_1$ can reach $d_2$ through $P_1 P_3 e^{32}_v P_7$; $s_1$ can reach $d_3$ through $P_1 P_3 P_5$; $s_2$ can reach $d_1$ through $P_2 P_3 P_4$; $s_2$ can reach $d_3$ through $P_2 P_3 P_5$; $s_3$ can reach $d_1$ through $P_6 e^{32}_u P_3 P_4$; and $s_3$ can reach $d_2$ through $P_6 e^{32}_u P_3 e^{32}_v P_7$. Moreover, $s_1$ can reach $d_1$ through either $P_{11}^\ast$ or $P_1 P_3 P_4$. As a result, \GS[16] must hold.

We now prove \GS[17]. To that end, we will show that there exists an edge $\tilde{e}\!\in\!P_{11}^\ast$ that cannot reach any of $\{d_2,d_3\}$, and cannot be reached from any of $\{s_2,s_3\}$. Note from \GS[22] being true that $P_{11}^\ast$ was chosen to be vertex-disjoint with $P_3$. Note that $P_{11}^\ast$ must also be vertex-disjoint with $P_2$ (resp. $P_6$) otherwise $s_2$ (resp. $s_3$) can reach $d_1$ without using $P_3$ (resp. $e^{32}_u P_3 e^{32}_v $). Similarly, $P_{11}^\ast$ must also be vertex-disjoint with $P_5$ (resp. $P_7$) otherwise $s_1$ can reach $d_3$ (resp. $d_2$) without using $P_3$ (resp. $e^{32}_u P_3 e^{32}_v$). Hence, among 7 path segments constructed above, the only path segments that can share a vertex with $P_{11}^\ast$ are $P_1$ and $P_4$. Without loss of generality, we also assume that $P_1$ is chosen such that it overlaps with $P_{11}^\ast$ in the beginning but then ``branches out". That is, let $u^\ast$ denote the most downstream vertex among those who are used by both $P_{1}$ and $P_{11}^\ast$ and we can then replace $P_1$ by $s_1 P_{11}^\ast u^\ast P_1 \tail[e^{23}_u]$. Note that the new construction still satisfies the requirement that $P_1$ and $P_2$ are edge-disjoint since $P_{11}^\ast$ is vertex-disjoint with $P_2$. Similarly, we also assume that $P_4$ is chosen such that it does not overlap with $P_{11}^\ast$ in the beginning but then ``merges" with $P_{11}^\ast$ whenever $P_4$ shares a vertex $v^\ast$ with $P_{11}^\ast$ for the first time. The new construction of $P_4$, i.e., $\head[e^{23}_v] P_4 v^\ast P_{11}^\ast d_1$ is still edge-disjoint from $P_5$. 
Then in the considered subgraph $\G[']$, in order for an edge $e\!\in\!P_{11}^\ast$ to reach $d_2$ or $d_3$, we must have $\head[e]\PRECEQ u^\ast$. Similarly, in order for an edge $e\!\in\!P_{11}^\ast$ to be reached from $s_2$ or $s_3$, this edge $e$ must satisfy $v^\ast \PRECEQ \tail[e]$. If there does not exist such an edge $\tilde{e}\!\in\!P_{11}^\ast$ satisfying \GS[17], then it means that $u^\ast\!=\!v^\ast$. This, however, contradicts the assumption that $\G[]$ is acyclic because now we can walk from $u^\ast$ through $P_1 P_3 P_4$ back to $v^\ast\!=\!u^\ast$. Therefore, we thus have \GS[17]. The proof of \RS[21] is thus complete.
\end{proof}

We prove \RS[22] as follows.
\begin{proof} Suppose \LNR\AND(\NotGS[3])\AND(\NotGS[4])\AND\GS[22]\AND\GS[23]\AND\\\noindent(\NotGS[25]) is true. Recall the definitions of $e^{23}_u$, $e^{32}_u$, $e^{23}_v$, and $e^{32}_v$ when (\NotGS[3])\AND(\NotGS[4]) is true. 
From Property~1 of both \NotGS[3] and \NotGS[4], $s_1$ reaches $e^{23}_u$ and $e^{32}_u$, respectively. From \GS[22] being true, there exists a \FromTo[1][1] path $P_{11}^\ast$ who does not use any vertex in-between $\tail[e^{23}_u]$ and $\head[e^{23}_v]$, and any vertex in-between $\tail[e^{32}_u]$ and $\head[e^{32}_v]$.

Note that \GS[23] implies $e^{23}_u \PREC e^{32}_u$. For the following, we prove that $\head[e^{23}_v]\PREC \tail[e^{32}_u]$.
To that end, we consider all the possible cases by \NotGS[25] being true: either $e^{23}_v\PREC e^{32}_u$ or $e^{23}_v\!=\!e^{32}_u$ or not reachable from each other. We first show that the second case is not possible. The reason is that if $e^{23}_v\!=\!e^{32}_u$, then we have $\Sover[2]\cap\Sover[3]\cap\Dover[2]\cap\Dover[3]\NotEqualEmpty$, which contradicts the assumption \LNR. The third case in which $e^{23}_v$ and $e^{32}_u$ are not reachable from each other is also not possible. The reason is that by our construction, there is always an \FromTo[1][1] path through $e^{23}_u$, $e^{32}_u$, and $e^{32}_v$ without using $e^{23}_v$. Note that by Property~3 of \NotGS[3], such \FromTo[1][1] path must use $e^{23}_v$, which is a contradiction. We have thus established $e^{23}_v\PREC e^{32}_u$. We still need to show that $e^{23}_v$ and $e^{32}_u$ are not immediate neighbors since we are proving $\head[e^{23}_v]\PREC\tail[e^{32}_u]$. We prove this by contradiction. Suppose not, i.e., $w\!=\!\head[e^{23}_v]\!=\!\tail[e^{32}_u]$. Since $e^{32}_u\!\in\!\Sover[3]\CAP\Dover[2]\!\subset\!\onecut[s_1][d_2]$, any \FromTo[1][2] path must use its tail $w$. By Property~3 of \NotGS[3] we have $e^{23}_v\!\in\!\onecut[s_1][w]$. This in turn implies that $e^{23}_v$ is also an $1$-edge cut separating $s_1$ and $d_2$. 
By symmetry, we can also prove $e^{32}_u\!\in\!\onecut[s_2][d_1]$. Jointly the above argument implies that $e^{23}_v\!\in\!\Sover[1]\cap\Sover[2]\cap\Dover[3]$ and $e^{32}_u\!\in\!\Sover[3]\cap\Dover[1]\cap\Dover[2]$, which contradicts the assumption \LNR\;by \PropRef{Prop4}.

Based on the above discussions, we construct the following 9 path segments.

\noindent \makebox[1cm][l]{$\bullet$ $P_{1}$:}a path from $s_1$ to $\tail[e^{23}_u]$. This is always possible due to \GS[3] being false.

\noindent \makebox[1cm][l]{$\bullet$ $P_{2}$:}a path from $s_2$ to $\tail[e^{23}_u]$ which is edge-disjoint with $P_1$. This is always possible due to \GS[3] being false and Property~2 of \NotGS[3].

\noindent \makebox[1cm][l]{$\bullet$ $P_{3}$:}a path starting from $e^{23}_u$ and ending at $e^{23}_v$. This is always possible due to \GS[3] being false.

\noindent \makebox[1cm][l]{$\bullet$ $P_{4}$:}a path from $\head[e^{23}_v]$ to $\tail[e^{32}_u]$. This is always possible from the above discussion.

\noindent \makebox[1cm][l]{$\bullet$ $P_{5}$:}a path starting from $e^{32}_u$ and ending at $e^{32}_v$. This is always possible due to \GS[4] being false.

\noindent \makebox[1cm][l]{$\bullet$ $P_{6}$:}a path from $\head[e^{32}_v]$ to $d_1$. This is always possible due to \GS[4] being false.

\noindent \makebox[1cm][l]{$\bullet$ $P_{7}$:}a path from $\head[e^{32}_v]$ to $d_2$ which is edge-disjoint with $P_6$. This is always possible due to \GS[4] being false and Property~2 of \NotGS[4].

\noindent \makebox[1cm][l]{$\bullet$ $P_{8}$:}a path from $s_3$ to $\tail[e^{32}_u]$. This is always possible due to \GS[4] being false.

\noindent \makebox[1cm][l]{$\bullet$ $P_{9}$:}a path from $\head[e^{23}_v]$ to $d_3$. This is always possible due to \GS[3] being false.

From \GS[22] being true, $P_{11}^\ast$ was chosen to be vertex-disjoint with $\{P_3, P_5\}$. Note that $P_{11}^\ast$ must also be vertex-disjoint with $P_2$ (resp. $P_8$) otherwise $s_2$ (resp. $s_3$) can reach $d_1$ without using $P_3$ (resp. $P_5$). Similarly, $P_{11}^\ast$ must also be vertex-disjoint with $P_7$ (resp. $P_9$) otherwise $s_1$ can reach $d_2$ (resp. $d_3$) without using $P_5$ (resp. $P_3$). Hence, among 9 path segments constructed above, the only path segments that can share a vertex with $P_{11}^\ast$ are $P_1$, $P_4$, and $P_6$.

We now consider the subgraph $\G[']$ induced by the above 9 path segments and $P_{11}^\ast$. First, one can easily check that $s_i$ can reach $d_j$ for all $i\!\neq\!j$. In particular, $s_1$ can reach $d_2$ through $P_1 P_3 P_4 P_5 P_7$; $s_1$ can reach $d_3$ through $P_1 P_3 P_9$; $s_2$ can reach $d_1$ through $P_2 P_3 P_4 P_5 P_6$; $s_2$ can reach $d_3$ through $P_2 P_3 P_9$; $s_3$ can reach $d_1$ through $P_8 P_5 P_6$; and $s_3$ can reach $d_2$ through $P_8 P_5 P_7$. Moreover, $s_1$ can reach $d_1$ through either $P_{11}^\ast$ or $P_1 P_3 P_4 P_5 P_6$. Thus we showed \GS[16].

{\bf Case~1:} $P_{11}^\ast$ is also vertex-disjoint with $P_4$. In this case, we will prove that \GS[17] is satisfied. Namely, we claim that there exists an edge $\tilde{e}\!\in\!P_{11}^\ast$ that cannot reach any of $\{d_2,d_3\}$, and cannot be reached from any of $\{s_2,s_3\}$. Note that only path segments that $P_{11}^\ast$ can share a vertex with are $P_1$ and $P_6$. Without loss of generality, we assume that $P_1$ is chosen such that it overlaps with $P_{11}^\ast$ in the beginning but then ``branches out". That is, let $u^\ast$ denote the most downstream vertex among those who are used by both $P_1$ and $P_{11}^\ast$ and we can then replace $P_1$ by $s_1 P_{11}^\ast u^\ast P_1 \tail[e^{23}_u]$. Note that the new construction still satisfies the requirement that $P_1$ and $P_2$ are edge-disjoint since $P_{11}^\ast$ is vertex-disjoint with $P_2$. Similarly, we also assume that $P_6$ is chosen such that it does not overlap with $P_{11}^\ast$ in the beginning but then ``merges" with $P_{11}^\ast$ whenever $P_6$ shares a vertex $v^\ast$ with $P_{11}^\ast$ for the first time. The new construction of $P_6$, i.e, $\head[e^{32}_v] P_6 v^\ast P_{11}^\ast d_1$, is still edge-disjoint from $P_7$. Then in the considered subgraph $\G[']$, in order for an edge $e\!\in\!P_{11}^\ast$ to reach $d_2$ or $d_3$, we must have $\head[e]\PRECEQ u^\ast$. Similarly, in order for an edge $e\!\in\!P_{11}^\ast$ to be reached from $s_2$ or $s_3$, this edge $e$ must satisfy $v^\ast \PRECEQ \tail[e]$. If there does not exist such an edge $\tilde{e}\!\in\!P_{11}^\ast$ satisfying \GS[17], then it means that $u^\ast\!=\!v^\ast$. This, however, contradicts the assumption that $\G[]$ is acyclic because now we can walk from $u^\ast$ through $P_1 P_3 P_4 P_5 P_6$ back to $v^\ast\!=\!u^\ast$. Therefore, we thus have \GS[17] for {\bf Case~1}.



{\bf Case~2:} $P_{11}^\ast$ shares a vertex with $P_4$. In this case, we will prove that \GS[18] is true. Since $P_{11}^\ast$ is vertex-disjoint with $\{P_3,P_5\}$, $P_{11}^\ast$ must share a vertex $w$ with $P_4$ where $\head[e^{23}_v]$ $\PREC w \PREC \tail[e^{32}_u]$. Choose the most downstream vertex among those who are used by both $P_{11}^\ast$ and $P_4$ and denote it as $w'$. Then, denote the path segment $\head[e^{23}_v] P_4 w' P_{11}^\ast d_1$ by $P_{10}$. Note that we do not introduce new paths but only introduce a new notation as shorthand for a combination of some existing path segments. We observe that there may be some edge overlap between $P_4$ and $P_9$ since both starts from $\head[e^{23}_v]$. Let $\tilde{w}$ denote the most downstream vertex that is used by both $P_4$ and $P_9$. We then replace $P_9$ by $\tilde{w} P_9 d_3$, i.e., we truncate $P_9$ so that $P_9$ is now edge-disjoint from $P_4$.

Since the path segment $w' P_{10} d_1$ originally comes from $P_{11}^\ast$, $w' P_{10} d_1$ is also vertex-disjoint with $\{P_2,P_3,P_5,P_7,P_8,P_9\}$. In addition, $P_8$ must be vertex-disjoint with $\{P_1, P_2, P_3, P_{10}\}$, otherwise $s_3$ can reach $d_1$ without using $P_5$. 

Now we consider the another subgraph $\G['']\!\subset\!\G[']$ induced by the path segments $P_1$ to $P_8$, the redefined $P_{9}$, and newly constructed $P_{10}$, i.e., when compared to $\G[']$, we replace $P_{11}^\ast$ by $P_{10}$. One can easily verify that $s_i$ can reach $d_j$ for all $i\neq j$, and $s_1$ can reach $d_1$ on this new subgraph $\G['']$. Using the above topological relationships between these constructed path segments, we will further show that the induced $\G['']$ satisfies $\ChGANA[1][1]\ChGANA[3][2]=\ChGANA[3][1]\ChGANA[1][2]$ and $\LRneq$.

Since $P_8$ is vertex-disjoint from $\{P_1,P_2\}$, one can see that removing $P_3$ separates $\{s_1,s_2\}$ and $\{d_1,d_3\}$. Thus, the considered $\G['']$ also satisfies $\ChGANA[1][1]\ChGANA[3][2] = \ChGANA[3][1]\ChGANA[1][2]$.

To prove $\LRneq$, we first show that $\G['']$ satisfies $\Sover[2]\cap\Sover[3]\EqualEmpty$. Note that any $\Sover[2]$ edge can exist only as one of three cases: (i) $P_2$; (ii) $P_3$; (iii) an edge that $P_{4}$ and $P_{10}$ share, whose head is in the upstream of or equal to $\tilde{w}$, i.e., $\{e\!\in\!P_4\cap P_{10} : \head[e]\PRECEQ \tilde{w}\}$ (may or may not be empty); and (iv) an edge that $P_6$, $P_9$, and $P_{10}$ share. Note also that any $\Sover[3]$ edge can exist only as on of three cases: (i) $P_8$; (ii) $P_5$; and (iii) an edge that $P_6$ and $P_7$ share. But since $P_6$ and $P_7$ were chosen to be edge-disjoint from the above construction, any $\Sover[3]$ edge can exist on either $P_8$ or $P_5$. We then notice that $P_8$ is vertex-disjoint with $\{P_2,P_3,P_{10}\}$. Also, $P_5$ was chosen to be vertex-disjoint with $P_{10}$ and both $P_2$ and $P_3$ are in the upstream of $P_5$. The above arguments show that no edge can be simultaneously in $\Sover[2]$ and $\Sover[3]$. We thus have
$\Sover[2]\cap\Sover[3]\EqualEmpty$ on $\G['']$.


Second, we show that $\G['']$ satisfies $\Dover[1]\cap\Dover[2]\EqualEmpty$. Note that any $\Dover[1]$ edge can exist only an edge that both $P_{6}$ and $P_{10}$ share since any of $\{P_5, P_8\}$ does not share an edge with any of $\{P_2, P_3, P_{10}\}$. Note also that any $\Dover[2]$ edge can exist only as one of three cases: (i) an edge that both $P_4$ and $P_8$ share; (ii) $P_5$; and (iii) $P_7$. However, $P_7$ was chosen to be edge-disjoint with $P_6$, and we have shown that $P_5$ is vertex-disjoint with $P_{10}$. Moreover, we already showed that $P_8$ is vertex-disjoint with $P_{10}$. Thus, $\Dover[1]\cap\Dover[2]\EqualEmpty$ on $\G['']$.

Third, we show that $\G['']$ satisfies $\Dover[1]\cap\Dover[3]\EqualEmpty$. Note that any $\Dover[1]$ edge can exist only on an edge that both $P_{10}$ and $P_6$ share. Note also that any $\Dover[3]$ edge can exist only as one of three cases: (i) a $P_3$ edge; (ii) a $P_4$ edge whose head is in the upstream of or equal to $\tilde{w}$, i.e., $\{e\!\in\!P_4 : \head[e]\PRECEQ \tilde{w}\}$ (may or may not be empty); and (iii) $P_9$. However, $P_6$ is in the downstream of $P_3$ and $P_4$. Moreover, $P_9$ is edge-disjoint with $P_{11}^\ast$ and thus edge-disjoint with $w'P_{10}d_1$. As a result, no edge can be simultaneously in $\Dover[1]$ and $\Dover[3]$. Thus $\Dover[1]\cap\Dover[3]\EqualEmpty$ on $\G['']$.

Hence, the above discussions, together with \PropRef{Prop4}, implies that the considered $\G['']$ satisfies $\LRneq$. We thus have proven \GS[18] being true for {\bf Case~2}.
\end{proof}

By swapping the roles of $s_2$ and $s_3$, and the roles of $d_2$ and $d_3$, the proofs of \RS[20] to \RS[22] can also be used to prove \RS[23] to \RS[25], respectively. More specifically, \GS[3] and \GS[4] are converted back and forth from each other when swapping the flow indices. The same thing happens between \GS[23] and \GS[24]; between \GS[25] and \GS[26]; and between \GS[18] and \GS[19]. Moreover, \LNR, \GS[1], \GS[16], \GS[17], and \GS[22] remain the same after the index swapping. Thus the above proofs of \RS[20] to \RS[22] can thus be used to prove \RS[23] to \RS[25].

\section{Proof of \SS[14]}\label{ProofSS[14]}

\subsection{The fifth set of logic statements}

To prove \SS[14], we need the fifth set of logic statements.

\noindent $\bullet$ \GS[27]{\bf:}\label{ref:GS[27]} $\Sover[2]\cap\Dover[1]\EqualEmpty$.

\noindent $\bullet$ \GS[28]{\bf:}\label{ref:GS[28]} $\Sover[3]\cap\Dover[1]\EqualEmpty$.

\noindent $\bullet$ \GS[29]{\bf:}\label{ref:GS[29]} $\Dover[2]\cap\Sover[1]\EqualEmpty$.

\noindent $\bullet$ \GS[30]{\bf:}\label{ref:GS[30]} $\Dover[3]\cap\Sover[1]\EqualEmpty$.

\noindent $\bullet$ \GS[31]{\bf:}\label{ref:GS[31]} $\Sover[i]\NotEqualEmpty$ and $\Dover[i]\NotEqualEmpty$ for all $i\!\in\!\{1,2,3\}$.

Several implications can be made when \GS[27] is true. We term those implications {\em the properties of \GS[27]}. Several properties of \GS[27] are listed as follows, for which their proofs are provided in \AppRef{ProofsG27G28}.

{\em Consider the case in which \GS[27] is true.} Use $e^\ast_2$ to denote the most downstream edge in $\onecut[s_2][d_1]\cap\onecut[s_2][d_3]$. Since the source edge $e_{s_2}$ belongs to both $\onecut[s_2][d_1]$ and $\onecut[s_2][d_3]$, such $e^\ast_2$ always exists. Similarly, use $e^\ast_1$ to denote the most upstream edge in $\onecut[s_2][d_1]\cap\onecut[s_3][d_1]$. The properties of \GS[27] can now be described as follows.

\noindent \makebox[3.45cm][l]{$\diamond$ {\bf Property~1 of \GS[27]:}}$e^\ast_2 \PREC e^\ast_1$ and the channel gains $\ChGANA[1][2]$, $\ChGANA[3][2]$, and $\ChGANA[1][3]$ can be expressed as $\ChGANA[1][2] = \ChG[e^\ast_2][e_{s_2}]\ChG[e^\ast_1][e^\ast_2]$ $\ChG[e_{d_1}][e^\ast_1]$, $\ChGANA[3][2] = \ChG[e^\ast_2][e_{s_2}]\ChG[e_{d_3}][e^\ast_2]$, and $\ChGANA[1][3] = \ChG[e^\ast_1][e_{s_3}]\ChG[e_{d_1}][e^\ast_1]$.

\noindent \makebox[3.45cm][l]{$\diamond$ {\bf Property~2 of \GS[27]:}}$\GCD[{\ChG[e^\ast_1][e_{s_3}]}][\;{\ChG[e^\ast_2][e_{s_2}]\ChG[e^\ast_1][e^\ast_2]}]\PolyEqual 1$,
$\GCD[{\ChG[e^\ast_1][e^\ast_2]\ChG[e_{d_1}][e^\ast_1]}][\;{\ChG[e_{d_3}][e^\ast_2]}]\PolyEqual 1$, $\GCD[{\ChGANA[1][3]}][\;{\ChG[e^\ast_1][e^\ast_2]}]\PolyEqual 1$, and $\GCD[{\ChGANA[3][2]}][\;{\ChG[e^\ast_1][e^\ast_2]}]\PolyEqual 1$.

On the other hand, when \GS[27] is false, we can also derive several implications, which are termed {\em the properties of \NotGS[27]}.

{\em Consider the case in which \GS[27] is false}.  Use $e^{21}_u$ (resp.\ $e^{21}_v$) to denote the most upstream (resp.\ the most downstream) edge in $\Sover[2]\cap\Dover[1]$. By definition, it must be $e^{21}_u\PRECEQ e^{21}_v$. We now describe the following properties of \NotGS[27].

\noindent \makebox[3.7cm][l]{$\diamond$ {\bf Property~1 of \NotGS[27]:}}The channel gains $\ChGANA[1][2]$, $\ChGANA[3][2]$, and $\ChGANA[1][3]$ can be expressed as $\ChGANA[1][2] = \ChG[e^{21}_u][e_{s_2}]\ChG[e^{21}_v][e^{21}_u]\ChG[e_{d_1}][e^{21}_v]$, $\ChGANA[3][2]\!=\!\ChG[e^{21}_u][e_{s_2}]\ChG[e^{21}_v][e^{21}_u]\ChG[e_{d_3}][e^{21}_v]$, and $\ChGANA[1][3]\!=\! \ChG[e^{21}_u][e_{s_3}]\ChG[e^{21}_v][e^{21}_u]$ $\ChG[e_{d_1}][e^{21}_v]$.

\noindent \makebox[3.7cm][l]{$\diamond$~{\bf Property~2 of \NotGS[27]:}}$\GCD[{\ChG[e^{21}_u][e_{s_2}]}][\;{\ChG[e^{21}_u][e_{s_3}]}]\PolyEqual 1$ and $\GCD[{\ChG[e_{d_1}][e^{21}_v]}][\;{\ChG[e_{d_3}][e^{21}_v]}]\,\PolyEqual 1$.


Symmetrically, we define the following properties of \GS[28] and \NotGS[28].

{\em Consider the case in which \GS[28] is true.} Use $e^\ast_3$ to denote the most downstream edge in $\onecut[s_3][d_1]\cap\onecut[s_3][d_2]$, and use $e^\ast_1$ to denote the most upstream edge in $\onecut[s_2][d_1]\cap\onecut[s_3][d_1]$. We now describe the following properties of \GS[28].

\noindent \makebox[3.45cm][l]{$\diamond$~{\bf Property~1 of \GS[28]:}}$e^\ast_3 \PREC e^\ast_1$ and
the channel gains $\ChGANA[1][3]$, $\ChGANA[2][3]$, and $\ChGANA[1][2]$ can be expressed as $\ChGANA[1][3] = \ChG[e^\ast_3][e_{s_3}]\ChG[e^\ast_1][e^\ast_3]$ $\ChG[e_{d_1}][e^\ast_1]$, $\ChGANA[2][3] = \ChG[e^\ast_3][e_{s_3}]\ChG[e_{d_2}][e^\ast_3]$, and $\ChGANA[1][2] = \ChG[e^\ast_1][e_{s_2}]\ChG[e_{d_1}][e^\ast_1]$.

\noindent \makebox[3.45cm][l]{$\diamond$~{\bf Property~2 of \GS[28]:}}$\GCD[{\ChG[e^\ast_1][e_{s_2}]}][\;{\ChG[e^\ast_3][e_{s_3}]\ChG[e^\ast_1][e^\ast_3]}]\PolyEqual 1$,
$\GCD[{\ChG[e^\ast_1][e^\ast_3]\ChG[e_{d_1}][e^\ast_1]}][\;{\ChG[e_{d_2}][e^\ast_3]}]\PolyEqual 1$, $\GCD[{\ChGANA[1][2]}][\;{\ChG[e^\ast_1][e^\ast_3]}]\PolyEqual 1$, and $\GCD[{\ChGANA[2][3]}][\;{\ChG[e^\ast_1][e^\ast_3]}]\,\PolyEqual 1$.

{\em Consider the case in which \GS[28] is false.} Use $e^{31}_u$ (resp.\ $e^{31}_v$) to denote the most upstream (resp.\ the most downstream) edge in $\Sover[3]\cap\Dover[1]$. By definition, it must be $e^{31}_u\PRECEQ e^{31}_v$. We now describe the following properties of \NotGS[28].

\noindent \makebox[3.7cm][l]{$\diamond$~{\bf Property~1 of \NotGS[28]:}}The channel gains $\ChGANA[1][3]$, $\ChGANA[2][3]$, and $\ChGANA[1][2]$ can be expressed as $\ChGANA[1][3] = \ChG[e^{31}_u][e_{s_3}]\ChG[e^{31}_v][e^{31}_u]\ChG[e_{d_1}][e^{31}_v]$, $\ChGANA[2][3]\!=\!\ChG[e^{31}_u][e_{s_3}]\ChG[e^{31}_v][e^{31}_u]\ChG[e_{d_2}][e^{23}_v]$, and $\ChGANA[1][2]\!=\! \ChG[e^{31}_u][e_{s_2}]\ChG[e^{31}_v][e^{31}_u]$ $\ChG[e_{d_1}][e^{31}_v]$.

\noindent \makebox[3.7cm][l]{$\diamond$~{\bf Property~2 of \NotGS[28]:}}$\GCD[{\ChG[e^{31}_u][e_{s_2}]}][\;{\ChG[e^{31}_u][e_{s_3}]}]\PolyEqual 1$ and $\GCD[{\ChG[e_{d_1}][e^{31}_v]}][\;{\ChG[e_{d_2}][e^{31}_v]}]\,\PolyEqual 1$.

\subsection{The skeleton of proving \SS[14]}

We prove the following relationships, which jointly prove \SS[14].

\noindent $\bullet$ \RS[26]{\bf:}\label{ref:RS[26]} \DS[3]\AND\DS[4]\IMPLY \GS[31].

\noindent $\bullet$ \RS[27]{\bf:}\label{ref:RS[27]} \LNR\AND(\NotGS[27])\AND(\NotGS[28])\AND(\NotGS[29])\AND(\NotGS[30])\IMPLY \\\noindent\CONT.

\noindent $\bullet$ \RS[28]{\bf:}\label{ref:RS[28]} \DS[3]\AND\DS[4]\AND\GS[27]\AND\GS[28]\IMPLY \CONT.

\noindent $\bullet$ \RS[29]{\bf:}\label{ref:RS[29]} \LNR\AND\GS[1]\AND\ES[0]\AND\DS[3]\AND\DS[4]\AND(\NotGS[27])\AND\GS[28]\IMPLY\CONT.

\noindent $\bullet$ \RS[30]{\bf:}\label{ref:RS[30]} \LNR\AND\GS[1]\AND\ES[0]\AND\DS[3]\AND\DS[4]\AND\GS[27]\AND(\NotGS[28])\IMPLY\CONT.

\noindent $\bullet$ \RS[31]{\bf:}\label{ref:RS[31]} \DS[3]\AND\DS[4]\AND\GS[29]\AND\GS[30]\IMPLY \CONT.

\noindent $\bullet$ \RS[32]{\bf:}\label{ref:RS[32]} \LNR\AND\GS[1]\AND\ES[0]\AND\DS[3]\AND\DS[4]\AND(\NotGS[29])\AND\GS[30]\IMPLY\CONT.

\noindent $\bullet$ \RS[33]{\bf:}\label{ref:RS[33]} \LNR\AND\GS[1]\AND\ES[0]\AND\DS[3]\AND\DS[4]\AND\GS[29]\AND(\NotGS[30])\IMPLY\CONT.

One can see that \RS[28] and \RS[31] imply, respectively,
\begin{align}
&\textbf{LNR}\wedge\textbf{G1}\wedge\textbf{E0}\wedge\textbf{D3}\wedge\textbf{D4}\wedge\textbf{G27}\wedge\textbf{G28}\Rightarrow\,\text{false}, \label{RS[17]-1} \\
&\textbf{LNR}\wedge\textbf{G1}\wedge\textbf{E0}\wedge\textbf{D3}\wedge\textbf{D4}\wedge\textbf{G29}\wedge\textbf{G30}\Rightarrow\,\text{false}. \label{RS[17]-2}
\end{align}

Also \RS[27] implies
\begin{equation}\label{RS[17]-3}
\begin{split}
&\textbf{LNR}\wedge\textbf{G1}\wedge\textbf{E0}\wedge\textbf{D3}\wedge\textbf{D4}\wedge(\neg\,\textbf{G27})\,\wedge \\
&\qquad\qquad\;\;\;\;(\neg\,\textbf{G28})\wedge(\neg\,\textbf{G29})\wedge(\neg\,\textbf{G30})\Rightarrow\,\text{false}.
\end{split}\end{equation}

\RS[29], \RS[30], \RS[32], \RS[33], \Ref{RS[17]-1}, \Ref{RS[17]-2}, and \Ref{RS[17]-3} jointly imply
\begin{align*}
&\textbf{LNR}\wedge\textbf{G1}\wedge\textbf{E0}\wedge\textbf{D3}\wedge\textbf{D4}\Rightarrow\,\text{false},
\end{align*}
which proves \SS[14]. The proofs of \RS[26] and \RS[27] are relegated to \AppRef{ProofsRS[26]-RS[27]}. The proofs of \RS[28], \RS[29], and \RS[30] are provided in Appendices~\ref{ProofRS[28]},~\ref{ProofRS[29]}, and~\ref{ProofRS[30]}, respectively.

The logic relationships \RS[31] to \RS[33] are the symmetric versions of \RS[28] to \RS[30]. Specifically, if we swap the roles of sources and destinations, then the resulting graph is still a 3-unicast ANA network; \DS[3] is now converted to \DS[4]; \DS[4] is converted to \DS[3]; \GS[27] is converted to \GS[29]; and \GS[28] is converted to \GS[30]. Therefore, the proof of \RS[28] can serve as a proof of \RS[31]. Further, after swapping the roles of sources and destinations, the \LNR\;condition (see \Ref{GTC1}) remains the same; \GS[1] remains the same (see \Ref{GTC2}); and \ES[0] remains the same. Therefore, the proof of \RS[29] (resp. \RS[30]) can serve as a proof of \RS[32] (resp. \RS[33]).

\section{Proofs of the properties of \GS[27], \GS[28], \NotGS[27], and \NotGS[28]}\label{ProofsG27G28}

We prove Properties 1 and 2 of \GS[27] as follows.
\begin{proof} By swapping the roles of $s_1$ and $s_3$, and the roles of $d_1$ and $d_3$, the proof of the properties of \GS[3] in \AppRef{ProofsG3G4} can be used to prove the properties of \GS[27].
\end{proof}

We prove Properties 1 and 2 of \NotGS[27] as follows.
\begin{proof} By swapping the roles of $s_1$ and $s_3$, and the roles of $d_1$ and $d_3$, the proof of Properties~1 and~2 of \NotGS[3] in \AppRef{ProofsG3G4} can be used to prove the properties of \NotGS[27].
\end{proof}

By swapping the roles of $s_2$ and $s_3$, and the roles of $d_2$ and $d_3$, the above proofs can also be used to prove Properties 1 and 2 of \GS[28] and Properties 1 and 2 of \NotGS[28].

\section{Proofs of \RS[26] and \RS[27]}\label{ProofsRS[26]-RS[27]}
We prove \RS[26] as follows.
\begin{proof}
Suppose \DS[3]\AND\DS[4] is true. By \CorRef{Cor2}, we know that any channel gain cannot have any other channel gain as a factor. Since \DS[3]\AND\DS[4] is true, any one of the four channel gains $\ChGANA[2][1]$, $\ChGANA[1][3]$, $\ChGANA[3][1]$, and $\ChGANA[1][2]$ must be reducible.

Since \DS[4] is true, we must also have for some positive integer $l_4$ such that
\begin{align}
\GCD[{\ChGANA[1][1]\ChGANA[2][1]^{l_4}\ChGANA[3][2]^{l_4}\ChGANA[1][3]^{l_4}}][\;{\ChGANA[1][2]}]= \ChGANA[1][2]. \label{RS[26]-1}
\end{align}

We first note that $\ChGANA[3][2]$ is the only channel gain starting from $s_2$ out of the four channel gains $\{\ChGANA[1][1], \ChGANA[2][1], \ChGANA[3][2], \ChGANA[1][3]\}$. Therefore, we must have $\GCD[{\ChGANA[3][2]}][{\ChGANA[1][2]}]\,\PolyNotEqual 1$ since ``we need to cover the factor of $\ChGANA[1][2]$ that emits from $s_2$." \LemRef{Lem7} then implies that $\Sover[2]\NotEqualEmpty$.

Further, \DS[4] implies $\GCD[{\ChGANA[1][1]\ChGANA[2][1]^{l_4}\ChGANA[3][2]^{l_4}\ChGANA[1][3]^{l_4}}][\,{\ChGANA[3][1]}]= \ChGANA[3][1]$ for some positive integer $l_4$, which, by similar arguments, implies $\GCD[{\ChGANA[3][2]}][{\ChGANA[3][1]}]\PolyNotEqual 1$. \LemRef{Lem7} then implies that $\Dover[3]\NotEqualEmpty$. By similar arguments but focusing on \DS[3] instead, we can also prove that $\Sover[3]\NotEqualEmpty$ and $\Dover[2]\NotEqualEmpty$.

We also notice that out of the four channel gains $\{\ChGANA[1][1], \ChGANA[2][1], \ChGANA[3][2], \ChGANA[1][3]\}$, both $\ChGANA[1][1]$ and $\ChGANA[2][1]$ are the only channel gains starting from $s_1$. By \DS[4], we thus have for some positive integer $l_4$ such that
\begin{align}
\GCD[{\ChGANA[1][1]\ChGANA[2][1]^{l_4}}][\;{\ChGANA[3][1]}]\PolyNotEqual 1. \label{RS[26]-2}
\end{align}

Similarly, by \DS[3] and \DS[4], we have for some positive integers $l_2$ and $l_4$ such that
\begin{align}
& \GCD[{\ChGANA[1][1]\ChGANA[1][3]^{l_4}}][\;{\ChGANA[1][2]}]\not\equiv 1, \label{RS[26]-3} \\
& \GCD[{\ChGANA[1][1]\ChGANA[3][1]^{l_2}}][\;{\ChGANA[2][1]}]\not\equiv 1, \label{RS[26]-4} \\
& \GCD[{\ChGANA[1][1]\ChGANA[1][2]^{l_2}}][\;{\ChGANA[1][3]}]\not\equiv 1. \label{RS[26]-5}
\end{align}

For the following, we will prove $\Sover[1]\NotEqualEmpty$. Consider the following subcases: Subcase 1: If $\GCD[{\ChGANA[2][1]}][\,{\ChGANA[3][1]}]\PolyNotEqual 1$, then by \LemRef{Lem7}, $\Sover[1]\NotEqualEmpty$. Subcase 2: If $\GCD[{\ChGANA[2][1]}][\,{\ChGANA[3][1]}]\PolyEqual 1$, then \Ref{RS[26]-2} and \Ref{RS[26]-4} jointly imply both $\GCD[{\ChGANA[1][1]}][\,{\ChGANA[3][1]}]\PolyNotEqual 1$ and $\GCD[{\ChGANA[1][1]}][\,{\ChGANA[2][1]}]\PolyNotEqual 1$. Then by first applying \LemRef{Lem7} and then applying \LemRef{Lem6}, we have $\Sover[1]\NotEqualEmpty$. The proof of $\Dover[1]\NotEqualEmpty$ can be derived similarly by focusing on \Ref{RS[26]-3} and \Ref{RS[26]-5}. The proof of \RS[26] is complete.
\end{proof}

We prove \RS[27] as follows.
\begin{proof} We prove an equivalent relationship: (\NotGS[27])\AND\\\noindent(\NotGS[28])\AND(\NotGS[29])\AND(\NotGS[30])\IMPLY\NotLNR. Suppose (\NotGS[27])\AND\\\noindent(\NotGS[28])\AND(\NotGS[29])\AND(\NotGS[30]) is true. By \LemRef{Lem4}, we know that (\NotGS[27])\AND(\NotGS[28]) is equivalent to $\Sover[2]\cap\Sover[3]\NotEqualEmpty$. Similarly, (\NotGS[29])\AND(\NotGS[30]) is equivalent to $\Dover[2]\CAP\Dover[3]\NotEqualEmpty$. By \PropRef{Prop4}, we have $\LReq$. The proof is thus complete.
\end{proof}

\section{Proof of \RS[28]}\label{ProofRS[28]}

\subsection{The additional set of logic statements}

To prove \RS[28], we need an additional set of logic statements. The following logic statements are well-defined if and only if \GS[27]\AND\GS[28] is true. Recall the definition of $e^\ast_2$, $e^\ast_3$, and $e^\ast_1$ in \AppRef{ProofSS[14]} when \GS[27]\AND\GS[28] is true.

\noindent $\bullet$ \GS[32]{\bf:}\label{ref:GS[32]} $e^\ast_2\!\neq\!e^\ast_3$ and  $\GCD[{\ChG[e^\ast_2][e_{s_2}]\ChG[e^\ast_1][e^\ast_2]}][\,{\ChG[e^\ast_3][e_{s_3}]\ChG[e^\ast_1][e^\ast_3]}]\,\PolyEqual 1$.

\noindent $\bullet$ \GS[33]{\bf:}\label{ref:GS[33]} $\GCD[{\ChGANA[1][1]}][\,{\ChG[e^\ast_1][e^\ast_2]}]\,\PolyEqual 1$.

\noindent $\bullet$ \GS[34]{\bf:}\label{ref:GS[34]} $\GCD[{\ChGANA[1][1]}][\,{\ChG[e^\ast_1][e^\ast_3]}]\,\PolyEqual 1$.

The following logic statements are well-defined if and only if \GS[27]\AND\GS[28]\AND\GS[31] is true.

\noindent $\bullet$ \GS[35]{\bf:}\label{ref:GS[35]} $\{e^\ast_2,e^\ast_1\}\!\subset\!\onecut[s_1][d_2]$.

\noindent $\bullet$ \GS[36]{\bf:}\label{ref:GS[36]} $\{e^\ast_3,e^\ast_1\}\!\subset\!\onecut[s_1][d_3]$.

\subsection{The skeleton of proving \RS[28]}

We prove the following logic relationships, which jointly proves \RS[28].

\noindent $\bullet$ \RS[34]{\bf:}\label{ref:RS[34]} \GS[27]\AND\GS[28]\IMPLY\GS[32].

\noindent $\bullet$ \RS[35]{\bf:}\label{ref:RS[35]} \DS[4]\AND\GS[27]\AND\GS[28]\AND\GS[31]\AND\GS[33]\IMPLY\GS[35].

\noindent $\bullet$ \RS[36]{\bf:}\label{ref:RS[36]} \DS[3]\AND\GS[27]\AND\GS[28]\AND\GS[31]\AND\GS[34]\IMPLY\GS[36].

\noindent $\bullet$ \RS[37]{\bf:}\label{ref:RS[37]} \GS[27]\AND\GS[28]\AND(\NotGS[33])\AND(\NotGS[34])\IMPLY\CONT.

\noindent $\bullet$ \RS[38]{\bf:}\label{ref:RS[38]} \GS[27]\AND\GS[28]\AND\GS[31]\AND(\NotGS[33])\AND\GS[36]\IMPLY\CONT.

\noindent $\bullet$ \RS[39]{\bf:}\label{ref:RS[39]} \GS[27]\AND\GS[28]\AND\GS[31]\AND(\NotGS[34])\AND\GS[35]\IMPLY\CONT.

\noindent $\bullet$ \RS[40]{\bf:}\label{ref:RS[40]} \GS[27]\AND\GS[28]\AND\GS[31]\AND\GS[35]\AND\GS[36]\IMPLY\CONT.

Specifically, \RS[35] and \RS[39] jointly imply that
\begin{align*}
\textbf{D3}\wedge\textbf{D4}\wedge\textbf{G27}\wedge\textbf{G28}\wedge\textbf{G31}\wedge\textbf{G33}\wedge(\neg\,\textbf{G34})\Rightarrow\,\text{false}. \end{align*}

Moreover, \RS[36] and \RS[38] jointly imply that
\begin{align*}
\textbf{D3}\wedge\textbf{D4}\wedge\textbf{G27}\wedge\textbf{G28}\wedge\textbf{G31}\wedge(\neg\,\textbf{G33})\wedge\textbf{G34}\Rightarrow\,\text{false}. \end{align*}

Furthermore, \RS[35], \RS[36], and \RS[40] jointly imply that
\begin{align*}
\textbf{D3}\wedge\textbf{D4}\wedge\textbf{G27}\wedge\textbf{G28}\wedge\textbf{G31}\wedge\textbf{G33}\wedge\textbf{G34}\Rightarrow\,\text{false}. \end{align*}

Finally, \RS[37] implies that
\begin{align*}
\textbf{D3}\wedge\textbf{D4}\wedge\textbf{G27}\wedge\textbf{G28}\wedge\textbf{G31}\wedge(\neg\,\textbf{G33})\wedge(\neg\,\textbf{G34})\Rightarrow\,\text{false}. \end{align*}


The above four relationships jointly imply \DS[3]\AND\DS[4]\AND\GS[27]\\\noindent\AND\GS[28]\AND\GS[31]\IMPLY\CONT. By \RS[26] in \AppRef{ProofSS[14]}, i.e., \DS[3]\AND\DS[4]\\\noindent\IMPLY\GS[31], we thus have \DS[3]\AND\DS[4]\AND\GS[27]\AND\GS[28]\IMPLY\CONT. The proof of \RS[28] is thus complete. The detailed proofs of \RS[34] to \RS[40] are provided in the next subsection.

\subsection{The proofs of \RS[34] to \RS[40]}

We prove \RS[34] as follows.
\begin{proof} Suppose \GS[27]\AND\GS[28] is true. Since $e^\ast_1$ is the most upstream $1$-edge cut separating $d_1$ from $\{s_2,s_3\}$, there must exist two edge-disjoint paths connecting $\{s_2,s_3\}$ and $\tail[e^\ast_1]$. By Property~1 of \GS[27] and \GS[28], one path must use $e^\ast_2$ and the other must use $e^\ast_3$. Due to the edge-disjointness, $e^\ast_2\!\neq\!e^\ast_3$. Since we have two edge-disjoint paths from $s_2$ (resp. $s_3$) to $\tail[e^\ast_1]$, we also have $\GCD[{\ChG[e_2^\ast][{\EDGE[s_2][]}]\ChG[e_1^\ast][e_2^\ast]}][{\,\ChG[e_3^\ast][{\EDGE[s_3][]}]\ChG[e_1^\ast][e_3^\ast]}]\,\PolyEqual 1$.
\end{proof}

We prove \RS[35] as follows.
\begin{proof} Suppose \DS[4]\AND\GS[27]\AND\GS[28]\AND\GS[31]\AND\GS[33] is true. By the Properties of \GS[27] and \GS[28] and by \GS[31], $e^\ast_2$ (resp. $e^\ast_3$) is the most downstream edge of $\Sover[2]$ (resp. $\Sover[3]$). And both $e^\ast_2$ and $e^\ast_3$ are in the upstream of $e^\ast_1$ where $e^\ast_1$ is the most upstream edge of $\Dover[1]$. Consider $\ChG[e^\ast_1][e^\ast_2]$, a factor of $\ChGANA[1][2]$. From Property 2 of \GS[27], we have $\GCD[{\ChGANA[3][2]}][\,{\ChG[e^\ast_1][e^\ast_2]}]\,\PolyEqual 1$. In addition, since \GS[27]\AND\GS[28]\IMPLY\GS[32] as established in \RS[34], we have $\GCD[{\ChGANA[1][3]}][\,{\ChG[e^\ast_1][e^\ast_2]}]\,\PolyEqual 1$. Together with the assumption that \DS[4] is true, we have for some positive integer $l_4$ such that
\begin{align}
\GCD[{\ChGANA[1][1]\ChGANA[2][1]^{l_4}}][\;{\ChG[e_1^\ast][e_2^\ast]}]=\ChG[e_1^\ast][e_2^\ast]. \label{RS11}
\end{align}

Since we assume that \GS[33] is true, \Ref{RS11} further implies $\GCD[{\ChGANA[2][1]^{l_4}}][\,{\ChG[e_1^\ast][e_2^\ast]}]\!=\!\ChG[e_1^\ast][e_2^\ast]$. By \PropRef{Prop3}, we must have \GS[35]: $\{e^\ast_2,e^\ast_1\}\!\subset\!\onecut[s_1][d_2]$. The proof is thus complete.
\end{proof}

\RS[36] is a symmetric version of \RS[35] and can be proved by relabeling $(s_2,d_2)$ as $(s_3,d_3)$, and relabeling $(s_3,d_3)$ as $(s_2,d_2)$ in the proof of \RS[35].

We prove \RS[37] as follows.
\begin{proof} Suppose \GS[27]\AND\GS[28]\AND(\NotGS[33])\AND(\NotGS[34]) is true. Since \GS[27]\AND\GS[28] is true, we have two edge-disjoint paths $P_{s_2 \tail[e^\ast_1]}$ through $e^\ast_2$ and $P_{s_3 \tail[e^\ast_1]}$ through $e^\ast_3$ if we recall \RS[34]. Consider $\ChG[e^\ast_1][e^\ast_2]$, a factor of $\ChGANA[1][2]$, and $\ChG[e^\ast_1][e^\ast_3]$, a factor of $\ChGANA[1][3]$. Since \NotGS[33] is true, there is an irreducible factor of $\ChG[e^\ast_1][e^\ast_2]$ that is also a factor of $\ChGANA[1][1]$. Since that factor is also a factor of $\ChGANA[1][2]$, by \PropRef{Prop3} and Property~1 of \GS[27], there must exist at least one edge $e'$ satisfying (i) $e^\ast_2 \PRECEQ e' \PREC e^\ast_1$; (ii) $e'\!\in\!\Dover[1][;\{1,2\}]$; and (iii) $e'\!\in\!P_{s_2 \tail[e^\ast_1]}$. 
Similarly, \NotGS[34] implies that there exists at least one edge $e''$ satisfying (i) $e^\ast_3 \PRECEQ e'' \PREC e^\ast_1$; (ii) $e''\!\in\!\Dover[1][;\{1,3\}]$; and (iii) $e''\!\in\!P_{s_3 \tail[e^\ast_1]}$. Then the above observation implies that $e'\!\in\!P_{s_2 \tail[e^\ast_1]}\cap\onecut[s_1][d_1]$ and $e''\!\in\!P_{s_3 \tail[e^\ast_1]}\cap\onecut[s_1][d_1]$. Since $P_{s_2 \tail[e^\ast_1]}$ and $P_{s_3 \tail[e^\ast_1]}$ are edge-disjoint paths, it must be $e'\!\neq\!e''$. But both $e'$ and $e''$ are $1$-edge cuts separating $s_1$ and $d_1$. Thus $e'$ and $e''$ must be reachable from each other: either $e'\PREC e''$ or $e'' \PREC e'$. However, both cases are impossible because one in the upstream can always follow the corresponding $P_{s_2 \tail[e^\ast_1]}$ or $P_{s_3 \tail[e^\ast_1]}$ path to $e^\ast_1$ without using the one in the downstream. For example, if $e' \PREC e''$, then $s_1$ can first reach $e'$ and follow $P_{s_2 \tail[e^\ast_1]}$ to arrive at $\tail[e^\ast_1]$ without using $e''$. Since $e^\ast_1\!\in\!\Dover[1]$ reaches $d_1$, this contradicts $e''\!\in\!\onecut[s_1][d_1]$. Since neither case can be true, the proof is thus complete.
\end{proof}

We prove \RS[38] as follows.
\begin{proof} Suppose \GS[27]\AND\GS[28]\AND\GS[31]\AND(\NotGS[33])\AND\GS[36] is true. By the Properties of \GS[27] and \GS[28] and by \GS[31], $e^\ast_2$ (resp. $e^\ast_3$) is the most downstream edge of $\Sover[2]$ (resp. $\Sover[3]$). And both $e^\ast_2$ and $e^\ast_3$ are in the upstream of $e^\ast_1$ where $e^\ast_1$ is the most upstream edge of $\Dover[1]$. Since $e^\ast_1$ is the most upstream $\Dover[1]$ edge, there exist three edge-disjoint paths $P_{s_2 \tail[e^\ast_1]}$, $P_{s_3 \tail[e^\ast_1]}$, and $P_{\head[e^\ast_1] d_1}$. Fix any arbitrary construction of these paths. Obviously, $P_{s_2 \tail[e^\ast_1]}$ uses $e^\ast_2$ and $P_{s_3 \tail[e^\ast_1]}$ uses $e^\ast_3$.

Since \NotGS[33] is true, there is an irreducible factor of $\ChG[{e^\ast_1}][{e^\ast_2}]$ that is also a factor of $\ChGANA[1][1]$. Since that factor is also a factor of $\ChGANA[1][2]$, by \PropRef{Prop3}, there must exist an edge $e$ satisfying (i) $e^\ast_2 \PRECEQ e \PREC e^\ast_1$; (ii) $e\!\in\!\onecut[s_1][d_1]\CAP\onecut[s_2][d_1]$. By (i), (ii), and the construction $e^\ast_1\!\in\!\Dover[1]\!\subset\!\onecut[s_2][d_1]$, the pre-defined path $P_{s_2 \tail[e^\ast_1]}$ must use such $e$.

Since \GS[36] is true, $e^\ast_3$ is reachable from $s_1$ and $e^\ast_1$ reaches to $d_3$. Choose arbitrarily one path $P_{s_1 \tail[e^\ast_3]}$ from $s_1$ to $\tail[e^\ast_3]$ and one path $P_{\head[e^\ast_1] d_3}$ from $\head[e^\ast_1]$ to $d_3$. We argue that $P_{s_1 \tail[e^\ast_3]}$ must be vertex-disjoint with $P_{s_2 \tail[e^\ast_1]}$. Suppose not and let $v$ denote a vertex shared by $P_{s_1 \tail[e^\ast_3]}$ and $P_{s_2 \tail[e^\ast_1]}$. Then there is a \FromTo[1][3] path $P_{s_1 \tail[e^\ast_3]}v P_{s_2 \tail[e^\ast_1]} e^\ast_1 P_{\head[e^\ast_1] d_3}$ without using $e^\ast_3$. This contradicts the assumption \GS[36] since \GS[36] implies $e^\ast_3\!\in\!\onecut[s_1][d_3]$. However, if $P_{s_1 \tail[e^\ast_3]}$ is vertex-disjoint with $P_{s_2 \tail[e^\ast_1]}$, then there is an \FromTo[1][1] path $P_{s_1 \tail[e^\ast_3]}$ $e^\ast_3 P_{s_3 \tail[e^\ast_1]} e^\ast_1 P_{\head[e^\ast_1] d_1\!}$ not using the edge $e$ defined in the previous paragraph since $e\!\in\!P_{s_2 \tail[e^\ast_1]}$ and $P_{s_2 \tail[e^\ast_1]}$ is edge-disjoint with $P_{s_3 \tail[e^\ast_1]}$. This also contradicts (ii). Since neither case can be true, the proof of \RS[38] is thus complete.
\end{proof}

\RS[39] is a symmetric version of \RS[38] and can be proved by swapping the roles of $s_2$ and $s_3$, and the roles of $d_2$ and $d_3$ in the proof of \RS[38].


We prove \RS[40] as follows.
\begin{proof} Suppose \GS[27]\AND\GS[28]\AND\GS[31]\AND\GS[35]\AND\GS[36] is true. By the Properties of \GS[27] and \GS[28] and by \GS[31], $e^\ast_2$ (resp. $e^\ast_3$) is the most downstream edge of $\Sover[2]$ (resp. $\Sover[3]$). Also $e^\ast_2 \PREC e^\ast_1$ and $e^\ast_3 \PREC e^\ast_1$ where $e^\ast_1$ is the most upstream $\Dover[1]$ edge.

By \GS[36], there exists a path from $s_1$ to $e^\ast_3$. Since $e^\ast_3\!\in\!\Sover[3]$, there exists a path from $e^\ast_3$ to $d_2$ without using $e^\ast_1$. As a result, there exists a path from $s_1$ to $d_2$ through $e^\ast_3$ without using $e^\ast_1$. This contradicts the assumption \GS[35] since \GS[35] implies $e^\ast_1\!\in\!\onecut[s_1][d_2]$. The proof is thus complete.
\end{proof}

\section{Proof of \RS[29]}\label{ProofRS[29]}

\subsection{The additional set of logic statements}

To prove \RS[29], we need some additional sets of logic statements. The following logic statements are well-defined if and only if \GS[28] is true. Recall the definition of $e^\ast_3$ and $e^\ast_1$ when \GS[28] is true.

\noindent $\bullet$ \GS[37]{\bf:}\label{ref:GS[37]} $e^\ast_3\in\onecut[s_1][d_1]$.

\noindent $\bullet$ \GS[38]{\bf:}\label{ref:GS[38]} $e^\ast_3\in\onecut[s_1][d_3]$.

\noindent $\bullet$ \GS[39]{\bf:}\label{ref:GS[39]} $e^\ast_1\in\onecut[s_1][d_1]$.

\noindent $\bullet$ \GS[40]{\bf:}\label{ref:GS[40]} $e^\ast_1\in\onecut[s_1][d_3]$.

\noindent $\bullet$ \GS[41]{\bf:}\label{ref:GS[41]} $e^\ast_3\in\onecut[s_1][d_2]$.

The following logic statements are well-defined if and only if (\NotGS[27])\AND\GS[28] is true. Recall the definition of $e^{21}_u$, $e^{21}_v$, $e^\ast_3$, and $e^\ast_1$ when (\NotGS[27])\AND\GS[28] is true.

\noindent $\bullet$ \GS[42]{\bf:}\label{ref:GS[42]} $e^\ast_1 = e^{21}_u$. 

\noindent $\bullet$ \GS[43]{\bf:}\label{ref:GS[43]} Let $e'$ be the most downstream edge of $\onecut[s_1][d_2]$ $\cap\,\onecut[s_1][{\tail[e^\ast_3]}]$ and also let $e''$ be the most upstream edge of $\onecut[s_1][d_2]\cap\onecut[{\head[e^\ast_3]}][d_2]$. Then, $e'$ and $e''$ simultaneously satisfy the following two conditions: (i) both $e'$ and $e''$ belong to $\onecut[s_1][d_3]$; and (ii) $e''\!\in\!\onecut[{\head[e^{21}_v]}][{\tail[e_{d_3}]}]$ and $e''\PREC e_{d_2}$.

\subsection{The skeleton of proving \RS[29]}

We prove the following relationships, which jointly proves \RS[29].

\noindent $\bullet$ \RS[41]{\bf:}\label{ref:RS[41]} (\NotGS[27])\AND\GS[28]\IMPLY\GS[42].

\noindent $\bullet$ \RS[42]{\bf:}\label{ref:RS[42]} \DS[3]\AND(\NotGS[27])\AND\GS[28]\AND\GS[31]\IMPLY(\GS[37]\OR\GS[38])\AND(\GS[39]\OR\\\noindent\GS[40]).

\noindent $\bullet$ \RS[43]{\bf:}\label{ref:RS[43]} \GS[1]\AND\GS[28]\AND\GS[31]\AND\GS[37]\IMPLY\NotGS[41].

\noindent $\bullet$ \RS[44]{\bf:}\label{ref:RS[44]} \DS[3]\AND(\NotGS[27])\AND\GS[28]\AND\GS[31]\AND\GS[37]\AND(\NotGS[41])\IMPLY\GS[43].

\noindent $\bullet$ \RS[45]{\bf:}\label{ref:RS[45]} \GS[1]\AND\ES[0]\AND\DS[3]\AND(\NotGS[27])\AND\GS[28]\AND\GS[31]\AND\GS[37]\IMPLY\CONT.

\noindent $\bullet$ \RS[46]{\bf:}\label{ref:RS[46]} (\NotGS[27])\AND\GS[28]\AND\GS[31]\AND(\NotGS[37])\AND\GS[38]\AND\GS[39]\IMPLY\CONT.

\noindent $\bullet$ \RS[47]{\bf:}\label{ref:RS[47]} \LNR\AND\DS[4]\AND(\NotGS[27])\AND\GS[28]\AND\GS[31]\AND(\NotGS[37])\AND\GS[38]\AND\\\noindent\GS[40]\IMPLY\CONT.

One can easily verify that jointly \RS[46] and \RS[47] imply
\begin{align*}
\textbf{LNR}\wedge\textbf{D4}&\wedge(\neg\,\textbf{G27})\wedge\textbf{G28}\wedge\textbf{G31} \\
& \wedge(\neg\,\textbf{G37})\wedge\,\textbf{G38}\wedge(\textbf{G39}\vee\textbf{G40})\Rightarrow\text{false}.
\end{align*}

From the above logic relationship and by \RS[42], we have
\begin{align*}
\textbf{LNR}\wedge\textbf{D3}\wedge\textbf{D4}\wedge(\neg\,\textbf{G27})&\wedge\textbf{G28}\wedge\textbf{G31} \\ &\wedge(\neg\,\textbf{G37})\wedge\textbf{G38}\Rightarrow\text{false}.
\end{align*}

From the above logic relationship and by \RS[45], we have
\begin{align*}
\textbf{LNR}\wedge\textbf{G1}\wedge\textbf{E0}\wedge\textbf{D3}\wedge\textbf{D4}&\wedge(\neg\,\textbf{G27})\wedge\textbf{G28}\wedge\textbf{G31} \\
& \wedge(\textbf{G37}\vee\textbf{G38})\Rightarrow\text{false}.
\end{align*}

By applying \RS[42] and \RS[26], we have \LNR\AND\GS[1]\AND\ES[0]\AND\DS[3]\\\noindent\AND\DS[4]\AND(\NotGS[27])\AND\GS[28]\IMPLY\CONT, which proves \RS[29]. The detailed proofs for \RS[41] to \RS[47] are provided in the next subsection.

\subsection{The proofs of \RS[41] to \RS[47]}

We prove \RS[41] as follows.
\begin{proof} Suppose (\NotGS[27])\AND\GS[28] is true. By \NotGS[27] being true and its Property~1, we have $e^{21}_u$ (resp. $e^{21}_v$), the most upstream (resp. downstream) edge of $\Sover[2]\cap\Dover[1]$. Since \NotGS[27] implies that $\Dover[1]\NotEqualEmpty$, by Property~1 of \GS[28], we also have $e^\ast_1$, the most upstream $\Dover[1]$ edge.

Since $\Dover[1]\CAP\Sover[2]\NotEqualEmpty$, we can partition the non-empty $\Dover[1]$ by $\Dover[1]\backslash\Sover[2]$ and $\Dover[1]\CAP\Sover[2]$. By the \SWAPSD-symmetric version of \LemRef{Lem3}, if $\Dover[1]\backslash\Sover[2]\NotEqualEmpty$, then any $\Dover[1]\backslash\Sover[2]$ edge must be in the downstream of $e^{21}_v\!\in\!\Dover[1]\cap\Sover[2]\!\subset\!\Sover[2]$. Thus, $e^{21}_u$, the most upstream $\Dover[1]\cap\Sover[2]$ edge, must also be the most upstream edge of $\Dover[1]$. Therefore, $e^\ast_1 = e^{21}_u$. The proof is thus complete.
\end{proof}

We prove \RS[42] as follows.
\begin{proof} Suppose \DS[3]\AND(\NotGS[27])\AND\GS[28]\AND\GS[31] is true. Since (\NotGS[27])\AND\GS[28]\AND\GS[31] is true, $e^\ast_3$ (resp. $e^\ast_1$) is the most downstream (resp. upstream) edge of $\Sover[3]$ (resp. $\Dover[1]$) and $e^\ast_3\PREC e^\ast_1$. By \RS[41], \GS[42] is also true and thus $e^\ast_1$ is also the most upstream edge of $\Sover[2]\cap\Dover[1]$.

Consider $\ChG[e^\ast_1][e^\ast_3]$, a factor of $\ChGANA[1][3]$. From Property~2 of \GS[28], $\GCD[{\ChGANA[2][3]}][\,{\ChG[e^\ast_1][e^\ast_3]}]\PolyEqual 1$. By \GS[42] being true and Property~2 of \NotGS[27], we also have $\GCD[{\ChG[e^\ast_1][e_{s_2}]}][\,{\ChG[e^\ast_3][e_{s_3}]\ChG[e^\ast_1][e^\ast_3]}]\PolyEqual 1$, which implies that $\GCD[{\ChGANA[1][2]}][\,{\ChG[e^\ast_1][e^\ast_3]}]\,\PolyEqual 1$. Then since \DS[3] is true, we have for some positive integer $l_2$ such that
\begin{equation*}
\GCD[{\ChGANA[1][1]\ChGANA[3][1]^{l_2}}][{\;\ChG[e^\ast_1][e^\ast_3]}]\!=\!\ChG[e^\ast_1][e^\ast_3].
\end{equation*}

\PropRef{Prop3} then implies that both $e^\ast_3$ and $e^\ast_1$ must be in $\onecut[s_1][d_1]\cup\onecut[s_1][d_3]$. This is equivalent to (\GS[37]\OR\GS[38])\\\noindent\AND(\GS[39]\OR\GS[40]) being true. The proof of \RS[42] is complete.
\end{proof}

%

We prove \RS[43] as follows.
\begin{proof} We prove an equivalent form: \GS[28]\AND\GS[31]\AND\GS[37]\\\noindent\AND\GS[41]\IMPLY\NotGS[1]. Suppose \GS[28]\AND\GS[31]\AND\GS[37]\AND\GS[41] is true. Since \GS[28]\AND\GS[31] is true, we have $e^\ast_3$ being the most downstream edge of $\Sover[3]$. Therefore $e^\ast_3\!\in\onecut[s_3][d_1]\cap\onecut[s_3][d_2]$. Since \GS[37]\AND\GS[41] is also true, $e^\ast_3$ belongs to $\onecut[s_1][d_1]\cap\onecut[s_1][d_2]$ as well. As a result, $\EC[\{s_1,s_3\}][\{d_1,d_2\}]\!=\!1$, which, by \CorRef{Cor2} implies \NotGS[1].
\end{proof}

We prove \RS[44] as follows.
\begin{proof} Suppose that \DS[3]\AND(\NotGS[27])\AND\GS[28]\AND\GS[31]\AND\GS[37]\AND\\\noindent(\NotGS[41]) is true, which by \RS[41] implies that \GS[42] is true as well. Since \GS[28]\AND\GS[31] is true, $e^\ast_3$ (resp. $e^\ast_1$) is the most downstream (resp. upstream) edge of $\Sover[3]$ (resp. $\Dover[1]$) and $e^\ast_3 \PREC e^\ast_1$. Recall the definition in \GS[43] that $e'$ is the most downstream edge of $\onecut[s_1][d_2]\cap\onecut[s_1][{\tail[e^\ast_3]}]$ and $e''$ is the most upstream edge of $\onecut[s_1][d_2]\cap\onecut[{\head[e^\ast_3]}][d_2]$. By the constructions of $e'$ and $e''$, we must have $e_{s_1}\PRECEQ e' \PREC e^\ast_3 \PREC e'' \PRECEQ e_{d_2}$. Then, we claim that the above construction together with \NotGS[41] implies $\EC[{\head[e']}][{\tail[e'']}]\!\geq\!2$. The reason is that if $\EC[{\head[e']}][{\tail[e'']}]\!=\!1$, then we can find an $1$-edge cut separating $\head[e']$ and $\tail[e'']$ and by \NotGS[41] such edge cut must not be $e^\ast_3$. Hence, such edge cut is either an upstream or a downstream edge of $e^\ast_3$. However, either case is impossible, because the edge cut being in the upstream of $e^\ast_3$ will contradict that $e'$ is the most downstream one during its construction. Similarly, the edge cut being in downstream of $e^\ast_3$ will contradict the construction of $e''$. The conclusion $\EC[{\head[e']}][{\tail[e'']}]\!\geq\!2$ further implies $\ChG[e''][e']$ is irreducible.


Further, because $e^\ast_3$ is the most downstream $\Sover[3]$ edge and $e''$, by construction, satisfies $e''\!\in\!\onecut[s_3][d_2]$, $e''$ must not belong to $\onecut[s_3][d_1]$, which in turn implies $e''\!\not\in\!\onecut[{\head[e^\ast_3]}][d_1]$. Since \GS[37] is true, $s_1$ can reach $e^\ast_3$. Therefore, there exists an \FromTo[1][1] path using $e^\ast_3$ but not using $e''$. As a result, $e''\!\not\in\!\onecut[s_1][d_1]$. Together with $\ChG[e''][e']$ being irreducible, we thus have $\GCD[{\ChGANA[1][1]}][\,{\ChG[e''][e']}]\,\PolyEqual 1$ by \PropRef{Prop3}.

Now we argue that $\GCD[{\ChGANA[1][2]}][\,{\ChG[e''][e']}]\PolyEqual 1$. Suppose not. Since $\ChG[e''][e']$ is irreducible, we must have $e'$ being an $1$-edge cut separating $s_2$ and $d_1$. Since $e^\ast_1$ is the most upstream $\Dover[1]$ edge, by Property~2 of \GS[28], there exists a \FromTo[2][1] path $P_{21}$ not using $e^\ast_3$. By the construction of $e'$, $s_1$ reaches $e'$. Choose arbitrarily a path $P_{s_1 e'}$ from $s_1$ to $e'$. Then, the following \FromTo[1][1] path $P_{s_1 e'}e'P_{21}$ does not use $e^\ast_3$, which contradicts \GS[37]. As a result, we must have $\GCD[{\ChGANA[1][2]}][\,{\ChG[e''][e']}]\,\PolyEqual 1$.

Now we argue that $\GCD[{\ChGANA[2][3]}][\,{\ChG[e''][e']}]\PolyEqual 1$. Suppose not. Since $\ChG[e''][e']$ is irreducible, both $e'$ and $e''$ must belong to $\onecut[s_3][d_2]$ and there is no $1$-edge cut of $\onecut[s_3][d_2]$ that is strictly being downstream to $e'$ and being upstream to $e''$. This, however, contradicts the above construction that $e' \PREC e^\ast_3 \PREC e''$ and $e^\ast_3\!\in\!\Sover[3]\!\subset\!\onecut[s_3][d_2]$. As a result, we must have $\GCD[{\ChGANA[2][3]}][\,{\ChG[e''][e']}]\,\PolyEqual 1$.

Together with the assumption that \DS[3] is true and the fact that $\ChG[e''][e']$ is a factor of $\ChGANA[2][1]$, we have for some positive integer $l_2$ such that
\begin{equation*}
\GCD[{\ChGANA[3][1]^{l_2}}][\;{\ChG[e''][e']}] = \ChG[e''][e'].
\end{equation*}

\PropRef{Prop3} then implies $\{e',e''\}\!\subset\!\onecut[s_1][d_3]$, which shows the first half of \GS[43].

Therefore, any \FromTo[1][3] path must use $e''$. Since $e^\ast_3 \PREC e''$ and $s_1$ can reach $e^\ast_3$, any path from $\head[e^\ast_3]$ to $d_3$ must use $e''$. Note that when we establish $\GCD[{\ChGANA[1][1]}][\,{\ChG[e''][e']}]\,\PolyEqual 1$ in the beginning of this proof, we also proved that $e''\!\not\in\onecut[s_1][d_1]$. Thus, there exists a path from $\head[e^\ast_3]$ to $d_1$ not using $e''$. Then such path must use $e^{21}_v$ because $e^{21}_v$ is also an $1$-edge cut separating $\head[e^\ast_3]$ and $d_1$ by the facts that $e^{21}_v\!\in\!\Sover[2]\,\cap$ $\Dover[1]\!\subset\!\onecut[s_3][d_1]$; $e^\ast_3\PREC e^{21}_v$; $s_3$ reaches $e^\ast_3$. Moreover, since $e^{21}_v\!\in\!\Sover[2]\CAP\Dover[1]\!\subset\!\onecut[s_2][d_3]$, $\head[e^{21}_v]$ can reach $d_3$. In sum, we have shown that (i) any path from $\head[e^\ast_3]$ to $d_3$ must use $e''$; (ii) there exists a path from $e^\ast_3$ to $e^{21}_v$ not using $e''$; (iii) $\head[e_v^{21}]$ can reach $d_3$. Jointly (i) to (iii) imply that any path from $\head[e^{21}_v]$ to $d_3$ must use $e''$. As a result, we have $e''\!\in\!\onecut[{\head[e^{21}_v]}][d_3]$. Also $e''$ must not be the $d_3$-destination edge $e_{d_3}$ since by construction $e''\PRECEQ e_{d_2}$, $e_{d_2}\!\neq\! e_{d_3}$, and $|\OUT[d_3]|\!=\!0$. This further implies that $e''$ must not be the $d_2$-destination edge $e_{d_2}$ since $e''\PREC e_{d_3}$ and $|\OUT[d_2]|\!=\!0$. We have thus proven the second half of \GS[43]: $e''\!\in\!\onecut[{\head[e^{21}_v]}][{\tail[{e_{d_3}}]}]$ and $e''\PREC e_{d_2}$. The proof of \RS[44] is complete.
\end{proof}

We prove \RS[45] as follows.
\begin{proof} Suppose \GS[1]\AND\ES[0]\AND\DS[3]\AND(\NotGS[27])\AND\GS[28]\AND\GS[31]\AND\\\noindent\GS[37] is true. By \RS[41], \RS[43], and \RS[44], we know that \GS[42], \NotGS[41], and \GS[43] are true as well. For the following we construct 10 path segments that interconnects $s_1$ to $s_3$, $d_1$ to $d_3$, and three edges $e''$, $e^\ast_3$, and $e^\ast_1$.

\noindent \makebox[1cm][l]{$\bullet$ $P_{1}$:}a path starting from $e_{s_1}$ and ending at $e'$. This is always possible due to \GS[43] being true.

\noindent \makebox[1cm][l]{$\bullet$ $P_{2}$:}a path from $s_2$ to $\tail[e^\ast_1]$ without using $e^\ast_3$. This is always possible due to the properties of \GS[28]. 

\noindent \makebox[1cm][l]{$\bullet$ $P_{3}$:}a path from $s_3$ to $\tail[e^\ast_3]$. This is always possible due to \GS[28] and \GS[31] being true. We also impose that $P_3$ is edge-disjoint with $P_2$. Again, this is always possible due to Property~2 of \GS[28].

\noindent \makebox[1cm][l]{$\bullet$ $P_{4}$:}a path from $\head[e']$ to $\tail[e'']$. This is always possible due to \GS[43] being true. We also impose the condition that $e^\ast_3\!\not\in\! P_4$. Again this is always possible since \NotGS[41] being true, which implies that one can always find a path from $s_1$ to $d_2$ not using $e^\ast_3$ but uses both $e'$ and $e''$ (due to the construction of $e'$ and $e''$ of \GS[43]).

\noindent \makebox[1cm][l]{$\bullet$ $P_{5}$:}a path from $\head[e^\ast_3]$ to $\tail[e^\ast_1]$. We also impose the condition that $P_5$ is edge-disjoint with $P_2$. The construction of such $P_5$ is always possible due to the Properties of \GS[28].

\noindent \makebox[1cm][l]{$\bullet$ $P_{6}$:}a path from $\head[e^\ast_1]$ to $d_1$. This is always possible due to (\NotGS[27])\AND\GS[28] being true. We also impose the condition that $e''\!\not\in\! P_6$. Again this is always possible. The reason is that $e^\ast_3$ is the most downstream $\Sover[3]$ edge and thus there are two edge-disjoint paths connecting $\head[e^\ast_3]$ and $\{d_1,d_2\}$. By our construction $e''$ must be in the latter path while we can choose $P_6$ to be part of the first path.

\noindent \makebox[1cm][l]{$\bullet$ $P_{7}$:}a path from $\head[e^\ast_3]$ to $\tail[e'']$, which is edge-disjoint with $\{P_5,e^\ast_1,P_6\}$. This is always possible due to the property of $e^\ast_3$ and the construction of \GS[43].

\noindent \makebox[1cm][l]{$\bullet$ $P_{8}$:}a path from $\head[e'']$ to $d_2$, which is edge-disjoint with $\{P_5,e^\ast_1,P_6\}$. This is always possible due to the property of $e^\ast_3$ and the construction of \GS[43].

\noindent \makebox[1cm][l]{$\bullet$ $P_{9}$:}a path from $\head[e^\ast_1]$ to $\tail[e'']$. This is always possible due to \GS[43] being true  (in particular the (ii) condition of \GS[43]).

\noindent \makebox[1.15cm][l]{$\bullet$ $P_{10}$:}a path from $\head[e'']$ to $d_3$. This is always possible due to \GS[43] being true (in particular the (ii) condition of \GS[43]).

Fig. \ref{Prop5+R21+Fig1} illustrates the relative topology of these 10 paths. We now consider the subgraph $\G[']$ induced by the 10 paths plus the three edges $e''$, $e^\ast_3$, and $e^\ast_1$. One can easily check that $s_i$ can reach $d_j$ for all $i\!\neq\!j$. In particular, $s_1$ can reach $d_2$ through $P_1 P_4 e'' P_8$; $s_1$ can reach $d_3$ through $P_1 P_4 e'' P_{10}$; $s_2$ can reach $d_1$ through $P_2 e^\ast_1 P_6$; $s_2$ can reach $d_3$ through $P_2 e^\ast_1 P_9 e'' P_{10}$; $s_3$ can reach $d_1$ through $P_3 e^\ast_3 P_5 e^\ast_1 P_6$; and $s_3$ can reach $d_2$ through either $P_3 e^\ast_3 P_5 e^\ast_1 P_9 e'' P_8$ or $P_3 e^\ast_3 P_7 e'' P_8$.  Furthermore, topologically, the 6 paths $P_5$ to $P_{10}$ are all in the downstream of $e^\ast_3$.

\begin{figure}[t]
\centering
\includegraphics[scale=0.2]{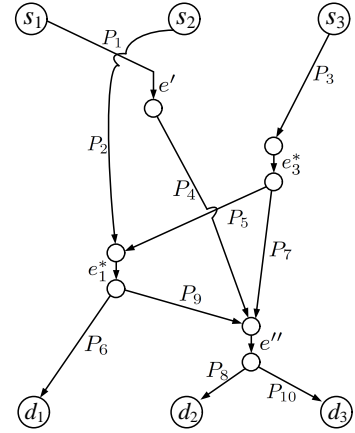}
\caption{The subgraph $\G[']$ of the 3-unicast ANA network $\GANA$ induced by 10 paths and three edges $e''$, $e^\ast_3$, and $e^\ast_1$ in the proof of \RS[45].}
\vspace{-0.04\columnwidth}
\label{Prop5+R21+Fig1}
\end{figure}

For the following we argue that $s_1$ cannot reach $d_1$ in the induced subgraph $\G[']$. To that end, we first notice that by \GS[37], $e^\ast_3\!\in \!\onecut[s_1][d_1]$ in the original graph. Therefore any \FromTo[1][1] path in the subgraph must use $e^\ast_3$ as well. Since $P_5$ to $P_{10}$ are in the downstream of $e^\ast_3$, we only need to consider $P_1$ to $P_4$.

By definition, $P_3$ reaches $e^\ast_3$. We now like to show that $e^\ast_3\!\not\in\!P_2$, and $\{P_2,P_3\}$ are vertex-disjoint paths. The first statement is done by our construction. Suppose $P_2$ and $P_3$ share a common vertex $v$ ($v$ can possibly be $\tail[e^\ast_3]$), then there exists a \FromTo[3][1] path $P_3 v P_2 e^\ast_1 P_6$ not using $e^\ast_3$. This contradicts \GS[28] (specifically $e^\ast_3\!\in\!\Sover[3]\!\subset\!\onecut[s_3][d_1]$). The above arguments show that the first time a path enters/touches part of $P_3$ (including $\tail[e^\ast_3]$) must be along either $P_1$ or $P_4$ (cannot be along $P_2$). As a result, when deciding whether there exists an \FromTo[1][1] path using $e^\ast_3$, we only need to consider whether $P_1$ (and/or $P_4$) can share a vertex with $P_3$. To that end, we will prove that (i) $e^\ast_3\!\not\in\! P_1$; (ii) $\{P_1,P_3\}$ are vertex-disjoint paths; (iii) $e^\ast_3\!\not\in\! P_4$; and (iv) $\{P_3,P_4\}$ are vertex-disjoint paths. Once (i) to (iv) are true, then there is no \FromTo[1][1] path in the subgraph $\G[']$.

We now notice that (i) is true since $e'\PREC e^\ast_3$; (iii) is true due to our construction; (ii) is true otherwise let $v$ denote the shared vertex and there will exist a \FromTo[3][2] path $P_3 v P_1 P_4 e'' P_8$ without using $e^\ast_3$, which contradicts \GS[28] ($e^\ast_3\!\in\!\Sover[3]\!\subset\!\onecut[s_3][d_2]$); and by the same reason, (iv) is true otherwise let $v$ denote the shared vertex and there will exist a \FromTo[3][2] path $P_3 v P_4 e'' P_8$ without using $e^\ast_3$. We have thus proven that there is no \FromTo[1][1] path in $\G[']$.


Since \ES[0] is true, $\GANA$ must satisfy \Ref{E0} with at least one non-zero coefficients $\alpha_i$ and $\beta_j$, respectively. Applying \PropRef{Prop2} implies that the subgraph $\G[']$ must satisfy \Ref{E0} with the same coefficient values. Note that there is no path from $s_1$ to $d_1$ on $\G[']$ but any channel gain $\ChGANA[j][i]$ for all $i\!\neq\!j$ is non-trivial on $\G[']$. Recalling the expression of \Ref{E0}, its LHS becomes zero since it contains the zero polynomial $\ChGANA[1][1]$ as a factor. We have $g(\{\ChGANA[j][i]:\forall\,(i,j)\!\in\!I_{\textrm{3ANA}}\})\,\psi^{(n)}_\beta(\BOLDb,\BOLDa) = 0$ and thus $\psi^{(n)}_\beta(\BOLDb,\BOLDa)=0$ with at least one non-zero coefficients $\beta_j$. This further implies that the set of  polynomials $\{\BOLDb^{n},\BOLDb^{n-1}\BOLDa, \cdots, \BOLDb\BOLDa^{n-1}, \BOLDa^{n}\}$ is linearly dependent on $\G[']$. Since this is the Vandermonde form, it is equivalent to that $\LReq$ holds on $\G[']$. However for the following, we will show that (a) $\Dover[1]\cap\Dover[2]\EqualEmpty$; (b) $\Sover[1]\cap\Sover[3]\EqualEmpty$; and (c) $\Sover[2]\cap\Sover[3]\EqualEmpty$ on $\G[']$, which implies by \PropRef{Prop4} that $\G[']$ indeed satisfies $\LRneq$. This is a contradiction and thus proves \RS[45].

(a) $\Dover[1]\cap\Dover[2]\EqualEmpty$ on $\G[']$: Note that any $\Dover[1]$ edge can exist on (i) $e^\ast_1$; and (ii) $P_6$. Note also that any $\Dover[2]$ edge can exist on (i) $e''$; and (ii) $P_8$. But from the above constructions, $P_6$ was chosen not to use $e''$. In addition, $P_8$ was chosen to be edge-disjoint with $\{e^\ast_1, P_6\}$. Moreover, $e^\ast_1\PREC e''$. Thus, we must have $\Dover[1]\cap\Dover[2]\EqualEmpty$ on $\G[']$.

(b) $\Sover[1]\cap\Sover[3]\EqualEmpty$ on $\G[']$: Note that any $\Sover[1]$ edge can exist on (i) $P_1$; (ii) $P_4$; (iii) $e''$; and (iv) an edge that $P_8$ and $P_{10}$ shares. Note also that any $\Sover[3]$ can exist on (i) $P_3$; and (ii) $e^\ast_3$. But $e^\ast_3$ is in the upstream of $e''$, $P_8$, and $P_{10}$. Also, $e^\ast_3$ is in the downstream of $e'$, ending edge of $P_1$. In addition, $P_4$ was chosen not to use $e^\ast_3$. Moreover, we already showed that $\{P_1,P_3\}$ are vertex-disjoint paths; and $\{P_3, P_4\}$ are vertex-disjoint paths. Thus, we must have $\Sover[1]\cap\Sover[3]\EqualEmpty$ on $\G[']$.

(c) $\Sover[2]\cap\Sover[3]\EqualEmpty$ on $\G[']$: Note that any $\Sover[2]$ edge can exist on (i) $P_2$; (ii) $e^\ast_1$; (iii) an edge that $P_6$ and $P_9$ shares; and (iv) an edge that $P_6$ and $P_{10}$ share. Note also that any $\Sover[3]$ edge can exist on (i) $P_3$; and (ii) $e^\ast_3$. However, $e^\ast_3$ is in the upstream of $e^\ast_1$, $P_6$, $P_9$, and $P_{10}$. In addition, $P_2$ was chosen not to use $e^\ast_3$. Moreover, we already showed that $\{P_2,P_3\}$ are vertex-disjoint paths. Thus, we must have $\Sover[2]\cap\Sover[3]\EqualEmpty$ on $\G[']$.
\end{proof}

We prove \RS[46] as follows.
\begin{proof} Suppose that (\NotGS[27])\AND\GS[28]\AND\GS[31]\AND(\NotGS[37])\AND\\\noindent\GS[38]\AND\GS[39] is true. By \RS[41], \GS[42] is true as well. Since \GS[28]\AND\\\noindent\GS[31] is true, $e^\ast_3$ (resp. $e^\ast_1$) is the most downstream (resp. upstream) edge of $\Sover[3]$ (resp. $\Dover[1]$). From (\NotGS[37])\AND\GS[38]\AND\GS[39] being true, we also have $e^\ast_3\!\in\!\onecut[s_1][d_3]\backslash\onecut[s_1][d_1]$ and $e^\ast_1\!\in\!\onecut[s_1][d_1]$.

Since \GS[42] is true, we have $e^\ast_1\!=\!e^{21}_u$ is in $\Sover[2]$. Any arbitrary \FromTo[2][3] path $P_{23}$ thus must use $e^\ast_1$. Since $e^\ast_3\!\not\in\!\onecut[s_1][d_1]$ and $e^\ast_1\!\in\!\onecut[s_1][d_1]$, there exists an \FromTo[1][1] path $Q_{11}$ using $e^\ast_1$ but not using $e_3^\ast$. Then, we can create a \FromTo[1][3] path $Q_{11} e^\ast_1 P_{23}$ not using $e^\ast_3$, which contradicts $e^\ast_3\!\in\!\onecut[s_1][d_3]$. The proof of \RS[46] is complete.
\end{proof}

We prove \RS[47] as follows.
\begin{proof} Suppose that \LNR\AND\DS[4]\AND(\NotGS[27])\AND\GS[28]\AND\GS[31]\AND\\\noindent(\NotGS[37])\AND\GS[38]\AND\GS[40] is true. Since \GS[28]\AND\GS[31] is true, $e^\ast_3$ (resp. $e^\ast_1$) is the most downstream (resp. upstream) edge of $\Sover[3]$ (resp. $\Dover[1]$). Since (\NotGS[27])\AND\GS[28] implies \GS[42], $e^\ast_1$ also belongs to $\Sover[2]$, which implies that $e^\ast_1\!\in\!\onecut[s_2][d_3]$. Since \GS[40] is true, we have $e^\ast_1\!\in\!\onecut[s_1][d_3]$. Jointly the above arguments imply $e^\ast_1\!\in\!\Dover[1]\cap\Dover[3]$. Also, \GS[38] being true implies $e^\ast_3\!\in\!\Sover[3]\cap\onecut[s_1][d_3]$. Since \LNR\, is true and $\Dover[1]\cap\Dover[3]\NotEqualEmpty$, by \PropRef{Prop4} we must have $\Sover[1]\cap\Sover[3]\EqualEmpty$, which implies that $e^\ast_3$ cannot belong to $\onecut[s_1][d_2]$.

Let a node $u$ be the tail of the edge $e^\ast_3$. Since $e^\ast_3\!\in\!\onecut[s_1][d_3]$, $u$ is reachable from $s_1$. Since $e^\ast_3\!\in\!\Sover[3]$, $u$ is also reachable form $s_3$. Consider the collection of edges, $\onecut[s_1][u]\cap\onecut[s_3][u]$ (may be empty), all edges of which are in the upstream of $e^\ast_3$ if non-empty. Note that $\left(\onecut[s_1][u]\cap\onecut[s_3][u]\right)\cup\{e^\ast_3\}$ is always non-empty (since it contains at least $e^\ast_3$). Then, we use $e''$ to denote the most upstream edge of $\left(\onecut[s_1][u]\cap\onecut[s_3][u]\right)\cup\{e^\ast_3\}$. Let $e'$ denote the most downstream edge among all edges in $\onecut[s_1][{\tail[e'']}]$. Such choice is always possible since $\onecut[s_1][{\tail[e'']}]$ contains at least one edge (the $s_1$-source edge $e_{s_1}$) and thus we have $e_{s_1} \PRECEQ e' \PREC e'' \PRECEQ e^\ast_3$. Since we choose $e'$ to be the most downstream one, by \PropRef{Prop3} the channel gain $\ChG[e''][e']$ must be irreducible. Moreover, since $e^\ast_3\!\in\! \onecut[s_1][d_3]$, any path from $s_1$ to $d_3$ must use $e^\ast_3$. Consequently since $e''\!\in\! \onecut[s_1][u]\cup\{e^\ast_3\}$, any path from $s_1$ to $d_3$ must also use $e''$. Consequently since $e'\!\in\! \onecut[s_1][{\tail[e'']}]$, any path from $s_1$ to $d_3$ must also use $e'$. As a result, $\{e', e''\}\!\subset\!\onecut[s_1][d_3]$. Therefore $\ChG[e''][e']$ is a factor of $\ChGANA[3][1]$.

Now we argue that $\GCD[{\ChGANA[1][3]}][\,{\ChG[e''][e']}]\PolyEqual 1$. Suppose not. Since $\ChG[e''][e']$ is irreducible, by \PropRef{Prop3} we must have $e'\!\in\!\onecut[s_3][d_1]$. Note that $e'\!=\!e_{s_1}$ cannot be a $1$-edge cut separating $s_3$ and $d_1$ from the definitions (i) and (ii) of the 3-unicast ANA network. Thus, we only need to consider the case when $e_{s_1} \PREC e'$ since $e_{s_1}\PRECEQ e'$ from the construction of $e'$. Since $e^\ast_3\!\in\!\onecut[s_3][d_1]$ and $e'\PREC e^\ast_3$ is an $1$-edge cut separating $s_3$ and $d_1$, we must have $e'\!\in\!\onecut[s_3][u]$. Note that the most downstream $\onecut[s_1][{\tail[e'']}]$ edge $e'$ also belongs to $\onecut[s_1][u]$ from our construction. Therefore, jointly, this contradicts the construction that $e''$ is the most upstream edge of $\left(\onecut[s_1][u]\CAP \onecut[s_3][u]\right)\CUP\{e^\ast_3\}$ since $e'\PREC e''$.

Now we argue that $\GCD[{\ChGANA[3][2]}][\,{\ChG[e''][e']}]\PolyEqual 1$. Suppose not. Since $\ChG[e''][e']$ is irreducible, we must have $e'\!\in\!\onecut[s_2][d_3]$ and thus $e_{s_1}\PREC e'$. Choose arbitrarily a path from $s_1$ to $e'$. Since we have already established $e^\ast_3\PREC e^\ast_1$ and $e^\ast_1$ is the most upstream edge of $\Dover[1]$, there exists a path $P_{s_2 \tail[e^\ast_1]}$ from $s_2$ to $\tail[e^\ast_1]$ not using $e^\ast_3$. Since $e^\ast_1$ is also in $\Dover[3]$, $\head[e^\ast_1]$ can reach $d_3$. Note that the chosen path $P_{s_2 \tail[e^\ast_1]}$ must use $e'$ since $e'\!\in\!\onecut[s_2][d_3]$. As a result, $s_1$ can reach $d_3$ by going to $e'$ first, and then following $P_{s_2 \tail[e^\ast_1]}$ to $e^\ast_1$, and then going to $d_3$, without using $e^\ast_3$. This contradicts the assumption that $e^\ast_3\!\in\! \onecut[s_1][d_3]$.

Now we argue that $\GCD[{\ChGANA[2][1]}][\,{\ChG[e''][e']}]\PolyEqual 1$. Suppose not. Since $\ChG[e''][e']$ is irreducible, we must have $e''\!\in\!\onecut[s_1][d_2]$. Since we have established \NotGS[41] (i.e., $e^\ast_3\!\not\in\!\onecut[s_1][d_2]$), we only need to consider the case when $e''\PREC e^\ast_3$. Then by construction there exists a \FromTo[1][2] path $P_{12}$ going through $e''$ but not $e^\ast_3$. However, since by construction $e''$ is reachable from $s_3$, there exists a path from $s_3$ to $e''$ first and then use $P_{12}$ to arrive at $d_2$. Such a \FromTo[3][2] path does not use $e^\ast_3$, which contradicts the assumption that $e^\ast_3\!\in\!\Sover[3]\!\subset\!\onecut[s_3][d_2]$.

Now we argue that $\GCD[{\ChGANA[1][1]}][\,{\ChG[e''][e']}]\PolyEqual 1$. Suppose not. Since $\ChG[e''][e']$ is irreducible, we must have $e''\!\in\!\onecut[s_1][d_1]$. Since \NotGS[37] is true (i.e., $e^\ast_3\!\not\in\! \onecut[s_1][d_1]$), we only need to consider the case when $e''\PREC e^\ast_3$. Then by construction there exists a \FromTo[1][1] path $P_{11}$ going through $e''$ but not $e^\ast_3$. However, since by construction $e''$ is reachable from $s_3$, there exists a path from $s_3$ to $e''$ first and then use $P_{11}$ to arrive at $d_1$. Such a \FromTo[3][1] path does not use $e^\ast_3$, which contradicts the assumption that $e^\ast_3\!\in\!\Sover[3]\!\subset\!\onecut[s_3][d_1]$.

The four statements in the previous paragraphs shows that
\begin{align*}
\GCD[{\ChGANA[1][1]\ChGANA[2][1]\ChGANA[3][2]\ChGANA[1][3]}][\;{\ChG[e''][e']}]\,\PolyEqual 1.
\end{align*}

This, however, contradicts the assumption that \DS[4] is true since we have shown that $\ChG[e''][e']$ is a factor of $\ChGANA[3][1]$. The proof of \RS[47] is thus complete.
\end{proof}

\section{Proof of \RS[30]}\label{ProofRS[30]}

If we swap the roles of $s_2$ and $s_3$, and the roles of $d_2$ and $d_3$, then the proof of \RS[29] in \AppRef{ProofRS[29]} can be directly applied to show \RS[30]. More specifically, note that both \DS[3] and \DS[4] are converted back and forth from each other when swapping the flow indices. Similarly, the index swapping also converts \GS[27] to \GS[28] and vice versa. Since \LNR, \GS[1], and \ES[0] remain the same after swapping the flow indices, we can see that \RS[29] becomes \RS[30] after swapping the flow indices. The proofs of \RS[29] in \AppRef{ProofRS[29]} can thus be used to prove \RS[30].

\bibliographystyle{IEEEtranS}
\bibliography{JMH_IEEEfull,jmhan,infoT,nc_sys,book}

\begin{thebibliography}{10}
\providecommand{\url}[1]{#1}
\csname url@samestyle\endcsname
\providecommand{\newblock}{\relax}
\providecommand{\bibinfo}[2]{#2}
\providecommand{\BIBentrySTDinterwordspacing}{\spaceskip=0pt\relax}
\providecommand{\BIBentryALTinterwordstretchfactor}{4}
\providecommand{\BIBentryALTinterwordspacing}{\spaceskip=\fontdimen2\font plus
\BIBentryALTinterwordstretchfactor\fontdimen3\font minus
  \fontdimen4\font\relax}
\providecommand{\BIBforeignlanguage}[2]{{%
\expandafter\ifx\csname l@#1\endcsname\relax
\typeout{** WARNING: IEEEtranS.bst: No hyphenation pattern has been}%
\typeout{** loaded for the language `#1'. Using the pattern for}%
\typeout{** the default language instead.}%
\else
\language=\csname l@#1\endcsname
\fi
#2}}
\providecommand{\BIBdecl}{\relax}
\BIBdecl

\bibitem{AhlswedeCaiLiYeung:IT00}
R.~Ahlswede, N.~Cai, S.-Y. Li, and R.~Yeung, ``Network information flow,''
  \emph{{IEEE} Transactions on Information Theory}, vol.~46, no.~4, pp.
  1204--1216, July 2000.

\bibitem{CadambeJafar:IT08a}
V.~R. Cadambe and S.~A. Jafar, ``Interference alignment and degrees of freedom
  of the k-user interference channel,'' \emph{{IEEE} Transactions on
  Information Theory}, vol.~54, no.~8, pp. 3425--3441, August 2008.

\bibitem{ChouWuJain:Allerton03}
P.~Chou, Y.~Wu, and K.~Jain, ``Practical network coding,'' in \emph{Proc. 41st
  Annual Allerton Conf. on Comm., Contr., and Computing.}, Monticello,
  Illinois, USA, October 2003.

\bibitem{DasVishwanathJafarMarkopoulou:ISIT10}
A.~Das, S.~Vishwanath, S.~Jafar, and A.~Markopoulou, ``Network coding for
  multiple unicasts: An interference alignement approach,'' in \emph{Proc.
  {IEEE} Int'l Symp. Inform. Theory.}, Austin, Texas, USA, June 2010, pp.
  1878--1882.

\bibitem{DoughertyFreilingZeger:IT05}
R.~Dougherty, C.~Freiling, and K.~Zeger, ``Insufficiency of linear coding in
  network information flow,'' \emph{{IEEE} Transactions on Information Theory},
  vol.~51, no.~8, pp. 2745--2759, August 2005.

\bibitem{HarveyKleinbergLehman:IT06}
N.~Harvey, R.~Kleinberg, and A.~Lehman, ``On the capacity of information
  network,'' \emph{{IEEE} Transactions on Information Theory}, vol.~52, no.~6,
  pp. 2345--2364, June 2006.

\bibitem{HoMedardKoetterKargerEffrosShiLeong:IT06}
T.~Ho, M.~Medard, R.~Koetter, D.~R. Karger, M.~Effros, J.~Shi, , and B.~Leong,
  ``A random linear network coding approach to multicast,'' \emph{{IEEE}
  Transactions on Information Theory}, vol.~52, no.~10, pp. 4413--4430, October
  2006.

\bibitem{KoetterMedard:TON03}
R.~Koetter and M.~Medard, ``An algebraic approach to network coding,''
  \emph{{IEEE/ACM} Transactions on Networking}, vol.~11, no.~5, pp. 782--795,
  October 2003.

\bibitem{Lehman:PhD05}
A.~Lehman, ``Network coding,'' Ph.D. dissertation, MIT, 2005.

\bibitem{LiYeungCai:IT03}
S.-Y. Li, R.~Yeung, and N.~Cai, ``Linear network coding,'' \emph{{IEEE}
  Transactions on Information Theory}, vol.~49, no.~2, pp. 371--381, February
  2003.

\bibitem{RmakrishnanMeng:ISIT12}
C.~Meng, A.~Ramakrishnan, A.~Markopoulou, and S.~A. Jafar, ``On the feasibility
  of precoding-based network alignment for three unicast sessions,'' in
  \emph{Proc. {IEEE} Int'l Symp. Inform. Theory.}, Cambridge, MA, USA, July
  2012, pp. 1907--1911.

\bibitem{RmakrishnanMeng:UCI-TecRep}
\BIBentryALTinterwordspacing
------, ``On the feasibility of precoding-based network alignment for three
  unicast sessions,'' in \emph{arXiV, Tech. Rep.}, February 2012. [Online].
  Available:
  \url{http://odysseas.calit2.uci.edu/wiki/doku.php/public:publication}
\BIBentrySTDinterwordspacing

\bibitem{RamakrishnanDaszMalekiMarkopoulouJafarVishwanath:Allerton10}
A.~Ramakrishnan, A.~Dasz, H.~Maleki, A.~Markopoulou, S.~Jafar, and
  S.~Vishwanath, ``Network coding for three unicast sessions: Interference
  alignment approaches,'' in \emph{Proc. 48th Annual Allerton Conf. on Comm.,
  Contr., and Computing.}, Monticello, Illinois, USA, September 2010, pp.
  1054--1061.

\bibitem{RouayhebCostas:ISIT08}
S.~E. Rouayheb, A.~Sprinton, and C.~Georghiades, ``On the relation between the
  index coding and the network coding problems,'' in \emph{Proc. {IEEE} Int'l
  Symp. Inform. Theory.}, Toronto, Canada, July 2008, pp. 1823--1827.

\bibitem{SandersEgnerTolhuizen:ACM-SPAA03}
P.~Sanders, S.~Egner, and L.~Tolhuizen, ``Polynomial time algorithms for
  network information flow,'' in \emph{Proc. 15th Annual ACM Symp. Parallel
  Algorithms and Architectures ({SPAA}'03)}, San Diego, California, USA, June
  2003.

\bibitem{Wang:IT12a}
C.-C. Wang, ``On the capacity of 1-to-k broadcast packet erasure channels with
  channel output feedback,'' \emph{{IEEE} Transactions on Information Theory},
  vol.~58, no.~2, pp. 931--956, February 2012.

\bibitem{WangShroff:IT10}
C.-C. Wang and N.~B. Shroff, ``Pairwise intersession network coding on directed
  networks,'' \emph{{IEEE} Transactions on Information Theory}, vol.~56, no.~8,
  pp. 3879--3900, August 2010.

\bibitem{YanYangZhang:IT06}
X.~Yan, J.~Yang, and Z.~Zhang, ``An outer bound for multisource multisink
  network coding with minimum cost consideration,'' \emph{{IEEE} Transactions
  on Information Theory}, vol.~52, no.~6, pp. 2373--2385, June 2006.

\bibitem{YazdiKramer:IT11}
S.~M.~S.~T. Yazdi, S.~A. Savari, and G.~Kramer, ``Network coding in
  node-constrained line and star networks,'' \emph{{IEEE} Transactions on
  Information Theory}, vol.~57, no.~7, pp. 4452--4468, July 2011.

\end{thebibliography}
\end{document}